\pdfoutput=1

\PassOptionsToPackage{names,dvipsnames}{xcolor} %
\documentclass[sigconf]{acmart}
\usepackage{amsmath,amsfonts}
\usepackage{algorithmic}
\usepackage{graphicx}
\usepackage{textcomp}
\usepackage{bmpsize}
\usepackage{lipsum}
\usepackage{paralist}

\usepackage[english]{babel}
\usepackage[inference]{semantic}
\usepackage{amsmath}\allowdisplaybreaks
\usepackage{mathpartir}
\usepackage{amsthm} 
\usepackage{stmaryrd} %
\usepackage{xspace}
\usepackage{latexsym}
\usepackage{ifthen}
\usepackage{bussproofs}
\usepackage{mathtools} 
\usepackage{color}
\usepackage{listings}
\usepackage{verbatim}
\usepackage[colorinlistoftodos]{todonotes}
\usepackage{tikz}
\usetikzlibrary{positioning,shadows,arrows,calc,backgrounds,fit,shapes,snakes,shapes.multipart,decorations.pathreplacing,shapes.misc,patterns}
\usepackage{xspace}
\usepackage[T1]{fontenc}
\usepackage[scaled=.83]{beramono}
\usepackage{epigraph}
\usepackage{booktabs}
\usepackage{float}
\usepackage{etoolbox}
\usepackage{wrapfig}
\usepackage{pgfplots}
\usepackage{nameref}
\usepackage{scalerel}
\usepackage{bm}
\usepackage{centernot}
\usepackage{mdframed}
\usepackage{enumitem}

\usepackage{natbib} %
\setcitestyle{number}

\usepackage{subfigure}

\usepackage{cleveref}

\newcommand{\MP}[1]{\todo[color=blue!30]{TODO: #1}}

\newcommand{\MG}[1]{\todo[color=green!30]{TODO: #1}}
\newcommand{\MGin}[1]{\todo[color=green!30,inline]{TODO: #1}}

\newcommand{\mi}[1]{\ensuremath{\mathit{#1}}}

\newcommand{\mtt}[1]{\ensuremath{\mathtt{#1}}}
\newcommand{\mf}[1]{\ensuremath{\mathbf{#1}}}

\newcommand{\ms}[1]{\ensuremath{\mathsf{#1}}}
\newcommand{\mb}[1]{\ensuremath{\mathbb{#1}}}

\newcommand{\isdef}[0]{\ensuremath{\mathrel{\overset{\makebox[0pt]{\mbox{\normalfont\tiny\sffamily def}}}{=}}}}

\newcommand{\relmiddle}[1]{\mathrel{}\middle#1\mathrel{}}

\newcommand\bnfdef{\ensuremath{\mathrel{::=}}}

\newcommand{\OB}[1]{\ensuremath{\overline{#1}}}

\newcommand{\myset}[2]{\ensuremath{\left\{#1 ~\relmiddle|~ #2\right\}}}

\newcommand*{\QEDA}{\hfill\ensuremath{\blacksquare}}%

\AtEndEnvironment{problem}{\null\hfill\QEDA}
\AtEndEnvironment{example}{}%

\Crefname{lstlisting}{Listing}{Listings}
\Crefname{problem}{Problem}{Problems}

\Crefname{equation}{Rule}{Rules}

\newcommand{\compskel}[3]{\ensuremath{\bl{\left\llbracket \src{#1} \right\rrbracket^{#2}_{#3}}}}
\newcommand{\comp}[1]{\compskel{\bl{#1}}{}{}}

\newcommand{\compslh}[1]{\compskel{#1}{\com{s}}{}}
\newcommand{\compsslh}[1]{\compskel{#1}{\com{ss}}{}}
\newcommand{\compslht}[1]{\compskel{#1}{\com{s}}{\com{n}}}
\newcommand{\compslhd}[1]{\compskel{#1}{\com{s}}{\com{3}}}
\newcommand{\complfence}[1]{\compskel{#1}{\com{f}}{}}

\newcommand{\comprp}[1]{\compskel{#1}{\com{r}}{}}

\newcommand{\funname}[1]{\mtt{#1}}
\newcommand{\fun}[2]{\ensuremath{{\bl{\funname{#1}\left(#2\right)}}}\xspace}
\newcommand{\dom}[1]{\fun{dom}{#1}}

\newcommand{\backtrskel}[3]{\ensuremath{\bl{\left\langle\!\left\langle {#1} \right\rangle\!\right\rangle^{#2}_{#3}}}}
\newcommand{\backtr}[1]{\backtrskel{#1}{}{}}

\newcommand{\backtrfencec}[1]{\backtrskel{#1}{\com{f}}{\com{c}}}

\newcommand{\SR}[0]{\src{L}\xspace}
\newcommand{\TR}[0]{\trg{T}\xspace}
\newcommand{\wSR}[0]{\src{\weak{L}}\xspace}
\newcommand{\wTR}[0]{\trg{\weak{T}}\xspace}

\newcommand{\strong}[1]{#1\ensuremath{^{\text{+}}}\xspace}
\newcommand{\weak}[1]{#1\ensuremath{^{\text{-}}}}%

\let\oldS\S
\renewcommand{\S}[0]{\src{{S}}\xspace}
\newcommand{\T}[0]{\trg{{T}}\xspace}

\newcommand{\contextletter}[0]{A}
\newcommand{\ctx}[1]{\ensuremath{\contextletter}} %
\newcommand{\ctxs}[1]{\src{\ctx{\contextletter}#1}\xspace}
\newcommand{\ctxt}[1]{\trg{\ctx{\contextletter}#1}\xspace}%
\newcommand{\ctxc}[1]{\com{\ctx{\contextletter}#1}\xspace}%
\newcommand{\hole}[1]{\ensuremath{\left[#1\right]}}

\newcommand{\trues}[0]{\src{{true}}\xspace}

\newcommand{\truet}[0]{\trg{{true}}\xspace}
\newcommand{\falset}[0]{\trg{false}\xspace}

\newcommand{\vdasht}[0]{\trgb{\vdash}}

\newcommand{\lskip}{\ensuremath{{skip}}}
\newcommand{\skips}{\ensuremath{\src{skip}}}
\newcommand{\skipc}{\ensuremath{\com{skip}}}
\newcommand{\skipt}{\ensuremath{\trg{skip}}}

\newcommand{\srce}[0]{\src{\emptyset}\xspace}
\newcommand{\trge}[0]{\trg{\emptyset}\xspace}

\newcommand{\come}[0]{\com{\emptyset}\xspace}

\newcommand{\SInit}[1]{\ensuremath{{\Omega_0}\left({#1}\right)}\xspace}
\newcommand{\SInits}[1]{\ensuremath{\src{\Omega_0}\left(\src{#1}\right)}\xspace}
\newcommand{\SInitt}[1]{\ensuremath{\trg{\Omega_0}\left(\trg{#1}\right)}\xspace}

\newcommand{\neutcol}[0]{black}
\newcommand{\stlccol}[0]{RoyalBlue}
\newcommand{\ulccol}[0]{RedOrange}

\newcommand{\commoncol}[0]{black}    %

\newcommand{\col}[2]{\ensuremath{{\color{#1}{#2}}}}

\newcommand{\src}[1]{\ms{\col{\stlccol}{#1}}}
\newcommand{\trg}[1]{{\mf{\col{\ulccol }{#1}}}}
\newcommand{\trgb}[1]{\ensuremath{\bm{\col{\ulccol }{#1}}}}

\newcommand{\bl}[1]{\col{\neutcol }{#1}}
\newcommand{\com}[1]{\mi{\col{\commoncol }{#1}}}

\newcounter{typerule}
\crefname{typerule}{rule}{rules}

\newcommand{\typeruleInt}[5]{%
	\def\thetyperule{#1}%
	\refstepcounter{typerule}%
	\label{tr:#4}%
  \ensuremath{\begin{array}{c}#5 \inference{#2}{#3}\end{array}} 
}
\newcommand{\typerule}[4]{%
  \typeruleInt{#1}{#2}{#3}{#4}{\textsf{\scriptsize ({#1})} \\      }
}

\pgfdeclarelayer{background}
\pgfdeclarelayer{veryback}
\pgfdeclarelayer{veryback2}
\pgfdeclarelayer{veryback3}
\pgfdeclarelayer{back2}
\pgfdeclarelayer{foreground}
\pgfsetlayers{veryback3,veryback2,veryback,background,back2,main,foreground}

\newcommand{\tikzpic}[1]{
\begin{tikzpicture}[shorten >=1pt,auto,node distance=6mm]
\tikzstyle{state} =[fill=white,minimum size=4pt]
\tikzstyle{field} =[fill=gray!5,draw=black!70, rectangle, minimum width={width("whiskersfieldww")+2pt}]]
#1
\end{tikzpicture}
}

\newcommand{\myfig}[3]{\begin{figure} [!ht]
#1
\caption{\label{fig:#2}#3}
\end{figure}}
\newcommand{\myfigonecol}[3]{\begin{figure*} [!hbt]
#1
\caption{\label{fig:#2}#3}
\end{figure*}}

\newcommand{\myfigp}[4]{\begin{figure} #4
#1
\caption{\label{fig:#2}#3}
\end{figure}}

\newcommand{\etal}[0]{\textit{et al.}\xspace} 

\newcommand{\BREAK}[0]{
\botrule
\begin{center}$\spadesuit$\end{center}
\botrule}

\newcommand{\mytoprule}[1]{\vspace{1mm}\noindent\hrulefill\ \raisebox{-0.5ex}{\fbox{\ensuremath{#1}}} \hrulefill\hrulefill\hrulefill\vspace{0.5mm}}

\def\botrule{\vspace{0mm}\hrule\vspace{2mm}}

\newcommand{\myparagraph}[1]{ \smallskip \noindent\noindent\textit{#1}~}

\newcommand{\hl}[1]{\colorbox{yellow}{#1}}
\newcounter{line}

\newcommand{\asm}[1]{\mtt{#1}}

\newcommand{\xto}[1]{\ensuremath{~\mathrel{\xrightarrow{~#1~}}~}}
\newcommand{\Xto}[1]{\ensuremath{~\mathrel{\xRightarrow{~#1~}}~}}

\newcommand{\xtos}[1]{\src{\xto{#1}}}
\newcommand{\Xtos}[1]{\src{\Xto{#1}}}

\newcommand{\xtot}[1]{\trg{\xto{#1}}}
\newcommand{\Xtot}[1]{\trg{\Xto{#1}}}

\definecolor{mygreen}{rgb}{0,0.6,0}
\definecolor{mygray}{rgb}{0.5,0.5,0.5}
\definecolor{mymauve}{rgb}{0.58,0,0.82}

\lstdefinelanguage{Java} %
{morekeywords={abstract, all, and, as, assert, but, disj, else, exactly, extends, fact, for, fun, iden, if, iff, implies, in, Int, void, int, let, lone, module, no, none, not, one, open, or, part, pred, run, seq, set, sig, some, sum, then, univ, package, class, public, private, null, return, new, interface, extern, object, implements, System, static, super, try , catch, throw, throws, Unit, var, val, of, principal, trust},
sensitive=true,
keywordstyle=\bfseries\color{\stlccol}, %
commentstyle=\itshape\color{purple!40!black},
morecomment=[l][\small\itshape\color{purple!40!black}]{//},
identifierstyle=\color{\stlccol},
stringstyle=\color{orange},
basicstyle=\small,
basicstyle={\small\ttfamily},
numbers=left,
numberstyle=\tiny\color{mygray},
tabsize=2,
numbersep=5pt,
breaklines=true,
lineskip=-2pt,
stepnumber=1,
captionpos=b,
breaklines=true,
breakatwhitespace=false,
showspaces=false,
showtabs=false,
float=!h,
columns=fullflexible,escapeinside={(*@}{@*)},
moredelim=**[is][\color{red!60}]{@}{@},
literate={->}{{$\to$}}1 {^}{{$\mspace{-3mu}\widehat{\quad}\mspace{-3mu}$}}1
{<}{$<$ }2 {>}{$>$ }2 {>=}{$\geq$ }2 {=<}{$\leq$ }2
{<:}{{$<\mspace{-3mu}:$}}2 {:>}{{$:\mspace{-3mu}>$}}2
{=>}{{$\Rightarrow$ }}2 {+}{$+$ }2 {++}{{$+\mspace{-8mu}+$ }}2
{<=>}{{$\Leftrightarrow$ }}2 {+}{$+$ }2 {++}{{$+\mspace{-8mu}+$ }}2
{\~}{{$\mspace{-3mu}\widetilde{\quad}\mspace{-3mu}$}}1
{!=}{$\neq$ }2 {*}{${}^{\ast}$}1 %
{\#}{$\#$}1
}

\lstdefinelanguage{General} %
{morekeywords={abstract, all, and, as, assert, but, disj, else, exactly, extends, fact, for, fun, iden, if, iff, implies, in, Int, void, int, let, lone, module, no, none, not, one, open, or, part, pred, run, seq, set, sig, some, sum, then, univ, package, class, public, private, null, return, new, interface, extern, object, implements, System, static, super, try , catch, throw, throws, Unit, var, val, of, principal, trust},
sensitive=true,
keywordstyle=\bfseries\color{\neutcol},
commentstyle=\itshape\color{purple!40!black},
morecomment=[l][\small\itshape\color{purple!40!black}]{//},
identifierstyle=\color{\neutcol},
stringstyle=\color{orange},
basicstyle=\small,
basicstyle={\small\ttfamily},
numbers=left,
numberstyle=\tiny\color{mygray},
tabsize=2,
numbersep=5pt,
breaklines=true,
lineskip=-2pt,
stepnumber=1,
captionpos=b,
breaklines=true,
breakatwhitespace=false,
showspaces=false,
showtabs=false,
float=!h,
columns=fullflexible,escapeinside={(*@}{@*)},
moredelim=**[is][\color{red!60}]{@}{@},
literate={->}{{$\to$}}1 {^}{{$\mspace{-3mu}\widehat{\quad}\mspace{-3mu}$}}1
{<}{$<$ }2 {>}{$>$ }2 {>=}{$\geq$ }2 {=<}{$\leq$ }2
{<:}{{$<\mspace{-3mu}:$}}2 {:>}{{$:\mspace{-3mu}>$}}2
{=>}{{$\Rightarrow$ }}2 {+}{$+$ }2 {++}{{$+\mspace{-8mu}+$ }}2
{<=>}{{$\Leftrightarrow$ }}2 {+}{$+$ }2 {++}{{$+\mspace{-8mu}+$ }}2
{\~}{{$\mspace{-3mu}\widetilde{\quad}\mspace{-3mu}$}}1
{!=}{$\neq$ }2 {*}{${}^{\ast}$}1 %
{\#}{$\#$}1
}

\lstdefinelanguage{Asm}
{morekeywords={abstract, all, and, as, assert, but, check, disj, else, exactly, extends, fact, for, fun, iden, if, iff, implies, in, Int, void, int, let, lone, module, no, none, not, one, open, part, pred, run, seq, set, sig, some, sum, then, univ, package, class, public, private, null, return, new, interface, extern, object, implements, System, static, super, try , catch, throw, throws, Unit, var, val, principal, trust, label, load, add, addi, into, test, mov, cmov, cmova, movzx, cmp, jbe, sar, cmovbe, or, jmp, shl, ret, jae, lea, lfence, jne},
sensitive=true,
identifierstyle=\color{\ulccol},
keywordstyle=\bfseries\color{\ulccol},
commentstyle=\itshape\color{purple!40!black},
morecomment=[l][\small\itshape\color{purple!40!black}]{//},
stringstyle=\color{orange},
basicstyle=\small,
basicstyle={\small},
numbers=left,
numberstyle=\tiny\color{mygray},
tabsize=2,
numbersep=5pt,
breaklines=true,
lineskip=-2pt,
stepnumber=1,
captionpos=b,
breaklines=true,
breakatwhitespace=false,
showspaces=false,
showtabs=false,
float=!h,
columns=fullflexible,escapeinside={(*@}{@*)},
moredelim=**[is][\color{red!60}]{@}{@},
literate={->}{{$\to$}}1 {^}{{$\mspace{-3mu}\widehat{\quad}\mspace{-3mu}$}}1
{<}{$<$ }2 {>}{$>$ }2 {>=}{$\geq$ }2 {=<}{$\leq$ }2
{<:}{{$<\mspace{-3mu}:$}}2 {:>}{{$:\mspace{-3mu}>$}}2
{=>}{{$\Rightarrow$ }}2 {+}{$+$ }2 {++}{{$+\mspace{-8mu}+$ }}2
{<=>}{{$\Leftrightarrow$ }}2 {+}{$+$ }2 {++}{{$+\mspace{-8mu}+$ }}2
{\~}{{$\mspace{-3mu}\widetilde{\quad}\mspace{-3mu}$}}1
{!=}{$\neq$ }2 {*}{${}^{\ast}$}1 %
{\#}{$\#$}1
}
\DeclareMathOperator\ceq{\ensuremath{\mathrel{\simeq_{\mi{ctx}}}}}

\DeclareMathOperator\loweq{\ensuremath{\mathrel{=_{\text{L}}}}}

\def\teqaux#1{\vcenter{\hbox{\ooalign{\hfil
       \raise6pt \hbox{\scriptsize{T}}\hfil\cr\hfil
       $=$}}}}

\def\relssa{\approx} %
\def\nrelssa{\not\approx} %

\DeclareMathOperator\rels{\ensuremath{\com{\relssa}}}
\DeclareMathOperator\nrels{\ensuremath{\com{\nrelssa}}}

\newcommand{\reltext}[0]{\relssa}
\newcommand{\relref}[0]{\ensuremath{\Cref{cr:tracerel}}\xspace}
\newcommand{\tracerel}[0]{\relref}
\newcommand{\reldef}[0]{\criteria{\reltext}{tracerel}}

\def\relssaux#1{\vcenter{\hbox{\ooalign{\hfil
       \raise6pt \hbox{\tiny{$\boldsymbol{+}$}}\hfil\cr\hfil
       $\approx$}}}}

\def\hrelssaux#1{\vcenter{\hbox{\ooalign{\hfil
       \raise6pt \hbox{\tiny{\com{H}}}\hfil\cr\hfil
       $\approx$}}}}
\def\hrelssb{\mathrel{\mathpalette\hrelssaux{}}}

\DeclareMathOperator\hrel{\ensuremath{\com{\hrelssb}}}

\newcommand{\hreltext}[0]{\hrel}
\newcommand{\hrelref}[0]{\ensuremath{\Cref{cr:htracerel}}\xspace}

\newcommand{\hreldef}[0]{\criteria{\hreltext}{htracerel}}

\def\vrelssaux#1{\vcenter{\hbox{\ooalign{\hfil
       \raise6pt \hbox{\tiny{\com{V}}}\hfil\cr\hfil
       $\approx$}}}}
\def\vrelssb{\mathrel{\mathpalette\vrelssaux{}}}

\DeclareMathOperator\vrel{\ensuremath{\com{\vrelssb}}}

\newcommand{\vreltext}[0]{\vrel}
\newcommand{\vrelref}[0]{\ensuremath{\Cref{cr:vtracerel}}\xspace}

\newcommand{\vreldef}[0]{\criteria{\vreltext}{vtracerel}}

\def\srelssaux#1{\vcenter{\hbox{\ooalign{\hfil
       \raise6pt \hbox{\tiny{\com{S}}}\hfil\cr\hfil
       $\approx$}}}}
\def\srelssb{\mathrel{\mathpalette\srelssaux{}}}

\DeclareMathOperator\srel{\ensuremath{\com{\srelssb}}}

\newcommand{\sreltext}[0]{\srel}
\newcommand{\srelref}[0]{\ensuremath{\Cref{cr:stracerel}}\xspace}

\newcommand{\sreldef}[0]{\criteria{\sreltext}{stracerel}}

\def\brelssaux#1{\vcenter{\hbox{\ooalign{\hfil
       \raise6pt \hbox{\tiny{\com{B}}}\hfil\cr\hfil
       $\approx$}}}}
\def\brelssb{\mathrel{\mathpalette\brelssaux{}}}

\DeclareMathOperator\brel{\ensuremath{\com{\brelssb}}}

\newcommand{\breltext}[0]{\brel}
\newcommand{\brelref}[0]{\ensuremath{\Cref{cr:btracerel}}\xspace}

\newcommand{\breldef}[0]{\criteria{\breltext}{btracerel}}

\def\crelssaux#1{\vcenter{\hbox{\ooalign{\hfil
       \raise6pt \hbox{\tiny{\com{C}}}\hfil\cr\hfil
       $\approx$}}}}
\def\crelssb{\mathrel{\mathpalette\crelssaux{}}}

\DeclareMathOperator\crel{\ensuremath{\com{\crelssb}}}

\newcommand{\creltext}[0]{\crel}
\newcommand{\crelref}[0]{\ensuremath{\Cref{cr:ctracerel}}\xspace}

\newcommand{\creldef}[0]{\criteria{\creltext}{ctracerel}}

\def\arelssaux#1{\vcenter{\hbox{\ooalign{\hfil
       \raise6pt \hbox{\tiny{\com{A}}}\hfil\cr\hfil
       $\approx$}}}}
\def\arelssb{\mathrel{\mathpalette\arelssaux{}}}

\DeclareMathOperator\arel{\ensuremath{\com{\arelssb}}}

\newcommand{\areltext}[0]{\arel}
\newcommand{\arelref}[0]{\ensuremath{\Cref{cr:atracerel}}\xspace}

\newcommand{\areldef}[0]{\criteria{\areltext}{atracerel}}

\newcommand{\relslh}[0]{\Bumpeq} %

\def\shrelssaux#1{\vcenter{\hbox{\ooalign{\hfil
       \raise6pt \hbox{\tiny{\com{H}}}\hfil\cr\hfil
       $\relslh$}}}}
\def\shrelssb{\mathrel{\mathpalette\shrelssaux{}}}

\DeclareMathOperator\shrel{\ensuremath{\com{\shrelssb}}}

\newcommand{\shreltext}[0]{\shrel}

\newcommand{\shreldef}[0]{\criteria{\shreltext}{shtracerel}}

\def\svrelssaux#1{\vcenter{\hbox{\ooalign{\hfil
       \raise6pt \hbox{\tiny{\com{V}}}\hfil\cr\hfil
       $\relslh$}}}}

\def\ssrelssaux#1{\vcenter{\hbox{\ooalign{\hfil
       \raise6pt \hbox{\tiny{\com{S}}}\hfil\cr\hfil
       $\relslh$}}}}
\def\ssrelssb{\mathrel{\mathpalette\ssrelssaux{}}}

\DeclareMathOperator\ssrel{\ensuremath{\com{\ssrelssb}}}

\newcommand{\ssreltext}[0]{\ssrel}
\newcommand{\ssrelref}[0]{\ensuremath{\Cref{cr:sstracerel}}\xspace}

\newcommand{\ssreldef}[0]{\criteria{\ssreltext}{sstracerel}}

\def\sbrelssaux#1{\vcenter{\hbox{\ooalign{\hfil
       \raise6pt \hbox{\tiny{\com{B}}}\hfil\cr\hfil
       $\relslh$}}}}
\def\sbrelssb{\mathrel{\mathpalette\sbrelssaux{}}}

\DeclareMathOperator\sbrel{\ensuremath{\com{\sbrelssb}}}

\newcommand{\sbreltext}[0]{\sbrel}

\newcommand{\sbreldef}[0]{\criteria{\sbreltext}{sbtracerel}}

\def\screlssaux#1{\vcenter{\hbox{\ooalign{\hfil
       \raise6pt \hbox{\tiny{\com{C}}}\hfil\cr\hfil
       $\relslh$}}}}

\def\sarelssaux#1{\vcenter{\hbox{\ooalign{\hfil
       \raise6pt \hbox{\tiny{\com{A}}}\hfil\cr\hfil
       $\relslh$}}}}

\newcommand{\relslhp}[0]{\propto} %

\def\cssrelssaux#1{\vcenter{\hbox{\ooalign{\hfil
       \raise6pt \hbox{\tiny{\com{S}}}\hfil\cr\hfil
       $\relslhp$}}}}
\def\cssrelssb{\mathrel{\mathpalette\cssrelssaux{}}}

\DeclareMathOperator\cssrel{\ensuremath{\com{\cssrelssb}}}

\newcommand{\cssreltext}[0]{\cssrel}
\newcommand{\cssrelref}[0]{\ensuremath{\Cref{cr:csstracerel}}\xspace}

\newcommand{\cssreldef}[0]{\criteria{\cssreltext}{csstracerel}}

\def\csbrelssaux#1{\vcenter{\hbox{\ooalign{\hfil
       \raise6pt \hbox{\tiny{\com{B}}}\hfil\cr\hfil
       $\relslhp$}}}}
\def\csbrelssb{\mathrel{\mathpalette\csbrelssaux{}}}

\DeclareMathOperator\csbrel{\ensuremath{\com{\csbrelssb}}}

\newcommand{\csbreltext}[0]{\csbrel}

\newcommand{\csbreldef}[0]{\criteria{\csbreltext}{csbtracerel}}

\def\cshrelssaux#1{\vcenter{\hbox{\ooalign{\hfil
       \raise6pt \hbox{\tiny{\com{H}}}\hfil\cr\hfil
       $\relslhp$}}}}
\def\cshrelssb{\mathrel{\mathpalette\cshrelssaux{}}}

\DeclareMathOperator\cshrel{\ensuremath{\com{\cshrelssb}}}

\def\ceqwaux#1{\vcenter{\hbox{\ooalign{\hfil
       \raise6pt \hbox{\scriptsize{w-b}}\hfil\cr\hfil
       $\ceq$}}}}

\def\praux#1{\vcenter{\hbox{\ooalign{\hfil
       \raise4pt \hbox{$\subset$}\hfil\cr\hfil
       $\sim$}}}}

\newcommand{\labelfont}[1]{\ensuremath{\asm{#1}}}

\newcommand{\clh}[3]{\ensuremath{\labelfont{call}~ #1~ #2{?}}} 	%
\newcommand{\cbh}[3]{\ensuremath{\labelfont{call}~ #1~ #2{!}}} 	%
\newcommand{\rth}[2]{\ensuremath{\labelfont{ret}{!}}}	 		%
\newcommand{\rbh}[2]{\ensuremath{\labelfont{ret}{?}}}			%

\newcommand{\rdl}[1]{\ensuremath{\labelfont{read}(#1)}}
\newcommand{\wrl}[1]{\ensuremath{\labelfont{write}(#1)}}

\newcommand{\ifl}[1]{\ensuremath{\labelfont{if}(#1)}}

\newcommand{\rollbl}[0]{\ensuremath{\labelfont{rlb}}}

\newcommand{\clgen}[2]{\ensuremath{\labelfont{call}~#1~#2}}
\newcommand{\rtgen}[1]{\ensuremath{\labelfont{ret}~#1}}

\newcommand{\behav}[1]{{Beh}{\left(#1\right)}}
\newcommand{\behavs}[1]{\src{\behav{#1}}}
\newcommand{\behavt}[1]{\trg{\behav{#1}}}
\newcommand{\behavc}[1]{\com{\behav{#1}}}

\newcommand{\length}[1]{\ensuremath{|#1|}}
\newcommand{\abs}[1]{\length{#1}}

\newcommand{\proc}[2]{\ensuremath{(#1)_{#2} }}

\newcommand{\op}[0]{\ensuremath{\oplus}}
\newcommand{\bop}[0]{\ensuremath{\otimes}}

\newcommand{\lfence}[0]{{lfence}}
\newcommand{\ret}[0]{{return;}}%
\newcommand{\letin}[3]{{let}~#1=#2~{in}~#3}

\newcommand{\letread}[3]{{let}~#1=rd~{#2}~{in}~#3}
\newcommand{\letreads}[3]{\src{let}~#1=\src{rd}~{#2}~\src{in}~#3}
\newcommand{\letreadt}[3]{\trg{let}~#1=\trg{rd}~{#2}~\trg{in}~#3}
\newcommand{\letreadp}[3]{{let}~#1=rd_{\prv}~{#2}~{in}~#3}
\newcommand{\letreadps}[3]{\src{let}~#1=\src{rd_{\prv}}~#2~\src{in}~#3}
\newcommand{\letreadpt}[3]{\trg{let}~#1=\trg{rd_{\prv}}~{#2}~\trg{in}~#3}

\newcommand{\prv}{pr}

\newcommand{\asgn}[2]{#1 := #2}

\newcommand{\asgnp}[2]{#1 :=_{\prv} #2}

\newcommand{\asgnpt}[2]{#1 :=_\trg{\prv} #2}

\newcommand{\letins}[3]{\src{let}~#1\src{=}#2~\src{in}~#3}
\newcommand{\letint}[3]{\trg{let}~#1\trg{=}#2~\trg{in}~#3}

\newcommand{\cmove}[4]{{let}~#1=#2~({if}~#3)~{in}~#4}
\newcommand{\cmovet}[4]{\trg{let}~#1=#2~(\trg{if}~#3)~\trg{in}~#4}
\newcommand{\call}[1]{{call}~#1}

\newcommand{\ifte}[3]{{if}~#1~{then}~#2~{else}~#3}

\newcommand{\ifzte}[3]{{ifz}~#1~{then}~#2~{else}~#3}
\newcommand{\ifztes}[3]{\src{ifz}~#1~\src{then}~#2~\src{else}~#3}
\newcommand{\ifztet}[3]{\trg{ifz}~#1~\trg{then}~#2~\trg{else}~#3}

\newcommand{\formatCompilers}[1]{\mf{\mi{#1}}\xspace}

\newcommand{\bigred}[0]{\ensuremath{\ \downarrow\ }}

\DeclareMathOperator\bigreds{\src{\bigred}}
\DeclareMathOperator\bigredt{\trg{\bigred}}

\newcommand{\xltot}[1]{\,\trg{\xleadsto{#1}}\,}

\AtEndEnvironment{example}{\null\hfill$\boxdot$}

\makeatletter
\xdef\@thefnmark{\@empty}

\newcommand{\Thmref}[1]{\Cref{#1}~(\nameref{#1})}
\makeatother

\renewcommand{\emptyset}[0]{\varnothing}

\newcounter{hps}
\crefname{hps}{}{}

\newcommand{\proven}[1]{\ensuremath{\checkmark}}

\Crefname{mydefinition}{Definition}{Definitions}
\crefname{mydefinition}{Definition}{Definitions}
\theoremstyle{definition}

\newcommand{\nspecProject}[1]{#1\!\!\upharpoonright_{\mi{nse}}}

\newcommand{\nspecproj}[1]{\nspecProject{#1}}

\newcommand{\sni}[0]{\Cref{cr:sni}\xspace} %
\newcommand{\snitext}[0]{\text{SNI}\xspace}
\newcommand{\snidef}[0]{\criteria{\snitext}{sni}}

\newcommand{\rsni}[0]{\Cref{cr:rsni}\xspace}
\newcommand{\rsnitext}[0]{\text{RSNI}\xspace}
\newcommand{\rsnidef}[0]{\criteria{\rsnitext}{rsni}}

\newcommand{\xleadsto}[1]{%
   \if\relax\detokenize{#1}\relax
   \rightsquigarrow
   \else
   \mathrel{%
     \begin{tikzpicture}[%
       baseline={(current bounding box.south)}
       ]
       \node[%
       ,inner sep=.44ex
       ,align=center
       ] (tmp) {$\scriptstyle #1$};
       \path[%
       ,draw,<-
       ,decorate,decoration={%
         ,zigzag
         ,amplitude=0.7pt
         ,segment length=1.2mm,pre length=3.5pt
       }
       ] 
       (tmp.south east) -- (tmp.south west);
     \end{tikzpicture}
   }
   \fi
 }

\newcommand{\tracesymbol}[0]{\lambda}
\newcommand{\actionsymbol}[0]{\alpha}

\newcommand{\terminationsymbol}[0]{\lightning}
\newcommand{\ts}[0]{\terminationsymbol}

\newcommand{\taintt}{\trgb{\sigma}}

\newcommand{\tra}[1]{\OB{\tracesymbol{#1}}}

\newcommand{\tras}[1]{\src{\tra{#1}}}

\newcommand{\trat}[1]{\trg{\OB{\trgb{\tracesymbol}{#1}}}}

\newcommand{\trac}[1]{\com{\tra{#1}}}

\newcommand{\acas}[1]{\src{\actionsymbol{#1}}}

\newcommand{\acat}[1]{\trg{\trgb{\actionsymbol}{#1}}}

\newcommand{\acac}[1]{\com{\actionsymbol{#1}}}

\newcommand{\safeta}[0]{\ensuremath{S}}
\newcommand{\unta}[0]{\ensuremath{U}}

\DeclareMathOperator\glb{\ensuremath{\sqcap}}

\DeclareMathOperator\sem{\ensuremath{\rightsquigarrow}}
\DeclareMathOperator\sems{\src{\rightsquigarrow}}
\DeclareMathOperator\semt{\trg{\boldsymbol{\rightsquigarrow}}}

\newcounter{criteria}
\crefname{criteria}{}{}
\newcommand{\criteria}[2]{%
	\def\thecriteria{\detokenize{#1}}%
  	\refstepcounter{criteria}%
  	\label{cr:#2}%
  	#1%
}

\newcommand{\snipcomp}[0]{\formatCompilers{RSNIP}}
\newcommand{\snip}[0]{\Cref{cr:snip}\xspace}
\newcommand{\rsnip}[0]{\snip}
\newcommand{\snipdef}[0]{\criteria{\snipcomp}{snip}}

\newcommand{\rdsscomp}[0]{\formatCompilers{RSSC}}
\newcommand{\rdss}[0]{\Cref{cr:rdss}\xspace}
\newcommand{\rssc}[0]{\rdss}
\newcommand{\rdssdef}[0]{\criteria{\rdsscomp}{rdss}}
\newcommand{\rsscdef}[0]{\rdssdef}

\newcommand{\rdsspcomp}[0]{\formatCompilers{RSSP}}
\newcommand{\rdssp}[0]{\Cref{cr:rdssp}\xspace}
\newcommand{\rssp}[0]{\rdssp}
\newcommand{\rdsspdef}[0]{\criteria{\rdsspcomp}{rdssp}}

\newcommand{\facomp}[0]{\formatCompilers{FAC}}
\newcommand{\facref}[0]{\Cref{cr:fac}\xspace}
\newcommand{\fac}[0]{\facref}
\newcommand{\facdef}[0]{\criteria{\facomp}{fac}}

\newcommand{\ctpcomp}[0]{\formatCompilers{CTPC}}
\newcommand{\ctpref}[0]{\Cref{cr:ctp}\xspace}
\newcommand{\ctp}[0]{\ctpref}
\newcommand{\ctpdef}[0]{\criteria{\ctpcomp}{ctp}}

\newcommand{\predState}[0]{{pr}}

\newcommand{\rsstext}[0]{\text{RSS}\xspace}
\newcommand{\rssref}[0]{\Cref{cr:rss}\xspace}
\newcommand{\rss}[0]{\rssref}
\newcommand{\rssdef}[0]{\criteria{\rsstext}{rss}}

\newcommand{\sstext}[0]{\text{SS}}
\newcommand{\ssref}[0]{\Cref{cr:ss}\xspace}
\renewcommand{\ss}[0]{\ssref}
\newcommand{\ssdef}[0]{\criteria{\sstext}{ss}}

\newcounter{proofr}
\crefname{proofr}{}{}
\newcommand{\proofref}[2]{%
	\def\theproofref{\detokenize{#1}}%
  	\refstepcounter{proofr}%
  	\label{prf:#2}%
}

\newcommand{\showproof}[1]{Proof \Cref{prf:#1}} %

\newcommand{\Bv}[0]{\ensuremath{B_v}\xspace}
\newcommand{\Bs}[0]{\ensuremath{B_{t}}\xspace}%
\newcommand{\Bt}[0]{\ensuremath{B_{t}}\xspace}%
\newcommand{\Hv}[0]{\ensuremath{H_v}\xspace}
\newcommand{\Hs}[0]{\ensuremath{H_{t}}\xspace}%
\newcommand{\Ht}[0]{\ensuremath{H_{t}}\xspace}%

\newcommand{\Ov}[0]{\ensuremath{\Omega_v}\xspace}
\newcommand{\Os}[0]{\ensuremath{\Omega_{t}}\xspace}%
\newcommand{\Ot}[0]{\ensuremath{\Omega_{t}}\xspace}%

\newcommand{\Pv}[0]{\ensuremath{\Phi_v}\xspace}
\newcommand{\Pt}[0]{\ensuremath{\Phi_{t}}\xspace}%

\newcommand{\Be}[0]{\ensuremath{B_e}\xspace}
\newcommand{\He}[0]{\ensuremath{H_e}\xspace}%
\newcommand{\Oe}[0]{\ensuremath{\Omega_e}\xspace}%
\newcommand{\Se}[0]{\ensuremath{\Sigma_e}\xspace}%
\newcommand{\Pe}[0]{\ensuremath{\Phi_e}\xspace}

\newcommand{\yes}{{\huge$\bullet$}}
\newcommand{\no}{{\huge$\circ$}}

\newcommand{\tainthl}[1]{\colorbox{gray!30}{\ensuremath{#1}}}

\DeclareMathOperator\lub{\bl{\sqcup}}

\colorlet{REDORANGE}{RedOrange}

\setcopyright{rightsretained}

\copyrightyear{2021}
\acmYear{2021}
\setcopyright{acmcopyright}\acmConference[CCS '21]{Proceedings of the 2021 ACM SIGSAC Conference on Computer and Communications Security}{November 15--19, 2021}{Virtual Event, Republic of Korea}
\acmBooktitle{Proceedings of the 2021 ACM SIGSAC Conference on Computer and Communications Security (CCS '21), November 15--19, 2021, Virtual Event, Republic of Korea}
\acmPrice{15.00}
\acmDOI{10.1145/3460120.3484534}
\acmISBN{978-1-4503-8454-4/21/11}

\begin{document}
\title{
	Exorcising Spectres with Secure Compilers
}
\author{Marco Patrignani}		
\orcid{0000-0003-3411-9678}			%
\affiliation{
  \institution{CISPA Helmholtz Center for Information Security}		%
  \country{Germany}
}
\additionalaffiliation{
  \institution{Stanford University}		%
  \country{USA}
}
\email{mp @ cs.stanford.edu}		%

\author{Marco Guarnieri}
\orcid{0000-0001-5767-555X}			%
\affiliation{
  \department{}			%
  \institution{IMDEA Software Institute}			%
  \country{Spain}
}
\email{marco.guarnieri@imdea.org}

\begin{abstract}
Attackers can access sensitive information of programs by exploiting the side-effects of speculatively-executed instructions using Spectre attacks.
To mitigate these attacks, popular compilers deployed a wide range of countermeasures whose security, however, has not been ascertained: while some are \emph{believed} to be secure, others are \emph{known} to be insecure and result in vulnerable programs.

This paper develops formal foundations for reasoning about the security of these defenses. 
For this, it proposes a framework of secure compilation criteria that characterise when compilers produce code resistant against Spectre v1 attacks.
With this framework, this paper performs a comprehensive security analysis of countermeasures against Spectre v1 attacks implemented in major compilers, deriving the first security proofs  of said countermeasures.

\begin{center}\small\it
	This paper uses
	a \src{blue}, \src{sans\text{-}serif} font for elements of the \src{source} language and an \trg{orange}, \trg{bold} font for elements of the \trg{target} language.
	Elements common to all languages are typeset in a \com{\commoncol}, \com{italic} font (to avoid repetitions).
	For a better experience, please print or view this in colour~\cite{patrignani2020use}.
\end{center}
\end{abstract}



\maketitle

\section{Introduction}\label{sec:intro}
By predicting the outcome of branching (and other) instructions, CPUs can trigger speculative execution and speed up computation by executing code based on such predictions.
When predictions are incorrect, CPUs roll back the effects of speculatively-executed instructions on the architectural state, i.e., memory, flags, and registers.
However, they do \emph{not} roll back effects on microarchitectural components like caches.

Exploiting microarchitectural leaks caused by speculative execution leads to Spectre attacks~\cite{maisuradze2018ret2spec, Kocher2018spectre,schwarz2019netspectre,220586,kiriansky2018speculative}.
Compilers support  a number of  countermeasures,
e.g.,
the insertion of \texttt{lfence} speculation barriers~\cite{Intel} 
and speculative load hardening~\cite{spec-hard}, that
\emph{can}  mitigate leaks introduced by speculation over branch instructions like those exploited in the Spectre v1 attack~\cite{Kocher2018spectre}.

Existing countermeasures, however, are often developed in an unprincipled way, %
that is, they are not \emph{proven} to be secure, and some of them fail in blocking \emph{speculative leaks}, i.e., leaks introduced by speculatively-executed instructions. %
For instance, the Microsoft Visual C++ compiler misplaces speculation barriers, thereby producing programs that are still vulnerable to Spectre attacks~\cite{spectector,kocher2018examples}.

In this paper,  
we propose a novel secure compilation framework for reasoning about speculative execution attacks and we use it to provide the first precise characterization of security for a comprehensive class of compiler countermeasures against Spectre v1 attacks. %
Let us now discuss our contributions more in detail:
\begin{asparaitem}[$\blacktriangleright$]

	\item We present a secure compilation framework tailored towards reasoning about speculative execution attacks (\Cref{sec:modeling}).
	The distinguishing feature of our framework is that compilers translate programs from a source language \SR, with a standard imperative semantics, into a target language \TR  equipped with a speculative semantics 
	capturing the effects of speculatively-executed instructions. %
	This matches a programmer's mental model: programmers do not think about speculative execution when writing source code (and they should not!) since speculation only exists in processors ( captured by \TR's  speculative semantics).
	It is the duty of a (secure) compiler to ensure the features of \TR cannot be exploited. %
	Our framework encompasses two different security models for speculative execution:
	\begin{inparaenum}[(1)]
	\item 
	\emph{(Strong) speculative non-interference}~\cite{spectector} (SNI), which considers \emph{all} leaks derived from speculatively-executed instructions as harmful, and
	\item  
	\emph{Weak speculative non-interference}~\cite{contracts}, which  considers harmful \emph{only} leaks of speculatively-accessed data. %
	\end{inparaenum}

	\item We introduce \emph{speculative safety} (\ss, \Cref{sec:robust-speculative-safety}), a novel safety property that implies the absence of classes of speculative leaks.
	The key features of \ss are that
	\begin{inparaenum}[(1)]
	\item it is parametric in a taint-tracking mechanism, which we leverage to reason about security by focusing on single traces, and 
	\item it is formulated to simplify proving that a compiler preserves it.
	\end{inparaenum}
	We instantiate \ss using two different taint-tracking mechanisms obtaining \textit{strong \ss} and \textit{weak \ss}.
	We precisely characterize the security guarantees of \ss by showing that strong (resp. weak) \ss over-approximates  strong (resp. weak) \sni. %

	\item We define two novel secure compilation criteria:  \emph{Robust Speculative Safety Preservation} (\rssp) and  \emph{Robust Speculative Non-In\-ter\-fer\-ence Preservation} (\rsnip, \Cref{sec:rob-saf-comp}).
	These criteria respectively ensure that compilers preserve (strong or weak) \ss and \sni \emph{robustly}, i.e, even when linked against arbitrary (potentially malicious) code.
	Satisfying these criteria implies that compilers correctly place countermeasures to prevent speculative leaks.
	However, \rssp requires preserving a safety property (\ss) and it is  simpler to prove than \rsnip, which requires preserving a hyperproperty~\cite{ClarksonS10}.
	To the best of our knowledge, these are the first criteria that concretely instantiate a recent theory that phrases security of compilers as the preservation of (hyper)properties~\cite{rhc,rhc-rel,rsc-j} to reason about a concrete security property, that is, the absence of speculative leaks.

\item Using our framework, we perform a comprehensive security analysis of  countermeasures against Spectre v1 attacks implemented in major C compilers  (\Cref{sec:frame-instances-insec}). 
	Specifically, we focus on \begin{inparaenum}[(1)]
	\item automated insertion of \lstinline{lfence}s (implemented in the Microsoft Visual C++ and the Intel ICC compilers~\cite{Intel-compiler,microsoft}), and
	\item speculative load hardening (SLH, implemented in Clang~\cite{spec-hard}).
	\end{inparaenum}
	We prove that:
	\begin{compactitem}
		\item The Microsoft Visual C++ implementation of (1) violates weak \rsnip and is thus insecure. %
		\item The Intel ICC implementation of (1) provides strong \rsnip, so compiled programs have \textit{no} speculative leaks.
		\item SLH provides weak \rsnip, so compiled programs do not leak speculatively-accessed data.
			This prevents Spectre-style attacks, but compiled programs might still speculatively leak data accessed non-speculatively. 
		\item The non-interprocedural variant of SLH violates weak \rsnip and is thus insecure.
		\item Our novel variant of SLH, called strong SLH, provides strong \rsnip and blocks all speculative leaks.
	\end{compactitem}
	All our security proofs follow a common methodology (see \Cref{sec:methodology}) whose  key insight is that proving a countermeasure to be \rssp is sufficient to ensure its security since \ss over-approximates \sni.
	This allows us to  leverage \ss to simplify our  proofs.

\end{asparaitem}
We conclude by discussing limitations and extensions of our approach  (\Cref{sec:disc-new}) and related work (\Cref{sec:rw}).
For simplicity, we only discuss  key aspects of our formal models. Full details and proofs  are in the companion report~\cite{guarnieri2019exorcising}.

\section{Modelling Speculative Execution} %
\label{sec:modelling_speculative_execution}\label{sec:modeling}
To illustrate our speculative execution model, we first introduce  Spectre v1 (\Cref{lis:spectrev1}). %
Using that, we define the threat model that we consider (\Cref{sec:threat}).
Then, we present the syntax of our languages (\Cref{sec:syn}) and their trace model (\Cref{sec:trace-model}).
This is followed by the operational semantics of our languages (\Cref{sec:non-spec}).
Next, we present the source (non-speculative) trace semantics (\Cref{sec:non-spec-tr-source}) and the target (speculative) trace semantics (\Cref{sec:trg}).
This formalisation focuses on the strong \sni model, so we conclude by defining the changes necessary for weak \sni  (\Cref{sec:weak-additions}).

\begin{lstlisting}[mathescape,label=lis:spectrev1,caption=The classic Spectre v1 snippet.] 
void get (int y) 
	if (y < size) then 
		temp = B[A[y]*512] 
\end{lstlisting}\label{sec:spectre-listing}
Consider the standard Spectre v1 example~\cite{Kocher2018spectre} in \Cref{lis:spectrev1}.
Function \lstinline{get} checks whether the index stored in variable \lstinline{y} is less than the size of array \lstinline{A}, stored in the global variable \lstinline{size}.
If so, the program retrieves \lstinline{A[y]}, multiplies it by the cache line size (here: 512), and uses the result to access array \lstinline{B}.
If \lstinline{size} is not cached, modern processors predict the guard's outcome  and speculatively continue the execution. %
Thus, line 3 might be executed even if \lstinline{y} $\geq$ \lstinline{size}.
When \lstinline{size} becomes available, the processor checks whether the prediction was correct.
If not, it rolls back all changes to the architectural state and executes the correct branch.
However, the speculatively-executed memory accesses leave a footprint in the cache, which enables an attacker to retrieve \lstinline{A[y]} even for \lstinline{y} $\geq$ \lstinline{size}.

\subsection{Threat Model}\label{sec:threat}

We study compiler countermeasures that translate source programs into (hardened) target programs.
In our setting, an attacker is an arbitrary program at target level that is linked against a (compiled) partial program of interest. %
The partial program (or, \emph{component}) stores sensitive information in a private heap that the attacker cannot access.
For this, we assume that attacker and component run on separate processes and OS-level memory protection restricts access to the private heap. %
For example, in \Cref{lis:spectrev1}, the array \lstinline{A} would be stored in the private heap and the attacker is code that runs before and after function \lstinline{get}.

While attackers cannot directly access the private heap, they can mount confused deputy attacks~\cite{confused,confused-dg} to trick components into leaking sensitive information despite the memory protection.%
\footnote{
    Adopting such an active attacker model turns our security definitions into `robust' ones, as we discuss in \Cref{sec:robustness-attackers}.
}
We focus on preventing \textit{only} speculative leaks, i.e., those caused by speculatively-executed instructions. %
For this, our attacker can observe the program counter and the locations of memory accesses during program execution.
This attacker model is commonly used to formalise code that has no timing side-channels~\cite{molnar2005program,almeida2016verifying} without requiring microarchitectural models.
Following~\citet{spectector}, we capture this model in our semantics through traces that record the address of all memory accesses (e.g., the address of  \lstinline{B[A[y]*512]} in  \Cref{lis:spectrev1}) and  the outcome of all control-flow instructions.

To model the  effects of speculative execution, our target language mispredicts the outcome of all branch instructions in the component.
This is the worst-case scenario in terms of leakage regardless of how  attackers poison the branch predictor~\cite{spectector}.

\subsection{Languages \SR and \TR}\label{sec:syn}

Technically, we have a pair of source and target languages (\SR and \TR) for studying security in the strong \sni model and a pair of source and target languages (\src{\wSR} and \trg{\wTR}) for studying weak \sni.
Strong (\SR-\TR) and weak (\wSR-\wTR) languages have the same syntax and a very similar semantics, which differ \textit{only} in the security-relevant observations produced during the computation.
We focus this section and the following ones on the strong languages \SR-\TR; we introduce the small changes for the weak languages \src{\wSR}-\trg{\wTR} in \Cref{sec:weak-additions}.

The source (\SR) and target (\TR) languages are  single-threaded While languages with a heap, a stack to lookup local variables, and a notion of components (our unit of compilation).
We focus on such a setting, instead of an assembly-style language like~\cite{spectector,cauligi2019towards}, to reason about speculative leaks without getting bogged down in complications like unstructured control flow.
This does not limit the power of attackers: since attackers reside in another process, they would not be able to exploit the additional features of assembly languages (e.g., unstructured control flow) to compromise  components.

The common syntax of \SR and \TR is presented below; we indicate sequences of elements $e_1,\cdots,e_n$ as \OB{e} and \com{\OB{e}\cdot e} denotes a stack with top element \com{e} and rest of the stack \com{\OB{e}}.
{
\begin{gather*}
\begin{aligned}
	\mi{Programs}\, \com{W}, \com{P} \bnfdef&\ \com{ H , \OB{F} , \OB{I}}
	&
	\mi{Codebase}~\com{C} \bnfdef&\ \com{\OB{F} , \OB{I}}
	&
	\mi{Imports}~\com{I} \bnfdef&\ \com{f}
	\end{aligned}
	\\
	\begin{aligned}
	\mi{Functions}~\com{F} \bnfdef&\ \com{f(x)\mapsto s;\ret}
	&
	\mi{Attackers}~\ctxc{} \bnfdef&\ \com{H , \OB{F}\hole{\cdot}}
	\\
	\mi{Heaps}~\com{H} \bnfdef&\ \come \mid \com{H ; n\mapsto v} 
	&\text{ where }& \com{n}\in\mb{Z}
	\\
	\mi{Expressions}~\com{e} \bnfdef&\ \com{x} \mid \com{v} \mid \com{e \op e} %
	&
		\mi{Values}~\com{v} \bnfdef&\ \com{n}\in\mb{N} 
	\end{aligned}
	\\
	\begin{aligned}
	\mi{Statements}~\com{s} \bnfdef&\ \lskip \mid \com{s;s} \mid \com{\letin{x}{e}{s}} \mid \com{\call{f}~e} 
			\mid \com{\asgn{e}{e}} 
		\\
		\mid&\
			\com{\asgnp{e}{e}} 
			\mid \com{\letread{x}{e}{s}} \mid \com{\letreadp{x}{e}{s}} 
		\\
			\mid&\ \com{\ifzte{e}{s}{s}} \mid \com{\cmove{x}{e}{e}{s}} \mid \com{\lfence}
\end{aligned}
\end{gather*}
}
We model \textit{components}, i.e., partial programs (\com{P}), and \textit{attackers} (\ctxc{}).
A (partial) program \com{P} defines its heap \com{H}, a list of functions \com{\OB{F}}, and a list of imports \com{\OB{I}}, which are all the functions an attacker can define.
An attacker \ctxc{} just defines its heap and its functions.
We indicate the code base of a program (its functions and imports) as \com{C}.

Functions are untyped, and their bodies are sequences of statements \com{s} that include standard instructions: skipping, sequencing, let-bindings, function calls, writing the public and the private heap, reading the public and private heap, conditional branching, conditional assignments and speculation barriers.
Statements can contain expressions \com{e}, which include program variables \com{x}, natural numbers \com{n}, arithmetic and comparison operators \com{\op}.
Heaps \com{H} map memory addresses $\com{n}\in\mb{Z}$ to values $\com{v}$.
Heaps are partitioned in a public part (when the domain $\com{n}\geq0$) and a private part (if $\com{n}<0$).
An attacker \com{A} can only define and access the public heap.
A program \com{P} defines a private heap and it can access both private and public heaps.

\subsection{Labels and Traces}\label{sec:trace-model}
Computation steps in \SR and \TR are \textit{labelled} with labels $\com{\lambda}$, which can be the \textit{empty label} $\com{\epsilon}$, 
an \textit{action} $\com{\alpha?}$ or \com{\alpha!} recording the control-flow between attacker and code (as required for secure compilation proofs~\cite{rsc-j,rhc,catalinRSC,PatrignaniASJCP15}), or a \textit{$\mu$arch$\ldotp$ action} $\com{\delta}$  capturing what a microarchitectural attacker can observe.
\begin{gather*}
\begin{aligned}
	\mi{\mu arch.\ Acts.}~\com{\delta} \bnfdef
	    &\ 
	    \com{\rdl{n}} \mid \com{\wrl{n}} \mid \com{\rdl{n\mapsto v}} 
	   \\
	   	\mid&\ \com{\wrl{n\mapsto v}} \mid \com{\ifl{v}} \mid \com{\rollbl}
	    \\
\end{aligned}
\\
\begin{aligned}
	\mi{Actions}~\acac{} \bnfdef&\ \com{\clgen{f}{v}{}}
		\mid \com{\rtgen{v}{}} 
	&
	\mi{Labels}~\com{\lambda} \bnfdef
		&\
		\com{\epsilon} \mid \com{\alpha?} \mid \com{\alpha!} \mid \com{\delta}
\end{aligned}
\end{gather*}

Action $\com{\clh{f}{v}{H}}$ represents a call to a function $f$ in the component with value $v$.
Dually, $\com{\cbh{f}{v}{H}}$ represents a call(back) to the attacker with value $v$.
Action $\com{\rth{v}{H}}$ represents a return to the attacker and $\com{\rbh{v}{H}}$ a return(back) to the component.

The $\com{\rdl{n}}$ and $\com{\wrl{n}}$ actions denote respectively read and write accesses to the private heap location \com{n}.
Dually, the $\com{\rdl{n\mapsto v}}$ and $\com{\wrl{n\mapsto v}}$ actions denote respectively read and write accesses to the public heap location $\com{n}$ where $\com{v}$ is the value read from/written to memory.
In these actions, locations $\com{n}$ model leaks through the data cache whereas values $\com{v}$, which only appear in operations on the public heap, model that attackers have access to the public heap.
In contrast, the $\com{\ifl{v}}$ action denotes the outcome of branch instructions and the $\com{\rollbl}$ action indicates the roll-back of speculatively-executed instructions.
These actions implicitly expose which instruction we are currently executing, and thus the instruction cache content.
Traces \com{\trac{}} are sequences of labels.
The semantics only track $\mu$arch$\ldotp$ actions executed inside the component $\com{P}$, whereas those executed in the attacker-controlled context $\ctxc{}$ are ignored (\Cref{tr:eus-tr-sin} later on).
The reason is that $\mu$arch$\ldotp$ actions produced by $\ctxc{}$ can be safely ignored since $\ctxc{}$ cannot access the private heap (this is analogous to other robust safety works~\citep{autysec,tydisa,cca,davidcaps}).

\subsection{Operational Semantics for \SR and \TR}\label{sec:non-spec}

Both languages are given a labelled operational semantics that describes how statements execute.
This semantics is defined in terms of program states $\com{C, H, \OB{B}\triangleright \proc{s}{\OB{f}} } $ that consist of a codebase \com{C}, a heap \com{H}, a stack of local variables \com{\OB{B}}, a statement \com{s}, and a stack of function names \com{\OB{f}}.
\com{C} is used to look up function bodies, %
whereas function names \com{\OB{f}}, which we often omit for simplicity, are used to infer if the code that is executing comes from the attacker or from the component, and this determines the produced labels. %
\begin{align*}
	\mi{Bindings}~\com{B} \bnfdef&\ \come \mid \com{B; x\mapsto v}
	&
	\mi{Prog.\ States}~\com{\Omega}\bnfdef&\ \com{C, H, \OB{B}\triangleright \proc{s}{\OB{f}} } 
\end{align*}
Both \SR and \TR have a big-step operational semantics for expressions and a small-step, structural operational semantics for statements that generates labels.
The former follows judgements $\com{B \triangleright e\bigred v}$ meaning: ``according to variables \com{B}, expression \com{e} reduces to value \com{v}.''
The latter follows judgements $\com{\Omega \xto{\lambda} \Omega'} $ meaning: ``state \com{\Omega} reduces in one step to \com{\Omega'} emitting label \com{\lambda}.''

We remark that values are computed as expected (though we use \com{0} for true in \com{ifz} statements; see \Cref{tr:eus-ift}) and expressions access only local variables in $\com{B}$ (reading from the heap is treated as a statement); therefore, we omit the expression semantics. 
Similarly, many of the rules for the statement semantics are standard and thus omitted; the most illustrative ones are given below.
We use \abs{n} for the absolute value of \com{n} and $\com{H}(n)$ to look up the binding for \com{n} in \com{H}.
\looseness=-1
\begin{center}
	\typerule{E-if-true}{
		\com{B \triangleright e\bigred 0}
	}{
		\begin{multlined}
			\com{C,H, \OB{B}\cdot B \triangleright \ifzte{e}{s}{s'}} \xto{(\ifl{0})} 
			\com{C, H, \OB{B}\cdot B \triangleright s}
		\end{multlined}
	}{eus-ift}
	\smallskip
	\typerule{E-read-prv}{
		\com{B \triangleright e\bigred n}
		&
		\com{H(-\abs{n})= v}
	}{
		\begin{multlined}
			\com{C, H, \OB{B}\cdot B \triangleright \letreadp{x}{e}{s}} \xto{\rdl{-\abs{n}}} 
			\\
			\com{C, H, \OB{B}\cdot B\cup x\mapsto v \triangleright s}
		\end{multlined}
	}{eus-rd-com-p}
	\smallskip	
	\typerule{E-write-prv}{
		\com{B \triangleright e\bigred n}
		&
		\com{H}=\com{{H}_1; -\abs{n}\mapsto v' ; H_2}
		\\
		\com{B \triangleright e'\bigred v}
		&
		\com{H'}=\com{{H}_1; -\abs{n}\mapsto v ; {H}_2}
	}{
		\com{C, H, \OB{B}\cdot B \triangleright \asgnp{e}{e'} \xto{\wrl{-\abs{n}}} C, H', \OB{B}\cdot B \triangleright \lskip} 
	}{eus-up-com-p}
\end{center}
The rules of conditionals, read, and write emit the related $\mu$arch$\ldotp$~actions (from \Cref{sec:trace-model}).
Specifically, conditionals produce observations recording the outcome of the condition (\Cref{tr:eus-ift}), whereas memory operations produce observations recording the accessed memory address (\Cref{tr:eus-rd-com-p} and \Cref{tr:eus-up-com-p}).

\subsection{Non-speculative Semantics for \SR}\label{sec:non-spec-tr-source}

We now define the non-speculative semantics of \SR, which describes how (whole) programs behave when executed on a processor without speculative execution.
A component \com{P} and an attacker \ctxc{} can be linked to obtain a whole program \com{W \equiv \ctxc{}\hole{P}} that contains the functions and heaps of \ctxc{} and \com{P}.
Only whole programs can run, and a program is whole only if it defines all functions that are called and if the attacker defines all the functions in the  interfaces of \com{P}. %

For this, we define the big-step semantics  $\Xtos{}$ of \SR, which concatenates single steps (defined by $\src{\to}$) into multiple ones and single labels into traces.  
The judgement $\src{\Omega} \Xtos{\tras{}} \src{\Omega'}$ is read: ``state \src{\Omega} emits trace \tras{} and becomes \src{\Omega'}''.
The most interesting rule is  below.
As mentioned in \Cref{sec:trace-model}, the trace does not contain $\mu$arch$\ldotp$ actions performed by the attacker (see the `then' branch, recall that functions in \src{\OB{I}} are defined by the attacker).
\begin{center}
	\typerule{E-\SR-single}{
		\src{\Omega} \equiv \src{\OB{F},\OB{I},H,B\triangleright \proc{s}{\OB{f}\cdot f}}
		&
		\src{\Omega'} \equiv \src{\OB{F},\OB{I},H',B'\triangleright \proc{s'}{\OB{f'}\cdot f'}}
		\\
		\src{\Omega}
		\xtos{\alpha}\src{\Omega'}
		&
		\text{ if } \src{f} == \src{f'} \text{ and } \src{f} \in \src{\OB{I}} \text{ then } \src{\tra{}} = \src{\epsilon} \text{ else } \src{\tra{}} = \src{\alpha}
	}{
		\src{\Omega}\Xtos{\tra{}}\src{\Omega'}
	}{eus-tr-sin}
\end{center}

Finally, the behaviour \src{\behavs{W}} of a whole program \src{W} is the trace \src{\tras{}} generated from the \Xtos{} semantics starting from the initial state of \src{W} (indicated as \SInits{W}) until termination. 
Intuitively, a program's initial state  is the \src{main} function, which is defined by the attacker.
\begin{example}[\,\SR  trace for \Cref{lis:spectrev1}]\label{ex:traces-src}
	Consider \lstinline{size}=\lstinline{4}.
	Trace \src{t_{ns}} indicates a valid execution of the code in \SR (without speculation).	
	\begin{align*}
	    \src{t_{ns}} =
	    	&\ 
	    	\src{ \clh{get}{0}{} \cdot \ifl{0} \cdot \rdl{n_A}{} \cdot \rdl{n_B+v_A^0} \cdot \rth{}{}}	
	\end{align*}
	We indicate the addresses of arrays \lstinline{A} and \lstinline{B} in the \SR heap with \src{n_A} and \src{n_B} respectively and the value stored at \lstinline{A[i]} with \src{v_A^i}.
	\lstset{language=Asm}
	\lstset{language=Java}
\end{example}

\subsection{Speculative Semantics for \TR}\label{sec:trg}

Our semantics for \TR models the effects of speculatively-executed instructions.
This semantics is inspired by the ``always mispredict'' semantics of \citet{spectector}, which captures the worst-case scenario (from an information theoretic perspective) independently of the branch prediction outcomes. %
Whenever the semantics executes a branch instruction, it first mis-speculates by executing the wrong branch for a fixed number \trg{w} of steps (called \textit{speculation window}). 
After speculating for \trg{w} steps, the speculative execution is terminated, the changes to the program state are rolled back, and the semantics restarts by executing the correct branch.
The $\mu$arch$\ldotp$~effects of speculatively-executed instructions are recorded on the trace as actions.
Speculative program states (\trg{\Sigma}) are defined as stacks of speculation instances (\trg{\Phi} \com{=} \trg{(\Omega,\trg{w})}), each one recording the program state \trg{\Omega} and the remaining speculation window \trg{w}.
The speculation window is a natural number \trg{n} or \trgb{\bot} when no speculation is happening; its maximum length is a global constant \trgb{\omega} that depends on physical characteristics of the CPU like the size of the reorder buffer.
\begin{align*}
	\mi{Speculative\ States}~\trg{\Sigma} \bnfdef&\ \trg{\OB{\Phi}}
	&
	\mi{Speculation\ Instance}~\trg{\Phi} \bnfdef&\ \trg{(\Omega,\trg{w})}
\end{align*}
The execution of program \trg{W} starts in state $ \trg{(\SInitt{W},\trg{\bot})}$, i.e., in the same initial state that \SR starts in. %
In the small-step operational semantics $\trg{\OB{\Phi} \xltot{\trgb{\lambda}} \OB{\Phi'}}$, reductions happen at the top of the stack: %
\begin{center}
	\typerule{E-\TR-speculate-if}{
		\trg{\Omega}\equiv\trg{C, H, \OB{B} \cdot B \triangleright \proc{s;s'}{\OB{f}\cdot f}}
		&
		\trg{s}\equiv\trg{\ifzte{e}{s''}{s'''}}
		\\
		\trg{\Omega \xtot{\acat} \Omega'}	
		&
		\trg{C}\equiv\trg{\OB{F};\OB{I}}
		&
		\trg{f}\notin\trg{\OB{I}}
		&
		\trg{j} = \fun{min}{\trgb{\omega},\trg{n}}
		\\
		\text{ if }
			\trg{B \triangleright e\bigred 0} 
		\text{ then }
			\trg{\Omega''}\equiv\trg{C, H, \OB{B} \cdot B \triangleright s''';s'}
		\\
		\text{ else }
			\trg{\Omega''}\equiv\trg{C, H, \OB{B} \cdot B \triangleright s'';s'}
	}{
		\trg{\OB{\Phi} \cdot (\Omega,n+1) \xltot{ \acat{} } \OB{\Phi} \cdot (\Omega',n)\cdot (\Omega'',j)}
	}{et-sp-if}
	\typerule{E-\TR-speculate-action}{
		\trg{\Omega \xtot{\trgb{\lambda}} \Omega'}	
		&
		\trg{\Omega}\equiv\trg{C, H, \OB{B} \triangleright \proc{s;s'}{\OB{f}\cdot f} }
		\\
		(\trg{s}\not\equiv\trg{ifz \cdots} \text{ and } \trg{s}\not\equiv\trg{\lfence})
        \text{ or }
        (\trg{C}\equiv\trg{\OB{F};\OB{I}} \text{ and } \trg{f}\in\trg{\OB{I}})
	}{
		\trg{\OB{\Phi} \cdot (\Omega,n+1) \xltot{ \trgb{\lambda}} \OB{\Phi} \cdot (\Omega',n)}
	}{et-sp-act}
	\typerule{E-\TR-speculate-lfence}{
		\trg{\Omega \xtot{\epsilon} \Omega'}	
		\\
		\trg{\Omega}\equiv\trg{C, H, \OB{B} \triangleright s;s'}
		\\
		\trg{s}\equiv\trg{\lfence}
	}{
		\trg{\OB{\Phi} \cdot (\Omega,n+1) \xltot{\epsilon} \OB{\Phi} \cdot (\Omega',0)}
	}{et-sp-lf}
	\typerule{E-\TR-speculate-rollback}{
		\trg{n}=\trg{0}\ \text{or}\ \trg{\Omega}\ \text{is stuck}
	}{
		\trg{\OB{\Phi} \cdot (\Omega,n) \xltot{ \rollbl } \OB{\Phi} }
	}{et-sp-rb}
\end{center}

Executing a statement updates the program state on top of the state and reduces the speculation window by $\trg{1}$ (\Cref{tr:et-sp-act}).
Mis-speculation pushes the mis-speculating state on top of the stack (\Cref{tr:et-sp-if}).
Note that speculation does not happen in attacker code (condition $\trg{f}\notin\trg{\OB{I}}$, recall that \trg{f} is the function executing now and \trg{\OB{I}} are all attacker-defined functions).
This is without loss of generality since (1) attackers cannot directly access  the private heap, and (2) our security definitions (\Cref{sec:robust-speculative-safety}) will consider any possible attacker, so the speculative behavior of an attacker (i.e., the speculative execution of the `wrong branch') will be captured by another one who has the same branches but inverted (e.g., the `then' code of one attacker is the `else' code of another).
When the speculation window is exhausted (or if the speculation reaches a stuck state), speculation ends and the top of the stack is popped (\Cref{tr:et-sp-rb}).
The role of the \trg{\lfence} instruction is setting to zero the speculation window, so that rollbacks are triggered (\Cref{tr:et-sp-lf}).

As before, the behaviour \trg{\behavt{W}} of a whole program \trg{W} is the trace \trg{\trat{}} generated, according to the \Xtot{} semantics, starting from the initial state of $\trg{W}$ until termination.

\begin{example}[\,\TR Trace for \Cref{lis:spectrev1}]\label{ex:traces-trg}
	Consider the same setting as \Cref{ex:traces-src}. 
	Trace \trg{t_{sp}} is a valid execution of the code in \TR, and therefore with speculation.
	As before, we indicate the addresses of arrays \lstinline{A} and \lstinline{B} in the source and target heaps with \com{n_A} and \com{n_B} respectively and the value stored at \lstinline{A[i]} with \com{v_A^i}.

	{
	\small
	\vspace{.2em}
	\hfill
	    \trg{t_{sp}} =
	    \trg{ \clh{get}{8}{} \cdot \ifl{1} \cdot \rdl{n_A+8}{} \cdot }\trg{  \rdl{n_B+v_A^8} \cdot \rollbl \cdot \rth{}{} }
	\hfill\hfill
	\vspace{.2em}
	}
	\lstset{language=Asm}

	Differently from \src{t_{ns}} in \Cref{ex:traces-src}, trace \trg{t_{sp}}  contains speculatively executed instructions whose side effects are represented by the actions \trg{\rdl{n_A+8}{}} and \trg{  \rdl{n_B+v_A^8}}.
	\lstset{language=Java}
\end{example}

\subsection{Weak Languages \src{\wSR} and \trg{\wTR}}\label{sec:weak-additions}

To conclude, we now introduce the weak languages \src{\wSR} and \trg{\wTR}, which we use to study security in the weak \sni model.
Following~\cite{contracts}, these languages differ from \src{L} and \trg{T} in a single aspect, that is, in the actions produced by memory reads.
Specifically, in  \src{\wSR} and \trg{\wTR}, non-speculatively reading from the private heap produces an action $\com{\rdl{n \mapsto v}}$ that contains the read value $\com{v}$ as well as the accessed memory address $\com{n}$.
As we show next, this difference allows us to precisely characterize \textit{only} the leaks of transiently loaded data, which are exactly those leaks exploited in speculative disclosure gadgets like \Cref{lis:spectrev1}, rather than all speculative leak.
\section{Security Definitions for Secure Speculation} %
\label{sec:robust-speculative-safety}

We now present  \emph{semantic} security definitions against speculative leaks.
We start by presenting (robust) speculative non-interference (\rsni, \Cref{sec:rsni}).
Next, we introduce (robust) speculative safety (\rss, \Cref{sec:rss}).
These definitions can be applied to programs in the four languages \SR, \TR, \src{\wSR}, and \trg{\wTR}.
Therefore,  we write \rsni{}\com{(L)} and \rss{}(\com{L}) to indicate which language \com{L} the definitions are referring to.
Since these languages have the same syntax but different semantics, we also study the relationships between \rsni and \rss for weak and strong languages.
We depict these results below (only for \trg{\TR} and \trg{\wTR} since all  security definitions trivially hold for the source non-speculative languages \src{\SR} and \src{\wSR}) and discuss them further down.
	\begin{center}
		
		\tikzpic{
				\node[](rss){ \rss{}(\TR)};
				\node[below = of rss, yshift = .5em](wrss){\rss{}(\trg{\wTR})};
				\node[right = of rss, xshift = 3 em](rsni){\rsni{}(\trg{\TR})};
				\node[]at (rsni |- wrss) (wrsni) {\rsni{}(\trg{\wTR})};
				
				\draw[=,double,-implies,double equal sign distance] (rss) to node[midway, above, font = \scriptsize](i1){\Cref{thm:rss-overapp-rsni}} (rsni);
				\draw[=,double,-implies,double equal sign distance] (wrss) to node[midway, above, font = \scriptsize](i2){\Cref{thm:rss-overapp-rsni-weak}} (wrsni);
				\draw[=,double,-implies,double equal sign distance] (rss) to node[midway, left, font = \scriptsize](i3){\Cref{thm:strong-impl-weak:ss}} (wrss);
				\draw[=,double,-implies,double equal sign distance] (rsni) to node[midway, right, font = \scriptsize](i4){\Cref{thm:strong-impl-weak:sni}} (wrsni);
	
				\node[left = 1.5 of rss](ms){\small most secure};
				\node[below = of ms](ls){\small least secure};
				\node[above = .1 of rss](mp){\small least precise};
				\node[above = .1 of rsni](lp){\small most precise};
			}
	\end{center}
\subsection{Robust Speculative Non-Interference}\label{sec:rsni}
Speculative non-interference (\sni) is a class of security properties~\cite{spectector,contracts} that is based on comparing the information leaked by instructions executed speculatively and non-speculatively.
\sni requires that speculatively-executed instructions do not leak more information than what is leaked by executing the program without speculative execution, which is obtained by ignoring observations produced speculatively.
Hence, \sni semantically characterize the information leaks that are introduced by speculative execution, that is, those leaks that are exploited in Spectre-style attacks. %

\subsubsection*{Property}
Here, we instantiate robust speculative non-interference in our framework by following \sni's trace-based characterization~\citep[Proposition 1]{spectector}.
Thus we need to introduce two concepts:
\begin{asparaitem}
\item \sni is parametric in a policy denoting sensitive information. 
As mentioned in \Cref{sec:threat}, we assume that only the private heap is sensitive. 
Hence, whole programs \com{W} and \com{W'} are \emph{low-equivalent}, written \com{W' \loweq W}, if they differ only in their private heaps.

\item \sni requires comparing the leakage resulting from non-spec\-ula\-ti\-ve and speculative instructions.
The \textit{non-speculative projection} $\nspecProject{\com{t}}$~\cite{spectector} of a trace \com{t} extracts the observations associated  with non-speculatively-executed instructions.
We obtain $\nspecProject{\com{t}}$ by removing from $\com{t}$ all sub-strings enclosed between $\com{\ifl{v}}$ and $\com{\rollbl}$ observations.
We illustrate this using an example: $\nspecProject{\cdot}$ applied to \trg{t_{sp}} from \Cref{ex:traces-trg} produces $\nspecProject{\trg{t_{sp}}} =\
\trg{ \clh{get}{8}{} \cdot \ifl{1} \cdot \rth{}{} }$.
\end{asparaitem}

We now formalise \sni{}.
A whole program \com{W} is \sni{} if its  traces do not leak more than their non-speculative projections.
 That is, if an attacker can distinguish the traces produced by \com{W} and a low-equivalent program \com{W'}, the distinguishing observation must be made by an instruction that does not result from mis-speculation.
\begin{definition}[Speculative Non-Interference (\snidef)]\label{def:sni}
	\begin{align*}
		\vdash \com{W} : \sni \isdef&\ 
			\forall \com{W'}.\ 	
			\text{ if } 
			\com{W' \loweq W}
			\\
			\text{ and } 
			&\
			\nspecproj{\behavc{\SInit{W}}} = \nspecproj{\behavc{\SInit{W'}}}
			\\
			\text{ then }
			&\ \behavc{\SInit{W}} = \behavc{\SInit{W'}}
	\end{align*}
\end{definition}

A component \com{P} is robustly speculatively non-interferent if it is \sni no matter what valid attacker it is linked to (\Cref{def:rsni}), where an attacker \com{A} is valid ($\vdash \ctxc{} : \com{atk}$) if it does not define a private heap and does not contain instructions to read and write it. %
\begin{definition}[Robust Speculative Non-Interference (\rsnidef)]\label{def:rsni}
	\begin{align*}
		\vdash \com{P} : \rsni \isdef&\ \forall \ctxc{} . \text{ if }\vdash \ctxc{} : \com{atk} \text{ then } \vdash \ctxc{}\hole{P} : \sni
	\end{align*}
\end{definition}

\begin{example}[\Cref{lis:spectrev1} is \rsni in \SR and not in \TR]\label{ex:sni}
	\lstset{language=Asm}
	Consider the code of \Cref{lis:spectrev1}.
	As expected, this code is \rsni in \SR. 
	Indeed, \SR does not support speculative execution and, therefore, for any trace $\src{t_{ns}}$ produced by an \SR-program $\nspecProject{\src{t_{ns}}} = {\src{t_{ns}}}$.
\lstset{language=Asm}
	The same code, however, is not \rsni in \TR.
Consider the code of \Cref{lis:spectrev1} (indicated as \trg{P_1}) and an  attacker \ctxt{^{\!8}} that calls function \lstinline{get}  with \trg{8}.
Since array \lstinline{A} is in the private heap, the low-equivalent program required by \Cref{def:sni} is the same \ctxt{^{\!8}} linked with some \trg{P_N}, which is the same \trg{P_1} with some array \lstinline{N} with contents different from \lstinline{A} in the heap such that \trg{\lstinline{A[8]}}$\neq$\trg{\lstinline{N[8]}}.
Whole program \trg{\ctxt{^{\!8}}\hole{P_1}} generates trace \trg{t_{sp}} from \Cref{ex:traces-trg} while \trg{\ctxt{^{\!8}}\hole{P_N}} generates \trg{t_{sp}'} below.
We indicate the address of array \lstinline{N} as \trg{n_N} and the content of \trg{\lstinline{N[i]}} as \trg{v_N^i}.
Low-equivalence yields that addresses are the same ($\trg{n_A+8}=\trg{n_N+8}$) but contents are not ($\trg{v_A^8}\neq\trg{v_N^8}$), and thus \lstinline{B} is accessed at different offsets ($\trg{n_B+v_A^8}\neq\trg{n_B+v_N^8}$).
\begin{align*}
	\trg{t_{sp}'} =&\ \trg{ \clh{get}{8}{} \cdot \ifl{1} \cdot \rdl{n_N+8}{} \cdot \rdl{n_B+v_N^8} \cdot \rollbl \cdot \rth{}{} }
\end{align*}
\lstset{language=General}%
\Cref{lis:spectrev1} is not \rsni in \TR (and neither in \wTR) since the non-speculative projections of \trg{t_{sp}'} and of \trg{t_{sp}} are the same (see above) while \trg{t_{sp}'} and \trg{t_{sp}} are \emph{different} ($\trg{\rdl{n_B+v_A^8}}\neq\trg{\rdl{n_B+v_N^8}}$).
\end{example}

\subsubsection*{Security Guarantees}
Since \rsni is defined in terms of traces, its security guarantees depend on which of the four languages  \SR, \TR, \src{\wSR}, and \trg{\wTR} we consider. %
As expected, for the source languages \SR and \src{\wSR}, \rsni is trivially satisfied; there is no speculative execution in \SR and \wSR and all traces are identical to their non-speculative projections.

\begin{theorem}[All \SR and \src{\wSR} programs are \rsni]\label{thm:all-s-rsni}
		\begin{align*}
			\forall \src{P} \ldotp 
			&\ 
			\vdash\src{P}:\rsni{}(\SR)
			\text{ and } 
			\vdash\src{P}:\rsni{}(\src{\wSR})
		\end{align*}
	\end{theorem}

For the target languages \TR and \trg{\wTR}, which support speculative execution, \rsni provides different security guarantees.

	\rsni{}(\trg{\TR}) corresponds to speculative non-interference~\cite{spectector,contracts}, which ensures the absence of \textit{all} speculative leaks.
	In our setting, the only allowed leaks are those depending either on information from the public heap or information from the private heap that is disclosed through actions produced non-speculatively, e.g., as an address of a non-speculative memory access. %
	Any other speculative leak of information from the private heap is disallowed by \rsni{}(\trg{\TR}).

	\rsni{}(\trg{\wTR}), in contrast, corresponds to weak speculative non-interference~\cite{contracts}, which allows speculative leaks of information that has been retrieved non-speculatively.
	Indeed, in \trg{\wTR} non-speculative reads from the private heap produce actions $\trg{\rdl{n \mapsto v}}$ that additionally disclose the value $\trg{v}$ read from the heap as part of the non-speculative projection.
	As a result, data retrieved non-speculatively from the private heap can influence speculative actions, which are not part of the non-speculative projection of the trace, without violating \rsni{}(\trg{\wTR}).
	That is, \rsni{}(\trg{\wTR}) ensures the absence only of leaks of speculatively-accessed data.

	Since \rsni{}(\trg{\TR}) ensures the absence of \textit{all} speculative leaks while \rsni{}(\trg{\wTR}) only ensures the absence of \emph{some} of them, any \rsni{}(\trg{\TR}) program is also \rsni{}(\trg{\wTR}). 
	
	\begin{theorem}[\rsni{}(\trg{\TR}) Implies \rsni{}(\trg{\wTR})]\label{thm:strong-impl-weak:sni}
		\begin{align*}
			\forall \trg{P}.&\ \text{ if } \vdash \trg{P} : \rsni{}(\trg{\TR}) \text{ then } \vdash \trg{P} : \rsni{}(\trg{\wTR}) 
		\end{align*}
	\end{theorem}

	As shown in~\cite{contracts}, strong and weak speculative non-interference (that is, \rsni{}(\trg{\TR}) and \rsni{}(\trg{\wTR})) have different implications for secure programming.
	In particular, programs that are traditionally constant-time (i.e., constant-time under the non-speculative semantics) and satisfy strong speculative non-interference are also constant-time w.r.t. the speculative semantics.
	Similarly, programs that are traditionally sandboxed (i.e., do not access out-of-the-sandbox data non-speculatively) and satisfy weak speculative non-interference are also sandboxed w.r.t. the speculative semantics.

\subsection{Robust Speculative Safety}\label{sec:rss}

We now introduce \textit{speculative safety} (\ss), a safety property that soundly over-approximates \sni.
To enable reasoning about security using single traces (rather than pairs of traces as in \sni), we extend our languages with a taint-tracking mechanism that (1) taints values as ``safe'' (denoted by  \com{\safeta}) whenever they can be leaked speculatively without violating \sni (e.g., the public heap is ``safe'') or  ``unsafe'' (denoted by \com{\unta}) otherwise, and (2) propagates taints to labels across computations.
Speculatively safe programs produce traces containing only safe labels.

\myparagraph{Taint tracking}
Taint-tracking is at the foundation of our speculative safety definition and it enables reasoning about security on single traces. 
For this, we extend the semantics of our languages \src{\SR}, \src{\wSR}, \trg{\TR}, and \trg{\wTR} with a taint tracking mechanism.
We consider two taint-tracking mechanisms, a strong and a weak one, that lead to different security guarantees, as we show later.
Each mechanism is adopted in the related pair of languages: strong (resp. weak) languages use the strong (resp. weak) taint-tracking.
Our taint-tracking is rather standard, so we provide an informal overview of its key features below using the rules for reading from the private heap as an example; full details are  \Cref{sec:taint-app}.
These rules simply extend \Cref{tr:eus-rd-com-p} with taint, which is highlighted in gray.
\begin{center}\small
	\typerule{T-read-prv}{
		\com{B \triangleright e\bigred n : \tainthl{\sigma'}}
		&
		\com{H}(-\abs{n}) = v : \tainthl{\sigma''}
		&
		\tainthl{\com{\sigma}} = \tainthl{\com{\sigma''\glb\sigma'}}
	}{
		\begin{multlined}
			\com{\tainthl{\sigma_{pc}}; C, H, \OB{B}\cdot B \triangleright \letreadp{x}{e}{s} \xto{ \rdl{-\abs{n}}^{\tainthl{\sigma\lub\sigma_{pc}}}}}
				\\
				\com{C, H, \OB{B}\cdot B\cup x\mapsto v : \tainthl{\unta} \triangleright s}
		\end{multlined}
	}{tus-rd-com-p-strong}
	\typerule{T-read-prv-weak}{
		\com{B \triangleright e\bigred n : \tainthl{\sigma'}}
		&
		\com{H}(-\abs{n}) = v : \tainthl{\sigma''}
		&
		\tainthl{\com{\sigma}} = \tainthl{\com{\sigma''\glb\sigma'}}
	}{
		\begin{multlined}
			\com{\tainthl{\sigma_{pc}}; C, H, \OB{B}\cdot B \triangleright \letreadp{x}{e}{s} \xto{ \rdl{-\abs{n}\mapsto v}^{\tainthl{ \sigma\lub\sigma_{pc}}}}}
				\\
				\com{C, H, \OB{B}\cdot B\cup x\mapsto v : \tainthl{\sigma' \lub \sigma_{pc}} \triangleright s}
		\end{multlined}
	}{tus-rd-com-p-weak}
\end{center}

\begin{asparaitem}
\item All values $\com{v}$ are tainted with a taint $\com{\sigma} \in \{\com{\safeta}, \com{\unta}\}$. %
Heaps $\com{H}$ and variable bindings $\com{B}$ are extended to record the taint of values.
Taints form the usual \emph{integrity} lattice $\safeta\leq\unta$ (which is the dual of the lattice used for non-interference) and are combined using the least-upper-bound ($\glb$) and greatest-lower-bound ($\lub$) operators.
For simplicity, we report here the key cases: $\safeta\glb\unta = \unta$ and $\safeta\lub\unta=\safeta$.

\item The public part of the initial heap is tainted as safe, and its private part is tainted as unsafe.

\item The taint-tracking mechanism also tracks the taint $\com{\sigma_{pc}}$ associated with the program counter.
The program counter taint is $\com{\safeta}$ whenever we are not speculating and it is raised to $\com{\unta}$ whenever we are executing instructions speculatively.
The latter can happen only in the \trg{\TR} and \trg{\wTR} languages, where it is represented by the speculative state containing more than one speculation instance.
In the source languages, instead, $\src{\sigma_{pc}}$ is always $\src{\safeta}$.

\item Taint is propagated in the standard way across computations.
For example, expressions combine taints using the least-upper-bound $\glb$, i.e.,  expressions involving unsafe values are tainted \com{U}.

The strong and weak taint-tracking mechanisms differ, however, in how they handle memory reads from the private heap. 
When reading from the private heap, the strong mechanism used in \src{\SR} and \trg{\TR} taints the variable where the data is stored as unsafe (\com{U}) (\Cref{tr:tus-rd-com-p-strong}).
In contrast, the weak mechanism of \src{\wSR} and \trg{\wTR}, taints the target value with the greatest-lower-bound of the taints of the memory address and of the program counter (\Cref{tr:tus-rd-com-p-weak}).
This ensures that information retrieved non-speculatively from the private heap (i.e., the program counter taint is $\com{\safeta}$) is tainted $\com{\safeta}$.
\looseness=-1

\item 
The taint tracking does not keep track of implicit flows.
Since the program counter is part of the actions, any sensitive implicit flow would appear in the trace due to the corresponding \com{\ifl{v}} action.

\item %
The taint of labels is the greatest-lower-bound of the taint of the expressions generating the label and the program counter taint (\Cref{tr:tus-rd-com-p-strong} and \Cref{tr:tus-rd-com-p-weak}).
This ensures that non-speculative labels are tainted as safe (\com{\safeta}), while speculative labels are tainted as unsafe (\com{\unta}) if they depend on unsafe data and safe otherwise.

\end{asparaitem}

With a slight abuse of notation, in the following we refer to the languages \src{\SR}, \src{\wSR}, \trg{\TR}, and \trg{\wTR} extended with the corresponding taint-tracking mechanisms outlined above whenever we talk about speculative safety.
That is, for speculative safety, programs in  \src{\SR}, \src{\wSR}, \trg{\TR}, and \trg{\wTR} produce traces $\trac{^\sigma}$ of tainted labels $\com{\lambda^\sigma}$, where taints $\com{\sigma}$ are computed as described above.

\subsubsection*{Property}
Speculative safety ensures that \emph{whole} programs \com{W} generate only safe (\safeta) actions in their traces.
As we show later, \ss security guarantees depend on the underlying language (and on its taint-tracking mechanism). 

\begin{definition}[Speculative Safety (\ssdef)]\label{def:ss}
	\begin{align*}
	\vdash\com{W} : \ss \isdef&\
		\forall \trac{^\sigma}\in\behavc{W}\ldotp \forall \acac{^\sigma}\in\trac{^\sigma}\ldotp \com{\sigma}\equiv\com{\safeta}
	\end{align*}
\end{definition}
A component \com{P} is \rss if it upholds \ss when linked against arbitrary valid attackers (\Cref{def:rss}).
\begin{definition}[Robust Speculative Safety (\rssdef)]\label{def:rss}
	\begin{align*}
		\vdash \com{P} : \rss \isdef&\ \forall \ctxc{}. \text{ if } \vdash \ctxc{} : \com{atk} \text{ then } \vdash \com{\ctxc{}\hole{P}} : \ss
	\end{align*}
\end{definition}

\begin{example}[\Cref{lis:spectrev1} is \rss in \SR and not in \TR]\label{ex:rss}
The code of   \Cref{lis:spectrev1} is \rss in \SR because $\src{\sigma_{pc}}$ is always $\src{\safeta}$ and, therefore, all actions are tainted as $\src{\safeta}$.
The code, however, is neither \rss in \TR nor in \wTR.
For this, consider the trace from \Cref{ex:traces-trg}. 
The taint-tracking mechanism taints the actions as follows:
	{\small
	\begin{align*}
	    \trg{t_{sp}} =&\ \trg{ \clh{get}{8}{}^\safeta \cdot \ifl{1}^\safeta \cdot \rdl{A[8]}{}^\safeta \cdot \rdl{B[A[8]]}^\unta \cdot  \rollbl^\safeta \cdot \rth{}{}^\safeta }
	\end{align*}
	}
	\lstset{language=Asm}
The trace contains an unsafe action corresponding to the second memory access.
This happens because the action has been generated speculatively (that is, $\trg{\sigma_{pc}}$ is $\trg{\unta}$) and it depends on data retrieved from the private heap (which \TR's taint-tracking  taints as $\trg{\unta}$).
\lstset{language=Java}
\end{example}

\subsubsection*{Security Guarantees}
Similarly to \sni, the security guarantees of \ss depend on the underlying language.
As expected, \rss  trivially holds for  \SR and \src{\wSR} since they only produce labels tainted \src{\safeta}.

\begin{theorem}[All \SR and \src{\wSR} programs are \rss]\label{thm:all-s-rdss}
		\begin{align*}
			\forall \src{P} \ldotp 
			&\ 
			\vdash\src{P}:\rss{}(\SR)
			\text{ and } 
			\vdash\src{P}:\rss{}(\src{\wSR})
		\end{align*}
	\end{theorem}

In contrast, \rss' guarantees are different for  \trg{\TR} and \trg{\wTR}, which are equipped with distinct taint tracking mechanisms.

\rss{}(\TR) is a strict over-approximation of \rsni{}(\trg{\TR}) (and, thus, of speculative non-interference in terms of~\cite{spectector, contracts}) and its preservation through compilation is easier to prove than \rsni{}(\trg{\TR})-preservation.
\begin{theorem}[\rss{}(\TR) over-approximates \rsni{}(\trg{\TR})]\label{thm:rss-overapp-rsni}
	\begin{align*}
		1) &\ \forall \trg{P}. \text{ if } \vdash \trg{P} : \rss{}(\TR) \text{ then } \vdash \trg{P} : \rsni{}(\trg{\TR})
		\\
		2) &\ \exists \trg{P}. \vdash \trg{P} : \rsni{}(\trg{\TR}) \text{ and } \nvdash \trg{P} : \rss{}(\TR)
	\end{align*}
\end{theorem}
To understand point 1, observe that \rss{}(\TR) ensures that only safe observations are produced by a program $\trg{P}$.
This, in turn, ensures that no information originating from the private heap is leaked through speculatively-executed instructions in  $\trg{P}$.
Therefore,  $\trg{P}$ satisfies \rsni{}(\trg{\TR}) because everything except the private heap is visible to the attacker, i.e., there are no additional leaks due to speculatively-executed instructions.

\lstset{language=Asm}
To understand point 2, consider \lstinline{get_nc} from \Cref{lis:rel-rss-rsni}, which always accesses \trg{\lstinline{B[A[y]]}}.
This code is \rsni{}(\trg{\TR}) because  states that can be distinguished by  the traces can also be distinguished by  their non-speculative projections, i.e., speculatively-executed instructions do not leak additional information.
However, it is not \rss{}(\TR) because speculative memory accesses will produce \trg{\unta} actions.
\lstset{language=Asm}
\begin{lstlisting}[mathescape,label=lis:rel-rss-rsni,caption=Code that is \rsni but not \rss.] 
void get_nc (int y) 
	if (y < size) then B[A[y] ] else B[A[y] ]
\end{lstlisting}\label{sec:spectre-listing}

\rss{}(\trg{\wTR}), in contrast, is a strict over-approximation of \rsni{}(\trg{\wTR}) (and, therefore, of weak speculative non-interference in terms of~\cite{contracts}).
\begin{theorem}[\rss{}(\trg{\wTR}) over-approximates \rsni{}(\trg{\wTR})]\label{thm:rss-overapp-rsni-weak}
	\begin{align*}
		1) &\ \forall \trg{P}. \text{ if } \vdash \trg{P} : \rss{}(\trg{\wTR}) \text{ then } \vdash \trg{P} : \rsni{}(\trg{\wTR})
		\\
		2) &\ \exists \trg{P}. \vdash \trg{P} : \rsni{}(\trg{\wTR}) \text{ and } \nvdash \trg{P} : \rss{}(\trg{\wTR})
	\end{align*}
\end{theorem}

Finally, it is easy to see that any \rss{}(\TR) program is also \rss{}(\trg{\wTR}) since all actions  tainted \trg{\safeta} by the taint-tracking of \TR are  tainted \trg{\safeta} also by the taint-tracking of \wTR.
\begin{theorem}[\rss{}(\trg{\TR}) Implies \rss{}(\trg{\wTR})]\label{thm:strong-impl-weak:ss}
	\begin{align*}
		\forall \trg{P}.&\ \text{ if } \vdash \trg{P} : \rss{}(\TR)  \text{ then } \vdash \trg{P} : \rss{}(\trg{\wTR}) 
	\end{align*}
\end{theorem}

\section{Compiler Criteria for Secure Speculation}\label{sec:rob-saf-comp}
We now introduce our secure compilation criteria:
\emph{robust speculative safety preservation} (\rssp, \Cref{sec:rssp}), which  preserves  \rss, and \emph{robust speculative non-interference preservation} (\rsnip, \Cref{sec:rsnip}), which preserves \rsni.
We conclude by discussing how compilers can be proven secure or insecure using these criteria (\Cref{sec:methodology}).

As before, criteria can be instantiated using pairs of languages \src{L}-\trg{T} or \src{\wSR}-\trg{\wTR}.
Criteria instantiated with the strong languages (say \rssp{}(\SR,\TR)) are indicated with a \text{+} (that is, \strong{\rssp{}}). 
Those instantiated with weak languages (say \rsnip{}(\wSR,\wTR)) are indicated with a \text{-} (that is, \weak{\rsnip{}}).
When we omit the `sign', we refer to both criteria.
For simplicity, we only present the strong criteria (for \src{L}-\trg{T}), weak ones are defined identically (but for \src{L^-}-\trg{T^-}).

\subsection{Robust Speculative Safety Preservation}\label{sec:rssp}
The first criterion is clear: a compiler preserves \rss if given a source component that is \rss, the compiled counterpart is also \rss.
\begin{definition}[\strong{\rdsspdef{}}]\label{def:rssp}
		\begin{align*}
		\vdash \comp{\cdot} : \strong{\rssp} \isdef&\ 
		\forall\src{P} \in \SR\ldotp 
		\text{ if } \vdash\src{P}: \rss{}(\SR)
		\text{ then }
		\vdash\comp{\src{P}} : \rss(\TR)
	\end{align*}
\end{definition}

\Cref{def:rssp} is a ``property-ful'' criterion since it explicitly refers to the preserved property~\citep{rhc,rhc-rel}. %
Proving a ``property-ful'' criterion, however, can be fairly complex.
Fortunately, it is generally possible to turn a ``property-ful'' definition into an \emph{equivalent} ``property-free'' one~\cite{rhc,rsc-j,rhc-rel}, which come in so-called \emph{backtranslation} form with established proof techniques~\cite{rhc,catalinRSC,rsc-j,PatrignaniASJCP15,nonintfree,NewBA16}.
To state the equivalence of these criteria, we introduce a cross-language relation between traces of the two languages, which specifies when two possibly different traces have the same ``meaning''.
Our property-free security criterion (\rssc, \Cref{def:rdss}) states that a compiler is \rssc if for any target-level attacker \ctxt{} that generates a trace \trat{^{\taintt}}, we can build a source-level attacker \ctxs{} that generates a trace \tras{^\sigma} that is related to \trat{^{\taintt}}.
A source trace \tras{^\sigma} and a target trace \trat{^{\taintt}} are related (denoted with $\tras{^\sigma} \rels \trat{^{\taintt}}$) if the target trace contains all the actions of the source trace, plus possible interleavings of safe (\trg{\safeta}) actions (\Cref{tr:tr-rel-safe,tr:tr-rel-safe-h}).
All other actions must be the same (i.e., $\equiv$, \Cref{tr:tr-rel-same,tr:tr-rel-same-h}).

\begin{center}
	\typerule{Trace-Relation-Same}{
		\src{\tras{^\sigma} } \rels \trg{\trat{^{\taintt}} }	
		&
		\src{\alpha^\sigma} \equiv \trg{ \acat{^{\taintt}} }
	}{
		\src{\tras{^\sigma} \cdot \alpha^\sigma} \rels \trg{\trat{^{\taintt}} \cdot \acat{^{\taintt}} }
	}{tr-rel-same}
	\typerule{Trace-Relation-Same-Heap}{
		\src{\tras{^\sigma} } \rels \trg{\trat{^{\taintt}} }	
		&
		\src{\delta^\sigma} \equiv \trg{ \trgb{\delta}^{\taintt} }
	}{
		\src{\tras{^\sigma} \cdot \delta^\sigma} \rels \trg{\trat{^{\taintt}} \cdot \trgb{\delta}^{\taintt} }
	}{tr-rel-same-h}
	\typerule{Trace-Relation-Safe}{
		\src{\tras{^\sigma} } \rels \trg{\trat{^{\taintt}} }	
	}{
		\src{\tras{^\sigma} } \rels \trg{\trat{^{\taintt}} \cdot \acat{^{\safeta}} }
	}{tr-rel-safe}
	\typerule{Trace-Relation-Safe-Heap}{
		\src{\tras{^\sigma} } \rels \trg{\trat{^{\taintt}} }	
	}{
		\src{\tras{^\sigma} } \rels \trg{\trat{^{\taintt}} \cdot \trgb{\delta}^{\safeta} }
	}{tr-rel-safe-h}
	\vspace{.3em}
\end{center}

We are now ready to formalise \rdss, which intuitively states that compiled programs produce the same traces as their source counterparts with  possibly additional safe actions.
Crucially, \rdss is equivalent to \rdssp (\Cref{thm:rdssp-rdss-eq}), this result implies that our choice for the trace relation is correct; a relation that is too strong or too weak would not let us prove this equivalence.
\begin{definition}[\strong{\rsscdef{}}]\label{def:rdss}\label{def:rssc}
	\begin{align*}
		\vdash \comp{\cdot} : \strong{\rdss} \isdef&\
			\forall\src{P}\in\SR, 
			\ctxt{}, \trat{^{\taintt}}\ldotp
			\text{ if }
			\behavt{\trg{\ctxt{}\hole{\comp{\src{P}}}}} = \trat{^{\taintt}}
			\\
			\text{ then }&\
			\exists \ctxs{}, \tras{^\sigma} \ldotp 
			\behavs{\src{\ctxs{}\hole{P}}} = \tras{^\sigma}
			\text{ and }
			\tras{^\sigma}\rels\trat{^{\taintt}}
	\end{align*}
\end{definition}
\begin{theorem}[\rdssp and \rdss are equivalent]\label{thm:rdssp-rdss-eq}
\begin{align*}
	\forall \comp{\cdot}\ldotp&
	\vdash\comp{\cdot} : \strong{\rdssp{}}
	\iff \vdash\comp{\cdot} : \strong{\rdss{}} 
	\\
	\forall \comp{\cdot}\ldotp&
	\vdash\comp{\cdot} : \weak{\rdssp{}}
	\iff \vdash\comp{\cdot} : \weak{\rdss{}}
\end{align*}
\end{theorem}

\Cref{def:rssc} requires providing an existentially-quantified source attacker \ctxs{}.
The general proof technique for these criteria is called \emph{backtranslation}~\cite{scsurvey,rhc}, and it can either be attacker-based~\cite{DevriesePP16,nonintfree,NewBA16} or trace-based~\cite{PatrignaniASJCP15,catalinRSC,rsc-j}.
The distinction tells us what quantified element one can use  to build the source attacker \ctxs{}, either the target attacker \ctxt{} or the trace \trat{^{\taintt}} respectively.
In our proofs, we will use an attacker-based backtranslation.

\subsection{Robust Speculative Non-Interference Preservation}\label{sec:rsnip}

Here, we only present a property-ful criterion for the preservation of \rsni (\Cref{def:rsnip}).
The reason is that we only directly prove that compilers do \textit{not} attain \rsnip.
This kind of proof is simple already (\Cref{def:not-rsnip}), and we do not need a property-free criterion.

\begin{definition}[\strong{\snipdef{}}]\label{def:rsnip}
	\begin{align*}
		\vdash \comp{\cdot} : \strong{\rsnip} \isdef&\ 
		\forall\src{P} \in \SR \ldotp 
		\text{ if } \vdash\src{P} : \rsni{}(\SR)
		\text{ then }
		\vdash\comp{\src{P}} : \rsni{}(\TR)
	\end{align*}
\end{definition}
\begin{corollary}[\com{\nvdash \comp{\cdot}:\strong{\rsnip}}]\label{def:not-rsnip}
	\begin{align*}
		\nvdash \comp{\cdot} : \strong{\rsnip} \isdef&\ 
		\exists \src{P} \in \SR\ldotp 
		\text{ } \vdash\src{P} : \rsni{}(\SR)
		\text{ and }
		\nvdash\comp{\src{P}} : \rsni{}(\TR)
	\end{align*}
	Let us now unfold the corollary in order to understand what must be proven to show that a compiler is not \strong{\rsnip{}}.
	The crux is the second clause of the corollary, which gets unfolded to the following.
	Recall that low-equivalent programs simply differ in their private heap, so $\trg{\ctxt{}\hole{\comp{\src{P'}}}}$ is the same as $\trg{\ctxt{}\hole{\comp{\src{P}}}}$ but with a different private heap.
	\begin{align*}
		\nvdash\comp{\src{P}} : \rsni{}(\TR)
		=
		&\
		\exists \ctxt{} . \vdash \ctxt{} : \com{atk}
		\text{ and given } 
		\trg{\ctxt{}\hole{\comp{\src{P'}}}}\loweq \trg{\ctxt{}\hole{\comp{\src{P}}}}
		\\
		\text{ we have }
		&
		\nspecproj{\behavt{\SInit{\trg{\ctxt{}\hole{\comp{\src{P}}}}}}} = \nspecproj{\behavt{\SInitt{\trg{\ctxt{}\hole{\comp{\src{P'}}}}}}}
			\\
			\text{ and }
			&\ 
				\behavt{\SInitt{\trg{\ctxt{}\hole{\comp{\src{P}}}}}} \neq \behavt{\SInitt{\trg{\ctxt{}\hole{\comp{\src{P'}}}}}}
	\end{align*}
	That is, we need to find a  program $\src{P}$ and an attacker $\trg{A}$ that violate \rsni.
	Finding these existentially-quantified program (and attacker) may be hard. %
	Fortunately, failed attempts at proving \rssc often provide hints for how to do this.\looseness=-1 %
	\qed
\end{corollary}

We remark that the insecurity part of our methodology is used to show its completeness w.r.t. vulnerability to Spectre v1 attacks. 
Unfortunately, one still has to manually come up with the insecure counterexample and verify that it is not \rsni.

\subsection{A Methodology for Provably-(In)Secure Countermeasures}\label{sec:methodology}

To prevent speculative leaks, secure compilers should produce target programs that satisfy \rsni (cf. \Cref{sec:rsni}) whereas insecure compilers will produce some programs that fail to achieve \rsni.
In this section, we show how to combine the results from the previous sections to derive exactly these facts about compilers;
we depict this with the two chains of implications below.
The first one (1) lists the assumptions (black dashed lines) and logical steps (theorem-annotated implications) to conclude compiler security while the second one (2) lists assumptions and logical steps for compiler insecurity.
For simplicity, the diagram focuses on security definitions and compiler criteria for  \src{L} and \trg{T}.
There are similar chains of implications for \src{\wSR} and \trg{\wTR} that use \Cref{thm:rss-overapp-rsni-weak} instead of \Cref{thm:rss-overapp-rsni}.
\begin{center}
		\tikzpic{
			\node[draw, rounded corners, dashed](srss){$\vdash \src{P} : \rss{}(\src{L})$ };
			\node[below = of srss](trss){$\vdash \comp{\src{P}} : \rss{}(\trg{T}) $};
			\node[above = of srss.-150](n1){$\forall \src{P} \in \SR$};
			\node[right = .6 of trss](trsni){$\vdash \comp{\src{P}} : \rsni{}(\trg{\TR})$};

			\draw[=,double,-implies,double equal sign distance] (srss) to node[midway, right, font = \footnotesize, xshift=.5em](i1){ $\vdash \comp{\cdot} : \strong{\rssp}$ } (trss);
			\draw[=,double,-implies,double equal sign distance] (trss) to node[midway, below, font = \scriptsize,yshift=-.3em](i2){\Cref{thm:rss-overapp-rsni}} (trsni);

			\node[draw, rounded corners, dashed, above = of i1.30, font = \footnotesize](i0){ $\vdash \comp{\cdot} : \strong{\rssc}$ };
			\draw[=,double,-implies,double equal sign distance] (i0) to node[midway, right, font = \scriptsize](i){ \Cref{thm:rdssp-rdss-eq} } (i1.30);

			\node[draw, rounded corners, dashed, right = 3.5 of srss](srss2){$\vdash \src{P} : \rsni{}(\SR)$};
			\node[above = of srss2.-150,xshift=1em](n2){$\exists \src{P} \in \SR$};
			\node[below = of srss2,yshift=-.5em](trss2){${\nvdash} \comp{\src{P}} : \rsni{}(\TR)$ };

			\draw[=,double,-implies,double equal sign distance] (srss2.-140) to node[draw, rounded corners, dashed, midway, right, font = \footnotesize, xshift=.5em](i12){ $\nvdash \comp{\cdot} : \strong{\rsnip}$ } (srss2.-140|-trss2.90);

			\node[ left = .05 of trss2, yshift = -1.8em] (ss1) {};
			\node[yshift = 1em]at (ss1 |- n2) (ss2) {};
			\draw[-,draw] (ss1) to node[pos=1,left, font = \scriptsize, align = left] (o1) {(1)} (ss2);
			\draw[] (ss2) to node[pos=0,right,font=\scriptsize,align=left] (o2) {(2)} (ss1);

		}
\end{center}
To show security (1), we need to prove that any compiled component is \rsni in the target language.
Rather than directly reasoning about \rsni, we rely on \rss, which over-approximates \rsni (cf. \Cref{thm:rss-overapp-rsni}).
This significantly simplifies our security proofs since it allows us to reason about single traces rather than pairs of traces.
Thus, it suffices to show that any compiled component is \rss in the target.
This can be obtained by (i) an \strong{\rssp{}} compiler so long as (ii) any \src{P} is \rss in the source.
By \Cref{thm:rdssp-rdss-eq}, for point (i) it is sufficient to show that the compiler is \strong{\rssc{}}.
Point (ii) holds for any \src{P} (\Cref{thm:all-s-rdss}).
This direction highlights how \rss really is a working security definition that simplifies proving the more precise, semantic security definition which is \rsni.

To show insecurity (2), we need to prove that there exists a compiled component that is \emph{not} \rsni in the target language.
For this, we show (A) that the compiler is not \strong{\rsnip{}} given that (B) the source component \src{P} was \rsni in the source.
To show (A), we follow \Cref{def:not-rsnip}, whereas point (B) holds for any  \src{P}  (\Cref{thm:all-s-rdss}).

Our security criteria, instantiated for the strong (\src{L}-\trg{T}) and weak (\src{\wSR}-\trg{\wTR}) languages, provide a way of characterizing the security guarantees of any countermeasure \comp{\cdot}, which is what we do next.
In particular,  showing that \comp{\cdot} is \strong{\rssc{}} ensures that compiled code has no speculative leaks.
Similarly, showing that \comp{\cdot} is \weak{\rssc{}} (\emph{and not} \strong{\rsnip{}}) ensures that compiled code does not leak information about speculatively-accessed data, i.e., it would prevent Spectre attacks.
Finally, showing that \comp{\cdot} is not \weak{\rsnip{}} implies that compiled code leaks  speculatively accessed data, like in Spectre attacks. %

\paragraph{Preservation or Enforcement?}
\rsnip and \rssp focus on  \emph{preserving} the related security property.
Since their premise is always satisfied, we could also state them  in terms of \emph{enforcing}  \rsni and \rss over compiled programs.
We choose against this to be able to reuse established compiler theory~\cite{Leroy09b}, and since it is unclear how to prove \Cref{thm:rdssp-rdss-eq} with enforcement statements.
\section{Countermeasures Analysis %
}\label{sec:frame-instances-insec} %

In this section, we characterise the security guarantees of the main Spectre v1 countermeasures implemented by compiler vendors: insertion of speculation barriers (\texttt{lfence}) and speculative load hardening (\texttt{slh}).
For this, we develop formal models that capture the key aspects of these countermeasures as implemented by the Microsoft Visual C++ compiler~\cite{microsoft} (MSVC, \Cref{sec:countermeasures:microsoft}), the Intel C++ compiler~\cite{Intel-compiler} (ICC, \Cref{sec:countermeasures:intel}), and the Clang compiler (\Cref{sec:countermeasures:clang}), and we analyze their guarantees using our secure compilation criteria.
We continue the section with an overview of our proofs (\Cref{sec:countermeasures:proving}).
We conclude by discussing our analysis' results (\Cref{sec:countermeasures:summary}). %
For space constraints, compiled snippets, their formalisation, and full security proofs can be found in~\cite{guarnieri2019exorcising}.

\subsection{MSVC is Insecure}\label{sec:countermeasures:microsoft} 
\lstset{language=Asm}
Inserting speculation barriers---the \lstinline{lfence} x86 instruction---after branch instructions is a simple  countermeasure against Spectre v1~\cite{Intel,Intel-compiler,microsoft}. %
This instruction stops speculative execution at the price of significant performance overhead.

MSVC implements a countermeasure that tries to minimize the number of \lstinline{lfence}s by selectively determining which branches to patch~\cite{microsoft}.\footnote{The countermeasure can be activated with the \texttt{/Qspectre} flag.} 
However, MSVC fails in inserting some necessary \lstinline{lfence}s, thereby producing  insecure code that is not  \rsni{}(\trg{\wTR}) and that is vulnerable to Spectre-style attacks.

\lstset{language=Java}
To show this, we follow \Cref{def:not-rsnip} and provide a program that is \rsni{}(\src{\wSR}) and its compilation  is not \rsni{}(\trg{\wTR}).
The program we consider, which is \rsni{}(\src{\wSR}) (\Cref{thm:all-s-rdss}), is given in \Cref{lis:spectre-10-listing}.
\begin{lstlisting}[mathescape,label=lis:spectre-10-listing,caption={A variant of the classic Spectre v1 snippet (Example~10 from~\cite{kocher2018examples}).}] 
void get (int y) 
	if (y < size) then
		if (A[y] == 0) then
			temp = B[0];
\end{lstlisting}
\lstset{language=Asm}
As shown in~\cite{kocher2018examples,spectector}, MSVC fails in injecting an \lstinline{lfence} after the first branch instruction.
As a result, the compiled target program is identical to \Cref{lis:spectre-10-listing}, and it  speculatively leaks  whether \lstinline{A[y]} is \lstinline{0} through the branch statement in line 3, i.e.,  it violates \rsni{}(\trg{\wTR}).
We refer to~\cite{kocher2018examples,spectector} for additional examples of MSVC's insecurity.

\subsection{ICC is Secure} %
\label{sec:countermeasures:intel}

The Intel C++ compiler also implements a countermeasure that inserts \lstinline{lfence}s after each branch instruction~\cite{Intel-compiler}.%
\footnote{The countermeasure can be activated with flag: \texttt{-mconditional-branch=all-fix}}

We model this countermeasure with \complfence{\cdot}, a homomorphic compiler that takes a component in \SR and translates all of its subparts to \TR.
Its key feature is inserting an \trg{\lfence} statement at the beginning of every \trg{then} and \trg{else} branch of compiled code.  
All other statements are left unmodified by the compiler.
\begin{align*}
	\complfence{ \ifzte{e}{s}{s'} }\!\! &= \!
		\trg{ 
			\ifztet{
				\complfence{e}\!
			}{ 
				\!\trg{\{\lfence; \complfence{s} \}}\!
			}{
				\!\trg{\{\lfence; \complfence{s'}\}}\!
			}
		}
\end{align*}
It should come at no surprise that \complfence{\cdot} is \strong{\rssc{}} (\Cref{thm:lfence-comp-rdss}).
In \TR, the only source of speculation are branches (\Cref{tr:et-sp-if}) but any branch, whether it evaluates to true or false, will execute an \trg{\lfence} (\Cref{tr:et-sp-lf}), triggering a rollback (\Cref{tr:et-sp-rb}).
Since compiled code performs no action during speculation, it can only perform actions when the program counter is tainted as \trg{\safeta}, which makes all actions \trg{\safeta}.
These actions are easy to relate to their source-level counterparts since they are generated according to the non-speculative semantics.
\begin{theorem}[ICC is secure for \SR-\TR]\label{thm:lfence-comp-rdss}
$		\vdash \complfence{\cdot} : \strong{\rssc} $
\end{theorem}

\subsection{Speculative Load Hardening}\label{sec:countermeasures:clang}

Clang implements a countermeasure called speculative load hardening~\cite{spec-hard} (SLH) that works as follows:%
\footnote{
	SLH can be activated with flag: \texttt{-mllvm -x86-speculative-load-hardening}
}

\begin{asparaitem}
\item Compiled code keeps track of a \emph{predicate bit} that records whether the processor is mis-speculating (predicate bit set to \trg{1}) or not (predicate bit set to \trg{0}).
This is done by replicating the behaviour of all branch instructions using branch-less \lstinline{cmov} instructions, which do not trigger speculation.
SLH-compiled code tracks the predicate bit inter-procedurally by storing it into the most-significant bits of the stack pointer register, which are always unused.
Note that when all speculative transactions have been rolled back, the predicate bit is reset to \trg{0} by the rollback capabilities of the processor.

\item Compiled code uses the predicate bit to initialise a mask whose usage is detailed below.
At the beginning of a function, SLH-compiled code retrieves the predicate bit from the stack and uses it to initialize a mask either to \trg{0xF..F} if predicate bit is \trg{1} or to \trg{0x0..0} otherwise.
During the computation, SLH-compiled code uses \lstinline{cmov} instructions to conditionally update the mask and preserve the invariant that $\mathit{mask} = \trg{0xF..F}$ if code is mis-speculating and $\mathit{mask} = \trg{0x0..0}$  otherwise.
Before returning from a function, SLH-compiled code pushes the most-significant bit of the current mask to the stack; thereby preserving the predicate bit.

\item All inputs to control-flow and store instructions are hardened by masking their values with $\mathit{mask}$ (i.e., by \lstinline{or}-ing their value with $\mathit{mask}$).
That is, whenever code is mis-speculating (i.e., $\mathit{mask} = \trg{0xF..F}$) the inputs to these statements are ``F-ed'' to $\trg{0xF..F}$, otherwise they are left unchanged.
This prevents speculative leaks through control-flow and store statements.

\item The outputs of memory loads instructions are hardened by \lstinline{or}-ing their value with $\mathit{mask}$.
So, when code is mis-speculating, the result of load instructions is ``F-ed'' to \trg{0xF..F}.
This prevents leaks of speculatively-accessed memory locations.
Inputs to load instructions, however, are \emph{not} masked.
\end{asparaitem}

In the following, we analyse the security guarantees of SLH.

\subsubsection{SLH is not $\strong{\rsnip}$}
We show that SLH is not $\strong{\rsnip}$, i.e., it does not preserve (strong) speculative non-interference and thus it allows speculative leaks of data retrieved non-speculatively.

Following \Cref{def:not-rsnip}, we do this by providing a program that is $\rsni{}(\SR)$ and that is compiled to a program that is not $\rsni{}(\TR)$.
The program in \Cref{lis:spectre-variant-listing} differs from \Cref{lis:spectrev1} in that the first memory access is performed non-speculatively (line~2).
\lstset{language=Java}
\begin{lstlisting}[mathescape,label=lis:spectre-variant-listing,caption={Another variant of the classic Spectre v1 snippet.}] 
void get (int y) 
	x = A[y];
	if (y < size) then
		temp = B[x];
\end{lstlisting}
\lstset{language=Asm}
In its compilation, SLH hardens the value of \trg{\lstinline{A[y]}} using the mask retrieved from the stack pointer.
When the \lstinline{get} function is invoked non-speculatively, the mask is set to \trg{0x0..0} and \trg{\lstinline{A[y]}} is not masked.
Thus, speculatively-executing the load in (the compiled counterpart of) line 4  leaks the value of \trg{\lstinline{A[y]}}, which might  differ for  low-equivalent states, and violates $\rsni{}(\TR)$.
\subsubsection{SLH is $\weak{\rdss}$}\label{sec:slh:weak:rssc}
We now show that SLH is $\weak{\rdss}$, that is, it 
prevents leaks of speculatively-accessed data.

We formalise SLH using the $\compslh{\cdot }$ compiler, whose most interesting cases are given in the top of \Cref{fig:sslh}.
The compiler takes components in \src{\wSR} and outputs compiled code in \trg{\wTR}. 
The compiler keeps track of the predicate bit in a cross-procedural way, masks inputs to control-flow and store instructions, and masks outputs of load instructions as described before.

\myfigonecol{
	\begin{gather*}
	\begin{aligned}
		&
		\begin{aligned}
			\compslh{ H ; \OB{F} ; \OB{I}} &= \trg{  \compslh{H}\cup(-1\mapsto 1 : \safeta) ; \compslh{\OB{F}} ; \compslh{\OB{I}}}
			&
			\compslh{H , -n\mapsto v:\unta} &= \trg{\compslh{H}, -\compslh{n}-1\mapsto\compslh{v}:\unta}
		\end{aligned}
	\\
	\compslh{ 
			\ifztes{
				\src{e}
			}{
				\src{s}
			}{\src{s'}}
	} &= \trg{ 
			\letint{\trg{x_f}}{\compslh{e}}{
				\letreadpt{\trg{\predState}}{\trg{-1}}{
					\cmovet{\trg{x_f}}{\trg{0}}{\trg{\predState}}{
						\ifztet{\trg{x_f}
							}{ 
							\asgnpt{\trg{-1}}{\trg{\predState \vee \neg x_f}} ;
							\compslh{s}
						}{
							\asgnpt{\trg{-1}}{\trg{\predState \vee x_f}} ;
							\compslh{s'}
						}
					}
				}
			}
	} 
	\\
	\compslh{ \asgnp{e}{e'} } &= 
		\trg{
				\letint{\trg{x_f}}{\compslh{e}\trg{+1}}{
					\letint{\trg{x_f'}}{\compslh{e'}}{
						\letreadpt{\trg{\predState}}{\trg{-1}}{
							\cmovet{\trg{x_f}}{\trg{0}}{\trg{\predState}}{ 
								\cmovet{\trg{x_f'}}{\trg{0}}{\trg{\predState}}{ 
									\trg{\asgnp{x_f}{x_f'}}
								}
							}
						}
					}
				}
		}
	\\
	\compslh{
			\letreadps{
				\src{x}
			}{
				\src{e}
			}{
				\src{s}
			}
	} &= 
		\trg{
				\letint{\trg{x_f}}{\compslh{e}\trg{+1}}{
					\letreadpt{\trg{\predState}}{\trg{-1}}{
						\letreadpt{\trg{x}}{\trg{x_f}}{
							\cmovet{\trg{x}}{\trg{0}}{\trg{\predState}}{
								\compslh{s}
							}
						}
					}	
				}	
		} 
	\end{aligned}
\\
	\rule{\textwidth}{0.4pt}
\\
	\begin{aligned}
	\compsslh{
			\letreadps{
				\src{x}
			}{
				\src{e}
			}{
				\src{s}
			}
	} &= 
		\trg{
				\letint{\trg{x_f}}{\compsslh{e}\trg{+1}}{
					\letreadpt{\trg{\predState}}{\trg{-1}}{
						\cmovet{\trg{x_f}}{\trg{0}}{\trg{\predState}}{
							\letreadpt{\trg{x}}{\trg{x_f}}{\compsslh{s}}
						}
					}	
				}	
		} 
	\end{aligned}
\end{gather*}
}{sslh}{Key bits of the SLH compiler \compslh{\cdot} (above). The SSLH compiler $\compsslh{\cdot}$ (below) differs in the compilation of memory reads.}

Since the stack pointer is not accessible from an attacker residing in another process, $\compslh{\cdot }$ tracks the predicate bit  in the first location of the private heap which attackers cannot access.
So location \trg{-1} is initialised to \trg{1} (false) and updated to \trg{0} whenever we are speculating. %
Compiled code must update the predicate bit right after the \trg{then} and \trg{else} branches (statements \trg{\asgnp{-1}{\cdots}}).
Since location \trg{-1} is reserved for the predicate bit, all private memory accesses and the private heap are shifted by~1.

Several statements may leak information to the attacker: calling attacker functions, reading and writing the public and private heap,  and branching.
For function calls, memory writes, and branch instructions, $\compslh{\cdot }$ masks the input to these statement.
That is, we evaluate the sub-expressions used in those statements and store them in auxiliary variables (called \trg{x_f}).
Then, we look up the predicate bit (via statement \trg{\letreadpt{{\predState}}{{-1}}{\cdots}}) and store it in variable \trg{\predState}.
Finally, using the conditional assignment, we set the result of those expressions to \trg{0} (tainted \trg{\safeta} as all constants) if the predicate bit is \trg{0} (true).
In contrast, for memory reads, $\compslh{\cdot }$ masks the output of these statement
based on the predicate bit stored in \trg{\predState}.

As stated in \Cref{thm:slh-comp-rdss}, programs compiled with SLH are \rss{}(\wTR) and, therefore, \rsni{}(\wTR) (\Cref{thm:rss-overapp-rsni}).
Hence, they are free of leaks of speculatively-accessed data, which is sufficient to stop Spectre-style leaks like those in \Cref{lis:spectrev1}.

\begin{theorem}[SLH is secure for \wSR-\wTR]\label{thm:slh-comp-rdss}
$		\vdash \compslh{\cdot} : \weak{\rdss}$
\end{theorem}
$\compslh{\cdot}$ is $\weak{\rdss}$ 
for two reasons.
First, location \trg{-1} (and thus variable \trg{pr} where its contents are loaded) always correctly tracks whether speculation is ongoing or not.
This holds because location \trg{-1} and \trg{pr} cannot be tampered by the attacker, the compiler initializes \trg{-1} correctly, and the assignments right after the branches correctly update location \trg{-1} (via the negation of the guard \trg{x_f}).
Second, whenever speculation is happening, the result of load operations is set to a constant \trg{0} whose taint is \trg{\safeta}. %
So, computations happening during speculation either depend on data loaded non-speculatively, which are tainted as \trg{\safeta} by the taint-tracking of \trg{\wTR}, or on masked values, which are also tainted \trg{\safeta}.
Speculative actions are tainted with glb ($\lub$) of data taint (\trg{\safeta}) and pc taint (\trg{\unta}). 
Since $\safeta\lub\unta=\safeta$ (see \Cref{sec:rss}), speculative actions are tainted \trg{\safeta}, satisfying $\rss{}(\trg{\wTR})$.

\subsubsection{Making SLH More Secure}
We now show how to modify SLH to prevent \textit{all} speculative leaks. 
We do so by introducing \textit{strong SLH} (SSLH for short) that differs from  standard SLH in that it masks the input (rather than the output) of memory read operations 
(as such, we expect an implementation of SSLH to have a small overhead caused by the newly introduced data dependencies that might delay some masked loads).
We model SSLH using the $\compsslh{\cdot}$ compiler that  takes components in \src{L} and outputs compiled code in \trg{T}.
$\compsslh{\cdot}$ differs from $\compslh{\cdot}$ in how memory reads are compiled (\Cref{fig:sslh}).
The compiler  masks the input of memory loads by evaluating the sub-expressions and storing them in auxiliary variables (called \trg{x_f}), retrieving the predicate bit and storing it in variable \trg{\predState}, conditionally masking the value of \trg{x_f}, and, finally, performing the memory access using \trg{x_f} as address.

As stated in \Cref{thm:sslh-comp-rdss}, programs compiled using SSLH are \rss{}(\TR) and, thanks to \Cref{thm:rss-overapp-rsni}, \rsni{}(\TR).
Therefore, they are free of all speculative leaks.

\begin{theorem}[SSLH is secure for \SR-\TR]\label{thm:sslh-comp-rdss}
$		\vdash \compsslh{\cdot} : \strong{\rdss}$
\end{theorem}
$\compsslh{\cdot}$ satisfies $\strong{\rdss}$ for two reasons.  
First, the compiler correctly tracks whether speculation is ongoing (cf. \oldS\ref{sec:slh:weak:rssc}).
Second, whenever speculation is happening, the result of any possibly-leaking expression is set to a constant \trg{0} whose taint is  \trg{\safeta}.
That is, labels during speculation are tainted as  \trg{\safeta}, and $\rss{}(\TR)$ holds.
 \subsubsection{Non-interprocedural SLH is insecure}\label{sec:non-interp-slh}%
 We conclude by showing that the  non-interprocedural variant of SLH, where the predicate bit is set to $\trg{0}$ at the beginning of each function, is insecure and does not prevent all speculative leaks.\footnote{Flags: { \texttt{-mllvm} \texttt{-x86-speculative-load-hardening} \texttt{-mllvm} \texttt{--x86-slh-ip=false}} }
 \lstset{language=Java}
 Consider the program in \Cref{lis:spectre-variant-listing-proc} that splits the memory accesses of \lstinline{A} and \lstinline{B} of the classical Spectre v1 snippet across functions \lstinline{get} and \lstinline{get_2}.
\begin{lstlisting}[mathescape,label=lis:spectre-variant-listing-proc,caption={Inter-procedural variant of Spectre v1 snippet~\cite{crossproc}}]
void get (int y) 
	x = A[y]; 
	if (y < size) then  get_2 (x);

void get_2 (int x)  temp = B[x];
\end{lstlisting}
\lstset{language=Asm}
Once compiled, \lstinline{get} starts the speculative execution (line 3), then the compiled code corresponding to \lstinline{get_2} is executed speculatively.
However, the predicate bit of \lstinline{get_2} is set to $\trg{0}$ upon calling the function.
Hence, the memory access corresponding to \lstinline{B[x]} is not masked leading to the leak of \lstinline{x} (which contains \lstinline{A[y]}), so the target program violates $\rsni{}(\trg{\wTR})$.

It is also possible to secure the non-interprocedural variant of SLH.
We model NISLH as \compslht{\cdot} by having the predicate bit initialized at the beginning of each function to \trg{1} (false) in a local variable \trg{\predState}.
As before, compiled code updates \trg{\predState}  after every branching instruction.
To ensure that \trg{\predState} correctly captures whether we are mis-speculating, we place an \trg{\lfence} as the first instruction of every compiled function.
\begin{align*}
    \compslht{ 
        \begin{aligned}
            &
            \src{f(x)\mapsto s; }
            \\
            &\
            \src{\ret} 
        \end{aligned}
    } =&\ 
        \trg{f(x)\mapsto  
            \left|\begin{aligned}
                &
                \trg{\lfence;}\ 
                \letint{\trg{\predState}}{\trg{1}}{
                \\
                &\ 
                \compslht{s}};
                \trg{\ret}
            \end{aligned}\right.
        }
    \\
    \compslht{ 
        \begin{aligned}
            &
            \ifztes{
                \src{e}
            \\&\
            }{
                \src{s}
            \\&\
            }{\src{s'}}
        \end{aligned}
    } =&\ \trg{ 
        \left|\begin{aligned}
            &
            \letint{\trg{x_f}}{\compslht{e}}{
            \\&\
                \ifztet{\trg{x_f}
                    }{ 
                    \letint{\trg{\predState}}{\trg{\predState \vee \neg x_f}}{\compslht{s}}
                \\
                &\
                }{
                    \letint{\trg{\predState}}{\trg{\predState \vee x_f}}{\compslht{s'}}
                }
            }
        \end{aligned}\right.
    }
\end{align*}

This compiler is also \weak{\rssc{}} since (1) it correctly tracks whether we are speculating (this time using local variable \trg{\predState} rather than location \trg{-1} as in \compslh{\cdot}), (2) speculation across function boundaries is blocked by \trg{lfence} statements, and (3) masking is done as in \compslh{\cdot}.
\begin{theorem}[The NISLH compiler is \weak{\rdss{}}]\label{thm:slh-comp-rdss-ni}
$       \vdash \compslht{\cdot} \! :\! \weak{\rdss}$
\end{theorem}

In a similar way, one can construct a secure, non-interprocedural version of \compsslh{\cdot} that satisfies \strong{\rssc{}}.

\subsection{How to Prove \rssc}\label{sec:countermeasures:proving}
We now illustrate the backtranslation proof technique used to prove SLH-related countermeasures secure.
Our backtranslation is a simple adaptation of the general backtranslation proof technique~\cite{rsc-j}.
To prove that a compiler is \rdss, we \emph{backtranslate} a target attacker (\ctxt{}) to create a source attacker (\ctxs{}=\backtr{\ctxt{}}) so that they produce traces related by the relation of \Cref{sec:rob-saf-comp}.
Our backtranslation function (\backtr{\cdot}), which is the same for all proofs, homomorphically translates target heaps, functions, statements etc$\ldotp$ into source ones. %

We depict our proof approach in \Cref{fig:proof-slh}.
There, circles and contoured statements represent source and target states.
A black dotted connection between source and target states indicates that they are related; dashed target states are not related to any source state.
In our setup, execution happens either on the attacker side or on the component side, coloured connections between same-colour states represent reductions.
\myfig{
	\centering
	\tikzpic{
		\node[draw=\stlccol,rounded corners,font=\footnotesize](sc1){};
		\node[draw=\stlccol,rounded corners,font=\footnotesize, right = .5 of sc1](sc2){};
		\node[draw=\stlccol,rounded corners,font=\footnotesize, right = .7 of sc2](sc3){};
		\node[draw=\stlccol,rounded corners,font=\footnotesize, right = .5 of sc3](sc4-if){\src{ifz}};

		\node[draw=\ulccol,rounded corners,font=\footnotesize, below = .8 of sc1](tc1){};
		\node[draw=\ulccol,rounded corners,font=\footnotesize,] at(sc2|-tc1) (tc2){};
		\node[draw=\ulccol,rounded corners,font=\footnotesize,] at(sc3|-tc1) (tc3){};
		\node[draw=\ulccol,rounded corners,font=\footnotesize,] at(sc4-if|-tc1) (tc4-if){\compslh{\src{ifz}}};

		\node[draw=\ulccol,dashed,rounded corners,font=\footnotesize, below= .5 of tc4-if](ts1){};
		\node[draw=\ulccol,dashed,rounded corners,font=\footnotesize, right = .5 of ts1](ts3){};
		\node[draw=\ulccol,dashed,rounded corners,font=\footnotesize, right = .5 of ts3](ts4){};
		\node[draw=\ulccol,dashed,rounded corners,font=\footnotesize, right = .5 of ts4](ts2){\trg{w}=\trg{0}};		
		\node[font=\footnotesize, left = .5 of ts1](ph){};

		\node[draw=\ulccol,rounded corners,font=\footnotesize,] at(ts2|-tc1) (tc5){};
		\node[draw=\ulccol,rounded corners,font=\footnotesize, right = .5 of tc5](tc6){};
		\node[draw=\ulccol,rounded corners,font=\footnotesize, right = .7 of tc6](tc7){};
		\node[draw=\ulccol,rounded corners,font=\footnotesize, right = .5 of tc7](tc8){};

		\node[draw=\stlccol,rounded corners,font=\footnotesize,] at(tc5|-sc1) (sc5){};
		\node[draw=\stlccol,rounded corners,font=\footnotesize,] at(tc6|-sc1) (sc6){};
		\node[draw=\stlccol,rounded corners,font=\footnotesize,] at(tc7|-sc1) (sc7){};
		\node[draw=\stlccol,rounded corners,font=\footnotesize,] at(tc8|-sc1) (sc8){};
		\node[font=\footnotesize, ] at (ph -| sc6)(ph2){};

		\draw[dashed, draw = \stlccol] (sc1) to (sc2);
		\draw[dashed, draw = \stlccol, bend left] (sc2) to node[above,font=\footnotesize](as1){\src{\acas{?^\sigma}}} (sc3);
		\draw[dashed, draw = \stlccol] (sc3) to (sc4-if);
		\draw[->, draw = \stlccol] (sc4-if) to node[](sd0){} (sc5);
		\draw[dashed, draw = \stlccol] (sc5) to (sc6);
		\draw[dashed, draw = \stlccol, bend left] (sc6) to node[above,font=\footnotesize](as2){\src{\acas{!^\sigma}}} (sc7);
		\draw[dashed, draw = \stlccol] (sc7) to (sc8);
		\draw[draw = \stlccol, - ] ([yshift=.6em]sc3.center) to ([yshift=.9em]sc3.center) to node[above,font=\footnotesize](sd1){\src{\OB{\delta_1^\sigma}}} ([yshift=.9em]sc4-if.center) to ([yshift=.6em]sc4-if.center);
		\draw[draw = \stlccol, - ] ([yshift=.3em]sc5.north) to ([yshift=.6em]sc5.north) to node[above,font=\footnotesize](sd2){\src{\OB{\delta_2^\sigma}}} ([yshift=.6em]sc6.north) to ([yshift=.3em]sc6.north);

		\draw[dashed, draw = \ulccol] (tc1) to (tc2);
		\draw[dashed, draw = \ulccol, bend left] (tc2) to node[above,font=\footnotesize](at1){\trg{\acat{?^{\taintt}}}} (tc3);
		\draw[dashed, draw = \ulccol] (tc3) to (tc4-if);
		\draw[dashed, draw = \ulccol] (tc5) to (tc6);
		\draw[dashed, draw = \ulccol, bend left] (tc6) to node[above,font=\footnotesize](at2){\src{\acat{!^{\taintt}}}} (tc7);		
		\draw[dashed, draw = \ulccol] (tc7) to (tc8);
		\draw[draw = \ulccol, - ] ([yshift=.6em]tc3.center) to ([yshift=.9em]tc3.center) to node[above,font=\footnotesize](td1){\trg{\OB{\trgb{\delta}_1^{\taintt}}}} ([yshift=.9em]tc4-if.center) to ([yshift=.6em]tc4-if.center);
		\draw[draw = \ulccol, - ] ([yshift=.3em]tc5.north) to ([yshift=.6em]tc5.north) to node[above,font=\footnotesize](td2){\trg{\OB{\trgb{\delta}_2^{\taintt}}}} ([yshift=.6em]tc6.north) to ([yshift=.3em]tc6.north);

		\draw[draw = \ulccol, - ] ([yshift=-.6em]ts1.center) to ([yshift=-.9em]ts1.center) to node[below,font=\footnotesize](td0){\trat{^{\taintt}}} ([yshift=-.9em]ts2.center) to ([yshift=-.6em]ts2.center);

		\draw[draw = \ulccol, decoration={zigzag, segment length=4, amplitude=.9, post=lineto, post length=2pt},decorate,->] (tc4-if) to (ts1);
		\draw[draw = \ulccol, decoration={zigzag, segment length=4, amplitude=.9, post=lineto, post length=2pt},decorate,->] (ts1) to (ts3);
		\draw[draw = \ulccol, loosely dotted, thick ] (ts3) to (ts4);
		\draw[draw = \ulccol, decoration={zigzag, segment length=4, amplitude=.9, post=lineto, post length=2pt},decorate,->] (ts4) to (ts2);
		\draw[draw = \ulccol, decoration={zigzag, segment length=4, amplitude=.9, post=lineto, post length=2pt},decorate,->] (ts2) to  node[pos=.4,left,font=\footnotesize](ar){\trg{\rollbl}} (tc5);

		\draw[rounded corners, dotted, fill=yellow, opacity = .2 ] (as1.north) -| ([xshift=.2em]sc3.east) |- ([yshift=-.2em]tc2.south) -| ([xshift=-.2em]sc1.west) |- (as1.north);
		\draw[rounded corners, dotted, fill=yellow, opacity = .2 ] ([yshift=.2em]sc7.north) -| ([xshift=.2em]sc8.east) |- ([yshift=-.2em]tc8.south) -| ([xshift=-.2em]sc7.west) |- ([yshift=.2em]sc7.north);

		\draw[rounded corners, dotted, fill=green, opacity = .1 ] (sd1.north) -| ([xshift=.2em]sc7.east) |- ([yshift=-.2em]tc4-if.south) -| ([xshift=-.2em]sc3.west) |- (sd1.north);

		\draw[rounded corners, dotted, fill=black, opacity = .1 ] (tc4-if.south) -| ([xshift=.2em]ts2.east) |- (td0.south) -| ([xshift=-.2em]ts1.north west) |- (tc4-if.south);		

		\draw[-] ([yshift=3em]sc3.west) -- ([yshift=-1em]tc3.west);
		\draw[-] ([yshift=3em]sc7.east) -- ([yshift=-1em]tc7.east);
		\draw[-] ([yshift=-.9em,xshift=-4.5em]tc3.west) -- ([yshift=-.9em,xshift=2em]tc7.east);

		\node[font = \footnotesize,above = .5 of sd0.north](y1){\src{P} / \compslh{P} executes};
		\node[font = \footnotesize, left = 2 of y1, align = center](yl){\backtr{\ctxt{}} / \ctxt{} \\ executes};
		\node[font = \footnotesize,right = 2 of y1, align = center](yr){\backtr{\ctxt{}} / \ctxt{} \\ executes};

		\node[font = \footnotesize, align = center] at (as1 |- ts2) (yx){either \ctxt{} \\ or \compslh{P} \\ executes};

		\draw[-,thin, dotted] (sc1) to (tc1);
		\draw[-,thin, dotted] (sc2) to (tc2);
		\draw[-,thin, dotted] (sc3) to (tc3);
		\draw[-,thin, dotted] (sc4-if) to (tc4-if);
		\draw[-,thin, dotted] (sc5) to (tc5);
		\draw[-,thin, dotted] (sc6) to (tc6);
		\draw[-,thin, dotted] (sc7) to (tc7);
		\draw[-,thin, dotted] (sc8) to (tc8);
		\draw[-,thin, dotted] (as1) to (at1);
		\draw[-,thin, dotted] (as2) to (at2);
		\draw[-,thin, dotted] (sd1) to (td1);
		\draw[-,thin, dotted] (sd2) to (td2);
	}
}{proof-slh}{Diagram depicting the proof that \compslh{\cdot} is \weak{\rssc{}}.}

To prove that source and target traces are related, we set up a cross-language relation between source and target states and prove that reductions both preserve this relation and generate related traces.
The state relation we use is strong: a source state is related to a target one if the latter is a singleton stack and all the sub-part of the state are identical, i.e., heaps bind the same locations to the same values and bindings bind the same variables to the same values.
To reason about attacker reductions, we use a lock-step simulation: we show that starting from related states, if \ctxt{} does a step, then \backtr{\ctxt{}} does the same step and ends up in related states (yellow areas).
To reason about component reductions, we adapt a reasoning from compiler correctness results~\cite{Leroy09b,barthe-ct2}.
That is, if \src{s} steps and emits a trace, then \compslh{s} does one or more steps and emits a trace such that both ending states and traces are related (green areas, related traces are connected by black-dotted lines).
This proof is straightforward except for the compilation of $\ifztes{\src{e} }{ \src{s} }{\src{s'}}$ 
since it triggers speculation in \TR (grey area).
After \compslh{\ifztes{\src{e} }{ \src{s} }{\src{s'}}} is executed, speculation starts and the cross-language state relation is temporarily broken (the stack of target states is not a singleton, so the cross-language state relation cannot hold).
Speculative execution continues for \trg{w} steps in both attacker and compiled code and generating a trace \trat{^{\taintt}}.
We then prove that \trat{^{\taintt}} is related to the empty source trace because 
all actions in \trat{^{\taintt}} are tainted \trg{\safeta}, and so they do not leak.
This fact follows from proving that while speculating, bindings always contain \trg{\safeta} values and therefore any generated action is \trg{\safeta}.
In turn, this follows from proving that \trg{\predState} correctly captures if speculation is ongoing or not and that the  mask 
is \trg{\safeta}.
As mentioned, both of these hold for \compslh{\cdot} and \compsslh{\cdot}, so they are secure.

The compiler \complfence{\cdot} can be proved secure in a simpler way since speculative reductions immediately trigger an \trg{\lfence}, which rolls the speculation back (the speculation window \trg{w} is \trg{0}) reinstating the cross-language state relation right away. %

\subsection{Summary}\label{sec:countermeasures:summary}

Our security analysis is the first rigorous characterization of the security guarantees provided by Spectre v1 compiler countermeasures, and it complements existing results that focus on selected code snippets~\cite{kocher2018examples,spectector}.
The table below
depicts the results of our analysis in terms of the security properties satisfied by compiled programs.
There, \yes{} denotes that \emph{all} compiled programs satisfy the criterion and \no{} denotes that some compiled programs violates it.
\begin{center}
	\begin{tabular}{l | c c}
								&	\rsni(\TR)	&	\rsni(\wTR)	\\
	\hline
		\texttt{lfence}(MSVC), \ \ \texttt{SLH-no-interp}
			&	\no			&	\no			\\
		\texttt{lfence}(ICC)/\complfence{\cdot}, \ \ \texttt{SSLH}/\compsslh{\cdot}	
			&	\yes		&	\yes		\\
		\texttt{SLH}(Clang)/\compslh{\cdot}, \texttt{NISLH}/\compslht{\cdot}				&	\no			&	\yes		
	\end{tabular}
\end{center}
The main findings of our security analysis are summarized below: 
\begin{asparaitem}
	\item The \texttt{lfence} countermeasure implemented in MSVC, denoted  \texttt{lfence}(MSVC), is insecure.  
	It violates $\weak{\rsnip}$ and produces programs that are not speculatively non-interference, i.e., that violate both \rsni(\TR) and \rsni(\wTR).
	Hence, compiled programs still contain speculative leaks and might be vulnerable to Spectre attacks.
	
	\item The \texttt{lfence} countermeasure implemented in ICC, denoted \texttt{lfence}(ICC) and modelled by \complfence{\cdot}, is secure.
	The model satisfies $\strong{\rssp}$ (\Cref{thm:lfence-comp-rdss}) and, as a result, produces only compiled programs that satisfy speculative non-interference, that is, \rsni{}(\TR).
	Hence, compiled programs are free of speculative leaks.
	
	\item The speculative load hardening countermeasure implemented in Clang, denoted \texttt{SLH}(Clang) and modelled by \compsslh{\cdot} is secure for \wSR-\wTR.
	The model satisfies $\weak{\rssp}$ (\Cref{thm:slh-comp-rdss}) and, as a result, produces only compiled programs that satisfy weak speculative non-interference, that is, \rsni{}(\wTR).
	Hence, compiled programs are free of speculatively leaks that involve speculatively-accessed data.
	While this is sufficient for preventing Spectre-style attacks, compiled programs may still speculatively leak data retrieved non-speculatively, which might result in breaking properties like constant-time (see~\cite{contracts}).

	\item The strong variant of \texttt{SLH}, denoted \texttt{SSLH}  and modelled by \compsslh{\cdot} is secure for \SR-\TR.
	The model satisfies $\strong{\rssp}$ (\Cref{thm:sslh-comp-rdss}) and produces compiled programs that satisfy  speculative non-interference, that is, \rsni{}(\TR).
	Thus, compiled programs have \emph{no} speculative leaks.

	\item Non-interprocedural \texttt{SLH}, denoted \texttt{SLH-no-interp}, is insecure.
	It violates $\weak{\rsnip}$ and produces programs that violate both \rsni(\TR) and \rsni(\wTR).	
	Hence, compiled programs  might still be vulnerable to Spectre attacks.
	
	\item Non-interprocedural \texttt{SLH} can be made secure as we show in \Cref{sec:non-interp-slh}.
	This variant, denoted \texttt{NISLH} and modelled by \compslht{\cdot}, is secure for \wSR-\wTR and it produces programs that are free of speculatively leaks involving speculatively-accessed data.
\end{asparaitem}

\paragraph{Additional security guarantees}
In addition to \rsnip, the secure compilers \complfence{\cdot}, \compslh{\cdot}, \compsslh{\cdot}, and \compslht{\cdot} also preserve the non-speculative behavior of source programs.
That is, if two source programs $\src{W}$ and $\src{W'}$ produce the same traces, then their compiled counterparts produce traces with the same non-speculative projection.
This directly follows from the  compilers only modifying the speculative behavior of programs, either through \trg{lfence}s or conditional masking.

By combining \rsnip{} with the preservation of non-speculative behaviors, we can derive an additional security guarantee for our compilers: preservation of non-interference.
For simplicity, we only focus on whole programs and we use $\complfence{\cdot}$ as an example; the same argument applies to \compslh{\cdot}, \compsslh{\cdot}, and \compslht{\cdot}.
We say that a program \emph{$\com{W}$ is non-interferent} (NI) if all programs $\com{W'}$ that differ from $\com{W}$ only in the private heap (i.e., they are low-equivalent) produce the same traces as $\com{W}$.
Given a source program $\src{W} \in \SR$ that is NI, we obtain that \complfence{\src{W}} is NI if we restrict ourselves to the non-speculative projection of traces since \complfence{\src{W}} preserves the non-speculative behavior of $\src{W}$.
Since $\complfence{\src{W}}$ is $\rsni{}(\TR)$, the full traces do not leak more than their non-speculative projections and thus $\complfence{\src{W}}$ is also non-interferent.

The security guarantees of NI depend on the underlying language.
For strong languages \SR-\TR, NI ensures that programs are \emph{constant-time} with respect to the private heap
	(in \SR, we have classical constant-time~\cite{molnar2005program,almeida2016verifying} while in \TR we have speculative constant-time~\cite{cauligi2019towards}).
Indeed, information from the private heap cannot influence the traces where $\com{\rdl{n}}$, $\com{\wrl{n}}$, and $\com{\ifl{v}}$ actions correspond to the standard constant-time observer.
For the weak languages \wSR-\wTR, NI ensures a form of \emph{sandboxing} where programs (1)  cannot access information from the private heap non-speculatively (because reading values from the private heap violates NI through actions $\com{\rdl{n \mapsto v}}$), and (2) cannot speculatively leak information about the private heap. 
We leave exploring these additional security results as future work.

\section{Scope and Limitations of the Model}\label{sec:disc-new}

Lifting our analysis to real CPUs is only valid to the extent that our attacker model and speculative semantics capture the target system.\looseness=-1

Our attacker  observes the location of memory accesses and the outcome of control-flow statements.
This attacker model offers a good trade-off between precision and simplicity~\cite{almeida2016verifying,molnar2005program}, and it has proven to capture interesting microarchitectural leaks, like those resulting from caches and port contention.
Other classes of microarchitectural leaks, like those resulting from internal buffers~\cite{ridl2019} or hardware prefetchers~\cite{prefetcher}, might not be captured by our model.

We also assume that attackers cannot access the private heap since there can be no protection against same-process attackers.
This can be achieved by running attacker and component in separate processes and leveraging OS-level memory protection.

Finally, the semantics of our target languages are adequate to reason only about Spectre v1-style attacks.
These semantics ignore the effects of out-of-order execution.
As a result, they cannot be used to reason about countermeasures that rely only on data dependencies to restrict speculatively executed instructions~\cite{DBLP:journals/corr/abs-1805-08506}.
For a similar reason, our analysis of SLH might be too pessimistic in that the data dependencies resulting from the injected masking operations might effectively limit the scope of speculative execution.
Our semantics also ignore  other sources of speculation (e.g., indirect jumps) that are exploited by other Spectre variants, as we discuss next.

\paragraph{Beyond Spectre v1}\label{sec:disc} %
Spectre v1 (also called Spectre-PHT) is just one of the (many) Spectre variants, we recount other variants below and discuss how to extend this work to reason about them. %
\begin{asparaitem}
\item Spectre BTB~\cite{Kocher2018spectre} exploits speculation over indirect jump instructions.
The \emph{retpoline} compiler countermeasure~\cite{retpoline} replaces indirect jumps with a return-based trampoline that leads to code that perform busy waiting.
As a result, the speculated jump executes no code and thus cannot leak anything.

\item Spectre-RSB~\cite{maisuradze2018ret2spec}, in contrast, exploits speculation over return addresses (through \texttt{ret}  instructions).
To prevent it, Intel deployed a microcode update~\cite{retpoline} that renders \emph{retpoline} a valid countermeasure also against Spectre-RSB~\cite{spectreSoK}.

\item Spectre-STL~\cite{spectrev4} exploits speculation over data dependencies between in-flight store and load operations. %
To mitigate it, ARM introduced a dedicated \texttt{SSBB} speculation barrier to prevent store bypasses that could  be injected by compilers.
\end{asparaitem}

To reason about these Spectre variants, we need to extend the speculative semantics of \TR to capture the new kinds of speculative execution; this is analogous to other semantics~\cite{McIlroy19,cauligi2019towards,blade, balliu2019inspectre}.
Crucially, the traces  must capture events that are meaningful for the related variant (e.g., reads and writes for Spectre-STL, returns for Spectre-RSB).
These actions are already present in traces of \TR, so the new semantics can reuse the trace model presented here.
This, in turn, ensures that we can use the  secure compilation criteria and  trace relation from \Cref{sec:rob-saf-comp}  to reason about whether compiler-inserted countermeasures for these variants are secure or not.
Any proof that countermeasures for these variants are \rssp should follow the overview in \Cref{sec:countermeasures:proving}.
Specifically, proofs for \emph{retpoline} would follow the approach of \Cref{fig:proof-slh} since speculative execution gets diverted to code that does not produce observations (we provide an in-depth discussion on \emph{retpoline} in \Cref{sec:v2-details}).
In contrast, reasoning about \texttt{SSBB} would be similar to reasoning about \complfence{\cdot} since \texttt{SSBB}s instructions act as speculation barriers.
We leave investigating these topics in detail for future work.

\section{Related Work}\label{sec:rw}%

\paragraph{Speculative execution attacks}
Many attacks analogous to Spectre~\cite{Kocher2018spectre, kiriansky2018speculative} exist; they differ in the exploited speculation sources~\cite{maisuradze2018ret2spec, 220586, spectrev4}, the covert channels~\cite{trippel2018meltdownprime,schwarz2019netspectre, stecklina2018lazyfp}, or the target platforms~\cite{chen2018sgxpectre}.
We refer the reader to~\cite{spectreSoK} for a survey of existing attacks.

\myparagraph{Speculative semantics}
These semantics model the effects of specu\-la\-ti\-vely-executed instructions.
Several semantics~\cite{McIlroy19,cauligi2019towards,blade, balliu2019inspectre, contracts} explicitly model microarchitectural details like multiple pipeline stages, reorder buffers, caches, and  predictors. %
These semantics are significantly more complex than ours (which is inspired by~\cite{spectector}), and they would lead to much harder proofs.

\myparagraph{Security definition against Spectre attacks}
\sni~\cite{spectector} has been used as security definition against speculative leaks also by~\cite{balliu2019inspectre,contracts}.
Cheang \etal~\cite{Cheang19} propose \textit{trace property-dependent observational determinism}, a property similar to \sni.
Cauligi \etal~\cite{cauligi2019towards} present speculative constant-time (SCT), i.e., constant-time w.r.t. the speculative semantics.
Differently from \sni, SCT captures leaks under the non-speculative \textit{and} the speculative semantics, and it is inadequate for reasoning about  countermeasures that \emph{only} modify a program's speculative behaviour. 
More generally, Guarnieri \etal~\cite{contracts} presents a secure programming framework that subsumes both \sni and SCT.

\myparagraph{Compiler countermeasures for Spectre v1}
Apart from the insertion of speculation barriers~\cite{Intel,amd} and SLH~\cite{DBLP:journals/corr/abs-1805-08506, spec-hard}, few countermeasures for Spectre v1 exist.
Replacing branch instructions with branchless computations (using \texttt{cmov} and bit masking) is effective~\cite{Webkit} but not generally applicable.
oo7~\cite{oo7} is a tool that automatically patches speculative leaks by injecting speculation barriers.
However, oo7 misses some speculative leaks~\cite{spectector} 
and  violates $\weak{\rsnip}$.

Blade~\cite{blade} is a compiler countermeasure that aims at optimising compiled code performance.
It finds the minimal set of variables that need to be masked in order to eliminate paths between sources (i.e., speculative memory reads) and sinks (i.e., operations resulting in microarchitectural side-effects).
Similarly to our framework, Blade consider a source language \emph{without} speculation and a target language \emph{with} speculation and it preserves constant-time from source to target~\cite[Corollary 1]{blade}.
This is different from the compilers we study, which block (classes of) speculative leaks regardless of whether the source program is constant-time.
Blade's design relies on fine-grained barriers whose scope are single instructions.
Since these barriers are not available in current CPUs, Blade's prototype realises them via both \lstinline{lfence}s and masking.
We believe that our framework can be applied to reason about both Blade's design and prototype, but we leave this for future work.
The challenges are extending the target languages with fine-grained barriers and formalising the optimal placement of those barriers.

Recent work~\cite{kocher2018examples, spectector} studied the security of compiler countermeasures by inspecting specific compiled code snippets and detected insecurities in MSVC.
Our work extends and complements these results by providing the first rigorous characterization of these countermeasures' security guarantees.
In particular, we prove the security of countermeasures for all source programs, rather than simply detecting insecurities on specific examples.

\myparagraph{Secure compilation}
\rdss and \rdssp are instantiations of robustly-safe compilation~\cite{rsc-j,rhc,rhc-rel,catalinRSC}.
Like \cite{rsc-j,rhc-rel}, we relate source and target traces using a cross-language relation; however, our target language has a speculative semantics.
While program behaviors are sets of traces due to non-determinism in~\cite{rhc,rhc-rel}, behaviors are single traces for our (deterministic) languages~\cite{Leroy09b}.

Fully abstract compilation (\facdef) is a  widely used secure compilation criterion ~\cite{PatrignaniASJCP15,JuglaretHAEP16,fstar2js,fab,domicapsj,scsurvey}.
\fac compilers must preserve (and reflect) observational equivalence of source programs in their compiled counterparts~\cite{DBLP:conf/icalp/Abadi98,scsurvey}.
While \fac has been used to reason about microarchitectural side-effects~\cite{busi}, it is unclear whether \fac is well-suited for   speculative leaks as it would require explicitly modelling microarchitectural components that are modified speculatively (like caches). %

Constant-time-preserving compilation (\ctpdef) has been used to show that compilers preserve constant-time~\cite{barthe-ct2,barthe-ct1,jasmin}.
Similarly to \rsnip, proving \ctp requires proving the preservation of a hypersafety property, which is more challenging than preserving safety properties like \rss.
Additionally, \ctp has been devised for whole programs only (like \sni), and it cannot be used to reason about countermeasures like SLH that do not preserve constant-time.

\myparagraph{Verifying Hypersafety as Safety}\label{sec:hyper-to-safe}
Verifying if a program satisfies a 2-hypersafety property~\cite{ClarksonS10} (like \rsni) is notoriously challenging. %
Approaches for this include taint-tracking~\cite{ni-taint1,ni-taint2} (which over-approximates the 2-hypersafety property with a safety property), secure multi-execution~\cite{smu} (which runs the code twice in parallel) and self-composition~\cite{barthe-smc,aiken-smc} (which runs the code twice sequentially).
Our criteria leverage taint-tracking (\rss); we leave investigating criteria  based on the other approaches as future work. 

\section{Conclusion}\label{sec:conc}
The paper presented a comprehensive and precise characterization of the security guarantees of  compiler countermeasures against Spectre v1, as well as the first proofs of security for such countermeasures.
For this, it introduced \ss, a safety property implying the absence of (classes of) speculative leaks.
\ss provides precise security guarantees in that it can be instantiated to over-approximate both  strong~\cite{spectector} and weak~\cite{contracts} \sni, and it is tailored towards simplifying secure compilation proofs.
As a basis for security proofs, the paper formalised secure compilation criteria capturing the robust preservation of  \ss and \sni. %

{\small\subsubsection*{Acknowledgements}
This work was partially supported by the German Federal Ministry of Education and Research (BMBF) through funding for the CISPA-Stanford Center for Cybersecurity (FKZ: 13N1S0762), 
by the Community of Madrid under the  project S2018/TCS-4339 BLOQUES and the Atracci\'on de Talento Investigador grant 2018-T2/TIC-11732A,
by the Spanish Ministry of Science, Innovation, and University under the project RTI2018-102043-B-I00 SCUM and the Juan de la Cierva-Formaci\'on grant FJC2018-036513-I, and by a gift from Intel Corporation.
}

\newpage

\bibliographystyle{ACM-Reference-Format}
\bibliography{./../biblio.bib}

\newpage
\onecolumn
\appendix 
\crefalias{section}{appendix}

\section{Taint Tracking Overview}\label{sec:taint-app}
The language semantics we devise contains two kinds of semantics that operate in parallel: the operational semantics, presented in the paper, and the taint tracking semantics, presented here.
Thus, technically, the top-level semantics is parametric in  the taint tracking semantics. %
The semantics of strong languages \SR and \TR uses the strong form of taint tracking while the semantics of weak languages \wSR and \wTR uses the weak form of taint tracking.
We now give an in-depth overview of our taint-tracking semantics; see~\cite{guarnieri2019exorcising} for the full models. %

To add taint-tracking to our semantics, we enrich the program state with taint information and devise a taint-tracking semantics that determines how taint is propagated.
The top-level semantic judgement is then expressed in terms of the extended program states.
An extended state steps if its operational part steps according to the semantics of \Cref{sec:non-spec} and if its taint part steps according to the rules of the taint semantics.

We now define all the elements needed to define the extended program states: extended heaps and extended bindings.
In this appendix, we indicate the heap, state, and bindings used by the operational semantics with a $v$ suffix, so the \com{H}, \com{\Omega} and \com{B} from \Cref{sec:non-spec} are denoted as \com{\Hv}, \com{\Ov} and \com{\Bv} respectively.
Formally, we indicate taint as $\sigma\bnfdef \safeta \mid \unta$.
Extended heaps \He extend heaps with the taint of each location, whereas taint heaps \Hs only track the taint.
Extended heaps \com{\He} can be split/merged in their value-only part \com{\Hv} (used for the language semantics) and their taint-only part \com{\Hs} (used for taint-tracking).
We denote this split as $\com{\He} \equiv \com{\Hv+\Hs}$. 
Just like heaps, extended variable bindings \com{\Be} extend the binding with the taint of the variable, whereas taint bindings \Bs only track the taint.
Still like heaps, bindings can be split/merged as $\com{\Be}\equiv\com{\Bv + \Bs}$.

\begin{align*}
	\mi{Extended\ Heaps}~\com{\He} \bnfdef&\ \come \mid \com{\He ; n\mapsto v : \sigma} \quad\text{ where } \com{n}\in\mb{Z}
	\\
	\mi{Taint\ Heaps}~\com{\Hs} \bnfdef&\ \come \mid \com{\Hs ; n\mapsto \sigma} \quad\text{ where } \com{n}\in\mb{Z}
	\\
	\mi{Extended\ Bindings}~\com{\Be} \bnfdef&\ \come \mid \com{\Be; x\mapsto v:\sigma}
	\\
	\mi{Taint\ Bindings}~\com{\Bs} \bnfdef&\ \come \mid \com{\Bs; x\mapsto \sigma}
	\\
	\mi{Exended\ Prog.\ States}~\com{\Oe}\bnfdef&\ \com{C, \He, \OB{\Be}\triangleright \proc{s}{\OB{f}} } 
	\\
	\mi{Taint\ States}~\com{\Os}\bnfdef&\ \com{C, \Hs, \OB{\Bv}\triangleright \proc{s}{\OB{f}} } 
\end{align*}

The taint semantics follows two judgements: 
\begin{compactitem}
\item Judgment	\com{\Bt \triangleright e \bigred \sigma} reads as ``expression \com{e} is tainted as \com{\sigma} according to the variable taints \Bt''.
\item Judgement \com{\sigma; \Ot \xto{\sigma'} \Ot'} reads as ``when the pc has taint \com{\sigma}, state \com{\Ot} single-steps to \com{\Ot'} producing a (possibly empty) action with taint \com{\sigma'}''.
\end{compactitem}
Below are the most representative rules for the taint tracking used by strong languages: %
\begin{center}\small
	\typerule{T-write-prv}{
		\com{\Be \triangleright e\bigred n : \sigma}
		&
		\com{\Be \triangleright e'\bigred \_ : \sigma''}
		&
		\com{\Hs'}=\com{\Hs \cup -\abs{n}\mapsto \sigma'' }
	}{
		\com{\sigma_{pc}; C, \Hs, \OB{\Be}\cdot \Be \triangleright \asgnp{e}{e'} \xto{\sigma\lub\sigma_{pc}} C, \Hs', \OB{\Be}\cdot \Be \triangleright \skipc} 
	}{tus-up-com-p}
	\typerule{T-read-prv}{
		\com{B \triangleright e\bigred n : \sigma'}
		&
		\com{n_a} = \com{-\abs{n}}
		&
		\com{\Hs}(n_a) = \sigma''
		&
		\com{\sigma} = \com{\sigma''\glb\sigma'}
	}{
		\begin{multlined}
			\com{\sigma_{pc}; C, \Hs, \OB{\Be}\cdot \Be \triangleright \letreadp{x}{e}{s} \xto{\sigma\lub\sigma_{pc}}}
				\\
				\com{C, \Hs, \OB{\Be}\cdot \Be\cup x\mapsto 0 : \unta \triangleright s}
		\end{multlined}
	}{tus-rd-com-p-strong}
\end{center}
Writing to the private heap (\Cref{tr:tus-up-com-p}) taints the location (\com{-\abs{n}}) with the taint of the written expression (\com{\sigma''}).
In contrast, reading from the private heap (\Cref{tr:tus-rd-com-p-strong}) taints the variable where the content is stored as unsafe (\com{\unta}) and the  read value is set to \com{0} (this information is not used by the taint-tracking).

For taint-tracking of the weak languages, we replace \Cref{tr:tus-rd-com-p-strong} with the one below that taints the read variable  with the glb of the taints of the pc and of the  read value (\com{\sigma'\lub\sigma_{pc}}) instead of  \com{\unta}.
\begin{center}\small
	\typerule{T-read-prv-weak}{
		\com{B \triangleright e\bigred n : \sigma'}
		&
		\com{n_a} = \com{-\abs{n}}
		&
		\com{\Hs}(n_a) = \sigma''
		&
		\com{\sigma} = \com{\sigma''\glb\sigma'}
	}{
		\begin{multlined}
			\com{\sigma_{pc}; C, \Hs, \OB{B}\cdot B \triangleright \letreadp{x}{e}{s} \xto{\sigma\lub\sigma_{pc}}}
				\\
				\com{C, \Hs, \OB{B}\cdot B\cup x\mapsto 0 : \sigma' \lub \sigma_{pc} \triangleright s}
		\end{multlined}
	}{tus-rd-com-p-weak}
\end{center}

To correctly taint memory accesses, we need to evaluate expression \com{e} to derive the accessed location \com{\abs{n}}; see, for instance, \Cref{tr:tus-up-com-p}.
This is why taint-tracking states \com{\Ot} contain the full stack of bindings \com{\Bv} and not just the taints \com{\Bt}.
The rules above rely on a judgement \com{\Be \triangleright e\bigred n : \sigma} which is obtained by joining the result of the expression semantics on the values of \com{\Be} and of the taint-tracking semantics on the taints of \com{\Be}.
\begin{center}\small
	\typerule{Combine-B}{
		\com{\Bv + \Bs} \equiv \com{\Be}
		&
		\com{\Bv \triangleright e \bigred v}
		&
		\com{\Bs \triangleright e \bigred \sigma}
	}{
		\com{\Be \triangleright e \bigred v : \sigma}
	}{e-comb}
\end{center}

The operational and taint single-steps from \Cref{sec:non-spec} are combined according to the judgement $\src{\Oe} \xtos{\lambda{^{\sigma}}} \src{\Oe'}$ below.
\begin{center}\small
	\typerule{Combine-s-\SR}{
		\src{\Ov + \Os} \equiv \src{\Oe}
		&
		\src{\Ov' + \Os'} \equiv \src{\Oe'}
		&
		\src{\Ov \xto{\lambda} \Ov'}
		&
		\src{\safeta ; \Os \xto{\sigma} \Os'}
	}{
		\src{\Oe \xto{\lambda^{\sigma}} \Oe'}
	}{s-comb}

	\typerule{Merge-$\Omega$}{
		\com{\Hv+\Hs}\equiv\com{\He}
		&
		\com{\OB{\Bv'} + \OB{\Bs}} \equiv \com{\OB{\Be}}
		&
		\com{\OB{\Bv}+\OB{\Bs}} \equiv \com{\OB{\Be'}}
	}{
		\com{C;\Hv;\OB{\Bv}\triangleright s + C;\Hs;\OB{\Be}\triangleright s'} \equiv \com{C ; \He ; \OB{\Be'} \triangleright s}
	}{merge-s}
\end{center}
The operational semantics determines how  states reduce (\src{\Ov \xto{\lambda} \Ov'}), whereas the taint-tracking semantics determines the action's label and how taints are updated (\src{\safeta ; \Os \xto{\sigma} \Os'}).
As already mentioned, the pc taint is always safe since there is no speculation in \SR.
Moreover, merging states \src{\Ov + \Os} results in ignoring the value information accumulated in \src{\Os} since  we rely on the computation performed by the operational semantics for values (\Cref{tr:merge-s}).

In the speculative semantics, as for the non-speculative one, we decouple the operational aspects from the taint-tracking ones.
At the top level, speculative program states (\trg{\Se}) are defined as stacks of extended speculation instances (\trg{\Pe}), which can be merged/split in their operational (\trg{\Pv}) and taint (\trg{\Pt}) sub-parts.
The operational part (\trg{\Pv}) was presented in \Cref{sec:modelling_speculative_execution}.
The taint part (\trg{\Pt}) keeps track of the taint part of the program state (\trg{\Ot}) and the taint of the pc (\taintt).
As before,  \trg{\Pv} and \trg{\Pt} can be split/merged as  $\trg{\Pe} \equiv \trg{\Pv + \Pt}$.
\begin{align*}
	\mi{Speculative\ States}~\trg{\Se} \bnfdef&\ \trg{\OB{\Pe}}
	\\
	\mi{Extended\ Speculation\ Instance}~\trg{\Pe} \bnfdef&\ \trg{(\Oe,\trg{w},\trgb{\taintt})}
	\\
	\mi{Speculation\ Instance\ Taint}~\trg{\Pt} \bnfdef&\ \trg{({\Ot, \taintt})}
\end{align*}

In the taint tracking used by the speculative semantics, similarly to the operational one, reductions happen at the top of the stack: $\trg{\OB{\Pt} \xltot{\trgb{\sigma}} \OB{\Pt'}}$.
Selected  rules are  below:
\begin{center}\small
	\typerule{T-\TR-speculate-action}{
		\trg{\taintt'; \Ot \xtot{\taintt} \Ot'}	
		&
		\trg{\Ot}\equiv\trg{C, \Hs, \OB{B} \triangleright s;s'}
		\\
		\trg{s}\not\equiv\trg{\ifzte{\_}{\_}{\_}} \text{ and } \trg{s}\not\equiv\trg{\lfence}
	}{
		\trg{\OB{\Pt} \cdot (\Ot,\taintt) \xltot{ \taintt'\glb\taintt } \OB{\Pt} \cdot (\Ot',\taintt)}
	}{tt-sp-act}
	\typerule{T-\TR-speculate-if}{
		\trg{\Ot}\equiv\trg{C, \Ht, \OB{B} \cdot B \triangleright \proc{s;s'}{\OB{f}\cdot f}}
		&
		\trg{s}\equiv\trg{\ifzte{e}{s''}{s'''}}
		\\
		\trg{\taintt'; \Ot \xtot{\taintt} \Ot'}	
		&
		\trg{C}\equiv\trg{\OB{F};\OB{I}}
		&
		\trg{f}\notin\trg{\OB{I}}
		\\
		\text{ if }
			\trg{B \triangleright e\bigred 0:\taintt} 
		\text{ then }
			\trg{\Ot''}\equiv\trg{C, \Ht, \OB{B} \cdot B \triangleright s''';s'}
		\\
		\text{ if }
			\trg{B \triangleright e\bigred n:\taintt} 
			\text{ and }
			\trg{n}>\trg{0}
		\text{ then }
			\trg{\Ot''}\equiv\trg{C, \Ht, \OB{B} \cdot B \triangleright s'';s'}
	}{
		\trg{\OB{\Pt} \cdot  (\Ot,\taintt') \xltot{ \taintt\glb\taintt'} \OB{\Pt} \cdot  (\Ot',\taintt')\cdot(\Ot'',\unta)}
	}{tt-sp-if}
\end{center}

In these rules, \taintt{} is the program counter taint which is combined with the action taint \trg{\taintt'} (\Cref{tr:tt-sp-act,tr:tt-sp-if}).
Mis-speculation pushes a new state on top of the stack whose program counter is tainted \trg{\unta} denoting  the beginning of speculation (\Cref{tr:tt-sp-if}).

The two operational and taint-tracking single steps from \Cref{sec:trg} are combined in a single reduction as follows:
\begin{center}\small
	\typerule{Combine-\T}{
		\trg{\OB{\Pv} + \OB{\Pt}} \equiv \trg{\Se}
		&
		\trg{\OB{\Pv'} + \OB{\Pt'}} \equiv \trg{\Se'}
		&
		\trg{\OB{\Pv} \xltot{\trgb{\lambda}} \OB{\Pv'}}
		&
		\trg{\OB{\Pt} \xltot{\taintt} \OB{\Pt'}}
	}{
		\trg{\Se \xltot{\trgb{\lambda}^{\taintt}}  \Se' }
	}{t-comb}
\end{center}
This reduction is used by the big-step semantics $\trg{\Se \Xtot{\trat{^{\taintt}}} \Se'}$ that concatenates single labels into traces, which, as before, do not contain microarchitectural actions generated by the attacker.

\section{The Spectre v2 Case}\label{sec:v2-details}

This section describes how to apply our methodology to reason about countermeasures against the Spectre v2 attack.
The Spectre v2 attack relies on  speculation over the outcome of  indirect jumps, rather than branch instructions.
When an indirect jump is encountered, if the location where to jump is not present in the cache, heuristics are used in order to understand where to jump to.
As for the speculation over branches, these heuristics can be wrong, and when this is detected, execution is rolled back.
An attacker can therefore exploit this kind of speculative execution in order to make benign code execute malicious one.
The main countermeasure against this kind of attack is the use of a \emph{retpoline}, i.e., a \emph{ret}urn-based tram\emph{poline}.
Intuitively, the retpoline replaces indirect jumps with a return to dead code, where the program will effectively sleep until the speculation window is over.

In order to prove security of the retpoline countermeasure, we therefore need the following:
\begin{itemize}
	\item add indirect jumps to our languages and give them a regular semantics (\Cref{sec:ind-j});
	\item give a speculative reduction to jump in \TR such that the location where to jump is nondeterministically chosen; this will be the start of speculation (\Cref{sec:spec-j});
	\item change the call/return semantics in order to model retpolines, i.e., have the return address explicit (\Cref{sec:callret-explicit}).
\end{itemize}
With these changes, we can formalise a compiler that introduces the retpoline countermeasure (\Cref{sec:retpoline}) and reason about whether it is secure (\Cref{sec:sec-rp}).

\subsection{Indirect Jumps}\label{sec:ind-j}
The simplest way to add indirect jumps to our while languages is to treat function names \com{f} as natural numbers and add a statement \com{goto\ e} that jumps to function \com{f} where $\com{B\triangleright e\bigred f}$.
Additionally, we need to add the way for a component to specify private functions, i.e., functions that are not callable from the attacker.
This is still generic enough that one can model the assembly-level kind of attacks without having to add a pc to all instructions or labels to the language.

\subsection{Speculative Execution of Jumps}\label{sec:spec-j}
To focus only on speculation over jumps, we would replace \Cref{tr:et-sp-if} (handling the speculation over branch instructions) with a rule that checks that the statement being executed is a \com{goto\ e} where \com{e} evaluates to \com{f}.
In that case, the right state (jumping to \com{f}) is pushed on the stack of states, but on top of that we push another state with a jump to function \com{f'\neq f}, for a non-deterministically chosen \com{f'} that is valid.

\subsection{Explicit Call and Return Semantics}\label{sec:callret-explicit}
We need to add a return address, keep track of the return address in a stack of return addresses as well as a register where the return address can be read from.
The reason is that the retpoline countermeasure relies on another kind of speculation, the one on return addresses.
Normally, architectures push the return address on the stack and in a specific register \mtt{rsp}.
When it is time to return, if the value on top of the stack differs from that on \mtt{rsp}, speculation starts, and a return to the top of the stack is made.
When speculation ends, it is rolled back (as before, with the usual microarchitectural leaks) and a return to the value of \mtt{rsp} is done.

\subsection{The Retpoline Countermeasure}\label{sec:retpoline}
The retpoline countermeasure \comprp{\cdot} is a homomorphic compiler with a single salient case: the compilation of \src{goto\ e}, where we encode the implementation of retpolines from
Compiling a \src{goto} will not rely on target-level \trg{goto}, since they would trigger the goto-speculation and result in vulnerable code.
Instead, the compilation of \src{goto} will be turned into a call to an auxiliary function \trg{aux}.
Function \trg{aux} will change the contents of register \trg{rsp} to the function where the source \src{goto} wanted to jump.
Then, function \trg{aux} will contain code that sleeps.
This way, when the compiled \src{goto} is executed, function \trg{aux} is called and the address where to the \src{goto} should have jumped to to is pushed on the stack.
This function speculatively returns to the code that sleeps and then, when speculation ends, execution resumes from the address popped from the stack (the target of the \src{goto}). 

\subsection{Security of \comprp{\cdot}}\label{sec:sec-rp}
We believe \comprp{\cdot} is \strong{\rdss{}} and we can argue that using the same proof technique described in the paper.
As before, the key part of these proofs is reasoning when speculation happens, i.e., in the gray area of \Cref{fig:proof-slh}.
In the case of \comprp{\cdot}, we see that the only code executed during speculation is sleeping code.
Additionally, once the speculation window runs out, we need to prove that the state we end up in is the same as the source state that executed the \src{goto}.
However, this last step only amounts to proving that the retpoline is correct, i.e., that it jumps where it is supposed to.

\begin{comment}
	An interesting part of the securty against v2 is the ability of the attacker to speculate.
	%
	Here the attacker cannot speculate because effectively, 
\end{comment}

%
%
%
%
%

\section{Robustness and Attackers}\label{sec:robustness-attackers}

Typically, works that deal with Spectre attacks do not consider an active attacker, like us, but a passive one.
If we were to adopt the same view, we would have to elide the whole `robustness' aspect in our paper.
We believe that dealing with robustness and with an explicit representation of attackers has its merits (among them, applying existing secure compilation theory), and this is why we opted in favour of it.
As already mentioned, these attackers can mount confused deputy attacks~\cite{confused,confused-dg}, unlike passive ones.
Then, by using robustness we give a precise characterisation of the attacker and of its power. 
It is by having this characterisation that we can tell precisely that with a single memory shared between code and attacker, no defence mechanism is possible. 
Thus we need two memories (in the model), which gets justified in practice by saying that the attacker needs to reside in another process. 
Deriving this conclusion seems harder --if at all possible-- without a concrete notion of attacker.
Conversely, our precise definition of the attacker power also limits the scope of the attackers we can meaningfully reason about; see \Cref{sec:disc-new} for a discussion on this topic.

Thus, we need to ensure that the model faithfully comprises all attack vectors that practical attackers mounting Spectre attack rely on -- which is what we believe the model does.

Finally, this approach lets us apply existing secure compilation theory.

\newpage
\section{\SR : A Source Language Without Speculation}\label{sec:src}
\SR is a sequential, untyped while language with expressions and statements.
A \SR component (i.e., a partial program) is a collection of function definitions and imports (functions it requires of the programs it links against).
We could add more, but this suffices.
A component links against an attacker (a context) in order to create a whole program, which is then evaluated starting from its \src{main} function, which is defined in the attacker.
Expressions are given a big step semantics ($\bigreds$).
Statements are given a labelled structural operational semantics ($\xtos{}$) that records calls and returns between a component and the attacker.
The labels generated by a program are collected in a trace semantics ($\Xtos{}$), whose actions are tagged as secure or insecure.
In \SR, no speculation is possible, so only safe actions are produced; the metavariable \src{\sigma} is introduced for modularity of rules since it will be expanded in \TR.

The heap is a map from integers to values.
To prevent the attacker from operating directly on the heap of the component -- a power that is not given to him normally -- the heap consists of two parts.
The positive one is shared, while the negative one is provate to the component.
Instructions to access the private heap cannot be used in the context.

Importantly, the heap is preallocated, so all locations from 0 up to plus and minus infinity are already allocated and initialised to 0.

Call actions only contain the value passed because they model the attack where code is tricked into passing a speculatively-load value directly to the attacker.
Return actions are only needed as proof devices. 
Technically, we need two different write actions: writing on the private heap only leaks the content while writing on the public heap also leaks the value.

\subsection{Syntax}\label{sec:src-syn}
\begin{align*}
	\mi{Whole\ Programs}~\src{W} \bnfdef&\ \src{ H , \OB{F} , \OB{I}}
	\\
	\mi{Programs}~\src{P} \bnfdef&\ \src{ H , \OB{F} , \OB{I}}
	\\
	\mi{Components}~\src{C} \bnfdef&\ \src{\OB{F} , \OB{I}}
	\\
	\mi{Contexts}~\src{\ctxs{}} \bnfdef&\ \src{H , \OB{F}\hole{\cdot}}
	\\
	\mi{Imports}~\src{I} \bnfdef&\ \src{f}
	\\
	\mi{Functions}~\src{F} \bnfdef&\ \src{f(x)\mapsto s;\ret}
	\\
	\mi{Operations}~\src{\op} \bnfdef&\ \src{+} \mid \src{-} \mid \src{\cdot} %
	\\
	\mi{Comparisons}~\src{\bop} \bnfdef&\ \src{==} \mid \src{<} \mid \src{>}
	\\
	\mi{Values}~\src{v} \bnfdef&\  \src{n}\in\mb{N} %
	\\
	\mi{Expressions}~\src{e} \bnfdef&\ \src{x} \mid \src{v} \mid \src{e \op e} \mid \src{e \bop e} %
	\\
	\mi{Statements}~\src{s} \bnfdef&\ \skips \mid \src{s;s} \mid \src{\letin{x}{e}{s}} \mid \src{\ifzte{e}{s}{s}} 
		\\
		\mid&\ \src{\call{f}~e} \mid \src{\asgn{e}{e}} \mid \src{\letread{x}{e}{s}} %
			\mid \src{\letreadp{x}{e}{s}} \mid \src{\asgnp{e}{e}}
	\\
	\mi{Security\ Tags}~\src{\sigma} \bnfdef&\ \src{\safeta} \mid \src{\unta}
	\\
	\src{B} \bnfdef&\ \srce \mid \src{B ; x\mapsto v : \sigma} 
	\\
	\src{\Bv} \bnfdef&\ \srce \mid \src{\Bv ; x\mapsto v} 
	\\
	\src{\Bs} \bnfdef&\ \srce \mid \src{\Bs ; x\mapsto \sigma} 
	\\
	\src{H} \bnfdef&\ \srce \mid \src{H ; n\mapsto v : \sigma} 
	\\
	\src{\Hv} \bnfdef&\ \srce \mid \src{\Hv ; n\mapsto v} 
	\\
	\src{\Hs} \bnfdef&\ \srce \mid \src{\Hs ; n\mapsto \sigma} 
	\\
	\src{\Omega} \bnfdef&\ \src{C;H;\OB{B}\triangleright s}
	\\
	\src{\Ov} \bnfdef&\ \src{C;\Hv;\OB{\Bv}\triangleright s}
	\\
	\src{\Os} \bnfdef&\ \src{C;\Hs;\OB{B}\triangleright s}
	\\
	\mi{Labels}~\src{\lambda} \bnfdef&\ \src{\epsilon} \mid \src{\alpha} \mid \src{\delta} \mid \src{\ts}
	\\
	\mi{Actions}~\src{\alpha} \bnfdef&\ \src{(\clh{f}{v}{H})} \mid \src{(\cbh{f}{v}{H})} \mid \src{(\rth{}{H})} \mid \src{(\rbh{}{H})} 
	\\
	\mi{Heap\&Pc\ Act.s}~\src{\delta} \bnfdef&\ \src{(\rdl{n})} \mid \src{(\wrl{n})} \mid \src{(\ifl{v})} \mid \src{(\wrl{n\mapsto v})}
	\\
	\mi{Traces}~\src{\tras{^\sigma}} \bnfdef&\ \srce \mid \src{\tras{^\sigma}\cdot\alpha^\sigma} \mid \src{\tras{^\sigma}\cdot\delta^\sigma}
\end{align*}

Heaps and bindings contain the tags of the values they map.

The additional condition on a trace \src{\tras{}} is that it is list with this shape: $\src{\OB{\alpha?\OB{\delta}\alpha!}}$.
We do not filter nor reorder heap actions in traces because they represent cache-visible actions.
The attacker is assumed to operate concurrently to our program, so it can effectively observe a difference between $\src{\wrl{0}}$ and $\src{\wrl{0}}\cdot\src{\wrl{0}}$.

For simplicity of the trace semantics, reading a location is a statement (despite it being pure, and thus an expression).

In order to model conditional updates, we do not perform substitutions for variables, instead each function has its stack of bindings \src{B} where to allocate and lookup variables.
Each function can only access its stack frame for simplicity.
We concatenate whole stacks of bindings as $\src{B \cdot B'}$.
We update the bindings for \src{x} in a stack \src{B} to \src{v} by writing \src{B\cup x\mapsto v}.
If an update is made for a binding that is not in the stack, then the update just adds the binding.

The taints are safe $\safeta$ and unsafe $\unta$ and they are ordered in the usual safety lattice $\safeta\leq\unta$.
The tag of an action-generating expression is the tag of the data involved in that expression.
When data is generated (\Cref{tr:es-big-op,tr:es-big-bop}), it is tagged with the label resulting of the lub ($\sqcup$) of the label of all its subdatas.
A value (natural number, location or boolean) is safe (\Cref{tr:es-big-val}), a variable has the same tag of its content (\Cref{tr:es-big-var}).
Reading a value from the heap tags the value (and thus the variable) as unsafe (\Cref{tr:eus-rd-com}).

\subsection{Dynamic Semantics}\label{sec:src-sem}
\Cref{tr:us-aux-intern,tr:us-aux-in,tr:us-aux-out} dictate the kind of a jump between two functions: if internal to the component/attacker, in(from the attacker to the component) or out(from the component to the attacker).
\Cref{tr:plug-us} tells how to obtain a whole program from a component and an attacker.
\Cref{tr:whole-us} tells when a program is whole.
\Cref{tr:ini-us} tells the initial state of a whole program.

We change the way the list of imports is used between partial and whole programs.
For partial programs, imports are effectively imports, i.e., the functions that contexts define and that the program relies on.
So a context can define more functions.
In whole programs, we change the imports to be the list of all context defined function (\Cref{tr:plug-us}) to keep track of what is and what is not context.

\mytoprule{\text{Helpers}}
\begin{center}
	\typerule{Intfs}{
		\src{C} = \src{\OB{F} , \OB{I}}
	}{
		\src{C}.\mtt{intfs} = \src{\OB{I}}
	}{}
	\typerule{Funs}{
		\src{C} = \src{\OB{F} , \OB{I}}
	}{
		\src{C}.\mtt{funs} = \src{\OB{F}}
	}{}

	\typerule{\SR-Jump-Internal}{
		((\src{f'}\in\src{\OB{I}} \wedge \src{f}\in\src{\OB{I}}) \vee
				\\
		(\src{f'}\notin\src{\OB{I}} \wedge \src{f}\notin\src{\OB{I}}))
	}{
		\src{\OB{I}}\vdash\src{f,f'}:\src{internal}
	}{us-aux-intern}
	\typerule{\SR-Jump-IN}{
		\src{f}\in\src{\OB{I}} \wedge \src{f'}\notin\src{\OB{I}}
	}{
		\src{\OB{I}}\vdash\src{f,f'}:\src{in}
	}{us-aux-in}
	\typerule{\SR-Jump-OUT}{
		\src{f}\notin\src{\OB{I}} \wedge \src{f'}\in\src{\OB{I}}
	}{
		\src{\OB{I}}\vdash\src{f,f'}:\src{out}
	}{us-aux-out}

	\typerule{Call Label - call}{
		\src{C}.\mtt{intfs}\vdash\src{f',f} : \src{in}
	}{
		\src{C};\src{f',f} ; \src{call} ; \src{v} \vdash \src{\clh{f}{v}{}}
	}{us-aux-call}
	\typerule{Call Label - callback}{
		\src{C}.\mtt{intfs}\vdash\src{f',f} : \src{out}
	}{
		\src{C};\src{f',f} ; \src{call} ; \src{v} \vdash \src{\cbh{f}{v}{}}
	}{us-aux-callback}
	\typerule{Call Label - internal}{
		\src{C}.\mtt{intfs}\vdash\src{f',f} : \src{internal}
	}{
		\src{C};\src{f',f} ; \src{call} ; \src{v} \vdash \src{\epsilon}
	}{us-aux-call-int}

	\typerule{Ret Label - return}{
		\src{C}.\mtt{intfs}\vdash\src{f',f} : \src{in}
	}{
		\src{C};\src{f',f} ; \src{ret}  \vdash \src{\rth{f}{v}{}}
	}{us-aux-ret}
	\typerule{Ret Label - returnback}{
		\src{C}.\mtt{intfs}\vdash\src{f',f} : \src{out}
	}{
		\src{C};\src{f',f} ; \src{ret}  \vdash \src{\rbh{f}{v}{}}
	}{us-aux-retback}
	\typerule{Ret Label - internal}{
		\src{C}.\mtt{intfs}\vdash\src{f',f} : \src{internal}
	}{
		\src{C};\src{f',f} ; \src{ret}  \vdash \src{\epsilon}
	}{us-aux-ret-int}

	\typerule{\SR-Plug}{
		\src{\ctxs{}} \equiv \src{H , \OB{F}\hole{\cdot}}
		&
		\src{P}\equiv\src{ H' , \OB{F'} , \OB{I}} 
		&
		\vdash\src{\OB{F'};\OB{F} , \OB{I}}:\src{whole}
		&
		\src{main}\in\fun{names}{\src{\OB{F}}}
		\\
		\dom{\src{H}}\cap\dom{\src{H'}}=\emptyset
		\\
		\forall \src{n\mapsto v : \sigma} \in \src{H'}, \src{n}<\src{0} \text{ and } \src{\sigma}=\src{\unta}
	}{
		\src{\ctxs{}\hole{P}} = \src{H;H' , \OB{F;F'}, \dom{\src{\OB{F}}}}
	}{plug-us}
	\typerule{\SR-Whole}{
		\fun{names}{\src{\OB{F}}}\cap\fun{names}{\src{\OB{F'}}}=\emptyset
		\\
		\fun{names}{\src{\OB{I}}}\subseteq \fun{names}{\src{\OB{F}}}\cup\fun{names}{\src{\OB{F'}}}
		&
		\fun{fv}{\src{\OB{F}}}\cup\fun{fv}{\src{\OB{F'}}}=\srce
	}{
		\vdash\src{\OB{F'};\OB{F} , \OB{I}}:\src{whole}
	}{whole-us}
	\typerule{\SR-Initial State}{
		\src{H_0} = \src{H''}\cup\src{H}\cup\src{H'}
		\\
		\src{H'} = \myset{ \src{n\mapsto 0 : \safeta} }{ \src{n}\in\mb{N}\setminus\dom{\src{H}} } 
		\\
		\src{H''} = \myset{ \src{-n\mapsto 0 : \unta} }{ \src{n}\in\mb{N}, \src{-n}\notin\dom{\src{H}} }
	}{
		\SInits{(\src{H , \OB{F} , \OB{I}})} = \src{\OB{F} , \OB{I} , H_0 , \srce\cdot x\mapsto 0 \triangleright \call{main}~ x}
	}{ini-us}
	\typerule{\SR-Terminal State}{
		\nexists \src{\Omega'},\src{\lambda} . \src{\Omega \xto{\lambda} \Omega'}
	}{
		\vdash \src{\Omega} : \bot
	}{term-state}

	\typerule{Merge-B-base}{}{
		\srce + \srce = \srce
	}{}
	\typerule{Merge-B-ind}{
		\src{\Bv+\Bs}=\src{B}
	}{
		\src{\Bv ; x\mapsto v+\Bs ; x\mapsto \sigma} = \src{B ; x\mapsto v:\sigma}
	}{}
	\typerule{Merge-H-base}{}{
		\srce + \srce = \srce
	}{}
	\typerule{Merge-H-ind}{
		\src{\Hv+\Hs}=\src{H}
	}{
		\src{\Hv ; n\mapsto v+\Hs ; n\mapsto \sigma} = \src{H ; n\mapsto v:\sigma}
	}{}
	\typerule{Merge-s-\S}{
		\src{\Hv+\Hs}=\src{H}
		&
		\src{\OB{\Bv'} + \OB{\Bs}} = \src{\OB{B}}
		&
		\src{\OB{\Bv}+\OB{\Bs}}=\src{\OB{B'}}
	}{
		\src{C;\Hv;\OB{\Bv}\triangleright s + C;\Hs;\OB{B}\triangleright s'} = \src{C ; H ; \OB{B'} \triangleright s}
	}{}
\end{center}
\botrule

\subsubsection{Component Semantics}\label{src:src-sem-com}
\mytoprule{\text{Judgements}}
\begin{align*}
	&\src{\Bv \triangleright e\bigreds v}
	&&\text{Expression \src{e} big-steps to value \src{v}.}
	\\
	&\src{\Bs \triangleright e\bigreds \sigma}
	&&\text{Expression \src{e} is tainted \src{\sigma}.}
	\\
	&\src{B \triangleright e\bigreds v : \sigma}
	&&\text{Expression \src{e} big-steps to value \src{v} tagged \src{\sigma}.}
	\\
	&\src{\Ov \xtos{\lambda} \Ov'}
	&&\text{State \src{\Ov} small-steps to \src{\Ov'} and emits action \src{\lambda}.} 
	\\
	&\src{\sigma' ; \Os \xtos{\sigma} \Os'}
	&&\text{With pc tainted \src{\sigma'}, the action of state \src{\Ov} is tainted \src{\sigma}.} 
	\\
	&\src{\Omega \xtos{\lambda^\sigma} \Omega'} \text{ aka}
	\\
	&\src{C, H, \OB{B} \triangleright \proc{s}{\OB{f}}} \xtos{\lambda^\sigma} \src{C, H', \OB{B'} \triangleright \proc{s'}{\OB{f'}}} 
	&&\text{Statement \src{s} reduces to \src{s'} and evolves the rest}
	\\
	&&&\text{ accordingly, emitting tagged label \src{\lambda^\sigma}.}
	\\
	&\src{\Omega} \Xtos{\tras{^\sigma}} \src{\Omega'}
	&& \text{Program state \src{\Omega} steps to \src{\Omega'} }
	\\
	&&&\text{ emitting tagged trace \src{\tras{^\sigma}}.}
	\\
	&\src{P\sems \tras{^\sigma}}
	&& \text{Whole program \src{P} produces tagged trace \tras{^\sigma}}
\end{align*}
\botrule

\mytoprule{\src{\Bv\triangleright e \bigreds v} }
\begin{center}
	\typerule{E-\SR-val}{
	}{
		\src{\Bv \triangleright v \bigreds v}
	}{es-big-val}
	\typerule{E-\SR-var}{
		\src{\Bv}(\src{x}) = \src{v}
	}{
		\src{\Bv \triangleright x \bigreds v}
	}{es-big-var}
	\typerule{E-\SR-op}{
		\src{\Bv \triangleright e\bigreds n }	
		&
		\src{\Bv \triangleright e'\bigreds n' }
		&
		\src{n''} = [\src{n\op n'}]		
	}{
		\src{\Bv \triangleright e\op e' \bigreds n'' }
	}{es-big-op}
	\typerule{E-\SR-comparison}{
		\src{\Bv \triangleright e\bigreds n }	
		&
		\src{\Bv \triangleright e'\bigreds n' }
		&
		\src{n''} = [\src{n\bop n'}]		
	}{
		\src{\Bv \triangleright e\bop e' \bigreds n''}
	}{es-big-bop}
\end{center}
\botrule

\mytoprule{\src{\Bs\triangleright e \bigreds \sigma} }
\begin{center}
	\typerule{T-\SR-val}{
	}{
		\src{\Bs \triangleright v \bigreds \safeta}
	}{ts-big-val}
	\typerule{T-\SR-var}{
		\src{\Bs}(\src{x}) = \src{\sigma}
	}{
		\src{\Bs \triangleright x \bigreds \sigma}
	}{ts-big-var}
	\typerule{T-\SR-op}{
		\src{\Bs \triangleright e\bigreds \sigma}	
		&
		\src{\Bs \triangleright e'\bigreds \sigma'}
		&
		\src{\sigma''} = \src{\sigma}\sqcup\src{\sigma'}
	}{
		\src{\Bs \triangleright e\op e' \bigreds \sigma''}
	}{ts-big-op}
	\typerule{T-\SR-comparison}{
		\src{\Bs \triangleright e\bigreds \sigma}	
		&
		\src{\Bs \triangleright e'\bigreds \sigma'}
		&
		\src{\sigma''} = \src{\sigma}\sqcup\src{\sigma'}
	}{
		\src{\Bs \triangleright e\bop e' \bigreds \sigma''}
	}{ts-big-bop}
\end{center}
\botrule

\mytoprule{\src{B \triangleright e\bigreds v : \sigma}}
\begin{center}
	\typerule{Combine-e-\S}{
		\src{\Bv + \Bs} = \src{B}
		&
		\src{\Bv \triangleright e \bigred v}
		&
		\src{\Bs \triangleright e \bigred \sigma}
	}{
		\src{B \triangleright e \bigred v : \sigma}
	}{e-comb}
\end{center}
\botrule

The taint propagation for operations is standard, using the lub.
A more refined version is possible (i.e., not propagating taint for an op with an identity operand, or propagating taint with glb), but not needed in this case.

\mytoprule{\src{\Ov \xtos{\lambda} \Ov'} }
\begin{center}
	\typerule{E-\SR-sequence}{
	}{
		\src{C, \Hv, \OB{\Bv} \triangleright \skips;s} \xtos{\epsilon} \src{C, \Hv, \OB{\Bv} \triangleright s}
	}{eus-seq}
	\typerule{E-\SR-step}{
		\src{C, \Hv, \OB{\Bv} \triangleright s} \xtos{\lambda} \src{C, \Hv', \OB{\Bv'} \triangleright s'}
	}{
		\src{C, \Hv, \OB{\Bv} \triangleright s;s''} \xtos{\lambda} \src{C, \Hv', \OB{\Bv'} \triangleright s';s''}
	}{eus-step}
	\typerule{E-\SR-if-true}{
		\src{\Bv \triangleright e\bigreds 0}
	}{
		\src{C, \Hv, \OB{\Bv}\cdot \Bv \triangleright \ifzte{e}{s}{s'}} \xtos{(\ifl{0})} \src{C, \Hv, \OB{\Bv}\cdot \Bv \triangleright s}
	}{eus-ift}
	\typerule{E-\SR-if-false}{
		\src{\Bv \triangleright e\bigreds n} 
		& 
		\src{n}>\src{0}
	}{
		\src{C, \Hv, \OB{\Bv}\cdot \Bv \triangleright \ifzte{e}{s}{s'}} \xtos{(\ifl{n})} \src{C, \Hv, \OB{\Bv}\cdot \Bv \triangleright s}
	}{eus-iff}
	\typerule{E-\SR-letin}{
		\src{\Bv \triangleright e\bigreds v}
	}{
		\src{C, \Hv, \OB{\Bv}\cdot \Bv \triangleright \letin{x}{e}{s}} \xtos{\epsilon} \src{C, \Hv, \OB{\Bv}\cdot \Bv\cup x\mapsto v \triangleright s}
	}{eus-letin}
	\typerule{E-\SR-write}{
		\src{\Bv \triangleright e\bigreds n}
		&
		\src{\Bv \triangleright e'\bigreds v}
		\\
		\src{\Hv}=\src{{\Hv}_1; \abs{n}\mapsto v' ; {\Hv}_2}
		&
		\src{\Hv'}=\src{{\Hv}_1; \abs{n}\mapsto v ; {\Hv}_2}
	}{
		\src{C, \Hv, \OB{\Bv}\cdot \Bv \triangleright \asgn{e}{e'}} \xtos{\wrl{\abs{n} \mapsto v}} \src{C, \Hv', \OB{\Bv}\cdot \Bv \triangleright \skips} 
	}{eus-up-com}
	\typerule{E-\SR-read}{
		\src{\Bv \triangleright e\bigreds n}
		&
		\src{\Hv}=\src{{\Hv}_1; \abs{n}\mapsto v ; {\Hv}_2}
	}{
		\src{C, \Hv, \OB{\Bv}\cdot \Bv \triangleright \letread{x}{e}{s}} \xtos{\rdl{\abs{n}}} \src{C, \Hv, \OB{\Bv}\cdot \Bv\cup x\mapsto v \triangleright s}
	}{eus-rd-com}
	\typerule{E-\SR-write-prv}{
		\src{\Bv \triangleright e\bigreds n}
		&
		\src{\Bv \triangleright e'\bigreds v}
		&
		\src{n_a} = \src{-\abs{n}}
		\\
		\src{\Hv}=\src{{\Hv}_1; n_a\mapsto v' ; {\Hv}_2}
		&
		\src{\Hv'}=\src{{\Hv}_1; n_a\mapsto v ; {\Hv}_2}
	}{
		\src{C, \Hv, \OB{\Bv}\cdot \Bv \triangleright \asgnp{e}{e'}} \xtos{\wrl{n_a}} \src{C, \Hv', \OB{\Bv}\cdot \Bv \triangleright \skips} 
	}{eus-up-com-p}
	\typerule{E-\SR-read-prv}{
		\src{\Bv \triangleright e\bigreds n }
		&
		\src{n_a} = \src{-\abs{n}}
		&
		\src{\Hv}=\src{{\Hv}_1; n_a\mapsto v ; {\Hv}_2}
	}{
		\src{C, \Hv, \OB{\Bv}\cdot \Bv \triangleright \letreadp{x}{e}{s}} \xtos{\rdl{n_a}} \src{C, \Hv, \OB{\Bv}\cdot \Bv\cup x\mapsto v \triangleright s}
	}{eus-rd-com-p}
	\typerule{E-\SR-call}{
		\src{\OB{f'}} = \src{\OB{f''};f'}
		&
		\src{f(x)\mapsto s;\ret}\in\src{C}.\mtt{funs}
		\\
		\src{C};\src{f',f} ; \src{call} ; \src{v} \vdash \src{\lambda}
		&
		\src{\Bv \triangleright e\bigreds v }
	}{
		\src{C, \Hv, \OB{\Bv}\cdot \Bv \triangleright \proc{{\call{f}~e}}{\OB{f'}}} \xtos{\lambda} \src{C, \Hv, \OB{\Bv}\cdot \Bv\cdot x\mapsto v \triangleright \proc{{s;\ret}}{\OB{f'};f}}
	}{eus-call}
	\typerule{E-\SR-return}{
		\src{\OB{f'}} = \src{\OB{f''};f'}
		&
		\src{C};\src{f',f} ; \src{ret}  \vdash \src{\lambda}
	}{
		\src{C, \Hv, \OB{\Bv}\cdot \Bv \triangleright \proc{{\ret}}{\OB{f'};f}} %
		\xtos{\lambda}
		\src{C, \Hv, \OB{\Bv} \triangleright \proc{\skips}{\OB{f'}}}
	}{eus-ret}
	\end{center}
\botrule

\mytoprule{\src{\sigma_{pc}; \Os \xtos{\sigma} \Os'} }
\begin{center}
	\typerule{T-\SR-sequence}{
	}{
		\src{\sigma_{pc}; C, \Hs, \OB{B} \triangleright \skips;s \xtos{\epsilon} C, \Hs, \OB{B} \triangleright s}
	}{tus-seq}
	\typerule{T-\SR-step}{
		\src{\sigma_{pc}; C, \Hs, \OB{B} \triangleright s \xtos{\sigma} C, \Hs', \OB{\Bs'} \triangleright s'}
	}{
		\src{\sigma_{pc}; C, \Hs, \OB{B} \triangleright s;s'' \xtos{\sigma} C, \Hs', \OB{\Bs'} \triangleright s';s''}
	}{tus-step}
	\typerule{T-\SR-if-true}{
		\src{B \triangleright e\bigreds v : \sigma}
	}{
		\src{\sigma_{pc}; C, \Hs, \OB{B}\cdot B \triangleright \ifzte{e}{s}{s'}} \xtos{\sigma} \src{C, \Hs, \OB{B}\cdot B \triangleright s}
	}{tus-ift}
	\typerule{T-\SR-letin}{
		\src{B \triangleright e\bigreds v : \sigma}
	}{
		\src{\sigma_{pc}; C, \Hs, \OB{B}\cdot B \triangleright \letin{x}{e}{s}} \xtos{\epsilon} \src{C, \Hs, \OB{B}\cdot B\cup x\mapsto \sigma \triangleright s}
	}{tus-letin}
	\typerule{T-\SR-write}{
		\src{B \triangleright e\bigreds n : \sigma}
		&
		\src{B \triangleright e'\bigreds \_ : \sigma''}
		\\
		\src{\Hs'}=\src{\Hs \cup \abs{n}\mapsto \safeta }
	}{
		\src{\sigma_{pc}; C, \Hs, \OB{B}\cdot B \triangleright \asgn{e}{e'}} \xtos{\sigma\sqcup\sigma''} \src{C, \Hs', \OB{B}\cdot B \triangleright \skips} 
	}{tus-up-com}
	\typerule{T-\SR-read}{
		\src{B \triangleright e\bigreds n : \sigma}
		&
		\src{\Hs(n_a)} = \src{\sigma}
	}{
		\src{\sigma_{pc}; C, \Hs, \OB{B}\cdot B \triangleright \letread{x}{e}{s}} \xtos{\sigma'} \src{C, \Hs, \OB{B}\cdot B\cup x\mapsto 0 : \sigma \triangleright s}
	}{tus-rd-com}
	\typerule{T-\SR-write-prv}{
		\src{B \triangleright e\bigreds n : \sigma}
		&
		\src{B \triangleright e'\bigreds \_ : \sigma''}
		&
		\src{n_a} = \src{-\abs{n}}
		\\
		\src{\Hs'}=\src{\Hs \cup \abs{n}\mapsto \sigma'' }
	}{
		\src{\sigma_{pc}; C, \Hs, \OB{B}\cdot B \triangleright \asgnp{e}{e'}} \xtos{\sigma} \src{C, \Hs', \OB{B}\cdot B \triangleright \skips} 
	}{tus-up-com-p}
	\typerule{T-\SR-read-prv}{
		\src{B \triangleright e\bigreds n : \sigma'}
		&
		\src{n_a} = \src{-\abs{n}}
		&
		\src{\Hs(n_a)} = \src{\sigma}
		&
		\src{\sigma''} = \src{\sigma}\sqcup\src{\sigma'}
	}{
		\src{\sigma_{pc}; C, \Hs, \OB{B}\cdot B \triangleright \letreadp{x}{e}{s}} \xtos{\sigma''} \src{C, \Hs, \OB{B}\cdot B\cup x\mapsto 0 : \unta \triangleright s}
	}{tus-rd-com-p}
	\typerule{T-\SR-call-internal}{
		\src{C}.\mtt{intfs}\vdash\src{f,f'}:\src{internal}
		&
		\src{\OB{f'}} = \src{\OB{f''};f'}
		\\
		\src{f(x)\mapsto s;\ret}\in\src{C}.\mtt{funs}
		&
		\src{B \triangleright e\bigreds v : \sigma}
	}{
		\src{\sigma_{pc}; C, \Hs, \OB{B}\cdot B \triangleright \proc{{\call{f}~e}}{\OB{f'}}} \xtos{\epsilon} \src{C, \Hs, \OB{B}\cdot B\cdot x\mapsto \sigma \triangleright \proc{{s;\ret}}{\OB{f'};f}}
	}{tus-call-i}
	\typerule{T-\SR-call}{
		\src{\OB{f'}} = \src{\OB{f''};f'}
		&
		\src{f(x)\mapsto s;\ret}\in\src{C}.\mtt{funs}
		\\
		\src{C}.\mtt{intfs}\nvdash\src{f',f}:\src{internal}
		&
		\src{B \triangleright e\bigreds v : \sigma}
	}{
		\src{\sigma_{pc}; C, \Hs, \OB{B}\cdot B \triangleright \proc{{\call{f}~e}}{\OB{f'}}} \xtos{\sigma} \src{C, \Hs, \OB{B}\cdot B\cdot x\mapsto \safeta \triangleright \proc{{s;\ret}}{\OB{f'};f}}
	}{tus-call}
	\typerule{T-\SR-ret-internal}{
		\src{\OB{f'}} = \src{\OB{f''};f'}
		&
		\src{C}.\mtt{intfs}\vdash\src{f,f'}:\src{internal}
	}{
		\src{\sigma_{pc}; C, \Hs, \OB{B}\cdot B \triangleright \proc{{\ret}}{\OB{f'};f}} \xtos{\epsilon} \src{C, \Hs, \OB{B} \triangleright \proc{\skips}{\OB{f'}}}
	}{tus-ret-i}
	\typerule{T-\SR-return}{
		\src{\OB{f'}} = \src{\OB{f''};f'}
		&
		\src{C}.\mtt{intfs}\nvdash\src{f,f'}:\src{internal}
	}{
		\src{\sigma_{pc}; C, \Hs, \OB{B}\cdot B \triangleright \proc{{\ret}}{\OB{f'};f} \xtos{\safeta} C, \Hs, \OB{B} \triangleright \proc{\skips}{\OB{f'}}}
	}{tus-ret}
	\end{center}
\botrule
Note that we perform a trick for the taint tracking of reading values.
We do not know the value that is being read, and so we read a default 0, that, however, we taint correctly.
When the stack of bindings of the taint tracking is merged, the values stored here are not considered, only the taint, so it is ok to track 0 because it will not affect the semantics at all.

\mytoprule{\src{\Omega\xto{\lambda^\sigma}\Omega'}}
\begin{center}
	\typerule{Combine-s-\S}{
		\src{\Ov + \Os} = \src{\Omega}
		&
		\src{\Ov' + \Os'} = \src{\Omega'}
		\\
		\src{\Ov \xto{\lambda} \Ov'}
		&
		\src{\safeta ; \Os \xto{\sigma} \Os'}
	}{
		\src{\Omega \xto{\lambda^{\sigma\sqcap\safeta}} \Omega'}
	}{s-comb}
\end{center}
\botrule

Reading from the private heap taints the read value as unsafe, so no matter what you write in the private heap, it'll be always \unta.
Writing to the public heap taints the written value as safe, so no matter what you read from the public heap, it'll be always \safeta.
In both cases, the action is labelled with the lub.
If i read a private value, that value is \unta\ but the read itself may be \safeta.
If i write a public value, that value is \safeta\ but the write itself may be \unta, if done speculatively (and thus rollbacked).
Public writes taint their action with the lub of data and address because an attacker sees both effects via that read.
A private write taints only with the address because the attacker does not see the data, only the address.

\mytoprule{ \src{\Omega} \Xtos{\tra{^\sigma}} \src{\Omega'} }
\begin{center}
	\typerule{E-\SR-single}{
		\src{\Omega}
		\xtos{\alpha^\sigma}\src{\Omega'}
		&
		\src{\Omega} = \src{\OB{F},\OB{I},H,B\triangleright \proc{s}{\OB{f}\cdot f}}
		&
		\src{\Omega'} = \src{\OB{F},\OB{I},H',B'\triangleright \proc{s'}{\OB{f'}\cdot f'}}
		\\
		\text{ if } \src{f} == \src{f'} \text{ and } \src{f} \in \src{I} \text{ then } \src{\tra{^\sigma}} = \src{\epsilon} \text{ else } \src{\tra{^\sigma}} = \src{\alpha^\safeta}
	}{
		\src{\Omega}\Xtos{\tra{^\sigma}}\src{\Omega'}
	}{eus-tr-sin}
	\typerule{E-\SR-silent}{
		\src{\Omega}\xtos{\epsilon}\src{\Omega'}
	}{
		\src{\Omega}\Xtos{}\src{\Omega'}
	}{eus-tr-silent}
	\typerule{E-\SR-trans}{
		\src{\Omega}\Xtos{\tra{_1^{\sigma_1}}}\src{\Omega''}
		&
		\src{\Omega''}\Xtos{\tra{_2^{\sigma_2}}}\src{\Omega'}
	}{
		\src{\Omega}\Xtos{\tra{_1^{\sigma_1}}\cdot\tra{_2^{\sigma_2}}}\src{\Omega'}
	}{eus-tr-trans}
\end{center}
\botrule

\Cref{tr:eus-tr-sin} tells that when you have a single action, if that action is done within the context, then no action is shown.
Otherwise, if the action is done within the component, or if the action is done between component and context, then the action is shown and it is tagged as \src{\safeta}.
All generated actions are $\src{\safeta}$ despite their label (\src{\sigma}) because the only source of unsafety is speculation, which only exists in \TR.
Technically, we should take the $\sqcup$ of the pc (which here is always \src{\safeta}) and of the data, but since this always results ins \src{\safeta}, we just write \src{\safeta}.

\mytoprule{\src{P\sems \tra{^\sigma}}}
\begin{center}
	\typerule{E-\SR-trace}{
		\exists\src{\Omega}\ldotp \vdash \src{\Omega} : \bot 
		&
		\src{\SInit{P} \Xtos{\tra{^\sigma}} \Omega}
	}{
		\src{P\sems \tra{^\sigma}\ts}
	}{eus-trace}
	\typerule{E-\SR-behaviour}{
	}{
		\src{\behavs{P}} = \myset{\tras{^\sigma}}{ \src{P\sems \tra{^\sigma}} }
	}{eus-beh}
\end{center}
\botrule

All traces are finite and the set of traces includes all traces that terminate.

\paragraph{Alternative Semantics Definition}

\mytoprule{\src{P\sems \tra{^\sigma}}}
\begin{center}
	\typerule{E-\SR-trace}{
		\exists\src{\Omega}\ldotp \src{\SInit{P} \Xtos{\tra{^\sigma}} \Omega}
	}{
		\src{P\sems \tra{^\sigma}}
	}{2eus-trace}
\end{center}
\botrule
This defnition drops the condition of \src{\Omega} being final in \Cref{tr:2eus-trace}.
This way, the set of behaviours includes all prefixes of a prefix, i.e., it is subset closed.
All the compiler proofs work with this definition too without concerns.
Having this definition complicates relating \Thmref{def:dss} and \Thmref{def:rsni} so we do not use this.

\newpage
\section{\TR : Adding Speculation to \SR}
\TR extends \SR by adding the ability to speculate as well as the programming constructs that are used as countermeasures againsts speculation: a conditional move and the lfence.
\TR defines a new notion of program states in order to model speculative execution and it adds rules to the semantics of statements to capture speculation ($\trg{\xltot{}} $).
Thus, the trace alphabet of \TR is richer than the one if \SR.

Notation-wise, \TR includes all that is \SR, i.e., all that is typeset in \src{blue} also exists in \trg{red}.

\subsection{Syntax}\label{sec:trg-syn}
All elements from \SR also exist in this language.

\begin{align*}
	\mi{Statements}~\trg{s} \bnfdef&\ \cdots \mid \trg{\lfence} \mid \trg{\cmove{x}{e}{e}{s}}
	\\
	\mi{Speculative\ States}~\trg{\Sigma} \bnfdef&\ \trg{n;\OB{\Phi}}
	\\
	\mi{Speculation\ Instance}~\trg{\Phi} \bnfdef&\ \trg{(\Omega,\trg{w},\trgb{\taintt})}
	\\
	\mi{Speculation\ Instance\ Vals.}~\trg{\Pv} \bnfdef&\ \trg{n ; {\Ov , \omega }}
	\\
	\mi{Speculation\ Instance\ Taint}~\trg{\Pt} \bnfdef&\ \trg{{\Ot, \taintt}}
	\\
	\mi{Actions}~\acat{} \bnfdef&\ \trg{\rollbl} %
	\\
	\mi{Window}~\trgb{\omega} \bnfdef&\ \trg{n} \mid \trgb{\bot}
	\\
	\mi{Droppable\ Names}~\trg{D} \bnfdef&\ \trge \mid \trg{D, x}
\end{align*}

\paragraph{Reading the Notation:}
A stack of elements $e_1,\cdots, e_n$ is indicated with \OB{e}.
So, given an \OB{e}, its elements are possibly distinct.
If we want to refer to the top of a stack of elements, we will use notation $\OB{e}\cdot e$, where $e$ is the top, \OB{e} is the rest of the stack, and $e$ is possibly different from all elements in \OB{e}.

\bigskip

We define pattern-matching on windows as follows.
The form \trg{k+1} matches \trgb{\bot} and any number different from \trg{0} so that \trg{k} is interpreted as \trgb{\bot} if $\trgb{\omega}=\trgb{\bot}$ ad it is interpreted as \trg{k} if $\trgb{\omega}=\trg{k+1}$.

Droppable names are a technicality needed for cross-language relations, as explained in \Cref{sec:sim-code-fence}.

\subsection{Dynamic Semantics}\label{sec:trg-sem}
\subsubsection{Auxiliary Functions}
\begin{center}
	\typerule{Merge-S}{
		\trg{\OB{\Ov} + \OB{\Os}} = \trg{\OB{\Omega}}
	}{
		\trg{\OB{n ; \Ov , \omega } + \OB{\Os, \sigma}} = \trg{n ; \OB{\Omega , \omega , \sigma}}
	}{}
\end{center}

\subsubsection{Component Semantics}\label{src:trg-sem-com}
\begin{align*}
	&\trg{\Pv \xltot{\trat{}} \OB{\Pv'}}
	&& \text{Speculative state \trg{\Pv} evolves into stack \trg{\OB{\Pv'}} emitting action \acat{}. }
	\\
	&\trg{\OB{\Pv} \xltot{\trat{}} \OB{\Pv'}}
	&& \text{ Stack \trg{\OB{\Pv}} evolves into \trg{\OB{\Pv'}}.}
	\\
	&\trg{\Pt \xltot{\sigma} \OB{\Pt'}}
	&& \text{Speculative tainted state \trg{\Pt} evolves into stack \trg{\OB{\Pt'}} emitting taint {\taintt}. }
	\\
	&\trg{\OB{\Pt} \xltot{\taintt} \OB{\Pt'}}
	&& \text{ Stack \trg{\OB{\Pt}} evolves into \trg{\OB{\Pt'}}.}
	\\
	&\trg{\Sigma \xltot{\trat{^\sigma}} \Sigma'}
	&& \text{Speculative state \trg{\Sigma} evolves into \trg{\Sigma'} emitting tainted action \acat{^{\taintt}}. }
\end{align*}

\mytoprule{\trg{\Ov \xtot{\trgb{\lambda}} \Ov'} }
\begin{center}
	\typerule{E-\TR-lfence}{
	}{
		\trg{C, \Hv, \OB{\Bv} \triangleright \lfence} \xtot{\epsilon} \trg{C, \Hv, \OB{\Bv} \triangleright \skipt}
	}{eut-lf}
	\typerule{E-\TR-cmove-true}{
		\trg{x}\in\dom{\trg{\Bv}}
		&
		\trg{\Bv \triangleright e'\bigred 0}
		&
		\trg{\Bv \triangleright e\bigred v}
	}{
		\trg{C, \Hv, \OB{\Bv}\cdot \Bv \triangleright \cmove{x}{e}{e'}{s}} \xtot{\epsilon} \trg{C, \Hv, \OB{\Bv}\cdot \Bv\cup x\mapsto v \triangleright s}
	}{eut-cmv-t}
	\typerule{E-\TR-cmove-false}{
		\trg{x}\in\dom{\trg{\Bv}}
		&
		\trg{\Bv \triangleright e'\bigred n }
		&
		\trg{n}>\trg{0}
		&
		\trg{\Bv(x)} = \trg{v}
	}{
		\trg{C, \Hv, \OB{\Bv}\cdot \Bv \triangleright \cmove{x}{e}{e'}{s}} \xtot{\epsilon} \trg{C, \Hv, \OB{\Bv}\cdot \Bv\cup x\mapsto v \triangleright s }
	}{eut-cmv-f}
\end{center}
\botrule

\mytoprule{\trg{\taintt_{pc}; \Ot \xtot{\trgb{\taintt}} \Ot'} }
\begin{center}
	\typerule{T-\TR-lfence}{
	}{
		\trg{\taintt_{pc}; C, \Ht, \OB{B} \triangleright \lfence} \xtot{\epsilon} \trg{C, \Ht, \OB{B} \triangleright \skipt}
	}{tut-lf}
	\typerule{T-\TR-cmove-true}{
		\trg{x}\in\dom{\trg{B}}
		&
		\trg{B \triangleright e'\bigred 0 : \taintt'}
		&
		\trg{B \triangleright e\bigred v : \taintt}
		&
		\trg{\taintt'' = \taintt \sqcup \taintt'}
	}{
		\trg{\taintt_{pc}; C, \Ht, \OB{B}\cdot B \triangleright \cmove{x}{e}{e'}{s}} \xtot{\epsilon} \trg{C, H, \OB{B}\cdot B\cup x\mapsto v :\taintt'' \triangleright s}
	}{tut-cmv-t}
	\typerule{E-\TR-cmove-false}{
		\trg{x}\in\dom{\trg{B}}
		&
		\trg{B \triangleright e'\bigred n : \taintt'}
		&
		\trg{n}>\trg{0}
		&
		\trg{B(x)} = \trg{v : \taintt}
		&
		\trg{\taintt'' = \taintt \sqcup \taintt'}
	}{
		\trg{\taintt_{pc}; C,\Ht, \OB{B}\cdot B \triangleright \cmove{x}{e}{e'}{s}} \xtot{\epsilon} \trg{C, \Ht, \OB{B}\cdot B\cup x\mapsto v :\taintt'' \triangleright s }
	}{tut-cmv-f}
\end{center}
\botrule

In a conditional move we require that \trg{x} is always bound afterwards.

\mytoprule{\trg{\OB{\Pv} \xltot{\acat{}} \OB{\Pv}'}}
\begin{center}
	\typerule{E-\TR-speculate-stack}{
		\trg{\Pv \xltot{\acat{}} \OB{\Pv'}}
	}{
		\trg{\OB{\Pv}\cdot\Pv \xltot{\acat{}} \OB{\Pv}\cdot\OB{\Pv'}}
	}{}
\end{center}
\botrule

\mytoprule{\trg{\Pv \xltot{\acat{}} \OB{\Pv}'}}
\begin{center}
	\typerule{E-\TR-speculate-epsilon}{
		\trg{\Ov \xtot{\epsilon} \Ov'}	
		&
		\trg{\Ov}\equiv\trg{C, \Hv, \OB{\Bv} \triangleright s;s'}
		&
		\trg{s}\not\equiv\trg{\ifzte{\_}{\_}{\_}} \text{ and } \trg{s}\not\equiv\trg{\lfence}
	}{
		\trg{w, (\Ov,n+1) \xltot{\epsilon} w, (\Ov',n)}
	}{et-sp-eps}
	\typerule{E-\TR-speculate-lfence}{
		\trg{\Ov \xtot{\epsilon} \Ov'}	
		&
		\trg{\Ov}\equiv\trg{C, \Hv, \OB{\Bv} \triangleright s;s'}
		&
		\trg{s}\equiv\trg{\lfence}
	}{
		\trg{w, (\Ov,n+1) \xltot{\epsilon} w, (\Ov',0)}
	}{et-sp-lf}
	\typerule{E-\TR-speculate-action}{
		\trg{\Ov \xtot{\trgb{\lambda}} \Ov'}	
		&
		\trg{\Ov}\equiv\trg{C, \Hv, \OB{\Bv} \triangleright s;s'}
		&
		\trg{s}\not\equiv\trg{\ifzte{\_}{\_}{\_}} \text{ and } \trg{s}\not\equiv\trg{\lfence}
	}{
		\trg{w, (\Ov,n+1) \xltot{ \trgb{\lambda}} w, (\Ov',n)}
	}{et-sp-act}
	\typerule{E-\TR-speculate-if}{
		\trg{\Ov \xtot{\acat} \Ov'}	
		&
		\trg{\Ov}\equiv\trg{C, \Hv, \OB{\Bv} \cdot \Bv \triangleright \proc{s;s'}{\OB{f}\cdot f}}
		&
		\trg{s}\equiv\trg{\ifte{e}{s''}{s'''}}
		\\
		\trg{C}\equiv\trg{\OB{F};\OB{I}}
		&
		\trg{f}\notin\trg{\OB{I}}
		\\
		\text{ if }
			\trg{\Bv \triangleright e\bigred 0} 
		\text{ then }
			\trg{\Ov''}\equiv\trg{C, \Hv, \OB{\Bv} \cdot \Bv \triangleright s''';s'}
		\\
		\text{ if }
			\trg{\Bv \triangleright e\bigred n} 
			\text{ and }
			\trg{n}>\trg{0}
		\text{ then }
			\trg{\Ov''}\equiv\trg{C, \Hv, \OB{\Bv} \cdot \Bv \triangleright s'';s'}
		\\
		\trg{j} = \fun{min}{\trg{w},\trg{n}}
	}{
		\trg{w, (\Ov,n+1) \xltot{ \acat{} } w, (\Ov',n)\cdot w,(\Ov'',j)}
	}{et-sp-if}
	\typerule{E-\TR-speculate-rollback}{
	}{
		\trg{w, (\Ov,0) \xltot{ \rollbl } w, \trge}
	}{et-sp-rb}
	\typerule{E-\TR-speculate-rollback-stuck}{
		\vdash \trg{\Pv} : \bot
		&
		\trg{\Pv} = \trg{w, (\Ov,\omega)}
	}{
		\trg{\Pv \xltot{ \rollbl } w, \trge }
	}{et-sp-rb-s}

	\typerule{E-\TR-speculate-if-att}{
		\trg{\Ov \xtot{\acat} \Ov'}	
		&
		\trg{\Ov}\equiv\trg{C, \Hv, \OB{\Bv} \cdot \Bv \triangleright \proc{s;s'}{\OB{f}\cdot f}}
		&
		\trg{s}\equiv\trg{\ifte{e}{s''}{s'''}}
		\\
		\trg{C}\equiv\trg{\OB{F};\OB{I}}
		&
		\trg{f}\in\trg{\OB{I}}
	}{
		\trg{w, (\Ov,n+1) \xltot{ \acat{} } w, (\Ov',n)}
	}{et-sp-if}
\end{center}
\botrule

Note that our attacker essentially cannot speculate.
This is because speculation would only trigger execution of the `other' branch.
But recall that our attackers are universally quantified.
It is therefore needless to have an attacker speculate and do something while the same behaviour is anyway considered by another attacker.

If we allowed the attacker to speculate, we would have a problem: the pc taint would be \trg{\unta} and if the attacker called us, our actions would be \trg{\unta}.
However, this would be an overapproximation of our taint-tracking: as already said, those actions are perfectly valid since other attackers will make our code do that without speculation.
Thus, if we were to allow the attacker to speculate, we should not raise the pc to \trg{\unta} there, but keep it \trg{\safeta}.

\mytoprule{\trg{\OB{\Pt} \xltot{\taintt} \OB{\Pt}'}}
\begin{center}
	\typerule{T-\TR-speculate-stack}{
		\trg{\Pt \xltot{\taintt} \OB{\Pt'}}
	}{
		\trg{\OB{\Pt}\cdot\Pt \xltot{\taintt} \OB{\Pt}\cdot\OB{\Pt'}}
	}{}
\end{center}
\botrule

\mytoprule{\trg{\Pt \xltot{\taintt} \OB{\Pt'}}}
\begin{center}
	\typerule{T-\TR-speculate-epsilon}{
		\trg{\taintt; \Ot \xtot{\epsilon} \Ot'}	
		&
		\trg{\Ot}\equiv\trg{C, \Hs, \OB{B} \triangleright s;s'}
		&
		\trg{s}\not\equiv\trg{\ifzte{\_}{\_}{\_}}
	}{
		\trg{w, (\Ot,\taintt) \xltot{\epsilon} w, (\Ot',\taintt)}
	}{tt-sp-eps}
	\typerule{T-\TR-speculate-action}{
		\trg{\taintt; \Ot \xtot{\taintt'} \Ot'}	
		&
		\trg{\Ot}\equiv\trg{C, \Hs, \OB{B} \triangleright s;s'}
		&
		\trg{s}\not\equiv\trg{\ifzte{\_}{\_}{\_}} \text{ and } \trg{s}\not\equiv\trg{\lfence}
	}{
		\trg{w, (\Ot,\taintt) \xltot{ \taintt'\glb\taintt } w, (\Ot',\taintt)}
	}{tt-sp-act}
	\typerule{T-\TR-speculate-if}{
		\trg{\taintt'; \Ot \xtot{\taintt} \Ot'}	
		&
		\trg{\Ot}\equiv\trg{C, \Ht, \OB{B} \cdot B \triangleright \proc{s;s'}{\OB{f}\cdot f}}
		&
		\trg{s}\equiv\trg{\ifte{e}{s''}{s'''}}
		\\
		\trg{C}\equiv\trg{\OB{F};\OB{I}}
		&
		\trg{f}\notin\trg{\OB{I}}
		\\
		\text{ if }
			\trg{B \triangleright e\bigred 0:\taintt} 
		\text{ then }
			\trg{\Ot''}\equiv\trg{C, \Ht, \OB{B} \cdot B \triangleright s''';s'}
		\\
		\text{ if }
			\trg{B \triangleright e\bigred n:\taintt} 
			\text{ and }
			\trg{n}>\trg{0}
		\text{ then }
			\trg{\Ot''}\equiv\trg{C, \Ht, \OB{B} \cdot B \triangleright s'';s'}
	}{
		\trg{w, (\Ot,\taintt') \xltot{ \taintt\glb\taintt'} w, (\Ot',\taintt')\cdot(\Ot'',\unta)}
	}{tt-sp-if}
	\typerule{T-\TR-speculate-rollback}{
	}{
		\trg{w, (\Ot,\taintt) \xltot{ \safeta } w, \trge}
	}{tt-sp-rb}

	\typerule{T-\TR-speculate-if-attacker}{
		\trg{\taintt'; \Ot \xtot{\taintt} \Ot'}	
		&
		\trg{\Ot}\equiv\trg{C, \Ht, \OB{B} \cdot B \triangleright \proc{s;s'}{\OB{f}\cdot f}}
		&
		\trg{s}\equiv\trg{\ifte{e}{s''}{s'''}}
		\\
		\trg{C}\equiv\trg{\OB{F};\OB{I}}
		&
		\trg{f}\in\trg{\OB{I}}
	}{
		\trg{w, (\Ot,\taintt') \xltot{ \taintt\glb\taintt'} w, (\Ot',\taintt')}
	}{tt-sp-if}
\end{center}
\botrule

\mytoprule{\trg{\Sigma \xltot{\lambda^\sigma} \Sigma'}}
\begin{center}
	\typerule{Combine-\T}{
		\trg{\Sigma} = \trg{\OB{\Phi}}
		&
		\trg{\Sigma'} = \trg{\OB{\Phi'}}
		\\
		\trg{\OB{\Pv} + \OB{\Pt}} = \trg{\OB{\Phi}}
		&
		\trg{\OB{\Pv'} + \OB{\Pt'}} = \trg{\OB{\Phi'}}
		\\
		\trg{\OB{\Pv} \xltot{\lambda} \OB{\Pv'}}
		&
		\trg{\OB{\Pt} \xltot{\sigma} \OB{\Pt'}}
	}{
		\trg{\Sigma \xltot{\lambda^\sigma} \Sigma'}
	}{t-comb}
\end{center}
\botrule

We allow speculation only when it happens inside the component, not in the attacker (\Cref{tr:et-sp-if}), for this reason.
Suppose a context speculates and call us with a parameter \trg{v} that it loaded speculatively.
This is not a concern because we quantify over all attackers, so there exist also the attacker that would call us with \trg{v} without speculation.
This is the reason why the semantics simplifies the situation and does not let the context speculate.

If we allowed speculation in the context, we would also have to allow the context to possibly access the private memory.
This would be the case for a 'reverse' spectre attack, suppose the standard vulnerable snippet is in the context and that compiled code calls it with an out-of-bound parameter.
The context could effectively access the memory of the component.
Now, we do not model this for a simple reason: this attack can not be defended against if we link with things in the same address space.
If we allow to link with things in different address spaces, then this attack would not be possible anymore.
Since this is not possible, and since this is the only interesting attack that the context can mount, we do not let the context speculate.

We keep track of only those actions that occur inside the component or between component and context, not of those that happen inside the context (\Cref{tr:eut-tr-sin}).

Note that the rule for rollback is nondeterministic: this works because the + of value and taint states is only valid if its sub-stacks have the same cardinality.
Since the only operation that eliminates from the stack is a rollback, the last two rules can only be used in conjunction with their corresponding rule on value states.

\mytoprule{ \trg{\Sigma \Xtot{\trat{^\sigma}} \Sigma'} }
\begin{center}
	\typerule{E-\TR-single}{
		\trg{\Sigma \xltot{\acat{^{\taintt}}} \Sigma'}
		\\
		\trg{\Sigma} = \trg{(w, \OB{(\Omega,m,\taintt)}\cdot (\OB{F},\OB{I},H,\OB{B}\triangleright \proc{s}{\OB{f}\cdot f} , n, \taintt')) }
		\\
		\trg{\Sigma'} = \trg{(w, \OB{(\Omega,m,\taintt)}\cdot (\OB{F},\OB{I},H',\OB{B'}\triangleright \proc{s'}{\OB{f'}\cdot f'} , n', \taintt'')) }
		\\
		\text{ if } \trg{f} == \trg{f'} \text{ and } \trg{f} \in \trg{I} \text{ then } \trg{\trat{^{\taintt}}} = \trg{\epsilon} \text{ else } \trg{\trat{^{\taintt}}} = \trg{\acat{^{\taintt}}}
		\\
		\text{ if } \trg{\acat{^{\taintt}}} == \trg{\rollbl^\safeta} \text{ then } \trg{j}=\trg{n} \text{ else } \trg{j}=\trg{n+1}
	}{
		\trg{(n,\Sigma) \Xtot{\trat{^{\taintt}}} (j,\Sigma')}
	}{eut-tr-sin}
	\typerule{E-\TR-silent}{
		\trg{(n,\Sigma) \Xtot{\trat{_1^{\taintt}}} (n',\Sigma'')}
		&
		\trg{\Sigma'' \xltot{ \epsilon } \Sigma'}
	}{
		\trg{(n,\Sigma) \Xtot{\trat{_1^{\taintt}}} (n'+1,\Sigma')}
	}{eut-tr-silent}
	\typerule{E-\TR-init}
	{
		\
	}{
		\trg{(n,\Sigma) \Xtot{\epsilon} (n,\Sigma)}
	}{eut-tr-init}
\end{center}
\botrule

\mytoprule{\text{Helpers}}
\begin{center}
	\typerule{\TR-Initial State}{
		\trg{H_0} = \trg{H''}\cup\trg{H}\cup\trg{H'}
		\\
		\trg{H'} = \myset{ \trg{n\mapsto 0 : \safeta} }{ \trg{n}\in\mb{N}\setminus\dom{\trg{H}} } 
		\\
		\trg{H''} = \myset{ \trg{-n\mapsto 0 : \unta} }{ \trg{n}\in\mb{N}, \trg{-n}\notin\dom{\trg{H}}  }
		\\
		\trg{\Sigma}\equiv  \trg{ w, (\OB{F} ; \OB{I} ; H_0 ; \trge\cdot x\mapsto 0 \triangleright \call{main}~ x ; \skipt , \bot, \safeta)}
	}{
		\SInitt{(\trg{H ; \OB{F} ; \OB{I}})} = \trg{(0,\Sigma)}
	}{ini-ut}
	\typerule{E-\TR-trace}{
		\exists\trg{\Sigma}\ldotp 
		\trg{\Pv + \Pt} = \trg{\Sigma} \text{ and }
		\vdash \trg{\Pv} : \bot_f 
		&
		\trg{\SInitt{P} \Xtot{\trat{^{\taintt}}} (\_,\Sigma)}
	}{
		\trg{P\semt \trat{^{\taintt}}\trgb{\ts}}
	}{eut-trace}
	\typerule{E-\TR-behaviour}{
	}{
		\trg{\behavt{P}} = \myset{\trat{^{\taintt}}}{ \trg{P\semt \trat{^{\taintt}}} }
	}{eut-beh}
	\typerule{\TR-Terminal State}{
		\trg{\Pv} = \trg{w, \OB{(\Ov,\trgb{\omega})}\cdot(\Ov,\omega)}
		&
		\nexists \trg{\Ov'},\trgb{\lambda} . \trg{\Ov \xtot{\trgb{\lambda}} \Ov'}
	}{
		\vdash \trg{\Pv} : \bot
	}{term-state-t}
	\typerule{\TR-Terminal Ending State}{
		\trg{\Pv} = \trg{w, (\Ov,\bot)}
		&
		\vdash \trg{\Pv} : \bot
	}{
		\vdash \trg{\Pv} : \bot_f
	}{term-state-t}
\end{center}
\botrule

The speculative semantics starts with the pc tag as safe (\trg{\safeta}, \Cref{tr:ini-ut}).
Any misspeculation sets the pc tag as unsafe (\trg{\unta}, \Cref{tr:et-sp-if}).
When a speculation is rolled back, the pc returns to be the one set previously, so when all speculation is rolled back, the pc will return to be \trg{\safeta}, otherwise it'll be \trg{\unta}.
Roll backs (\Cref{tr:et-sp-rb}) are triggered by the window reaching 0. Rules \Cref{tr:et-sp-eps} and \Cref{tr:et-sp-act} decrement the window, whereas  \trg{\lfence} sets the remaining window to 0 (\Cref{tr:et-sp-lf}).

When an action is generated (\Cref{tr:et-sp-act}) it is tagged with the label resulting of the glb ($\glb$) of the label of the action and the label of the pc, thus:
\begin{itemize}
	\item a safe action-generating expression (\safeta) done while not speculating (\safeta) generates a safe action $\safeta\glb\safeta=\safeta$.
	\item an unsafe action-generating expression (\unta) done while not speculating (\safeta) generates a safe action $\unta\glb\safeta=\safeta$.
	\item a safe action-generating expression (\safeta) done while speculating (\unta) generates a safe action $\safeta\glb\unta=\safeta$.
	\item an unsafe action-generating expression (\unta) done while speculating (\unta) generates an unsafe action $\unta\glb\unta=\unta$.
\end{itemize}

\subsection{Alternative Modelling of Attacker Speculation}
We could have let the attacker always speculate.
Then we'd need to change how we keep track of the pc taint.
In fact, all actions done during the speculation started by an attacker can also be done by another attacker (whose code simply has the if guards negated).
Thus, those actions must not be considered unsafe, and thus the pc taint is left \trg{\safeta} (else branch on line 2).

\begin{center}
	\typerule{E-\TR-speculate-if-alt}{
		\trg{\Ov \xtot{\acat} \Ov'}	
		&
		\trg{\Ov}\equiv\trg{C, \Hv, \OB{\Bv} \cdot \Bv \triangleright \proc{s;s'}{\OB{f}\cdot f}}
		&
		\trg{s}\equiv\trg{\ifte{e}{s''}{s'''}}
		\\
		\text{ if }
			\trg{\Bv \triangleright e\bigred 0} 
		\text{ then }
			\trg{\Ov''}\equiv\trg{C, \Hv, \OB{\Bv} \cdot \Bv \triangleright s''';s'}
		\\
		\text{ if }
			\trg{\Bv \triangleright e\bigred n} 
			\text{ and }
			\trg{n}>\trg{0}
		\text{ then }
			\trg{\Ov''}\equiv\trg{C, \Hv, \OB{\Bv} \cdot \Bv \triangleright s'';s'}
		\\
		\trg{j} = \fun{min}{\trg{w},\trg{n}}
	}{
		\trg{w, (\Ov,n+1) \xltot{ \acat{} } w, (\Ov',n)\cdot w,(\Ov'',j)}
	}{et-sp-if-alt}

	\typerule{T-\TR-speculate-if-alt}{
		\trg{\taintt'; \Ot \xtot{\taintt} \Ot'}	
		&
		\trg{\Ot}\equiv\trg{C, \Ht, \OB{B} \cdot B \triangleright \proc{s;s'}{\OB{f}\cdot f}}
		&
		\trg{s}\equiv\trg{\ifte{e}{s''}{s'''}}
		\\
		\trg{C}\equiv\trg{\OB{F};\OB{I}}
		&
		\text{ if } \trg{f}\notin\trg{\OB{I}} \text{ then } \trg{\taintt_f}=\trg{\unta} \text{ else } \trg{\taintt_f} = \trg{\safeta}
		\\
		\text{ if }
			\trg{B \triangleright e\bigred 0:\taintt} 
		\text{ then }
			\trg{\Ot''}\equiv\trg{C, \Ht, \OB{B} \cdot B \triangleright s''';s'}
		\\
		\text{ if }
			\trg{B \triangleright e\bigred n:\taintt} 
			\text{ and }
			\trg{n}>\trg{0}
		\text{ then }
			\trg{\Ot''}\equiv\trg{C, \Ht, \OB{B} \cdot B \triangleright s'';s'}
	}{
		\trg{w, (\Ot,\taintt') \xltot{ \taintt\glb\taintt'} w, (\Ot',\taintt')\cdot(\Ot'',\taintt_f)}
	}{tt-sp-if-alt}
\end{center}

\newpage
\section{\weak{\SR}: a Language with Weak Taint-Tracking}\label{sec:weak-src}
This language considers a different read label, which is manifested in the read semantics rule, as well as a different form of taint tracking.
\begin{align*}
	\mi{Heap\&Pc\ Act.s}~\src{\delta} \bnfdef&\ \cdots \mid \src{(\rdl{n\mapsto v})}
\end{align*}
\begin{center}
	\typerule{E-\SR-read}{
		\src{\Bv \triangleright e\bigreds n}
		&
		\src{\Hv}=\src{{\Hv}_1; \abs{n}\mapsto v ; {\Hv}_2}
	}{
		\src{C, \Hv, \OB{\Bv}\cdot \Bv \triangleright \letread{x}{e}{s}} \xtos{\rdl{\abs{n}\mapsto v}} \src{C, \Hv, \OB{\Bv}\cdot \Bv\cup x\mapsto v \triangleright s}
	}{eus-rd-com-2}

	\typerule{T-\SR-read-prv-weak}{
		\src{B \triangleright e\bigred n : \sigma'}
		&
		\src{n_a} = \src{-\abs{n}}
		&
		\src{\Hs(n_a)} = \src{\sigma''}
		&
		\src{\sigma} = \src{\sigma''\sqcup\sigma'}
	}{
			\src{\sigma_{pc}; C, \Hs, \OB{B}\cdot B \triangleright \letreadp{x}{e}{s} \xto{\sigma\sqcap\sigma_{pc}}
				C, \Hs, \OB{B}\cdot B\cup x\mapsto 0 : \sigma' \sqcap \sigma_{pc} \triangleright s}
	}{tus-rd-com-p-weak}
\end{center}

\section{\weak{\TR}: a Language with Weak Leaks}\label{sec:weak-trg}
This language is obtained by adopting the changes made for \weak{\SR} in \TR. 
\newpage
\section{Examples}\label{sec:example}
In this section we write the classical spectre-susceptible program in both \SR and \TR and see what kind of semantics it yields.

\begin{example}[The classical Spectre V1 attack]\label{ex:spv1}
	We have this pseudocode:
	\begin{lstlisting}
		if (y < size) {
			temp = B[A[y] * 512]
		}
	\end{lstlisting}
	The program checks whether the index stored in the variable y is less than the size of the array \lstinline{A}, stored in the variable size. 
	If that is the case, the program retrieves \lstinline{A[y]}, amplifies it with a multiple (here: 512) of the cache line size, and uses the result as an address for accessing the array \lstinline{B}.
	The memory accesses in line 2 may be executed even if \lstinline{y >= size}.
	However, the speculatively executed memory accesses leave a footprint in the microarchitectural state, in particular in the cache, which enables an adversary to retrieve \lstinline{A[y]}, even for \lstinline{y >= size}, by probing the array \lstinline{B}.

	\begin{align*}
		&
		\src{
			\left.\begin{aligned}
				\src{
					\begin{aligned}[c]
					&
					\src{-1\mapsto 0 : \safeta} ;
					\\
					&
					\src{-2\mapsto 10 : \safeta}
					\\
					&
					\src{-3\mapsto 10 : \safeta}
					\\
					&
					\src{-4\mapsto 30 : \safeta}
					\end{aligned}
				}
				~,
				&~
				\src{
					\begin{aligned}[c]
						&
						\OB{B} \cdot
						\\
						&
						\src{y \mapsto 0 : \safeta}
						;
						\\
						&
						\src{ size \mapsto 3 : \safeta}
					\end{aligned}
				}
				~\triangleright
				&~
				\src{
					\begin{aligned}[c]
						&
						\ifztes{
							\src{(y < size)}
						}{
							\\
							&\ \
							\begin{aligned}
								&
								\src{\letreadps{x_a}{1+y}{}}
								\\
								&
								\src{\letreadps{x_b}{4+x_a}{}}
								\\
								&
								\src{\letins{temp}{x_b}{\skips}}
							\end{aligned}
							\\
							&
						}{
							\skips
						}		
					\end{aligned}
				}
			\end{aligned}\right\}\src{\Omega}
		}
		\\
		&\text{\Cref{tr:eus-ift}}
		\\
		&
		\text{ since } \src{\src{y \mapsto 0 ; size \mapsto 1} \triangleright (y < size) \bigreds 0 : \safeta}
		\\
		\xtos{(\ifl{0})^\safeta}
		&
		\src{
			\left.\begin{aligned}
				\src{
					\begin{aligned}[c]
					&
					\src{-1\mapsto 0 : \safeta} ;
					\\
					&
					\src{-2\mapsto 10 : \safeta}
					\\
					&
					\src{-3\mapsto 10 : \safeta}
					\\
					&
					\src{-4\mapsto 30 : \safeta}
					\end{aligned}
				}
				~,
				&~
				\src{
					\begin{aligned}[c]
						&
						\OB{B} \cdot
						\\
						&
						\src{y \mapsto 0 : \safeta}
						;
						\\
						&
						\src{ size \mapsto 3 : \safeta}
					\end{aligned}
				}
				~\triangleright
				&~
				\src{
					\begin{aligned}[c]
						&
						\src{\letreadps{x_a}{1+y}{}}
						\\
						&
						\src{\letreadps{x_b}{4+x_a}{}}
						\\
						&
						\src{\letins{temp}{x_b}{\skips}}
					\end{aligned}
				}
			\end{aligned}\right\}\src{\Omega_1}
		}	
		\\
		&\text{\Cref{tr:eus-rd-com}}
		\\
		&\text{ since } \src{\src{\cdots} \triangleright 1+y \bigreds 1 : \safeta}  \text{ and } \src{\safeta}\sqcup\src{\safeta} = \src{\safeta}
		\\
		\xtos{\rdl{-1}^\safeta}
		&
		\src{
			\left.\begin{aligned}
				\src{
					\begin{aligned}[c]
					&
					\src{-1\mapsto 0 : \safeta} ;
					\\
					&
					\src{-2\mapsto 10 : \safeta}
					\\
					&
					\src{-3\mapsto 10 : \safeta}
					\\
					&
					\src{-4\mapsto 30 : \safeta}
					\end{aligned}
				}
				~,
				&~
				\src{
					\begin{aligned}[c]
						&
						\OB{B} \cdot
						\\
						&
						\src{y \mapsto 0 : \safeta}
						;
						\\
						&
						\src{ size \mapsto 3 : \safeta}
						\\
						&
						\src{x_a \mapsto 0 : \unta}
					\end{aligned}
				}
				~\triangleright
				&~
				\src{
					\begin{aligned}[c]
						&
						\src{\letreadps{x_b}{4+x_a}{}}
						\\
						&
						\src{\letins{temp}{x_b}{\skips}}
					\end{aligned}
				}
			\end{aligned}\right\}\src{\Omega_2}
		}	
		\\
		&\text{\Cref{tr:eus-rd-com}}
		\\
		&\text{ since } \src{\src{\cdots} \triangleright x_a \bigreds 4 : \unta} \text{ and } \src{\safeta}\sqcup\src{\unta} = \src{\unta}
		\\
		\xtos{\rdl{-4}^{ \unta }}
		&
		\src{
			\left.\begin{aligned}
				\src{
					\begin{aligned}[c]
					&
					\src{-1\mapsto 0 : \safeta} ;
					\\
					&
					\src{-2\mapsto 10 : \safeta}
					\\
					&
					\src{-3\mapsto 10 : \safeta}
					\\
					&
					\src{-4\mapsto 30 : \safeta}
					\end{aligned}
				}
				~,
				&~
				\src{
					\begin{aligned}[c]
						&
						\OB{B} \cdot
						\\
						&
						\src{y \mapsto 0 : \safeta}
						;
						\\
						&
						\src{ size \mapsto 3 : \safeta}
						\\
						&
						\src{x_a \mapsto 0 : \unta}
						\\
						&
						\src{x_b \mapsto 30 : \unta}
					\end{aligned}
				}
				~\triangleright
				&~
				\src{
					\begin{aligned}[c]
						&
						\src{\letins{temp}{x_b}{\skips}}
					\end{aligned}
				}
			\end{aligned}\right\}\src{\Omega_3}
		}	
		\\
		&\text{\Cref{tr:eus-letin}}
		\\
		\xtos{\epsilon}
		&
		\src{
			\left.\begin{aligned}
				\src{
					\begin{aligned}[c]
					&
					\src{-1\mapsto 0 : \safeta} ;
					\\
					&
					\src{-2\mapsto 10 : \safeta}
					\\
					&
					\src{-3\mapsto 10 : \safeta}
					\\
					&
					\src{-4\mapsto 30 : \safeta}
					\end{aligned}
				}
				~,
				&~
				\src{
					\begin{aligned}[c]
						&
						\OB{B} \cdot
						\\
						&
						\src{y \mapsto 0 : \safeta}
						;
						\\
						&
						\src{ size \mapsto 3 : \safeta}
						\\
						&
						\src{x_a \mapsto 0 : \unta}
						\\
						&
						\src{x_b \mapsto 30 : \unta}
						\\
						&
						\src{temp \mapsto 30 : \unta}
					\end{aligned}
				}
				~\triangleright
				&~
				\src{
					\begin{aligned}[c]
						&
						\skips
					\end{aligned}
				}
			\end{aligned}\right\}\src{\Omega_4}
		}
	\end{align*}
	
	This is not going to be a problem for \SR, when calculating the trace semantics, all actions are then turned into safe since there is no speculation in \SR.
	Thus, this program performs the following action:
	\begin{center}\small
		\AxiomC{  }
		\UnaryInfC{ $\src{\Omega\Xtos{}\Omega}$}
			\AxiomC{ $\src{ \Omega \xtos {(\ifl{0})^\safeta} \Omega_1}$ }
		\BinaryInfC{ $\src{\Omega \Xtos{(\ifl{0})^\safeta} \Omega_1}$ }
			\AxiomC{ $\src{\Omega_1 \xtos{\rdl{\ell_A}^{ \safeta }} \Omega_2}$ }
		\BinaryInfC{ $\src{\Omega \Xtos{(\ifl{0})^\safeta \cdot\rdl{\ell_A}^{ \safeta }} \Omega_2}$ }
			\AxiomC{ $\src{\Omega_2 \xtos{\rdl{\ell_B}^{ \unta }} \Omega_3}$ }
		\BinaryInfC{ $\src{\Omega \Xtos{(\ifl{0})^\safeta \cdot\rdl{\ell_A}^{ \safeta } \cdot \rdl{\ell_B}^{ \safeta }} \Omega_3}$ }
			\AxiomC{ $\src{\Omega_3 \xtos{\epsilon} \Omega_4}$ }
		\BinaryInfC{ $\src{\Omega \Xtos{(\ifl{0})^\safeta \cdot\rdl{\ell_A}^{ \safeta } \cdot \rdl{\ell_B}^{ \safeta }} \Omega_4}$ }
		\DisplayProof
	\end{center}

	If we consider the same execution in \TR, however, something different happens when considering the trace semantics.
	Each individual action is generated as before, but the ``if-then-else'' will trigger a speculation, which raises the pc tag to \trg{\unta}.
	We start from a different state with \trg{y\mapsto 2}.

	\begin{align*}
		&
		\trg{
			\left.\begin{aligned}
				\trg{
					\begin{aligned}[c]
					&
					\trg{-1\mapsto 10 : \safeta} ;
					\\
					&
					\trg{-2\mapsto 0 : \safeta}
					\\
					&
					\trg{-3\mapsto 10 : \safeta}
					\\
					&
					\trg{-4\mapsto 30 : \safeta}
					\end{aligned}
				}
				~,
				&~
				\trg{
					\begin{aligned}[c]
						&
						\OB{B} \cdot
						\\
						&
						\trg{y \mapsto 2 : \safeta}
						;
						\\
						&
						\trg{ size \mapsto 1 : \safeta}
					\end{aligned}
				}
				~\triangleright
				&~
				\trg{
					\begin{aligned}[c]
						&
						\ifztet{
							\trg{(y < size)}
						}{
							\\
							&\ \
							\begin{aligned}
								&
								\trg{\letreadpt{x_a}{0+y}{}}
								\\
								&
								\trg{\letreadpt{x_b}{4+x_a}{}}
								\\
								&
								\trg{\letint{temp}{x_b}{\skipt}}
							\end{aligned}
							\\
							&
						}{
							\skipt
						}		
					\end{aligned}
				}
			\end{aligned}\right\}\trg{\Omega}
		}
		\\
		&
		\text{ since } \trg{\trg{y \mapsto 2 ; size \mapsto 1} \triangleright (y < size) \bigredt 1 : \safeta}
		\\
		\xtot{(\ifl{1})^\safeta}
		&
		\trg{
			\left.\begin{aligned}
				\trg{
					\begin{aligned}[c]
					&
					\trg{-1\mapsto 10 : \safeta} ;
					\\
					&
					\trg{-2\mapsto 0 : \safeta}
					\\
					&
					\trg{-3\mapsto 10 : \safeta}
					\\
					&
					\trg{-4\mapsto 30 : \safeta}
					\end{aligned}
				}
				~,
				&~
				\trg{
					\begin{aligned}[c]
						&
						\OB{B} \cdot
						\\
						&
						\trg{y \mapsto 2 : \safeta}
						;
						\\
						&
						\trg{ size \mapsto 1 : \safeta}
					\end{aligned}
				}
				~\triangleright
				&~
				\trg{
					\begin{aligned}[c]
						&
						\skipt
					\end{aligned}
				}
			\end{aligned}\right\}\trg{\Omega'}
		}
		\\
		\\
		&
		\text{ the rest of the states needed for speculation are :}
		\\
		\xtot{(\ifl{0})^\safeta}
		&
		\trg{
			\left.\begin{aligned}
				\trg{
					\begin{aligned}[c]
					&
					\trg{-1\mapsto 10 : \safeta} ;
					\\
					&
					\trg{-2\mapsto 0 : \safeta}
					\\
					&
					\trg{-3\mapsto 10 : \safeta}
					\\
					&
					\trg{-4\mapsto 30 : \safeta}
					\end{aligned}
				}
				~,
				&~
				\trg{
					\begin{aligned}[c]
						&
						\OB{B} \cdot
						\\
						&
						\trg{y \mapsto 2 : \safeta}
						;
						\\
						&
						\trg{ size \mapsto 1 : \safeta}
					\end{aligned}
				}
				~\triangleright
				&~
				\trg{
					\begin{aligned}[c]
						&
						\trg{\letreadpt{x_a}{0+y}{}}
						\\
						&
						\trg{\letreadpt{x_b}{4+x_a}{}}
						\\
						&
						\trg{\letint{temp}{x_b}{\skipt}}
					\end{aligned}
				}
			\end{aligned}\right\}\trg{\Omega_1}
		}	
		\\
		\\
		&\text{ since } \trg{\trg{\cdots} \triangleright 0+y \bigredt 2 : \safeta}  \text{ and } \trg{\safeta}\sqcup\trg{\safeta} = \trg{\safeta}
		\\
		\xtot{\rdl{-2}^\safeta}
		&
		\trg{
			\left.\begin{aligned}
				\trg{
					\begin{aligned}[c]
					&
					\trg{-1\mapsto 10 : \safeta} ;
					\\
					&
					\trg{-2\mapsto 0 : \safeta}
					\\
					&
					\trg{-3\mapsto 10 : \safeta}
					\\
					&
					\trg{-4\mapsto 30 : \safeta}
					\end{aligned}
				}
				~,
				&~
				\trg{
					\begin{aligned}[c]
						&
						\OB{B} \cdot
						\\
						&
						\trg{y \mapsto 2 : \safeta}
						;
						\\
						&
						\trg{ size \mapsto 1 : \safeta}
						\\
						&
						\trg{x_a \mapsto 0 : \unta}
					\end{aligned}
				}
				~\triangleright
				&~
				\trg{
					\begin{aligned}[c]
						&
						\trg{\letreadpt{x_b}{4+x_a}{}}
						\\
						&
						\trg{\letint{temp}{x_b}{\skipt}}
					\end{aligned}
				}
			\end{aligned}\right\}\trg{\Omega_2}
		}	
		\\
		\\
		&\text{ since } \trg{\trg{\cdots} \triangleright x_a \bigredt 4 : \unta} \text{ and } \trg{\safeta}\sqcup\trg{\unta} = \trg{\unta}
		\\
		\xtot{\rdl{-4}^{ \unta }}
		&
		\trg{
			\left.\begin{aligned}
				\trg{
					\begin{aligned}[c]
					&
					\trg{-1\mapsto 10 : \safeta} ;
					\\
					&
					\trg{-2\mapsto 0 : \safeta}
					\\
					&
					\trg{-3\mapsto 10 : \safeta}
					\\
					&
					\trg{-4\mapsto 30 : \safeta}
					\end{aligned}
				}
				~,
				&~
				\trg{
					\begin{aligned}[c]
						&
						\OB{B} \cdot
						\\
						&
						\trg{y \mapsto 2 : \safeta}
						;
						\\
						&
						\trg{ size \mapsto 1 : \safeta}
						\\
						&
						\trg{x_a \mapsto 2 : \unta}
						\\
						&
						\trg{x_b \mapsto 30 : \unta}
					\end{aligned}
				}
				~\triangleright
				&~
				\trg{
					\begin{aligned}[c]
						&
						\trg{\letint{temp}{x_b}{\skipt}}
					\end{aligned}
				}
			\end{aligned}\right\}\trg{\Omega_3}
		}	
		\\
		\\
		\xtot{\epsilon}
		&
		\trg{
			\left.\begin{aligned}
				\trg{
					\begin{aligned}[c]
					&
					\trg{-1\mapsto 10 : \safeta} ;
					\\
					&
					\trg{-2\mapsto 0 : \safeta}
					\\
					&
					\trg{-3\mapsto 10 : \safeta}
					\\
					&
					\trg{-4\mapsto 30 : \safeta}
					\end{aligned}
				}
				~,
				&~
				\trg{
					\begin{aligned}[c]
						&
						\OB{B} \cdot
						\\
						&
						\trg{y \mapsto 2 : \safeta}
						;
						\\
						&
						\trg{ size \mapsto 1 : \safeta}
						\\
						&
						\trg{x_a \mapsto 4 : \unta}
						\\
						&
						\trg{x_b \mapsto 30 : \unta}
						\\
						&
						\trg{temp \mapsto 30 : \unta}
					\end{aligned}
				}
				~\triangleright
				&~
				\trg{
					\begin{aligned}[c]
						&
						\skipt
					\end{aligned}
				}
			\end{aligned}\right\}\trg{\Omega_4}
		}
	\end{align*}
	We now take a look at the speculating semantics for this program.
	We assume we are not already speculating and that \trg{\Omega} is our starting state.
	For simplicity, we consider a speculation window of $4$ steps.

	Given these states, we have the following reductions
	\begin{align*}
		\trg{\Sigma} =&\ \trg{(4, \trge\cdot(\Omega,4,\safeta))}
		\\
		\trg{\Sigma_1} =&\ \trg{(4, \trge\cdot(\Omega',3,\safeta)\cdot(\Omega_1,3,\unta) )}
		\\
		\trg{\Sigma_2} =&\ \trg{(4, \trge\cdot(\Omega',3,\safeta)\cdot(\Omega_2,2,\unta) )}
		\\
		\trg{\Sigma_3} =&\ \trg{(4, \trge\cdot(\Omega',3,\safeta)\cdot(\Omega_3,1,\unta) )}
		\\
		\trg{\Sigma_4} =&\ \trg{(4, \trge\cdot(\Omega',3,\safeta)\cdot(\Omega_4,0,\unta) )}
		\\
		\trg{\Sigma'} =&\ \trg{(4, \trge\cdot(\Omega',3,\safeta))}
	\end{align*}
	\begin{align*}
		\trg{\underbrace{\trg{(4, \trge\cdot(\Omega,4,\safeta))}}_{\Sigma} \xltot{(\ifl{0})^\safeta} \underbrace{\trg{(4, \trge\cdot(\Omega',3,\safeta)\cdot(\Omega_1,3,\unta) )}}_{\Sigma_1}} &\text{ by \Cref{tr:et-sp-if}}
		\\
		\trg{\underbrace{\trg{(4, \trge\cdot(\Omega',3,\safeta)\cdot(\Omega_1,3,\unta) )}}_{\Sigma_1} \xltot{(\rdl{\ell_A})^\safeta} \underbrace{\trg{(4, \trge\cdot(\Omega',3,\safeta)\cdot(\Omega_2,2,\unta) )}}_{\Sigma_2}} &\text{ by \Cref{tr:et-sp-act}}
		\\
		\trg{\underbrace{\trg{(4, \trge\cdot(\Omega',3,\safeta)\cdot(\Omega_2,2,\unta) )}}_{\Sigma_2} \xltot{(\rdl{\ell_B})^\unta} \underbrace{\trg{(4, \trge\cdot(\Omega',3,\safeta)\cdot(\Omega_3,1,\unta) )}}_{\Sigma_3}} &\text{ by \Cref{tr:et-sp-act}}
		\\
		\trg{\underbrace{\trg{(4, \trge\cdot(\Omega',3,\safeta)\cdot(\Omega_3,1,\unta) )}}_{\Sigma_3} \xltot{} \underbrace{\trg{(4, \trge\cdot(\Omega',3,\safeta)\cdot(\Omega_4,0,\unta) )}}_{\Sigma_4}} &\text{ by \Cref{tr:et-sp-eps}}
		\\
		\trg{\underbrace{\trg{(4, \trge\cdot(\Omega',3,\safeta)\cdot(\Omega_4,0,\unta) )}}_{\Sigma_4} \xltot{\rollbl} \underbrace{\trg{(4, \trge\cdot(\Omega',3,\safeta))}}_{\Sigma'}} &\text{ by \Cref{tr:et-sp-rb}}
	\end{align*}
	The crucial reduction is the third one, where the label is tagged as \trg{\unta} since the pc tag is \trg{\unta} and the label itself is \trg{\unta}.
	The second reduction is tagged \trg{\safeta} since the action itself is \trg{\safeta} and so it is ok to perform it even under speculation.

	For simplicity, we omit the \trg{n} parameter in the $\trg{(n,\Sigma)\Xtot{}(n',\Sigma')}$ reductions.

	\begin{center}\small
		\AxiomC{}
		\UnaryInfC{
			$\trg{\Sigma\Xtot{}\Sigma}$
		}
			\AxiomC{$ \trg{\Sigma \xltot{(\ifl{0})^\safeta} \Sigma_1} $}
		\BinaryInfC{$ \trg{\Sigma \Xtot{(\ifl{0})^\safeta} \Sigma_1} $}
			\AxiomC{$ \trg{\Sigma_1 \xltot{(\rdl{\ell_A})^\safeta} \Sigma_2} $}
		\BinaryInfC{$ \trg{\Sigma \Xtot{(\ifl{0})^\safeta\cdot(\rdl{\ell_A})^\safeta} \Sigma_2} $}
			\AxiomC{$ \trg{\Sigma_2 \xltot{(\rdl{\ell_B})^\unta} \Sigma_3} $}
		\BinaryInfC{$ \trg{\Sigma \Xtot{(\ifl{0})^\safeta\cdot(\rdl{\ell_A})^\safeta\cdot(\rdl{\ell_B})^\unta} \Sigma_3} $}
			\AxiomC{$ \trg{\Sigma_3 \xltot{} \Sigma_4} $}
		\BinaryInfC{$ \trg{\Sigma \Xtot{(\ifl{0})^\safeta\cdot(\rdl{\ell_A})^\safeta\cdot(\rdl{\ell_B})^\unta} \Sigma_4} $}
			\AxiomC{$ \trg{\Sigma_4 \xltot{\rollbl} \Sigma'} $}
		\BinaryInfC{$ \trg{\Sigma \Xtot{(\ifl{0})^\safeta\cdot(\rdl{\ell_A})^\safeta\cdot(\rdl{\ell_B})^\unta\cdot\rollbl} \Sigma'} $}
		\DisplayProof
	\end{center}
\end{example}

\begin{example}[Spectre V1 with implicit flow]
	We have this pseudocode:
	\begin{lstlisting}
		y = A[x]
		if (y < size) {
			temp = B[y * 512]
		}
	\end{lstlisting}	
	This is analogous to the code above, but the leak is implicit.

	We could support the thesis that this is not a leak and to do so our semantics would turn the taint of the private heap to \unta only when speculating.
	We do not share that idea so we do not do it.
\end{example}

\subsection{Additional Snippets}
\lstset{language=Java}

\begin{lstlisting}[mathescape,label=lis:spectrev1,caption=The classic Spectre v1 snippet.] 
void get (int y) 
	if (y < size) then 
		temp = B[A[y]*512] 
\end{lstlisting}\label{sec:spectre-listing}

\begin{lstlisting}[mathescape,label=lis:rel-rss-rsni,caption=Code that is \rsni but not \rss.] 
void get_nc (int y) 
	if (y < size) then B[A[y] * 512] else B[A[y] * 512]
\end{lstlisting}\label{sec:spectre-listing}

\begin{lstlisting}[mathescape,label=lis:spectre-10-listing,caption={A variant of the classic Spectre v1 snippet (Example~10 from~\cite{kocher2018examples}).}] 
void get (int y) 
	if (y < size) then
		if (A[y] == 0) then
			temp = B[0];
\end{lstlisting}

\begin{lstlisting}[mathescape,label=lis:spectre-variant-listing,caption={Another variant of the classic Spectre v1 snippet.}] 
void get (int y) 
	x = A[y];
	if (y < size) then
		temp = B[x];
\end{lstlisting}
 
\newpage

\section{Formalisation of Code Snippets}\label{sec:target-snippets}
\lstset{language=General}
Generally, \src{n_A}, \src{n_B}, \trg{n_A} and \trg{n_B} indicate the addresses of arrays \lstinline{A} and \lstinline{B} in the source and target heaps respectively.
Assume variable \lstinline{size} is passed through location \com{1}.

\subsection{Snippet of \Cref{lis:spectrev1}}\label{test}
\begin{align*}
	&
	\src{get(y)\mapsto}
		\letreads{\src{size}}{\src{1}}{
		\ifztes{
			\src{y<size}
		\\
		&\quad
		}{
			\src{\letreadps{{x}}{{n_A + y}}{}}
				\src{\letreads{{temp}}{{n_B + x} }{\skips}}
		\\
		&\quad
		}{
			\skips
		}
	}
\end{align*}

\subsection{Snippet of \Cref{lis:rel-rss-rsni}}
\begin{align*}
	&
	\src{get(y)\mapsto}
		\letreads{\src{size}}{\src{1}}{
		\ifztes{
				\src{y<size}
			\\
			&\quad
			}{
				\src{\letreadps{{x}}{{n_A + y}}{}}
					\src{\letreads{{temp}}{{n_B + x} }{\skips}}
			\\
			&\quad
			}{
				\src{\letreadps{{x}}{{n_A + y}}{}}
					\src{\letreads{{temp}}{{n_B + x} }{\skips}}
			}
		}
\end{align*}

\subsection{Snippet of \Cref{lis:spectre-1-listing-compiled-msr}}
\begin{align*}
	&
	\trg{get(y)\mapsto}
		\letreadt{\trg{size}}{\trg{1}}{
		\ifztet{
			\trg{y<size}
			\\
			&\quad
			}{
				\trg{ \lfence ; }\
					\trg{\letreadpt{\trg{x}}{\trg{n_A + y}}{}}
					\\
					&\qquad
					\trg{\letreadt{\trg{temp}}{\trg{n_B + x} }{\skipt}}
			\quad
			\quad
			}{
				\skipt
			}
		}
\end{align*}

\subsection{Snippet of \Cref{lis:spectre-10-listing}}
\begin{align*}
	&
	\src{get(y)\mapsto}
		\letreads{\src{size}}{\src{1}}{}
		\ifztes{
			\src{y<size}
			\\
			&\quad
			}{
				\src{\letreadps{\src{x}}{\src{n_A + y}}{}}
					\ifztes{
						\src{x == 0}
					\\
					&\qquad\enskip
					}{
						\src{\letreads{\src{temp}}{\src{n_B + 0} }{\skips}}
					\\
					&\qquad\enskip
					}{
						\skips
					}
			\quad
			\quad
			}{
				\skips
			}
\end{align*}

\subsection{Snippet of \Cref{lis:spectre-10-listing-compiled}}
\begin{align*}
	&
	\trg{get(y)\mapsto}
	\\
	&\enskip
		\letreadt{\trg{size}}{\trg{1}}{}
		\ifztet{
				\trg{y<size}
			\\
			&\quad
			}{
				\trg{\letreadpt{\trg{x}}{\trg{n_A + y}}{}}
					\ifztet{
						\trg{x == 0}
					\\
					&\qquad\enskip
					}{
						\trg{\lfence;}\
						\trg{\letreadt{\trg{temp}}{\trg{n_B + 0} }{\skipt}}
					\\
					&\qquad\enskip
					}{
						\skipt
					}
			\quad\quad
			}{
				\skipt
			}
\end{align*}

\subsection{Snippet of \Cref{lis:spectre-1-listing-compiled-clang}}
Here we are saving the predicate bit in location \trg{-1} so source \src{n_a} is stored at \trg{n_a -1}, which is the address that gets calculated in the reads.
\begin{align*}
	&
	\trg{get(y)\mapsto}
	\\
	&\enskip
		\letreadt{\trg{size}}{\trg{1}}{
			\letint{
				\trg{x_g}
			}{
				\trg{y<size}
			}{
				\ifztet{
					\trg{x_g}
				\\
				&\quad
				}{
					\letreadpt{
						\trg{x}
					}{
						\trg{-1}
					}{
						\trg{\asgnpt{-1}{x \vee \neg x_g} ; }\
				\\
				&\qquad
						\letreadpt{
							\trg{x}
						}{
							\trg{n_a + y + 1}
						}{
							\letreadpt{
								\trg{\predState}
							}{
								\trg{-1}
							}{
				\\
				&\qquad
								\cmovet{
									\trg{x}
								}{
									\trg{0}
								}{
									\trg{\predState}
								}{
									\letreadt{
										\trg{temp}
									}{
										\trg{n_b + x}
									}{
										\skipt
									}
								}
							}
						}
					}
				\\
				&\quad
				}{
					\letreadpt{
						\trg{x}
					}{
						\trg{-1}
					}{
						\trg{\asgnpt{-1}{x \vee x_g} ; }\ 
						\skipt
					}
				}
			}
		}
\end{align*}

\subsection{Snippet of \Cref{lis:spectre-variant-listing}}
\begin{align*}
	&
	\src{get(y)\mapsto}
	\letreads{
		\src{size}
	}{
		\src{1}
	}{
		\letreadps{
			\src{x}
		}{
			\src{n_a+y}
		}{
			\ifztes{
				\src{y<size}
			\\
			&\quad
			}{
				\src{\letreads{{temp}}{{n_B + x} }{\skips}}
			\quad \quad
			}{
				\skips
			}
		}
	}
\end{align*}

\subsection{Snippet of \Cref{lis:spectre-variant-listing-compiled-clang}}\label{sec:code2}
\begin{align*}
	&
	\trg{get(y)\mapsto}
	\\
	&\enskip
		\letreadt{
			\trg{size}
		}{
			\trg{1}
		}{
			\letreadpt{
				\trg{x}
			}{
				\trg{n_a + y + 1}
			}{
			\\
			&\enskip
				\letreadpt{
					\trg{\predState}
				}{
					\trg{-1}
				}{
					\cmovet{
						\trg{x}
					}{
						\trg{
							0
						}
					}{
						\trg{\predState}
					}{
						\ifztet{
							\trg{y < size}
							\\
							&\quad
						}{
							\letreadpt{
								\trg{x_f}
							}{
								\trg{-1}
							}{
								\trg{\asgnpt{-1}{x_f \vee \neg x_g} ; }\
							\\
							&\qquad
								\letreadt{
									\trg{temp}
								}{
									\trg{n_b + x}
								}{
									\skipt
								}
							}
							\\
							&\quad
						}{
							\letreadpt{
								\trg{x_f}
							}{
								\trg{-1}
							}{
								\trg{\asgnpt{-1}{x_f \vee x_g} ; }\ 
								\skipt
							}
						}
					}
				}
			}
		}
\end{align*}

\subsection{Snippet of \Cref{lis:spectre-variant-listing-proc}}\label{sec:code1}
In this case the predicate bit is stored in each function in a local variable \trg{\predState}, so heap accesses are not shifted by 1 in the target.
Below is the source code and its compiled counterpart.
\begin{align*}
	&
	\src{get(y)\mapsto}
	\letreads{
		\src{size}
	}{
		\src{1}
	}{
		\letreadps{
			\src{x}
		}{
			\src{n_a+y}
		}{
			\ifztes{
				\src{y<size}
			\\
			&\quad
			}{
				\src{\call{get2}~ x}
			\quad
			\quad
			}{
				\skips
			}
		}
	}
	\\
	&
	\src{get2(x)\mapsto}
	\
	\src{
		\letreads{
			{temp}
		}{
			{n_B + x} 
		}{
			\skips
		}
	}
\\
	&
	\trg{get(y) \mapsto}
	\\
	&\enskip
		\letint{
			\trg{\predState}
		}{
			\trg{1}
		}{
			\letreadt{
				\trg{size}
			}{
				\trg{1}
			}{
				\letreadpt{
					\trg{x}
				}{
					\trg{n_a + y}
				}{
				\\
				&\enskip
					\letint{
						\trg{x_g}
					}{
						\trg{y < size}
					}{
						\ifztet{
							\trg{x_g}
						\\
						&\quad
						}{
							\letint{
								\trg{\predState}
							}{
								\trg{\predState \vee \neg x_g }
							}{
								\trg{\call{get2}~x}
							}
						\\
						&\quad
						}{
							\letint{
								\trg{\predState}
							}{
								\trg{\predState \vee x_g }
							}{
								\skipt
							}
						}
					}
				}
			}
		}
	\\
	&
	\trg{get2(x) \mapsto}
	\trg{
		\letint{
			\trg{\predState}
		}{
			\trg{1}
		}{
			\letreadt{
				\trg{temp}
			}{
				\trg{n+b + x}
			}{
				\skipt
			}
		}
	}
\end{align*}

\newpage
\section{Speculative Safety}

We need to define some helpers first.

A heap is valid if and only if its private part contains unsafe values.
We do not enforce this on the public part too because a program may write a private value in the public heap.
\begin{definition}[Valid Heap]\label{def:val-heap}
	\begin{align*}
		\vdash \com{H} :\com{vld} \isdef&\
			\forall \com{n}\mapsto \com{v}:\com{\sigma} \in \com{H}, \text{ if } \com{n}<0 \text{ then } \com{\sigma}=\com{\unta}
	\end{align*}
\end{definition}

An attacker is one that has no private heap nor instructions to manipulate it directly.
\begin{definition}[Attacker]\label{def:atk}
	\begin{align*}
		\vdash\com{\ctxc{}} :\com{atk} \isdef&\
			\com{\ctxc{}} \equiv \com{H};\com{\OB{F}} 
			\text{ and }
			\forall \com{n \mapsto v : \sigma} \in \com{H}, \com{n}\geq0
			\\&\
			\text{ and }
			\forall \com{f(x)\mapsto s;\ret} \in \com{\OB{F}}, \com{\letreadp{x}{e}{s'}}, \com{\asgnp{e}{e}} \notin \com{s}
	\end{align*}
\end{definition}

A whole program of some language \com{L} is speculatively safe (\ssdef(L)) if all its actions are safe.
\begin{definition}[Speculative Safety (\ss{}(L))]\label{def:dss}
\begin{align*}
	\vdash
	\com{P} : \ss(L) \isdef&\
		\forall \trac{^\sigma}\in\behavc{P}\ldotp \forall \acac{^\sigma}\in\trac{^\sigma}\ldotp \com{\sigma}\equiv\com{\safeta}
\end{align*}
\end{definition}

A component is robustly speculatively safe (\rssdef{}(L)) if it is safe no matter what attacker it is linked againt.
\begin{definition}[Robust Speculative Safety (\rss{}(L))]\label{def:rdss}
\begin{align*}
	\vdash\com{P} : \rss(L) \isdef&\
		\forall \ctxc{} \ldotp \text{ if } \vdash\ctxc{}:\com{atk} \text{ then } \vdash\com{\ctxc{}\hole{P}} : \ss(L)
\end{align*}	
\end{definition}

If we consider the code of \Cref{ex:spv1}, we see that the \SR execution only produces \src{\safeta} actions.
This is intuitive, since there is no speculation in \SR.

Still in \Cref{ex:spv1}, the \TR execution can produce \trg{\unta} actions, precisely when a memory access is done speculatively and based on other memory-accessed data.

\newpage
\section{Review: Speculative Non-Interference}

We first define what it means for heaps and programs to be low equivalent.
\begin{definition}[Low-equivalence for heaps and programs]\label{def:loweq}
	\begin{align*}
		\com{H \loweq H'} \isdef&\
			\vdash \com{H} : \com{vld} \wedge \vdash \com{H'} : \com{vld} \wedge \dom{H} = \dom{H'}\wedge 
			\\&\
			\forall \com{n\mapsto v:\sigma}\in \com{H}, \text{ if } \com{n}\geq0 \text{ then } \com{n\mapsto v:\sigma} \in \com{H'}
			\\&\
			\text{ if } \com{n}<0 \text{ then } \exists v'.\ \com{n\mapsto v':\sigma} \in \com{H'}
		\\
		\com{P\loweq P'} \isdef&\
			\com{P}\equiv \com{H;\OB{F};\OB{I}} \text{ and } \com{P'}\equiv \com{H';\OB{F};\OB{I}} \text{ and } \com{H \loweq H'}
	\end{align*}
\end{definition}

We can define the non-speculative projection of a trace:

\begin{align*}
	\nspecproj{\trac{}} =&\ \nspecproj{(\trac{}, 0)}
	\\
	\nspecproj{(\come, n)} =&\ \come
	\\
	\nspecproj{(\acac{^\sigma} \cdot \trac{} , 0)} =&\  \acac{^\sigma} \cdot \nspecproj{(\trac{} , 0)}	\\
	\nspecproj{(\delta{^\sigma} \cdot \trac{} , 0)} =&\  \delta{^\sigma} \cdot \nspecproj{(\trac{} , 0)} &\text{ where } \com{\delta^\sigma}\neq\com{(\ifl{v})^\sigma}
	\\
	\nspecproj{((\ifl{v})^\sigma \cdot \trac{} , 0)} =&\  (\ifl{v})^\sigma \cdot \nspecproj{(\trac{} , 1)} 
	\\
	\nspecproj{(\acac{^\sigma} \cdot \trac{} , n+1)} =&\   \nspecproj{(\trac{} , 0)}	\\
	\nspecproj{(\delta{^\sigma} \cdot \trac{} , n+1)} =&\   \nspecproj{(\trac{} , n+1)} &\text{ where } \com{\delta^\sigma}\neq\com{(\ifl{v})^\sigma}
	\\
	\nspecproj{((\ifl{v})^\sigma \cdot \trac{} , n+1)} =&\   \nspecproj{(\trac{} , n+2)} 
	\\
	\nspecproj{((\rollbl)^\sigma \cdot \trac{}, n+1)} =&\  \nspecproj{(\trac{} , n)}
\end{align*}

A program of a language \com{L} is SNI if, taking any other program that is low-equivalent to itself, if their non-speculative traces are non-interferent, then their speculative traces must also be non-interferent.
\begin{definition}[Speculative Non-interference]\label{def:sni}
	\begin{align*}
		\vdash \com{P} : \sni(L) \isdef&\
			\forall \com{P' \loweq P},
			\forall \trac{_1}\in\behavc{\SInit{P}}, \trac{_2}\in\behavc{\SInit{P'}} 
			\\&\
			\text{ if } \nspecproj{\trac{_1}} = \nspecproj{\trac{_1}} \text{ then } \trac{_1} = \trac{_2}	
	\end{align*}
\end{definition}

A component it robustly SNI if it is SNI for any attacker it links against.
\begin{definition}[Robust Speculative Non-Interference]\label{def:rsni}
	\begin{align*}
		\vdash \com{P} : \rsni(L) \isdef&\
			\forall \ctxc{} \text{ if } \vdash \ctxc{} : \com{atk} \text{ then } \vdash \com{\ctxc{}\hole{P}} : \sni(L)
	\end{align*}
\end{definition}

\newpage
\section{Security Definitions and their Implications}

Here, we show the relationship between speculative safety and speculative non-interference.
We prove these relationships only for programs in $\TR$, since both definitions are trivially satisfied in $\SR$ (this immediately follows from $\SR$ not having any speculative behavior; see \Cref{thm:ss-sni-source}).

\newcommand{\safe}[1]{\vdash #1 : \mathit{safe}}
\newcommand{\unsafe}[1]{\vdash #1 : \mathit{unsafe}}

\begin{theorem}[\ssdef{} and \snidef hold for all source programs]\label{thm:ss-sni-source}
\begin{align*}
	\forall \src{P} \in \SR.
	 \vdash \src{P} : \ss(\SR) 
	\text{ and } \vdash \src{P} : \sni(\SR) \\
	\forall \src{P} \in \src{L^-}.
	 \vdash \src{P} : \ss(\src{L^-}) 
	\text{ and } \vdash \src{P} : \sni(\src{L^-}) 
\end{align*}
\end{theorem}
\begin{proof}
This trivially holds from (1) programs in $\SR$ and $\src{L^-}$ produce only actions labelled with $\src{\safeta}$ (for \ss{}(\SR) and \ss{}(\src{L^-})), and (2) traces are identical to their non-speculative projection for programs $\SR$ and $\src{L^-}$ (for \sni{}(\SR) and \sni{}(\src{L^-})).
\end{proof}

\subsection{SS implies SNI}

\subsubsection{Strong Variants}

\begin{theorem}[\ssdef{} implies \snidef (strong)]\label{thm:ss-impl-sni}
\begin{align*}
	\forall \trg{P} \in \TR.
	\text{ if } \vdash \trg{P} : \ss(\TR),  
	\text{ then } \vdash \trg{P} : \sni(\TR)
\end{align*}
\end{theorem}

\begin{proof}
Let $\trg{P}$ be an arbitrary program in $\TR$ such that $\vdash \trg{P}:\ss(\TR)$.
Assume, for contradiction's sake, that $\vdash \trg{P}:\sni(\TR)$ does not hold.
That is, there is another program $\trg{P'}$ and traces $\trat{_1}\in\behavt{\SInit{\trg{P}}}$, $\trat{_2}\in\behavt{\SInit{\trg{P'}}}$ such that $\trg{P} \loweq \trg{P'}$, $\nspecproj{\trat{_1}} = \nspecproj{\trat{_2}}$, and $\trat{_1} \neq \trat{_2}$.
By unrolling the $\TR$-Terminal State rule, we have that $\trg{P\semt \trat{_1}}$ and $\trg{P'\semt \trat{_2}}$.
By unrolling the rule E-$\TR$-trace, we have that there are $\trg{\Sigma}$, $\trg{\Sigma'}$, $\trg{n}$, $\trg{n'}$ such that $\vdash \trg{\Sigma} : \bot$, $\vdash \trg{\Sigma'} : \bot$, $\trg{(0,\SInitt{P}) \Xtot{\trat{_1}} (n,\Sigma)}$, and $\trg{(0,\SInitt{P'}) \Xtot{\trat{_2}} (n',\Sigma')}$.
From $\trg{P} \loweq \trg{P'}$ and \cref{lemma:low-equivalent-programs-low-equivalent-states}, we get $\trg{(0,\SInitt{P})} \approx \trg{(0,\SInitt{P'}) }$. 
There are two cases:
\begin{description}
	\item[$\trg{n} = \trg{n'}$:]
	Then,  we have $\safe{\trg{\Sigma_0}}$ and from  $\vdash \trg{P}:\ss(\TR)$, we get that $\safe{\trat{_1}}$.
	From $\trg{(0,\SInitt{P})} \approx \trg{(0,\SInitt{P'}) }$, $\trg{(0,\SInitt{P}) \Xtot{\trat{_1}} (n,\Sigma)}$,  $\trg{(0,\SInitt{P'}) \Xtot{\trat{_2}} (n',\Sigma')}$, and \cref{lemma:trans-semantics-preserve-low-equivalence-or-different-non-spec-projections}, we get that $\trg{\nspecProject{\trat{_1}}} <> \trg{\nspecProject{\trat{_2}}} \vee \trg{\trat{_1}} = \trg{\trat{_2}}$. 
	From this and $\nspecproj{\trat{_1}} = \nspecproj{\trat{_2}}$, we get $\trg{\trat{_1}} = \trg{\trat{_2}}$ leading to a contradiction.
	
	\item[$\trg{n} \neq \trg{n'}$:]
	Let $\trg{m}$ be the maximum value such that 
	$\trg{(0,\SInitt{P}) \Xtot{\lambda} (m,\Sigma_1)}$ and
	$\trg{(0,\SInitt{P'}) \Xtot{\lambda} (m,\Sigma_1')}$.
	Observe that (1) such an $\trg{m}$ always exists, and (2) $\trg{\lambda}$ is a prefix to both $\safe{\trat{_1}}$ and $\safe{\trat{_2}}$
	From $\vdash \trg{P}:\ss(\TR)$ and $\trg{\lambda}$ being a prefix of  $\safe{\trat{_1}}$, we have $\safe{\trg{\lambda}}$.
	From $\trg{(0,\SInitt{P})} \approx \trg{(0,\SInitt{P'}) }$, $\trg{(0,\SInitt{P}) \Xtot{\lambda} (m,\Sigma_1)}$,
	$\trg{(0,\SInitt{P'}) \Xtot{\lambda} (m,\Sigma_1')}$, $\safe{\trg{\lambda}}$, and \cref{lemma:trans-semantics-preserve-low-equivalence-or-different-non-spec-projections}, we have that $\trg{\Sigma_1} \approx \trg{\Sigma_1'}$.
	There are three cases:
	\begin{description}
	\item[$\neg \exists \trg{\alpha},\trg{\sigma},\trg{\Sigma_2}.\ \trg{(0,\SInitt{P}) \Xtot{\lambda \cdot \alpha^\sigma} (m+1,\Sigma_2)}$:]
	From this, we get that $\Sigma_1 = \Sigma$ and therefore $\vdash \trg{\Sigma} : \bot$.
	From this and $\trg{\Sigma_1} \approx \trg{\Sigma_1'}$, we get that $\vdash \trg{\Sigma'} : \bot$ holds as well.
	From this, we get that $\trg{m} = \trg{n} = \trg{n'}$, leading to a contradiction.
	
	\item[$\neg \exists \trg{\alpha'},\trg{\sigma'},\trg{\Sigma_2'}.\ \trg{(0,\SInitt{P'}) \Xtot{\lambda \cdot {\alpha'}^{\sigma'}} (m+1,\Sigma_2')}$:]
	The proof of this case is similar to the case $\neg \exists \trg{\alpha},\trg{\sigma},\trg{\Sigma_2}.\ \trg{(0,\SInitt{P}) \Xtot{\lambda \cdot \alpha^\sigma} (m+1,\Sigma_2)}$.
		
	\item[$\exists \trg{\alpha},\trg{\sigma},\trg{\Sigma_2},\trg{\alpha'},\trg{\sigma'},\trg{\Sigma_2'}.\ \trg{(0,\SInitt{P}) \Xtot{\lambda \cdot \alpha^\sigma} (m+1,\Sigma_2)} \wedge \trg{(0,\SInitt{P'}) \Xtot{\lambda \cdot {\alpha'}^{\sigma'}} (m+1,\Sigma_2')} \wedge \trg{\alpha^\sigma} \neq \trg{ {\alpha'}^{\sigma'} }$:]
	Observe that $\trg{\lambda \cdot \alpha^\sigma}$ is a prefix of $\trg{\trat{_1}}$.
	Therefore, $\safe{\trg{\lambda \cdot \alpha^\sigma}}$ follows from $\vdash \trg{P}:\ss$.
	From $\trg{(0,\SInitt{P})} \approx \trg{(0,\SInitt{P'}) }$, $\trg{(0,\SInitt{P}) \Xtot{\lambda \cdot \alpha^\sigma} (m+1,\Sigma_2)}$, $\trg{(0,\SInitt{P'}) }$ $\trg{\Xtot{\lambda \cdot {\alpha'}^{\sigma'}} (m+1,\Sigma_2')}$, $\safe{\trg{\lambda \cdot \alpha^\sigma}}$, and \cref{lemma:trans-semantics-preserve-low-equivalence-or-different-non-spec-projections}, we have $\trg{\nspecProject{(\lambda \cdot \alpha^\sigma)}}$ $<> \trg{\nspecProject{(\lambda \cdot {\alpha'}^{\sigma'})}} \vee (\trg{\Sigma_2} \approx \trg{\Sigma_2'} \wedge  \trg{\lambda \cdot \alpha^\sigma} = \trg{\lambda \cdot {\alpha'}^{\sigma'}})$.
	There are two cases:
	\begin{description}
	\item[$\trg{\nspecProject{(\lambda \cdot \alpha^\sigma)}} <> \trg{\nspecProject{(\lambda \cdot {\alpha'}^{\sigma'})}}$:]
	This contradicts $\nspecproj{\trat{_1}} = \nspecproj{\trat{_2}}$ given that $\trg{\lambda \cdot \alpha^\sigma}$ is a prefix of $\trg{\trat{_1}}$ and $\trg{\lambda \cdot {\alpha'}^{\sigma'}}$ is a prefix of $\trg{\trat{_2}}$.
	\item[$(\trg{\Sigma_2} \approx \trg{\Sigma_2'} \wedge  \trg{\lambda \cdot \alpha^\sigma} = \trg{\lambda \cdot {\alpha'}^{\sigma'}})$:]
	This leads to a contradiction with $\trg{\alpha^\sigma} \neq \trg{ {\alpha'}^{\sigma'} }$.
	\end{description}
	\end{description} 
\end{description}
Since all cases lead to a contradiction, this completes the proof of our theorem.
\end{proof}

\subsubsection{Low-equivalence and inital states}
\begin{lemma}[Low-equivalent programs have low-equivalent initial states]\label{lemma:low-equivalent-programs-low-equivalent-states}
\begin{align*}
\forall \trg{P}, \trg{P'}. &
	\text{ if } \trg{P} \loweq \trg{P'} 
	\text{ then } \SInitt{P} \approx \SInitt{P'}
\end{align*}
\end{lemma}

\begin{proof}
Let $\trg{P} = (\trg{H ; \OB{F} ; \OB{I}})$ and $\trg{P'} = (\trg{H' ; \OB{F'} ; \OB{I'}})$ be two arbitrary programs such that $\trg{P} \loweq \trg{P'}$.
Then, $\SInitt{P}$ and $\SInitt{P'}$ are as follows:
\begin{align*}
	\trg{H_0} &= \trg{H_1}\cup\trg{H}\cup\trg{H_2}\\
	\trg{H_1} &= \myset{ \trg{n\mapsto 0 : \safeta} }{ \trg{n}\in\mb{N}\setminus\dom{\trg{H}} } \\
	\trg{H_2} &= \myset{ \trg{-n\mapsto 0 : \unta} }{ \trg{n}\in\mb{N}, \trg{-n}\notin\dom{\trg{H}}  }\\
	\SInitt{P} \isdef & \trg{ w, (\OB{F} ; \OB{I} ; H_0 ; \trge\cdot x\mapsto 0 \triangleright \call{main}~ x ; \skipt , \bot, \safeta)}\\
	\trg{H_0'} &= \trg{H_1'}\cup\trg{H'}\cup\trg{H_2'}\\
	\trg{H_1'} &= \myset{ \trg{n\mapsto 0 : \safeta} }{ \trg{n}\in\mb{N}\setminus\dom{\trg{H'}} } \\
	\trg{H_2'} &= \myset{ \trg{-n\mapsto 0 : \unta} }{ \trg{n}\in\mb{N}, \trg{-n}\notin\dom{\trg{H'}}  }\\
	\SInitt{P'} \isdef & \trg{ w, (\OB{F} ; \OB{I} ; H_0' ; \trge\cdot x\mapsto 0 \triangleright \call{main}~ x ; \skipt , \bot, \safeta)}	
\end{align*}
Moreover, from $\trg{P} \loweq \trg{P'}$, we get $\trg{H} \approx \trg{H'}$.
Therefore, $\SInitt{P} \approx \SInitt{P'}$ immediately follows.
\end{proof}

Before continuing with our proof, we define what it means for program states to be safe-equivalent.

\begin{definition}[Safe-equivalence for heaps and program states]\label{def:safeeq}
	\begin{align*}
		\unsafe{w, ss_0 \cdot (\Omega, m, \sigma)} \isdef & \sigma = \unta\\
		\safe{w, (\Omega, m, \sigma)} \isdef & \sigma = \safeta\\
		\safe{\epsilon} \isdef & \mathit{true} \\
		\safe{\alpha^\sigma} \isdef & \sigma = \safeta \\
		\safe{\trac{^\sigma} \cdot \acac{^\sigma}} \isdef & 
			\safe{\trac{^\sigma}} \text{ and } \safe{\acac{^\sigma}}\\
		\vdash H(n) : \mathit{def} \isdef& \exists v, \sigma.\ H(n) = v:\sigma\\
		\vdash B(x) : \mathit{def} \isdef& \exists v, \sigma.\ B(x) = v:\sigma\\
		\com{v:\sigma \approx v':\sigma'} \isdef&\ 
			\sigma = \sigma' \text{ and }
			\text{if } \sigma = \safeta \text{ then } v = v'\\
		\com{H \approx H'} \isdef&\
			\forall \com{n}.\ 
				\vdash H(n) : \mathit{def} \text{ iff } \vdash H'(n) : \mathit{def} \text{ and }\\
				& \qquad 
				\text{ if }\vdash H(n) : \mathit{def} \text{ then }	H(n) \approx H'(n) \\
		\com{B \approx B'} \isdef&\
				\forall \com{x}.\
				\vdash B(x) : \mathit{def} \text{ iff } \vdash B'(x) : \mathit{def} \text{ and }\\
				& \qquad 
				\text{ if }\vdash B(x) : \mathit{def} \text{ then }	B(x) \approx B'(x) \\
		\com{\OB{B} \cdot B \approx \OB{B}' \cdot B'} \isdef&\
			\com{\OB{B} \approx \OB{B}'} \text{ and } \com{B \approx B'} \\
		\com{\Omega \approx \Omega'} \isdef&\
			\com{\Omega}\equiv \com{C, H, \OB{B} \triangleright s}  \text{ and } \com{\Omega'}\equiv \com{C, H', \OB{B}' \triangleright s} \text{ and } \com{H \approx H'} \text{ and } \com{\OB{B} \approx \OB{B}'}
		\\
		\com{\emptyset \approx \emptyset}
		\\
		\com{ss \approx ss'} \isdef&\
		\com{ss} \equiv \com{ss_0 \cdot (\Omega, m, \sigma)} \text{ and }
		\com{ss'} \equiv \com{ss_0' \cdot (\Omega', m, \sigma)} \text{ and }
		\com{ss_0 \approx ss_0'} \text{ and }
		\com{\Omega \approx \Omega'}
		\\
		\com{\Sigma \approx \Sigma'} \isdef&\
			\com{\Sigma} \equiv (w, ss) \text{ and } \com{\Sigma'} \equiv (w, ss') \text{ and }
				\com{ss} \approx \com{ss'}
	\end{align*}
\end{definition}

\subsubsection{Transitive-closure semantics $\Xtot{}$}

\begin{definition}
\begin{align*}
	\alpha^\sigma <> {\alpha'}^{\sigma'} &\text{if } \alpha \neq \alpha' \vee \sigma \neq \sigma'\\	
	\alpha^\sigma \cdot \lambda <> {\alpha'}^{\sigma'} \cdot \lambda' & \text{if } \alpha^\sigma <> {\alpha'}^{\sigma'} \vee \lambda <> \lambda'
\end{align*}	
\end{definition}

\begin{lemma}[Steps of $\Xtot{}$ preserve low-equivalence or produce distinct non-speculative projections]\label{lemma:trans-semantics-preserve-low-equivalence-or-different-non-spec-projections}
\begin{align*}
	\forall  \trg{P}, \trg{P'}, \trg{\Sigma_0}, \trg{\Sigma_0'}, \trg{\Sigma_1}, \trg{\Sigma_1'},\trg{\lambda}, \trg{\lambda'}, \trg{n}.
	& \text{ if } 
		\trg{\Sigma_0} = \trg{\SInit{P}}, \trg{\Sigma_0'} = \trg{\SInit{P'}}, \\
	& \quad
		\trg{(0,\Sigma_0) \Xtot{\lambda} (n,\Sigma_1)}, \trg{(0,\Sigma_0') \Xtot{\lambda'} (n,\Sigma_1')}, \\
	& \quad \trg{\Sigma_0} \approx \trg{\Sigma_0'}, \safe{\trg{\Sigma_0}}, \safe{\trg{\lambda}}, \\%
	& \text{ then }  
		\trg{\nspecProject{\lambda}} <> \trg{\nspecProject{\lambda'}}
	 \vee (\trg{\Sigma_1} \approx \trg{\Sigma_1'} \wedge  \trg{\lambda} = \trg{\lambda'})
\end{align*}
\end{lemma}

\begin{proof}
Let $\trg{P}$, $\trg{P'}$, $\trg{\Sigma_0}$, $\trg{\Sigma_0'}$, $\trg{\Sigma_1}$, $\trg{\Sigma_1'}$, $\trg{\lambda}$, $\trg{\lambda'}$, $\trg{n}$ be such that 	$\trg{\Sigma_0} = \trg{\SInit{P}}$, $\trg{\Sigma_0'} = \trg{\SInit{P'}}$, $\trg{(0,\Sigma_0) \Xtot{\lambda} (n,\Sigma_1)}$, $\trg{(0,\Sigma_0') \Xtot{\lambda'} (n,\Sigma_1')}$, $\trg{\Sigma_0} \approx \trg{\Sigma_0'}$, $\safe{\trg{\Sigma_0}}$, and $\safe{\trg{\lambda}}$. %
We prove the lemma by  induction on $n$:
\begin{description}
	\item [Base case:]
	For the base case, we consider $n = 0$.
	Then, both $\trg{(0,\Sigma_0) \Xtot{\lambda} (n,\Sigma_1)}$ and $\trg{(0,\Sigma_0') \Xtot{\lambda'} (n,\Sigma_1')}$ have been derived using the E-\TR-init rule.
	Therefore, $\trg{\lambda} = \trg{\epsilon}$, $\trg{\Sigma_1} = \trg{\Sigma_0}$, $\trg{\lambda'} = \trg{\epsilon}$, and $\trg{\Sigma_1'} = \trg{\Sigma_0'}$.
	As a result, $\trg{\lambda} = \trg{\lambda'}$ follows from $\trg{\lambda} = \trg{\epsilon}$ and $\trg{\lambda'} = \trg{\epsilon}$, whereas $\trg{\Sigma_1} \approx \trg{\Sigma_1'}$ follows from $\trg{\Sigma_1} = \trg{\Sigma_0}$, $\trg{\Sigma_1'} = \trg{\Sigma_0'}$, and $\trg{\Sigma_0} \approx \trg{\Sigma_0'}$.
	Therefore, $\trg{\nspecProject{\lambda}} <> \trg{\nspecProject{\lambda'}}
	 \vee (\trg{\Sigma_1} \approx \trg{\Sigma_1'} \wedge  \trg{\lambda} = \trg{\lambda'})$ holds for the base case.
	
	\item[Induction step:]
	For the induction step, we assume that the claim holds for all $\trg{n'}<\trg{n}$ and we show that it holds for $\trg{n}$ as well.
	We proceed by case distinction on the rule used to derive $\trg{(0,\Sigma_0) \Xtot{\lambda} (n,\Sigma_1)}$:
	\begin{description}
	\item[E-\TR-silent:]
	Then, $\trg{(0,\Sigma_0) \Xtot{\lambda} (n,\Sigma_1)}$ has been derived using the E-\TR-silent rule.
	Therefore, we have $\trg{(0,\Sigma_0) \Xtot{\trat{_1^{\taintt}}} (n-1,\Sigma_2)}$, $\trg{\Sigma_2 \xltot{ \epsilon } \Sigma_1}$, and $\trg{\lambda} = \trat{_1^{\taintt}} \cdot \trg{\epsilon}$.
	
	Since $\trg{n} > \trg{1}$, $\trg{(0,\Sigma_0') \Xtot{\lambda'} (n,\Sigma_1')}$ has been derived using the E-\TR-silent or E-\TR-single rules.
	Therefore, we have  $\trg{(0,\Sigma_0') \Xtot{\trat{_1'^{\taintt'}}} (n-1,\Sigma_2')}$ and $\trg{\Sigma_2' \xltot{ \lambda_1' } \Sigma_1'}$.
	From the induction hypothesis, there are two cases:
	\begin{description}
		\item [$\nspecProject{ \trat{_1^{\taintt}}} <> \nspecProject{\trat{_1'^{\taintt'}} }$:]
		Observe that $\trg{\nspecProject{\lambda}} = \nspecProject{ \trat{_1^{\taintt}}}$ and $ \trg{\nspecProject{\lambda'}}$ is either $ \trg{\nspecProject{ \trat{_1'^{\taintt'}} }}$ or $ \trg{\nspecProject{ \trat{_1'^{\taintt'}}  }} \cdot \trg{\lambda_1'}$.
		From this and $\nspecProject{ \trat{_1^{\taintt}}} <> \nspecProject{\trat{_1'^{\taintt'}} }$, we get $\trg{\nspecProject{\lambda}} <> \trg{\nspecProject{\lambda'}}$.
		
		\item[$\trg{\Sigma_2} \approx \trg{\Sigma_2'} \wedge \trg{\trat{_1^{\taintt}}} = \trg{\trat{_1'^{\taintt'}}}$:]
		From $\trg{\Sigma_2} \approx \trg{\Sigma_2'}$, $\trg{\Sigma_2 \xltot{ \epsilon } \Sigma_1}$, \cref{lemma:no-obs-preserve-low-equivalence-spec-semantics}, we get $\trg{\Sigma_2' \xltot{ \epsilon } \Sigma_1''}$ and $\trg{\Sigma_1} \approx \trg{\Sigma_1''}$.
		From this and the determinism of $\xltot{}$, we get that $\trg{\lambda_1'} = \trg{\epsilon}$ and $\trg{\Sigma_1'} = \trg{\Sigma_1''}$.
		Hence, $\trg{\lambda} = \trg{\lambda'}$ follows from $\trg{\lambda} = \trat{_1^{\taintt}}$, $\trg{\lambda'} = \trat{_1'^{\taintt'}} \cdot \trg{\lambda_1'}$, $ \trat{_1^{\taintt}} =  \trat{_1'^{\taintt'}}$, and $\trg{\lambda_1'} = \trg{\epsilon}$.
		Moreover, $\trg{\Sigma_1} \approx \trg{\Sigma_1'}$ follows from $\trg{\Sigma_1} \approx \trg{\Sigma_1''}$ and $\trg{\Sigma_1'} = \trg{\Sigma_1''}$.

	\end{description}
	Therefore, $\trg{\nspecProject{\lambda}} <> \trg{\nspecProject{\lambda'}}
	 \vee (\trg{\Sigma_1} \approx \trg{\Sigma_1'} \wedge  \trg{\lambda} = \trg{\lambda'})$ holds for this case.

	\item[E-\TR-single:]
	Then, $\trg{(0,\Sigma_0) \Xtot{\lambda} (n,\Sigma_1)}$ has been derived using the E-\TR-single rule.
	Therefore, we have $\trg{(0,\Sigma_0) \Xtot{\trat{_1^{\taintt_1}}} (n-1,\Sigma_2)}$, $\trg{\Sigma_2 \xltot{ \acat{^{\taintt}} } \Sigma_1}$, $\trg{\Sigma_2} = \trg{(w, \OB{(\Omega,m,\taintt)}\cdot (\OB{F},\OB{I},H_2,B_2\triangleright \proc{s}{\OB{f}\cdot f} , n, \taintt')) }$, $\trg{\Sigma_1} = \trg{(w, \OB{(\Omega,m,\taintt)}}$ $\trg{\cdot (\OB{F},\OB{I},H_1,B_1\triangleright \proc{s'}{\OB{f'}\cdot f'} , n', \taintt'')) }$, and $\trg{\lambda} = \trat{_1^{\taintt}} $ if $\trg{f} == \trg{f'}  \wedge \trg{f} \in \trg{I}$ and  $\trg{\lambda} = \trg{\trat{_1^{\taintt}}\cdot \trg{\acat{^{\taintt}}}} $ otherwise.
	
	Since $\trg{n} > \trg{1}$, $\trg{(0,\Sigma_0') \Xtot{\lambda'} (n,\Sigma_1')}$ has been derived using the E-\TR-silent or E-\TR-single rules.
	Therefore, we have  $\trg{(0,\Sigma_0') \Xtot{\trat{_1'^{\taintt'}}} (n-1,\Sigma_2')}$ and $\trg{\Sigma_2' \xltot{ \lambda_1' } \Sigma_1'}$.
	From the induction hypothesis, there are two cases:
	\begin{description}
		\item [$\trg{\nspecProject{ \trat{_1^{\taintt}}}} <> \trg{\nspecProject{\trat{_1'^{\taintt'}} }}$:]
		Observe that $\trg{\nspecProject{\lambda}}$ is either $ \nspecProject{ \trat{_1^{\taintt}}}$ or $\nspecProject{ \trat{_1^{\taintt}}} \cdot \trg{\acat{^{\taintt}}}$ and $ \trg{\nspecProject{\lambda'}}$ is either $ \trg{\nspecProject{ \trat{_1'^{\taintt'}} }}$ or $ \trg{\nspecProject{ \trat{_1'^{\taintt'}}  }} \cdot \trg{\lambda_1'}$.
		From this and $\nspecProject{ \trat{_1^{\taintt}}} <> \nspecProject{\trat{_1'^{\taintt'}} }$, we get $\trg{\nspecProject{\lambda}} <> \trg{\nspecProject{\lambda'}}$.
		
		\item[$\trg{\Sigma_2} \approx \trg{\Sigma_2'} \wedge \trg{\trat{_1^{\taintt}}} = \trg{\trat{_1'^{\taintt'}}}$:]
		From $\trg{\Sigma_2} \approx \trg{\Sigma_2'}$ and the determinism of  $\xltot{}$, then $\trg{(0,\Sigma_0') \Xtot{\lambda'} (n,\Sigma_1')}$ must have been derived using the E-\TR-single rule (since $\trg{\Sigma_2 \xltot{ \acat{^{\taintt}} } \Sigma_1}$ has produced an observation that is not $\trg{\epsilon}$).
		From this, $\trg{\Sigma_2} \approx \trg{\Sigma_2'}$, and the determinism of $\xltot{}$, we get that $\trg{\Sigma_2' \xltot{ 
				\acat{'{^{\taintt'}}} } \Sigma_1'}$, $\trg{\Sigma_2'} = \trg{(w, \OB{(\Omega',m,\taintt)}\cdot (\OB{F},\OB{I},H_2',B_2'\triangleright \proc{s}{\OB{f}\cdot f} , n, \taintt')) }$, $\trg{\Sigma_1'} = \trg{(w, \OB{(\Omega',m,\taintt)}}$ $\trg{\cdot (\OB{F},\OB{I},H_1',B_1'\triangleright \proc{s'}{\OB{f'}\cdot f'} , n', \taintt'')) }$, and  $\trg{\lambda'} = \trg{ \trat{_1'{^{\taintt'}}} } $ if $\trg{f} == \trg{f'}  \wedge \trg{f} \in \trg{I}$ and  $\trg{\lambda'} = \trg{\trat{_1'{^{\taintt'}}}\cdot \acat{'{^{\taintt'}}}} $ otherwise.
		There are two cases:
		\begin{description}
			\item[$\trg{f} == \trg{f'}  \wedge \trg{f} \in \trg{I}$:]
			Hence, $\trg{\lambda} = \trat{_1^{\taintt}}$ and $\trg{\lambda'} = \trat{_1'^{\taintt'}} $.
			Therefore, $\trg{\lambda} = \trg{\lambda'}$ follows from $\trg{\lambda} = \trat{_1^{\taintt}}$, $\trg{\lambda'} = \trat{_1'^{\taintt'}} $, and $\trg{\trat{_1^{\taintt}}} = \trg{\trat{_1'^{\taintt'}}}$.
			
			From $\trg{f} \in \trg{I}$, $\trg{(0,\Sigma_0) \Xtot{\trat{_1^{\taintt}}} (n-1,\Sigma_2)}$, $\trg{\Sigma_2} = \trg{(w, \OB{(\Omega,m,\taintt)}\cdot}$ $\trg{(\OB{F},\OB{I},H_2,B_2\triangleright \proc{s}{\OB{f}\cdot f} , n, \taintt')) }$, and \cref{lemma:trans-semantics-in-components-only-safe-values}, we get that for all $\trg{x}$, $\trg{v}$, $\trg{\sigma}$, if $\trg{B_2(x)} = \trg{v:\sigma}$, then $\trg{\sigma} = \trg{\safeta}$.
			From this, $\trg{\Sigma_2'} = \trg{(w, \OB{(\Omega',m,\taintt)}\cdot (\OB{F},\OB{I},H_2',B_2'\triangleright \proc{s}{\OB{f}\cdot f} , n, \taintt')) }$, and $\trg{\Sigma_2} \approx \trg{\Sigma_2'}$, we get that $\trg{B_2} = \trg{B_2'}$. 
			From this, $\trg{\Sigma_2 \xltot{ \acat{^{\taintt}} } \Sigma_1}$, and the determinism of $\xltot{}$, we get that $\trg{\Sigma_2' \xltot{ \acat{^{\taintt}} } \Sigma_1'}$.
			Then, $\trg{\Sigma_1} \approx \trg{\Sigma_1'}$ follows from 	$\trg{\Sigma_2 \xltot{ \acat{^{\taintt}} } \Sigma_1}$, $\trg{\Sigma_2' \xltot{ \acat{^{\taintt}} } \Sigma_1'}$, $\trg{\Sigma_2} \approx \trg{\Sigma_2'}$, and \cref{lemma:same-obs-preserve-low-equivalence-for-spec-semantics}.

			\item[$\neg(\trg{f} == \trg{f'}  \wedge \trg{f} \in \trg{I})$:]
			Hence, $\trg{\lambda} = \trg{\trat{_1^{\taintt}}\cdot \trg{\acat{^{\taintt}}}} $ and $\trg{\lambda'} = \trg{\trat{_1'{^{\taintt'}}}\cdot \acat{'{^{\taintt'}}}} $.
			There are two cases:
			\begin{description}
			\item[$\unsafe{\trg{\Sigma_2}}$:]
			Then, from $\safe{\trg{\lambda}}$, we have $\safe{\trg{\acat{^{\taintt}}}}$.
			From $\unsafe{\trg{\Sigma_2}}$, $\trg{\Sigma_2 \xltot{ \acat{^{\taintt}} } \Sigma_1}$, $\safe{\trg{\acat{^{\taintt}}}}$, $\trg{\Sigma_2} \approx \trg{\Sigma_2'}$, $\trg{\Sigma_2' \xltot{ \acat{'{^{\taintt}}} } \Sigma_1'}$, the determinism of $\xltot{}$, and \cref{lemma:safe-obs-preserve-low-equivalence-for-spec-semantics}, we get $\trg{\Sigma_1} \approx \trg{\Sigma_1'}$, $\trg{\acat{{^{\taintt}}} } = \trg{ \acat{'{^{\taintt'}}} }$, and $\trg{\lambda}= \trg{\lambda'}$ (from  $\trg{\acat{{^{\taintt}}} } = \trg{ \acat{'{^{\taintt'}}} }$ and $\trg{\trat{_1^{\taintt}}} = \trg{\trat{_1'^{\taintt'}}}$).
			
			\item[$\safe{\trg{\Sigma_2}}$:]
			From this and $\trg{\Sigma_2} \approx \trg{\Sigma_2'}$, we have $\safe{\trg{\Sigma_2'}}$.
			From this, $\trg{\Sigma_0} = \trg{\SInit{P}}$, $\trg{\Sigma_0'} = \trg{\SInit{P'}}$, and \cref{lemma:non-speculative-projection}, we get 
			 $\trg{\nspecProject{\lambda}} = \trg{ \nspecProject{ \trat{_1^{\taintt}} } \cdot \acat{^{\taintt} } } $ and $\trg{\nspecProject{\lambda'}} = \trg{ \nspecProject{ \trat{_1'{^{\taintt'}}} }\cdot \acat{'{^{\taintt'}}}} $.
			 There are two cases:
			 \begin{description}
			 \item[$\trg{ \acat{^{\taintt} }} = \trg{\acat{'{^{\taintt'}}}}$:]
			 Then, $\trg{\lambda} = \trg{\lambda'}$ follows from $\trg{ \acat{^{\taintt} }} = \trg{\acat{'{^{\taintt'}}}}$ and $\trg{\trat{_1^{\taintt}}} = \trg{\trat{_1'^{\taintt'}}}$.
			 Moreover, $\trg{\Sigma_1} \approx \trg{\Sigma_1'}$ follows from  $\trg{\Sigma_2} \approx \trg{\Sigma_2'}$,  	$\trg{\Sigma_2 \xltot{ \acat{^{\taintt}} } \Sigma_1}$, $\trg{\Sigma_2' \xltot{ \acat{'{^{\taintt'}}} } \Sigma_1'}$, $\trg{ \acat{^{\taintt} }} = \trg{\acat{'{^{\taintt'}}}}$, and \cref{lemma:same-obs-preserve-low-equivalence-for-spec-semantics}.

			 \item[$\trg{ \acat{^{\taintt} }} \neq \trg{\acat{'{^{\taintt'}}}}$:]
			 From this, $\trg{\nspecProject{\lambda}} = \trg{ \nspecProject{ \trat{_1^{\taintt}} } \cdot \acat{^{\taintt} } } $, $\trg{\nspecProject{\lambda'}} = \trg{ \nspecProject{ \trat{_1'{^{\taintt'}}} }\cdot \acat{'{^{\taintt'}}}} $, and $\trg{\nspecProject{ \trat{_1^{\taintt}}}} = \trg{\nspecProject{\trat{_1'^{\taintt'}} }}$, we get that $\trg{\nspecProject{\lambda}} <> \trg{\nspecProject{\lambda'}}$. 
			 \end{description}
			\end{description}
		\end{description}	
	Therefore, $\trg{\nspecProject{\lambda}} <> \trg{\nspecProject{\lambda'}}
	 \vee (\trg{\Sigma_1} \approx \trg{\Sigma_1'} \wedge  \trg{\lambda} = \trg{\lambda'})$ holds for this case.	
	\end{description}
	Therefore, $\trg{\nspecProject{\lambda}} <> \trg{\nspecProject{\lambda'}}
	 \vee (\trg{\Sigma_1} \approx \trg{\Sigma_1'} \wedge  \trg{\lambda} = \trg{\lambda'})$ holds for the induction step.
	\end{description}
\end{description}
This concludes the proof.
\end{proof}

\subsubsection{Speculative semantics $\xltot{}$}

\begin{lemma}[Steps of $\xltot{}$ with safe observations preserve low-equivalence for unsafe configurations]\label{lemma:safe-obs-preserve-low-equivalence-for-spec-semantics}
\begin{align*}
	\forall  \trg{\Sigma_0}, \trg{\Sigma_0'}, \trg{\Sigma_1}, \trg{\lambda}.
	& \text{ if } 
		\trg{\Sigma_0 \xltot{\lambda} \Sigma_1}, \trg{\Sigma_0} \approx \trg{\Sigma_0'}, \safe{\trg{\lambda}}, \unsafe{\trg{\Sigma_0}} \\
	& \text{ then } \exists \trg{\Sigma_1'}, \trg{\lambda'}.\ \trg{\Sigma_0' \xltot{\lambda'} \Sigma_1'},\ \trg{\Sigma_1} \approx \trg{\Sigma_1'},\ \trg{\lambda} = \trg{\lambda'}
\end{align*}
\end{lemma}

\begin{proof}
Let $\trg{\Sigma_0}$, $\trg{\Sigma_0'}$, $\trg{\Sigma_1}$, and $\trg{\lambda}$ be such that $\trg{\Sigma_0 \xltot{\lambda} \Sigma_1}$, $\trg{\Sigma_0} \approx \trg{\Sigma_0'}$, and $\safe{\trg{\lambda}}$.
We proceed by case distinction on the rule used to derive $\trg{\Sigma_0 \xltot{\lambda} \Sigma_1}$:
\begin{description}
\item[E-\TR-speculate-epsilon:]
Then, $\trg{\Sigma_0} \isdef \trg{w, \OB{(\Omega,\trgb{\omega},\taintt)}\cdot(\Omega_0,n+1,\taintt)}$, $\trg{\Omega_0 \xtot{\epsilon} \Omega_1}$, $\trg{\Omega_0}\equiv\trg{C, H, \OB{B} \triangleright s;s'}$, $\trg{s}\not\equiv\trg{s'';s'''} \text{ and } \trg{s}\not\equiv\trg{\lfence}$, and $\trg{\lambda} = \trg{\epsilon}$, and $\trg{\Sigma_1} \isdef \trg{w, \OB{(\Omega,\trgb{\omega},\taintt)}\cdot(\Omega_1,n,\taintt)}$.
From $\trg{\Sigma_0} \approx \trg{\Sigma_0'}$, we get that $\trg{\Sigma_0'} \isdef$ $\trg{w, \OB{(\Omega',\trgb{\omega},\taintt)}}$ $\trg{\cdot(\Omega_0',n+1,\taintt)}$ and $\trg{\Omega_0} \approx \trg{\Omega_0'}$.
From $\trg{\Omega_0 \xtot{\epsilon} \Omega_1}$, $\trg{\Omega_0} \approx \trg{\Omega_0'}$, $\safe{\trg{\epsilon}}$, and \cref{lemma:safe-obs-preserve-low-equivalence-for-non-spec-semantics}, we get that $\trg{\Omega_0' \xtot{\epsilon} \Omega_1'}$ and $\trg{\Omega_1} \approx \trg{\Omega_1'}$.
We can therefore apply the E-\TR-speculate-epsilon rule to derive $\trg{\Sigma_0'} \xltot{\trg{\lambda'}} \trg{\Sigma_1'}$ where $\trg{\lambda'} = \trg{\epsilon}$ and $\trg{\Sigma_1'} \isdef \trg{w, \OB{(\Omega,\trgb{\omega},\taintt)}\cdot(\Omega_1',n,\taintt)}$.
Hence, $\trg{\lambda} = \trg{\lambda'}$ follows from $\trg{\lambda'} = \trg{\epsilon}$ and $\trg{\lambda'} = \trg{\epsilon}$, whereas $\trg{\Sigma_1} \approx \trg{\Sigma_1'}$ follows from $\trg{\Sigma_0} \approx \trg{\Sigma_0'}$ and $\trg{\Omega_1} \approx \trg{\Omega_1'}$.

\item[E-\TR-speculate-lfence:]
The proof of this case is similar to that of E-\TR-speculate-epsilon.

\item[E-\TR-speculate-action:]
Then,  $\trg{\Sigma_0} \isdef \trg{w, \OB{(\Omega,\trgb{\omega},\taintt)}\cdot(\Omega_0,n+1,\taintt)  }$, $\trg{\Omega_0 \xtot{\trgb{\lambda_0}^{\taintt_0}} \Omega_1}$, $\trg{\Omega_0}\equiv\trg{C, H, \OB{B} \triangleright s;s'}$, $\trg{s}\not\equiv\trg{s'';s'''} \text{ and } \trg{s}\not\equiv\trg{\lfence}$, and $\trg{\lambda} = \trg{ \trgb{\lambda_0}^{{\taintt_0}\glb\taintt}}$, and $\trg{\Sigma_1} \isdef \trg{w, \OB{(\Omega,\trgb{\omega},\taintt)}\cdot(\Omega_1,n,\taintt)}$.
From $\trg{\Sigma_0} \approx \trg{\Sigma_0'}$, we get that $\trg{\Sigma_0'} \isdef \trg{w, \OB{(\Omega',\trgb{\omega},\taintt)}\cdot(\Omega_0',n+1,\taintt)}$ and $\trg{\Omega_0} \approx \trg{\Omega_0'}$.
From $\unsafe{\trg{\Sigma_0}}$ and $\safe{\trg{\lambda}}$, we get that $\safe{\trgb{\lambda}^{\taintt_0}}$.
From this, $\trg{\Omega_0} \approx \trg{\Omega_0'}$, and \cref{lemma:safe-obs-preserve-low-equivalence-for-non-spec-semantics}, we get that $\trg{\Omega_0' \xtot{\trgb{\lambda}^{\taintt_0}} \Omega_1'}$ and $\trg{\Omega_1} \approx \trg{\Omega_1'}$.
We can therefore apply the E-\TR-speculate-action rule to derive $\trg{\Sigma_0'} \xltot{\trg{\lambda'}} \trg{\Sigma_1'}$ where $\trg{\lambda'} = \trg{\trgb{\lambda_0}^{{\taintt_0}\glb\taintt}}$ and $\trg{\Sigma_1'} \isdef \trg{w, \OB{(\Omega,\trgb{\omega},\taintt)}\cdot(\Omega_1',n,\taintt)}$.
Hence, $\trg{\lambda} = \trg{\lambda'}$ follows from $\trg{\lambda'} = \trg{\trgb{\lambda_0}^{{\taintt_0}\glb\taintt}}$ and $\trg{\lambda'} = \trg{\trgb{\lambda_0}^{{\taintt_0}\glb\taintt}}$, whereas $\trg{\Sigma_1} \approx \trg{\Sigma_1'}$ follows from $\trg{\Sigma_0} \approx \trg{\Sigma_0'}$ and $\trg{\Omega_1} \approx \trg{\Omega_1'}$.

\item[E-\TR-speculate-if:]
Then,  $\trg{\Sigma_0} \isdef \trg{w, \OB{(\Omega,\trgb{\omega},\taintt)}\cdot(\Omega_0,n+1,\taintt) }$,
$\trg{\Omega_0} \isdef \trg{C, H, \OB{B} \cdot B \triangleright  }$ $\trg{\proc{\ifte{e}{s''}{s'''};s'}{\OB{f}\cdot f}}$,
$\trg{C}\equiv\trg{\OB{F};\OB{I}}$, $\trg{f}\notin\trg{\OB{I}}$, 
$\trg{\Omega_0 \xtot{\acat{^{\taintt_0}}} \Omega_1}$,
$\trg{\lambda} = \trg{\acat{^{\taintt_0\glb\taintt}}}$, and 
$\trg{\Sigma_1} \isdef \trg{w, \OB{(\Omega,\trgb{\omega},\taintt)}\cdot(\Omega_1,n,\taintt) \cdot (\Omega_2, \fun{min}{\trg{w},\trg{n}}, \unta )}$ where $\trg{\Omega_2} = \trg{C, H,}$ $\trg{\OB{B} \cdot B \triangleright s''';s'}$ if $\trg{B \triangleright e\bigred 0:{\taintt_0}}$ and   $\trg{\Omega_2} = \trg{C, H, \OB{B} \cdot B \triangleright s'';s'}$ if $\trg{B \triangleright e\bigred n:}$ $\trg{{\taintt_0} \wedge n > 0}$.
From $\trg{\Sigma_0} \approx \trg{\Sigma_0'}$, we get that $\trg{\Sigma_0'} \isdef  \trg{w, \OB{(\Omega',\trgb{\omega},\taintt)}\cdot(\Omega_0',n+1,\taintt) }$ where $\trg{\OB{(\Omega',\trgb{\omega},\taintt)}} \approx \trg{\OB{(\Omega,\trgb{\omega},\taintt)}}$  and $\trg{\Omega_0} \approx \trg{\Omega_0'}$.
Thus, $\trg{\Omega_0'} \isdef \trg{C, H', \OB{B'} \triangleright  }$ $\trg{\proc{\ifte{e}{s''}{s'''};s'}{\OB{f}\cdot f}}$ where $\trg{H} \approx \trg{H'}$ and $\trg{\OB{B}} \approx \trg{\OB{B'}}$.

From $\unsafe{\trg{\Sigma_0} }$ and $\safe{\trg{\lambda}}$, we get $\safe{\trg{\acat{^{\taintt_0}}}}$ and $\trg{\taintt_0} = \trg{\safeta}$.
From this, $\trg{\Sigma_0} \approx \trg{\Sigma_0'}$, and  \cref{lemma:safe-obs-preserve-low-equivalence-for-non-spec-semantics}, we get that $\trg{\Omega_0' \xtot{\acat{^{\taintt_0}}} \Omega_1'}$ and $\trg{\Omega_1} \approx \trg{\Omega_1'}$.

From $\trg{\Omega_0} \approx \trg{\Omega_0'}$ and  $\trg{\Omega_0} \isdef \trg{C, H, \OB{B} \cdot B \triangleright  }$ $\trg{\proc{\ifte{e}{s''}{s'''};s'}{\OB{f}\cdot f}}$, we have $\trg{\Omega_0} \isdef \trg{C, H', \OB{B'} \cdot B' \triangleright  }$ $\trg{\proc{\ifte{e}{s''}{s'''};s'}{\OB{f}\cdot f}}$ where $\trg{H} \approx \trg{H'}$, $\trg{B} \approx \trg{B'}$, and $\trg{\OB{B}} \approx \trg{\OB{B'}}$. 
From $\trg{B} \approx \trg{B'}$, $\trg{B \triangleright e\bigred v:{\taintt_0}}$, $\trg{\taintt_0} = \trg{\safeta}$, and \cref{lemma:low-equiv-bindings-low-equiv-results}, we have that  $\trg{B' \triangleright e\bigred v:{\taintt_0}}$.

We can therefore apply the E-\TR-speculate-if rule to derive $\trg{\Sigma_0'} \xltot{\trg{\lambda'}} \trg{\Sigma_1'}$ where $\trg{\lambda'} = \trg{ \acat{^{\taintt_0\glb\taintt'}}}$ and $\trg{\Sigma_1'} \isdef \trg{w, \OB{(\Omega',\trgb{\omega},\taintt)}\cdot(\Omega_1',n,\taintt')\cdot(\Omega_2',j,\unta)}$ where $\trg{\Omega_2'} =   \trg{C, H',}$ $\trg{\OB{B'} \cdot B' \triangleright s''';s'}$ if $\trg{B' \triangleright e\bigred 0:{\taintt_0}}$ and   $\trg{\Omega_2'} = \trg{C, H', \OB{B'} \cdot B' \triangleright s'';s'}$ if $\trg{B' \triangleright e\bigred n:{\taintt_0} \wedge n > 0}$.
Hence, $\trg{\lambda} = \trg{\lambda'}$ follows from $\trg{\lambda'} = \trg{\trg{\acat{^{\taintt_0\glb\taintt'}}}}$ and $\trg{\lambda'} = \trg{\trg{\acat{^{\taintt_0\glb\taintt'}}}}$, whereas $\trg{\Sigma_1} \approx \trg{\Sigma_1'}$ follows from $\trg{\Sigma_0} \approx \trg{\Sigma_0'}$, $\trg{\Omega_1} \approx \trg{\Omega_1'}$, and $\trg{\Omega_2} \approx \trg{\Omega_2'}$ (which follows from $\trg{B \triangleright e\bigred v:{\taintt_0}}$ and $\trg{B' \triangleright e\bigred v:{\taintt_0}}$).

\item[E-\TR-speculate-rollback:]
Then, $\trg{\Sigma_0} \isdef \trg{w, \OB{(\Omega,\trgb{\omega},\taintt)}\cdot(\Omega,0,\taintt) }$, $\trg{\lambda} = \trg{\rollbl}$, and $\trg{\Sigma_1} \isdef \trg{ w, \OB{(\Omega,\trgb{\omega},\taintt)}}$.
From $\trg{\Sigma_0} \approx \trg{\Sigma_0'}$, we get that $\trg{\Sigma_0'} \isdef \trg{w, \OB{(\Omega',\trgb{\omega},\taintt)}\cdot(\Omega_0',}$ $\trg{0,\taintt) }$ where $\trg{\OB{(\Omega',\trgb{\omega},\taintt)}} \approx \trg{\OB{(\Omega,\trgb{\omega},\taintt)}}$ and $\trg{\Omega_0} \approx \trg{\Omega_0'}$.
Hence, we can apply the E-\TR-speculate-rollback rule to derive $\trg{\Sigma_0'} \xltot{\trg{\lambda'}} \trg{\Sigma_1'}$ where $\trg{\lambda'} = \trg{\rollbl}$ and $\trg{\Sigma_1'} \isdef \trg{w, \OB{(\Omega',\trgb{\omega},\taintt)}}$.
Hence, $\trg{\lambda} = \trg{\lambda'}$ follows from $\trg{\lambda'} = \trg{\rollbl}$ and $\trg{\lambda'} = \trg{\rollbl}$, whereas $\trg{\Sigma_1} \approx \trg{\Sigma_1'}$ follows from $\trg{\Sigma_0} \approx \trg{\Sigma_0'}$.

\item[E-\TR-speculate-rollback-stuck:]
Then, $\trg{\Sigma_0} \isdef \trg{w, \OB{(\Omega,\trgb{\omega},\taintt)}\cdot(\Omega_0,W,\taintt)}$, $\vdash \trg{\Sigma_0} : \bot$,  $\trg{\lambda} = \trg{\rollbl}$, and $\trg{\Sigma_1} = \trg{(w, \OB{(\Omega,\trgb{\omega},\taintt)})}$.
From $\trg{\Sigma_0} \approx \trg{\Sigma_0'}$, we get that $\trg{\Sigma_0'} \isdef \trg{w, \OB{(\Omega',\trgb{\omega},\taintt)}\cdot(\Omega_0',0,\taintt) }$ where $\trg{\OB{(\Omega',\trgb{\omega},\taintt)}} \approx \trg{\OB{(\Omega,\trgb{\omega},\taintt)}}$ and $\trg{\Omega_0} \approx \trg{\Omega_0'}$.
From $\trg{\Sigma_0} \approx \trg{\Sigma_0'}$ and $\vdash \trg{\Sigma_0} : \bot$, we get $\vdash \trg{\Sigma_0'} : \bot$.
Therefore, we can apply the E-\TR-speculate-rollback-stuck to derive  $\trg{\Sigma_0'} \xltot{\trg{\lambda'}} \trg{\Sigma_1'}$ where $\trg{\lambda'} = \trg{\rollbl}$ and $\trg{\Sigma_1'} \isdef \trg{w, \OB{(\Omega',\trgb{\omega},\taintt)}}$.
Hence, $\trg{\lambda} = \trg{\lambda'}$ follows from $\trg{\lambda'} = \trg{\rollbl}$ and $\trg{\lambda'} = \trg{\rollbl}$, whereas $\trg{\Sigma_1} \approx \trg{\Sigma_1'}$ follows from $\trg{\Sigma_0} \approx \trg{\Sigma_0'}$.

\end{description} 
This concludes our proof.
\end{proof}

\begin{lemma}[Steps of $\xltot{}$ without observations preserve low-equivalence]\label{lemma:no-obs-preserve-low-equivalence-spec-semantics}
\begin{align*}
	\forall  \trg{\Sigma_0}, \trg{\Sigma_0'}, \trg{\Sigma_1}, \trg{\lambda}.
	& \text{ if } 
		\trg{\Sigma_0 \xltot{\epsilon} \Sigma_1}, \trg{\Sigma_0} \approx \trg{\Sigma_0'} \\
	& \text{ then } \exists \trg{\Sigma_1'}.\ \trg{\Sigma_0' \xltot{\epsilon'} \Sigma_1'},\ \trg{\Sigma_1} \approx \trg{\Sigma_1'}
\end{align*}
\end{lemma}

\begin{proof}
Proof is similart to  \cref{lemma:safe-obs-preserve-low-equivalence-for-spec-semantics} (only rules E-\TR-speculate-epsilon and E-\TR-speculate-lfence) but using \cref{lemma:no-obs-preserve-low-equivalence-for-non-spec-semantics} instead of \cref{lemma:safe-obs-preserve-low-equivalence-for-non-spec-semantics}.	
\end{proof}

\begin{lemma}[Steps of $\xltot{}$ with same observations preserve low-equivalence]\label{lemma:same-obs-preserve-low-equivalence-for-spec-semantics}
\begin{align*}
	\forall  \trg{\Sigma_0}, \trg{\Sigma_0'}, \trg{\Sigma_1}, \trg{\Sigma_1'},  \trg{\lambda}, \trg{\lambda'}.
	& \text{ if } 
		\trg{\Sigma_0 \xltot{\lambda} \Sigma_1}, \trg{\Sigma_0' \xltot{\lambda'} \Sigma_1'}, \trg{\Sigma_0} \approx \trg{\Sigma_0'}, \trg{\lambda} = \trg{\lambda'} \\%
	& \text{ then } \trg{\Sigma_1} \approx \trg{\Sigma_1'}
\end{align*}
\end{lemma}

\begin{proof}
Let $\trg{\Sigma_0}$, $\trg{\Sigma_0'}$, $\trg{\Sigma_1}$, $\trg{\Sigma_1'}$,  $\trg{\lambda}$, $\trg{\lambda'}$ be such that $\trg{\Sigma_0 \xltot{\lambda} \Sigma_1}$, $\trg{\Sigma_0' \xltot{\lambda'} \Sigma_1'}$, $\trg{\Sigma_0} \approx \trg{\Sigma_0'}$, and $\trg{\lambda} = \trg{\lambda'}$. %
We proceed by case distinction on the rule used to derive $\trg{\Sigma_0 \xltot{\lambda} \Sigma_1}$:
\begin{description}
\item[E-\TR-speculate-epsilon:]
Then, $\trg{\Sigma_0} \isdef \trg{w, \OB{(\Omega,\trgb{\omega},\taintt)}\cdot(\Omega_0,n+1,\taintt)}$, $\trg{\Omega_0 \xtot{\epsilon} \Omega_1}$, $\trg{\Omega_0}\equiv\trg{C, H, \OB{B} \triangleright s;s'}$, $\trg{s}\not\equiv\trg{s'';s'''} \text{ and } \trg{s}\not\equiv\trg{\lfence}$, and $\trg{\lambda} = \trg{\epsilon}$, and $\trg{\Sigma_1} \isdef \trg{w, \OB{(\Omega,\trgb{\omega},\taintt)}\cdot(\Omega_1,n,\taintt)}$.
From $\trg{\Sigma_0} \approx \trg{\Sigma_0'}$, we get that $\trg{\Sigma_0'} \isdef$ $\trg{w, \OB{(\Omega',\trgb{\omega},\taintt)}}$ $\trg{\cdot(\Omega_0',n+1,\taintt)}$ and $\trg{\Omega_0} \approx \trg{\Omega_0'}$.
From $\trg{\Omega_0} \approx \trg{\Omega_0'}$ and $\trg{\Omega_0}\equiv\trg{C, H, \OB{B} \triangleright s;s'}$, we get that $\trg{\Omega_0'}\equiv\trg{C, H, \OB{B} \triangleright s;s'}$.
Since $\trg{\lambda} = \trg{\lambda'}$ and $\trg{\lambda} = \trg{\epsilon}$, $\trg{\Sigma_0' \xltot{\lambda'} \Sigma_1'}$ has also been derived using the E-\TR-speculate-epsilon rule.
Thus, $\trg{\Sigma_1'} \isdef \trg{w, \OB{(\Omega',\trgb{\omega},\taintt)}\cdot(\Omega_1',n,\taintt)}$ and  $\trg{\Omega_0' \xtot{\epsilon} \Omega_1'}$.
From $\trg{\Omega_0} \approx \trg{\Omega_0'}$, $\trg{\Omega_0 \xtot{\epsilon} \Omega_1}$, $\trg{\Omega_0' \xtot{\epsilon} \Omega_1'}$, and \cref{lemma:same-obs-preserve-low-equivalence-for-non-spec-semantics}, we get $\trg{\Omega_1} \approx \trg{\Omega_1'}$.
Hence, $\trg{\Sigma_1} \approx \trg{\Sigma_1'}$ follows from $\trg{\Sigma_0} \approx \trg{\Sigma_0'}$ and $\trg{\Omega_1} \approx \trg{\Omega_1'}$.

\item[E-\TR-speculate-lfence:]
The proof of this case is similar to that of E-\TR-speculate-epsilon.

\item[E-\TR-speculate-action:]
Then,  $\trg{\Sigma_0} \isdef \trg{w, \OB{(\Omega,\trgb{\omega},\taintt)}\cdot(\Omega_0,n+1,\taintt)  }$, $\trg{\Omega_0 \xtot{\trgb{\lambda_0}^{\taintt_0}} \Omega_1}$, $\trg{\Omega}\equiv\trg{C, H, \OB{B} \triangleright s;s'}$, $\trg{s}\not\equiv\trg{s'';s'''} \text{ and } \trg{s}\not\equiv\trg{\lfence}$, and $\trg{\lambda} = \trg{ \trgb{\lambda_0}^{{\taintt_0}\glb\taintt}}$, and $\trg{\Omega_1} \isdef \trg{w, \OB{(\Omega,\trgb{\omega},\taintt)}\cdot(\Omega_1,n,\taintt)}$.
From $\trg{\Sigma_0} \approx \trg{\Sigma_0'}$, we get that $\trg{\Sigma_0'} \isdef \trg{w, \OB{(\Omega',\trgb{\omega},\taintt)}\cdot(\Omega_0',n+1,\taintt)}$ and $\trg{\Omega_0} \approx \trg{\Omega_0'}$.
From $\trg{\Omega_0} \approx \trg{\Omega_0'}$ and $\trg{\Omega_0}\equiv\trg{C, H, \OB{B} \triangleright s;s'}$, we get that $\trg{\Omega_0'}\equiv\trg{C, H, \OB{B} \triangleright s;s'}$.
Since $\trg{\lambda} = \trg{\lambda'}$ and $\trg{\lambda} = \trg{ \trgb{\lambda_0}^{{\taintt_0}\glb\taintt}}$, $\trg{\Sigma_0' \xltot{\lambda'} \Sigma_1'}$ has also been derived using the E-\TR-speculate-action rule.
Thus, $\trg{\Sigma_1'} \isdef \trg{w, \OB{(\Omega',\trgb{\omega},\taintt)}\cdot(\Omega_1',n,\taintt)}$ and  $\trg{\Omega_0' \xtot{\trgb{\lambda_0'}^{\taintt_0'}} \Omega_1'}$.
From  $\trg{\lambda} = \trg{ \trgb{\lambda_0}^{{\taintt_0}\glb\taintt}}$,  $\trg{\lambda'} = \trg{ \trgb{\lambda_0'}^{{\taintt_0'}\glb\taintt}}$, and $\trg{\lambda} = \trg{\lambda'}$, we get $\trg{\lambda_0} = \trg{\lambda_0'}$.
From $\trg{\Omega_0} \approx \trg{\Omega_0'}$, $\trg{\Omega_0 \xtot{\lambda_0^{\taintt_0}} \Omega_1}$, $\trg{\Omega_0' \xtot{\lambda_0^{\taintt_0'}} \Omega_1'}$, and \cref{lemma:same-obs-cannot-disagree-on-labels-for-low-equivalence-confs-non-spec-semantics}, we get $\trg{\taintt_0} = \trg{\taintt_0'}$.
From $\trg{\Omega_0} \approx \trg{\Omega_0'}$, $\trg{\Omega_0 \xtot{\lambda_0^{\taintt_0}} \Omega_1}$, $\trg{\Omega_0' \xtot{\lambda_0^{\taintt_0}} \Omega_1'}$, and \cref{lemma:same-obs-preserve-low-equivalence-for-non-spec-semantics}, we get $\trg{\Omega_1} \approx \trg{\Omega_1'}$.
Hence, $\trg{\Sigma_1} \approx \trg{\Sigma_1'}$ follows from $\trg{\Sigma_0} \approx \trg{\Sigma_0'}$ and $\trg{\Omega_1} \approx \trg{\Omega_1'}$.

\item[E-\TR-speculate-if:]
Then,  $\trg{\Sigma_0} \isdef \trg{w, \OB{(\Omega,\trgb{\omega},\taintt)}\cdot(\Omega_0,n+1,\taintt) }$,
$\trg{\Omega_0} \isdef \trg{C, H, \OB{B} \cdot B \triangleright  }$ $\trg{\proc{\ifte{e}{s''}{s'''};s'}{\OB{f}\cdot f}}$,
$\trg{C}\equiv\trg{\OB{F};\OB{I}}$, $\trg{f}\notin\trg{\OB{I}}$, 
$\trg{\Omega_0 \xtot{\acat{^{\taintt_0}}} \Omega_1}$,
$\trg{\lambda} = \trg{\acat{^{\taintt_0\glb\taintt}}}$, and 
$\trg{\Sigma_1} \isdef \trg{w, \OB{(\Omega,\trgb{\omega},\taintt)}\cdot(\Omega_1,n,\taintt) \cdot (\Omega_2, \fun{min}{\trg{w},\trg{n}}, \unta )}$ where $\trg{\Omega_2} = \trg{C, H,}$ $\trg{\OB{B} \cdot B \triangleright s''';s'}$ if $\trg{B \triangleright e\bigred 0:{\taintt_0}}$ and   $\trg{\Omega_2} = \trg{C, H, \OB{B} \cdot B \triangleright s'';s'}$ if $\trg{B \triangleright e\bigred n:}$ $\trg{{\taintt_0} \wedge n > 0}$.
From $\trg{\Sigma_0} \approx \trg{\Sigma_0'}$, we get that $\trg{\Sigma_0'} \isdef  \trg{w, \OB{(\Omega',\trgb{\omega},\taintt)}\cdot(\Omega_0',n+1,\taintt) }$ where $\trg{\OB{(\Omega',\trgb{\omega},\taintt)}} \approx \trg{\OB{(\Omega,\trgb{\omega},\taintt)}}$  and $\trg{\Omega_0} \approx \trg{\Omega_0'}$.
Thus, $\trg{\Omega_0'} \isdef \trg{C, H', \OB{B'} \triangleright  }$ $\trg{\proc{\ifte{e}{s''}{s'''};s'}{\OB{f}\cdot f}}$ where $\trg{H} \approx \trg{H'}$ and $\trg{\OB{B}} \approx \trg{\OB{B'}}$.
Hence, $\trg{\Sigma_0' \xltot{\lambda'} \Sigma_1'}$ has also been derived using the E-\TR-speculate-action rule.
Thus, $\trg{\Sigma_1'} \isdef \trg{w, \OB{(\Omega',\trgb{\omega},\taintt)}\cdot(\Omega_1',n,\taintt) \cdot (\Omega_2', \fun{min}{\trg{w},\trg{n}}, \unta )}$ and $\trg{\Omega_0' \xtot{{\acat'}{^{\taintt_0'}}} \Omega_1'}$.
From  $\trg{\lambda} = \trg{ \trgb{\alpha}^{{\taintt_0}\glb\taintt}}$,  $\trg{\lambda'} = \trg{ \trgb{\alpha'}^{{\taintt_0'}\glb\taintt}}$, and $\trg{\lambda} = \trg{\lambda'}$, we get $\trg{\alpha} = \trg{\alpha'}$.
From $\trg{\Omega_0} \approx \trg{\Omega_0'}$, $\trg{\Omega_0 \xtot{\alpha^{\taintt_0}} \Omega_1}$, $\trg{\Omega_0' \xtot{{\alpha'}^{\taintt_0'}} \Omega_1'}$, and \cref{lemma:same-obs-cannot-disagree-on-labels-for-low-equivalence-confs-non-spec-semantics}, we get $\trg{\taintt_0} = \trg{\taintt_0'}$.
From $\trg{\acat{^{\taintt_0'}}} = \trg{\acat'^{\taintt_0'}}$, $\trg{\Sigma_0} \approx \trg{\Sigma_0'}$, $\trg{\Omega_0 \xtot{{\acat{}}{^{\taintt_0}}} \Omega_1'}$, $\trg{\Omega_0' \xtot{{\acat'}{^{\taintt_0'}}} \Omega_1'}$, and \cref{lemma:same-obs-preserve-low-equivalence-for-non-spec-semantics}, we get $\trg{\Omega_1} \approx \trg{\Omega_1'}$.
Observe that from $\trg{\Omega_0} \approx \trg{\Omega_0'}$ and  $\trg{\Omega_0'} \isdef \trg{C, H', \OB{B'} \triangleright  }$ $\trg{\proc{\ifte{e}{s''}{s'''};s'}{\OB{f}\cdot f}}$, we get that 
$\trg{{\acat{}}{^{\taintt_0}}} = $ $\trg{(\ifl{n})^{\sigma_0} }$ 
$\leftrightarrow \trg{B \triangleright e\bigred n:}$ $\trg{\taintt_0}$ and 
$\trg{{\acat'}^{\taintt_0'}}= $ $\trg{(\ifl{n'})^{\sigma_0'} }$ 
$\leftrightarrow \trg{B' \triangleright e\bigred n':}$ $\trg{\taintt_0'} $.
From this and $\trg{\acat{^{\taintt_0'}}} = \trg{\acat'^{\taintt_0'}}$, we get $\trg{n:\taintt_0} = \trg{n':\taintt_0'}$.
Therefore, $\trg{\Omega_2} \approx \trg{\Omega_2'}$.
Hence, $\trg{\Sigma_1} \approx \trg{\Sigma_1'}$ follows from $\trg{\Sigma_0} \approx \trg{\Sigma_0'}$, $\trg{\Omega_1} \approx \trg{\Omega_1'}$, and $\trg{\Omega_2} \approx \trg{\Omega_2'}$.

\item[E-\TR-speculate-rollback:]
Then, $\trg{\Sigma_0} \isdef \trg{w, \OB{(\Omega,\trgb{\omega},\taintt)}\cdot(\Omega,0,\taintt) }$, $\trg{\lambda} = \trg{\rollbl}$, and $\trg{\Sigma_1} \isdef \trg{ w, \OB{(\Omega,\trgb{\omega},\taintt)}}$.
From $\trg{\Sigma_0} \approx \trg{\Sigma_0'}$, we get that $\trg{\Sigma_0'} \isdef \trg{w, \OB{(\Omega',\trgb{\omega},\taintt)}\cdot(\Omega_0',}$ $\trg{0,\taintt) }$ where $\trg{\OB{(\Omega',\trgb{\omega},\taintt)}} \approx \trg{\OB{(\Omega,\trgb{\omega},\taintt)}}$ and $\trg{\Omega_0} \approx \trg{\Omega_0'}$.
Hence, $\trg{\Sigma_0' \xltot{\lambda'} \Sigma_1'}$ has also been derived using the E-\TR-speculate-rollback rule.
Thus, $\trg{\Sigma_1'} \isdef \trg{ w, \OB{(\Omega',\trgb{\omega},\taintt)}}$.
Hence, $\trg{\Sigma_1} \approx \trg{\Sigma_1'}$ follows from $\trg{\OB{(\Omega',\trgb{\omega},\taintt)}} \approx \trg{\OB{(\Omega,\trgb{\omega},\taintt)}}$.

\item[E-\TR-speculate-rollback-stuck:]

Then, $\trg{\Sigma_0} \isdef \trg{w, \OB{(\Omega,\trgb{\omega},\taintt)}\cdot(\Omega_0,W,\taintt)}$, $\vdash \trg{\Sigma_0} : \bot$,  $\trg{\lambda} = \trg{\rollbl}$, and $\trg{\Sigma_1} = \trg{(w, \OB{(\Omega,\trgb{\omega},\taintt)})}$.
From $\trg{\Sigma_0} \approx \trg{\Sigma_0'}$, we get that $\trg{\Sigma_0'} \isdef \trg{w, \OB{(\Omega',\trgb{\omega},\taintt)}\cdot(\Omega_0',0,\taintt) }$ where $\trg{\OB{(\Omega',\trgb{\omega},\taintt)}} \approx \trg{\OB{(\Omega,\trgb{\omega},\taintt)}}$ and $\trg{\Omega_0} \approx \trg{\Omega_0'}$.
From $\trg{\Sigma_0} \approx \trg{\Sigma_0'}$ and $\vdash \trg{\Sigma_0} : \bot$, we get $\vdash \trg{\Sigma_0'} : \bot$.
Hence, $\trg{\Sigma_0' \xltot{\lambda'} \Sigma_1'}$ has also been derived using the E-\TR-speculate-rollback-stuck rule.
Thus, $\trg{\Sigma_1'} \isdef \trg{ w, \OB{(\Omega',\trgb{\omega},\taintt)}}$.
Hence, $\trg{\Sigma_1} \approx \trg{\Sigma_1'}$ follows from $\trg{\OB{(\Omega',\trgb{\omega},\taintt)}} \approx \trg{\OB{(\Omega,\trgb{\omega},\taintt)}}$.
\end{description} 
This concludes our proof.
\end{proof}

\subsubsection{Non-speculative semantics $\xtot{}$}

\begin{lemma}[Steps of $\xtot{}$ with same observations preserve low-equivalence]\label{lemma:same-obs-preserve-low-equivalence-for-non-spec-semantics}
\begin{align*}
	\forall \trg{\Omega_0}, \trg{\Omega_0'}, \trg{\lambda}, \trg{\lambda'}, \trg{\Omega_1}, \trg{\Omega_1'}.
	& \text{ if } 
		\trg{\Omega_0 \xtot{\lambda} \Omega_1},  \trg{\Omega_0' \xtot{\lambda'} \Omega_1'}, \trg{\Omega_0} \approx \trg{\Omega_0'}, \trg{\lambda} = \trg{\lambda'} \\
	& \text{ then } \trg{\Omega_1} \approx \trg{\Omega_1'}
\end{align*}
\end{lemma}

\begin{proof}
Let $\trg{\Omega_0}$, $\trg{\Omega_0'}$, $\trg{\lambda}$, $\trg{\lambda'}$, $\trg{\Omega_1}$,  $\trg{\Omega_1'}$ be such that $\trg{\Omega_0 \xtot{\lambda} \Omega_1}$,  $\trg{\Omega_0' \xtot{\lambda'} \Omega_1'}$, $\trg{\Omega_0} \approx \trg{\Omega_0'}$, and $\trg{\lambda} = \trg{\lambda'}$. 
We proceed by structural induction on the rule used to derive $\trg{\Omega_0 \xtot{\lambda} \Omega_1}$.

\begin{description}
\item[Base case:]
There are several cases depending on the rule used to derive $\trg{\Omega_0 \xtot{\lambda} \Omega_1}$:
\begin{description}
\item[E-\TR-sequence:]
Then, $\trg{\Omega_0} \isdef \trg{C, H, \OB{B} \triangleright \skipt;s}$, $\trg{\lambda} = \trg{\epsilon}$, and $\trg{\Omega_1} \isdef  \trg{C, H, \OB{B} \triangleright s}$.
From $\trg{\Omega_0} \approx \trg{\Omega_0'}$, we have that $\trg{\Omega_0'} \isdef \trg{C, H', \OB{B'} \triangleright \skipt;s}$ where $\trg{H} \approx \trg{H'}$ and $\trg{\OB{B}} \approx \trg{\OB{B'}}$.
Therefore, $\trg{\Omega_0' \xtot{\lambda'} \Omega_1'}$ has also been derived using the  E-\TR-sequence  where $\trg{\Omega_1'} \isdef \trg{C, H', \OB{B'} \triangleright s}$ and $\trg{\lambda'}= \trg{\epsilon}$.
Hence, $\trg{\Omega_1} \approx \trg{\Omega_1'}$ follows from  $\trg{H} \approx \trg{H'}$, $\trg{\OB{B}} \approx \trg{\OB{B'}}$, and $\trg{B} \approx \trg{B'}$.

\item[E-\TR-if-true:]
Then, $\trg{\Omega_0} \isdef \trg{C, H, \OB{B}\cdot B \triangleright \ifzte{e}{s}{s'}}$, $\trg{B \triangleright e\bigredt 0 : \sigma}$
$\trg{\lambda} = \trg{(\ifl{0})^\sigma}$, and $\trg{\Omega_1} \isdef \trg{C, H, \OB{B}\cdot B \triangleright s}$.
From $\trg{\Omega_0} \approx \trg{\Omega_0'}$, we have that $\trg{\Omega_0'} \isdef \trg{C, H', \OB{B'}\cdot B' \triangleright \ifzte{e}{s}{s'}}$ where $\trg{H} \approx \trg{H'}$, $\trg{\OB{B}} \approx \trg{\OB{B'}}$, and $\trg{B} \approx \trg{B'}$.
From this, $\trg{\lambda} = \trg{(\ifl{0})^\sigma}$, and $\trg{\lambda} = \trg{\lambda'}$, $\trg{\Omega_0' \xtot{\lambda'} \Omega_1'}$ has also been derived using the E-\TR-if-true rule.
Therefore,  $\trg{\Omega_1} \isdef \trg{C, H', \OB{B'}\cdot B' \triangleright s}$.
Hence, $\trg{\Omega_1} \approx \trg{\Omega_1'}$ follows from  $\trg{H} \approx \trg{H'}$, $\trg{\OB{B}} \approx \trg{\OB{B'}}$, and $\trg{B} \approx \trg{B'}$.
	
\item[E-\TR-if-false:]
The proof of this case is similar to that of the E-\TR-if-true case.

\item[E-\TR-letin:]
Then, $\trg{\Omega_0} \isdef \trg{C, H, \OB{B}\cdot B \triangleright \letin{x}{e}{s}}$, $\trg{\lambda} = \trg{\epsilon}$, $\trg{B \triangleright e\bigredt v : \sigma}$, and $\trg{\Omega_1} \isdef \trg{C, H, \OB{B}\cdot B\cup x\mapsto v : \sigma \triangleright s}$.
From $\trg{\Omega_0} \approx \trg{\Omega_0'}$, we have that  $\trg{\Omega_0'} \isdef \trg{C, H', \OB{B'}\cdot B' \triangleright \letin{x}{e}{s}}$ where $\trg{H} \approx \trg{H'}$, $\trg{\OB{B}} \approx \trg{\OB{B'}}$, and $\trg{B} \approx \trg{B'}$.
From this, $\trg{\Omega_0' \xtot{\lambda'} \Omega_1'}$ has also been derived using the E-\TR-letin rule.
Hence,  $\trg{\Omega_1'} \isdef \trg{C, H', \OB{B'}\cdot B'\cup x\mapsto v' : \sigma' \triangleright s}$ where $\trg{B' \triangleright e\bigredt v' : \sigma'}$.
From \cref{lemma:low-equiv-bindings-low-equiv-results1}, $\trg{B} \approx \trg{B'}$, $\trg{B \triangleright e\bigredt v : \sigma}$, and  $\trg{B' \triangleright e\bigredt v' : \sigma'}$, we have $\trg{v:\sigma} \approx \trg{v':\sigma'}$.
Hence, $\trg{\Omega_1} \approx \trg{\Omega_1'}$ follows from $\trg{H} \approx \trg{H'}$, $\trg{\OB{B}} \approx \trg{\OB{B'}}$,  $\trg{B} \approx \trg{B'}$, and $\trg{v:\sigma} \approx \trg{v':\sigma'}$.

\item[E-\TR-write:]
Then, $\trg{\Omega_0} \isdef \trg{C, H_0, \OB{B}\cdot B \triangleright \asgn{e_1}{e_2}}$, $\trg{\lambda} = \trg{\wrl{\abs{n}\mapsto v}^{\sigma_1\sqcup\sigma_2}}$, $\trg{B \triangleright e_1\bigredt n : \sigma_1}$, $\trg{B \triangleright e_2\bigredt v : \sigma_2}$, $\trg{H_0}=\trg{H_2; \abs{n}\mapsto v_0 : \sigma_0 ; H_3}$, $\trg{H_1}=\trg{H_2; \abs{n}\mapsto v : \safeta ; H_3}$, and $\trg{\Omega_1} \isdef \trg{C, H_1, \OB{B}\cdot B \triangleright \skipt}$.
From $\trg{\Omega_0} \approx \trg{\Omega_0'}$, $\trg{\Omega_0'} \isdef \trg{C, H_0', \OB{B}\cdot B \triangleright \asgn{e_1}{e_2}}$ where $\trg{H_0} \approx \trg{H_0'}$, $\trg{\OB{B}} \approx \trg{\OB{B'}}$, and $\trg{B} \approx \trg{B'}$. 
From $\trg{H_0} \approx \trg{H_0'}$ and $\trg{H_0}=\trg{H_2; \abs{n}\mapsto v_0 : \sigma_0 ; H_3}$, it follows that there are $\trg{H_2'}, \trg{H_3'}$ such that $\trg{H_0'}=\trg{H_2'; \abs{n}\mapsto v_0' : \sigma_0 ; H_3'}$, $\trg{H_2} \approx \trg{H_2'}$, and $\trg{H_3} \approx \trg{H_3}'$.
Therefore, $\trg{\Omega_0' \xtot{\lambda'} \Omega_1'}$ has also been derived using the E-\TR-write rule.
From this, we get $\trg{\Omega_1'} \isdef \trg{C, \trg{H_2'; \abs{n'}\mapsto v' : \safeta ; H_3'}, \OB{B'}\cdot B' \triangleright \skipt}$ where $\trg{B' \triangleright e_1\bigredt n' : \sigma_1'}$ and $\trg{B' \triangleright e_2\bigredt v' : \sigma_2'}$.
From $\trg{\lambda} = \trg{\wrl{\abs{n}\mapsto v}^{\sigma_1\sqcup\sigma_2}}$, $\trg{\lambda'} = \trg{\wrl{\abs{n'}\mapsto v'}^{\sigma_1'\sqcup\sigma_2'}}$, and $\trg{\lambda} = \trg{\lambda'}$, we get that $\trg{|n|} = \trg{|n'|}$ and $\trg{v} = \trg{v'}$.
Hence, $\trg{\Omega_1} \approx \trg{\Omega_1'}$ follows from $\trg{H_2} \approx \trg{H_2'}$, $\trg{H_3} \approx \trg{H_3}'$, $\trg{\OB{B}} \approx \trg{\OB{B'}}$,  $\trg{B} \approx \trg{B'}$, $\trg{|n|} = \trg{|n'|}$, and $\trg{v} = \trg{v'}$ (the latter is needed since the label of the written value is $\trg{\safeta}$).

\item[E-\TR-read:]
Then, $\trg{\Omega_0} \isdef \trg{C, H, \OB{B}\cdot B \triangleright \letread{x}{e}{s}}$, $\trg{B \triangleright e\bigredt n : \sigma_1}$, $\trg{H}=\trg{H_1; \abs{n}\mapsto v : \sigma_0 ; H_2}$,  $\trg{\lambda} = \trg{\rdl{\abs{n}}^{\sigma_1}}$, and $\trg{\Omega_1} \isdef \trg{C, H, \OB{B}\cdot B\cup x\mapsto v : \sigma_0 \triangleright s}$.
From $\trg{\Omega_0} \approx \trg{\Omega_0'}$, we get that $\trg{\Omega_0'} \isdef \trg{C, H', \OB{B'}\cdot B' \triangleright \letread{x}{e}{s}}$, $\trg{H} \approx \trg{H'}$, $\trg{\OB{B}} \approx \trg{\OB{B'}}$, and $\trg{B} \approx \trg{B'}$.
Therefore, $\trg{\Omega_0' \xtot{\lambda'} \Omega_1'}$ has also been derived using the E-\TR-read rule.
Thus, $\trg{\Omega_1'} \isdef \trg{C, H', \OB{B'}\cdot B'\cup x\mapsto v' : \sigma_0' \triangleright s}$ where $\trg{B' \triangleright e\bigredt n' : \sigma_1'}$, $\trg{H'(\abs{n'})} = \trg{v' : \sigma_0'}$, and $\trg{\lambda'} = \trg{\rdl{\abs{n'}}^{\sigma_1'}}$.
From $\trg{\lambda} = \trg{\rdl{\abs{n}}^{\sigma_1}}$, $\trg{\lambda'} = \trg{\rdl{\abs{n'}}^{\sigma_1'}}$, and $\trg{\lambda} = \trg{\lambda'}$, we get $\trg{\abs{n}} = \trg{\abs{n'}}$.
From $\trg{\abs{n}} = \trg{\abs{n'}}$, $\trg{H} \approx \trg{H'}$, $\trg{H(\abs{n})} = \trg{v : \sigma_0}$, and $\trg{H'(\abs{n'})} = \trg{v' : \sigma_0'}$, we get $\trg{v:\sigma_0} \approx \trg{v':\sigma_0'}$.
Hence, $\trg{\Omega_1} \approx \trg{\Omega_1'}$ follows from $\trg{H} \approx \trg{H'}$, $\trg{\OB{B}} \approx \trg{\OB{B'}}$,  $\trg{B} \approx \trg{B'}$, $\trg{|n|} = \trg{|n'|}$, and $\trg{v:\sigma_0} \approx \trg{v':\sigma_0'}$.

\item[E-\TR-write-prv:]
Then, we have that $\trg{\Omega_0} \isdef \trg{C, H, \OB{B}\cdot B \triangleright \asgnp{e}{e'}}$, $\trg{B \triangleright e\bigredt}$ $\trg{n : \sigma_0}$, $\trg{B \triangleright e'\bigredt v : \sigma_1}$, $\trg{\lambda} = \trg{\wrl{-\abs{n}}^{\sigma_0}}$, $\trg{H}=\trg{H_2; -\abs{n}\mapsto v_0 : \sigma_2 ;}$ $\trg{H_3}$, $\trg{H_1}=\trg{H_2; -\abs{n}\mapsto v : \sigma_1 ; H_3}$, and $\trg{\Omega_1} \isdef \trg{ {C, H_1, \OB{B}\cdot B \triangleright \skipt} }$.
From $\trg{\Omega_0} \approx \trg{\Omega_0'}$, we get that  $\trg{\Omega_0'} \isdef \trg{C, H', \OB{B'}\cdot B' \triangleright \asgnp{e}{e'}}$, $\trg{H} \approx \trg{H'}$, $\trg{\OB{B}} \approx \trg{\OB{B'}}$, and $\trg{B} \approx \trg{B'}$.
Hence, $\trg{\Omega_0' \xtot{\lambda'} \Omega_1'}$ has also been derived using the E-\TR-write-prv rule.
Thus, $\trg{\Omega_1'} \isdef \trg{ {C, H_1', \OB{B'}\cdot B' \triangleright \skipt} }$ where  $\trg{B' \triangleright e\bigredt}$ $\trg{n' : \sigma_0'}$, $\trg{B' \triangleright e'\bigredt v' : \sigma_1'}$,  $\trg{H_1'} = \trg{H_2'; -\abs{n'}\mapsto v' : \sigma_1' ; H_3}$, and $\trg{\lambda'} = \trg{\wrl{-\abs{n'}}^{\sigma_0'}}$.
From $\trg{\lambda} = \trg{\lambda'}$, $\trg{\lambda'} = \trg{\wrl{-\abs{n'}}^{\sigma_0'}}$, and $\trg{\lambda} = \trg{\wrl{-\abs{n}}^{\sigma_0}}$, we get $\trg{\abs{n}:\sigma_0} = \trg{\abs{n'} : \sigma_0'}$.
From  $\trg{B} \approx \trg{B'}$,  $\trg{B' \triangleright e'\bigredt v' : \sigma_1'}$,  $\trg{B \triangleright e'\bigredt v : \sigma_1}$, and \cref{lemma:low-equiv-bindings-low-equiv-results1}, we get $\trg{v : \sigma_1 } \approx \trg{v' : \sigma_1'}$.
Hence, $\trg{\Omega_1} \approx \trg{\Omega_1'}$ follows from $\trg{H} \approx \trg{H'}$, $\trg{\OB{B}} \approx \trg{\OB{B'}}$,  $\trg{B} \approx \trg{B'}$, $\trg{\abs{n}:\sigma_0} = \trg{\abs{n'} : \sigma_0'}$, and $\trg{v : \sigma_1 } \approx \trg{v' : \sigma_1'}$.

\item[E-\TR-read-prv:]
Then, $\trg{\Omega_0} \isdef \trg{C, H, \OB{B}\cdot B \triangleright \letreadp{x}{e}{s}}$, $\trg{B \triangleright e\bigredt n :}$ $\trg{\sigma_0}$, $\trg{H}=\trg{H_1; -\abs{n}\mapsto v : \sigma_1 ; H_2}$, 	$\trg{\lambda} = \trg{\rdl{-\abs{n}}^{\sigma_0}}$, and $\trg{\Omega_1} \isdef \trg{C, H, \OB{B}}$ $\trg{\cdot B\cup x\mapsto v : \unta \triangleright s}$.
From $\trg{\Omega_0} \approx \trg{\Omega_0'}$, we get that $\trg{\Omega_0'} \isdef \trg{C, H', \OB{B'}\cdot B' \triangleright}$ $\trg{\letreadp{x}{e}{s}}$ where $\trg{H} \approx \trg{H'}$, $\trg{\OB{B}} \approx \trg{\OB{B'}}$, and $\trg{B} \approx \trg{B'}$.
Hence, $\trg{\Omega_0' \xtot{\lambda'} \Omega_1'}$ has also been derived using the E-\TR-read-prv rule.
Thus, $\trg{\Omega_1'} \isdef \trg{C, H', \OB{B'}\cdot B'\cup x\mapsto v' : \unta \triangleright s}$ where $\trg{B' \triangleright e\bigredt n' : \sigma_0'}$, $\trg{\lambda'} = \trg{\rdl{-\abs{n'}}^{\sigma_0'}}$, and $\trg{H'(-\abs{n'})} = \trg{v' : \sigma_1'}$.
From $\trg{\lambda} = \trg{\lambda'}$, $\trg{\lambda} = \trg{\rdl{-\abs{n}}^{\sigma_0}}$, and $\trg{\lambda'} = \trg{\rdl{-\abs{n}}^{\sigma_0'}}$, we have $\trg{\abs{n}:\sigma_0} = \trg{\abs{n'}:\sigma_0'}$.
From this and $\trg{H} \approx \trg{H'}$, we have $\trg{H(-\abs{n})} \approx \trg{H'(-\abs{n'})}$.
Hence, $\trg{\Omega_1} \approx \trg{\Omega_1'}$ follows from $\trg{H} \approx \trg{H'}$, $\trg{\OB{B}} \approx \trg{\OB{B'}}$,  $\trg{B} \approx \trg{B'}$, $\trg{\abs{n}:\sigma_0} = \trg{\abs{n'} : \sigma_0'}$, and $\trg{v : \sigma_1 } \approx \trg{v' : \sigma_1'}$ (observe that the latter is enough since $\trg{x}$ is tagged $\trg{\unta}$).

\item[E-\TR-call-internal:]
Then, $\trg{\Omega_0} \isdef \trg{C, H, \OB{B}\cdot B \triangleright \proc{{\call{f}~e}}{\OB{f'}}}$, $\trg{\lambda} = \trg{\epsilon}$, $		\trg{C}.\mtt{intfs}\vdash\trg{f,f'}:\trg{internal}$, $\trg{\OB{f'}} = \trg{\OB{f''};f'}$, $\trg{f(x)\mapsto s;\ret}\in\trg{C}.\mtt{funs}$,  $\trg{B \triangleright e\bigredt v : \sigma}$, and $\trg{\Omega_1} \isdef \trg{C, H, \OB{B}\cdot B\cdot x\mapsto v : \sigma \triangleright \proc{{s;\ret}}{\OB{f'};f}}$.
From $\trg{\Omega_0} \approx \trg{\Omega_0'}$, we get that $\trg{\Omega_0'} \isdef \trg{C, H', \OB{B'}\cdot B' \triangleright \proc{{\call{f}~e}}{\OB{f'}}}$ where $\trg{H} \approx \trg{H'}$, $\trg{\OB{B}} \approx \trg{\OB{B'}}$, and $\trg{B} \approx \trg{B'}$.
Hence, $\trg{\Omega_0' \xtot{\lambda'} \Omega_1'}$ has also been derived using the E-\TR-call-internal rule.
Thus, $\trg{\Omega_1'} \isdef \trg{C, H', \OB{B'}\cdot B'\cdot x\mapsto v' : \sigma' \triangleright \proc{{s;\ret}}{\OB{f'};f}}$ and $\trg{B' \triangleright e\bigredt v' : \sigma'}$.
From $\trg{B} \approx \trg{B'}$, $\trg{B \triangleright e\bigredt v : \sigma}$, $\trg{B' \triangleright e\bigredt v' : \sigma'}$, and \cref{lemma:low-equiv-bindings-low-equiv-results1}, we get $\trg{v : \sigma} \approx \trg{v' : \sigma'}$.
Hence, $\trg{\Omega_1} \approx \trg{\Omega_1'}$ follows from $\trg{H} \approx \trg{H'}$, $\trg{\OB{B}} \approx \trg{\OB{B'}}$,  $\trg{B} \approx \trg{B'}$,  and $\trg{v : \sigma } \approx \trg{v' : \sigma'}$.

\item[E-\TR-callback:]
Then, $\trg{\Omega_0} \isdef \trg{C, H, \OB{B}\cdot B \triangleright \proc{{\call{f}~e}}{\OB{f'}}}$, $\trg{\lambda} = \trg{\cbh{f}{v}{H}^\sigma}$, $\trg{\OB{f'}} = \trg{\OB{f''};f'}$, $\trg{f(x)\mapsto s;\ret}\in\trg{C}.\mtt{funs}$, ${\trg{C}}.\mtt{intfs}\vdash\trg{f',f}:\trg{out}$, $\trg{B \triangleright e \bigredt v:\sigma}$, and $\trg{\Omega_1} = \trg{C, H, \OB{B}\cdot B\cdot x\mapsto v:\sigma \triangleright \proc{{s;\ret}}{\OB{f'};f}}$.
From $\trg{\Omega_0} \approx \trg{\Omega_0'}$, we get that $\trg{\Omega_0'} \isdef \trg{C, H', \OB{B'}\cdot B' \triangleright \proc{{\call{f}~e}}{\OB{f'}}}$  where $\trg{H} \approx \trg{H'}$, $\trg{\OB{B}} \approx \trg{\OB{B'}}$, and $\trg{B} \approx \trg{B'}$.
Hence, $\trg{\Omega_0' \xtot{\lambda'} \Omega_1'}$ has also been derived using the E-\TR-callback rule.
Thus,  $\trg{\Omega_1'} = \trg{C, H', \OB{B'}\cdot B'\cdot x\mapsto v':}$ $\trg{\sigma' \triangleright \proc{{s;\ret}}{\OB{f'};f}}$, $\trg{B' \triangleright e \bigredt v':\sigma'}$, and $\trg{\lambda'} = \trg{\cbh{f}{v'}{H}^{\sigma'}}$.
From $\trg{\lambda} = \trg{\lambda'}$, $\trg{\lambda} = \trg{\cbh{f}{v}{H}^\sigma}$, and $\trg{\lambda'} = \trg{\cbh{f}{v'}{H}^{\sigma'}}$, we get $\trg{v:\sigma} = \trg{v':\sigma'}$.
Hence, $\trg{\Omega_1} \approx \trg{\Omega_1'}$ follows from $\trg{H} \approx \trg{H'}$, $\trg{\OB{B}} \approx \trg{\OB{B'}}$,  $\trg{B} \approx \trg{B'}$,  and $\trg{v : \sigma } = \trg{v' : \sigma'}$.

\item[E-\TR-call:]
The proof of this case is similar to that of the case E-\TR-callback.

\item[E-\TR-ret-internal:]
Then, $\trg{\Omega_0} \isdef \trg{C, H, \OB{B}\cdot B \triangleright \proc{{\ret}}{\OB{f'};f}}$, $\trg{\OB{f'}} = \trg{\OB{f''};f'}$, $\trg{{C}}.\mtt{intfs}\vdash\trg{f,f'}:\trg{internal}$, $\trg{\lambda} = \trg{\epsilon}$, and $\trg{\Omega_1} = \trg{C, H, \OB{B} \triangleright \proc{\skipt}{\OB{f'}}}$.
From $\trg{\Omega_0} \approx \trg{\Omega_0'}$, we get that $\trg{\Omega_0'} \isdef  \trg{C, H', \OB{B'}\cdot B' \triangleright \proc{{\ret}}{\OB{f'};f}}$  where $\trg{H} \approx \trg{H'}$, $\trg{\OB{B}} \approx \trg{\OB{B'}}$, and $\trg{B} \approx \trg{B'}$.
Hence, $\trg{\Omega_0' \xtot{\lambda'} \Omega_1'}$ has also been derived using the E-\TR-ret-internal rule.
Thus,  $\trg{\Omega_1} = \trg{C, H', \OB{B'} \triangleright \proc{\skipt}{\OB{f'}}}$.
Hence, $\trg{\Omega_1} \approx \trg{\Omega_1'}$ follows from $\trg{H} \approx \trg{H'}$ and $\trg{\OB{B}} \approx \trg{\OB{B'}}$.

\item[E-\TR-retback:]
The proof of this case is similar to that of E-\TR-retback.

\item[E-\TR-return:]
The proof of this case is similar to that of E-\TR-retback.

\item[E-\TR-lfence:]
The proof of this case is similar to that of E-\TR-skip.

\item[E-\TR-cmove-true:]

Then, $\trg{\Omega_0} \isdef \trg{C, H, \OB{B}\cdot B \triangleright \cmove{x}{e_0}{e_1}{s}}$, $\trg{x}\in\dom{\trg{B}}$, $\trg{\lambda} = \trg{\epsilon}$,  $\trg{\Omega_1} \isdef \trg{C, H, \OB{B}\cdot B\cup x\mapsto v_0 : \taintt \triangleright s}$, $\trg{B \triangleright e_0\bigred v_0 : \taintt_0}$, $\trg{B \triangleright e_1\bigred 0}$ $\trg{ : \taintt_1}$, and $\trg{\taintt} = \trg{\taintt_0} \sqcup \trg{\taintt_1}$.
From $\trg{\Omega_0} \approx \trg{\Omega_0'}$, we have that $\trg{\Omega_0'} \isdef \trg{C, H', \OB{B'}\cdot B' \triangleright \cmove{x}{e_0}{e_1}{s}}$ where $\trg{H} \approx \trg{H'}$, $\trg{\OB{B}\cdot B} \approx \trg{\OB{B}'\cdot B'}$, and $\trg{B} \approx \trg{B'}$.
From $\trg{B} \approx \trg{B'}$, $\trg{B \triangleright e_1\bigred 0}$ $\trg{ : \taintt_1}$, and \cref{lemma:low-equiv-bindings-low-equiv-results}, we get that $\trg{B' \triangleright e_1\bigred n'}$ $\trg{ : \taintt_1'}$ and $\trg{0:\taintt_1} \approx \trg{n':\taintt_1'}$.
There are two cases:
\begin{description}
\item[$\trg{\taintt_1} = \trg{\safeta}$:]
Then, $\trg{\taintt_1'} = \trg{\safeta}$ and $\trg{v_1'} = \trg{0}$ follows from $\trg{0:\taintt_1} \approx \trg{n':\taintt_1'}$.
Hence, $\trg{\Omega_0' \xtot{\lambda'} \Omega_1'}$ has also been derived using the E-\TR-cmove-true rule.
Thus, $\trg{\Omega_1} \isdef \trg{C, H', \OB{B'}\cdot B'\cup x\mapsto v_0' : \taintt' \triangleright s}$ where $\trg{B' \triangleright e_0\bigred v_0' : \taintt_0'}$ and $\trg{\taintt'} = \trg{\taintt_0'} \sqcup \trg{\taintt_1'}$.
From $\trg{B} \approx \trg{B'}$, $\trg{B \triangleright e_0\bigred v_0}$ $\trg{: \taintt_0}$, $\trg{B' \triangleright e_0\bigred v_0' : \taintt_0'}$, and \cref{lemma:low-equiv-bindings-low-equiv-results1}, we get $\trg{v_0:\taintt_0} \approx \trg{v_0':\taintt_0'}$.
Hence,  $\trg{\Omega_1} \approx \trg{\Omega_1'}$ follows from $\trg{H} \approx \trg{H'}$, $\trg{\OB{B}} \approx \trg{\OB{B'}}$, $\trg{B} \approx \trg{B'}$, and $\trg{v_0:\taintt_0} \approx \trg{v_0':\taintt_0'}$.

\item[$\trg{\taintt_1} = \trg{\unta}$:]
Then, $\trg{\taintt_1'} = \trg{\unta}$ holds as well.
There are two cases:
\begin{description}
\item[$\trg{v_1'} = \trg{0}$:]
Hence, $\trg{\Omega_0' \xtot{\lambda'} \Omega_1'}$ has also been derived using the E-\TR-cmove-true rule.
Thus, $\trg{\Omega_1} \isdef \trg{C, H', \OB{B'}\cdot B'\cup x\mapsto v_0' : \taintt' \triangleright s}$ where $\trg{B' \triangleright e_0\bigred v_0' : \taintt_0'}$ and $\trg{\taintt'} = \trg{\taintt_0'} \sqcup \trg{\taintt_1'}$.
Hence,  $\trg{\Omega_1} \approx \trg{\Omega_1'}$ follows from $\trg{H} \approx \trg{H'}$, $\trg{\OB{B}} \approx \trg{\OB{B'}}$, $\trg{B} \approx \trg{B'}$, and $\sigma = \sigma' = \trg{\unta}$.

\item[$\trg{v_1'} > \trg{0}$:]
Hence, $\trg{\Omega_0' \xtot{\lambda'} \Omega_1'}$ has been derived using the E-\TR-cmove-false rule.
Thus, $\trg{\Omega_1'} \isdef \trg{C, H', \OB{B'}\cdot B'\cup x\mapsto v_0 :\taintt' \triangleright s }$ where $\trg{B(x)} = \trg{v_0' : \taintt_0'}$ and $\trg{\taintt' = \taintt_0' \sqcup \taintt_1'}$.
Hence,  $\trg{\Omega_1} \approx \trg{\Omega_1'}$ follows from $\trg{H} \approx \trg{H'}$, $\trg{\OB{B}} \approx \trg{\OB{B'}}$, $\trg{B} \approx \trg{B'}$, and $\sigma = \sigma' = \trg{\unta}$.

\end{description}

\end{description}

\item[E-\TR-cmove-false:]
The proof of this case is similar to that of the E-\TR-cmove-true case.
\end{description}

\item[Induction step:]

Then, $\trg{\Omega_0 \xtot{\lambda} \Omega_1}$ has been derived using the E-\TR-step rule.
Then, $\trg{\Omega_0} \isdef \trg{C, H, \OB{B} \triangleright s;s''}$, $\trg{C,}$ $\trg{H,}$ $\trg{\OB{B} \triangleright s} \xtot{\trg{\lambda}} \trg{C,}$ $\trg{H_1, \OB{B_1} \triangleright s_1}$, and $\trg{\Omega_1} = \trg{C, H_1, \OB{B_1} \triangleright s_1;s''}$.
From $\trg{\Omega_0} \approx \trg{\Omega_0'}$, we get that $\trg{\Omega_0} \isdef \trg{C, H', \OB{B'} \triangleright s;s''}$ where $\trg{H} \approx \trg{H'}$ and $\trg{\OB{B}} \approx \trg{\OB{B'}}$.
Hence,  $\trg{\Omega_0' \xtot{\lambda'} \Omega_1'}$ has also been derived using the E-\TR-step rule.
Therefore, we get $\trg{\Omega_0'} \isdef \trg{C, H', \OB{B'} \triangleright s;s''}$, $\trg{C,}$ $\trg{H',}$ $\trg{\OB{B'} \triangleright s} \xtot{\trg{\lambda'}} \trg{C,}$ $\trg{H_1', \OB{B_1'} \triangleright s_1}$, and $\trg{\Omega_1'} = \trg{C, H_1', \OB{B_1'} \triangleright s_1;s''}$.
From $\trg{H} \approx \trg{H'}$, $\trg{\OB{B}} \approx \trg{\OB{B'}}$, $\trg{C,}$ $\trg{H,}$ $\trg{\OB{B} \triangleright s} \xtot{\trg{\lambda}} \trg{C,}$ $\trg{H_1, \OB{B_1} \triangleright s_1}$, $\trg{C,}$ $\trg{H',}$ $\trg{\OB{B'} \triangleright s} \xtot{\trg{\lambda'}} \trg{C,}$ $\trg{H_1', \OB{B_1'} \triangleright s_1}$, $\trg{\lambda} = \trg{\lambda'}$, and the induction hypothesis, we get that $\trg{H_1} \approx \trg{H_1'}$ and $\trg{\OB{B_1}} \approx \trg{\OB{B_1'}}$.
Hence, $\trg{\Omega_1} \approx \trg{\Omega_1'}$ follows from $\trg{H_1} \approx \trg{H_1'}$ and $\trg{\OB{B_1}} \approx \trg{\OB{B_1'}}$.
\end{description}
\end{proof}

\begin{lemma}[Steps of $\xtot{}$ with safe observations preserve low-equivalence]\label{lemma:safe-obs-preserve-low-equivalence-for-non-spec-semantics}
\begin{align*}
	\forall \trg{\Omega_0}, \trg{\Omega_0'}, \trg{\lambda}, \trg{\Omega_1}.
	& \text{ if } 
		\trg{\Omega_0 \xtot{\lambda} \Omega_1}, \trg{\Omega_0} \approx \trg{\Omega_0'}, \safe{\trg{\lambda}} \\
	& \text{ then } \exists \trg{\lambda'},\trg{\Omega_1}.\  \trg{\Omega_0' \xtot{\lambda'} \Omega_1'}, \trg{\Omega_1} \approx \trg{\Omega_1'},\ \trg{\lambda} = \trg{\lambda'}
\end{align*}
\end{lemma}

\begin{proof}
Let  $\trg{\Omega_0}$, $\trg{\Omega_0'}$, $\trg{\lambda}$, and $\trg{\Omega_1}$ be such that $\trg{\Omega_0} \approx \trg{\Omega_0'}$, $\trg{\Omega_0 \xtot{\lambda} \Omega_1}$, and $\safe{\trg{\lambda}}$.
We proceed by structural induction on the rules used to derive $\trg{\Omega_0 \xtot{\lambda} \Omega_1}$:
\begin{description}
\item[Base case:]
There are several cases depending on the rule used to derive $\trg{\Omega_0 \xtot{\lambda} \Omega_1}$:
\begin{description}
\item[E-\TR-sequence:]
Then, $\trg{\Omega_0} \isdef \trg{C, H, \OB{B} \triangleright \skipt;s}$, $\trg{\lambda} = \trg{\epsilon}$, and $\trg{\Omega_1} \isdef  \trg{C, H, \OB{B} \triangleright s}$.
From $\trg{\Omega_0} \approx \trg{\Omega_0'}$, we have that $\trg{\Omega_0'} \isdef \trg{C, H', \OB{B'} \triangleright \skipt;s}$ where $\trg{H} \approx \trg{H'}$ and $\trg{\OB{B}} \approx \trg{\OB{B'}}$.
We can apply the E-\TR-sequence rule to derive $\trg{\Omega_0' \xtot{\lambda'} \Omega_1'}$ where $\trg{\Omega_1} \isdef \trg{C, H', \OB{B'} \triangleright s}$ and $\trg{\lambda'}= \trg{\epsilon}$.
Hence, $\trg{\Omega_1} \approx \trg{\Omega_1'}$ follows from $\trg{\Omega_0} \approx \trg{\Omega_0'}$ and $\trg{\lambda} = \trg{\lambda'}$ follows from $\trg{\lambda} = \trg{\epsilon}$ and $\trg{\lambda'}= \trg{\epsilon}$.

\item[E-\TR-if-true:]
Then, $\trg{\Omega_0} \isdef \trg{C, H, \OB{B}\cdot B \triangleright \ifzte{e}{s}{s'}}$, $\trg{B \triangleright e\bigredt 0 : \sigma}$
$\trg{\lambda} = \trg{(\ifl{0})^\sigma}$, and $\trg{\Omega_1} \isdef \trg{C, H, \OB{B}\cdot B \triangleright s}$.
From $\trg{\Omega_0} \approx \trg{\Omega_0'}$, we have that $\trg{\Omega_0'} \isdef \trg{C, H', \OB{B'}\cdot B' \triangleright \ifzte{e}{s}{s'}}$ where $\trg{H} \approx \trg{H'}$, $\trg{\OB{B}} \approx \trg{\OB{B'}}$, and $\trg{B} \approx \trg{B'}$.
From $\trg{B \triangleright e\bigredt 0 : \sigma}$, $\trg{B} \approx \trg{B'}$, and \cref{lemma:low-equiv-bindings-low-equiv-results}, we get that $\trg{B' \triangleright e\bigredt v' : \sigma'}$ and $\trg{0 :\sigma} \approx \trg{v':\sigma'}$.
From $\safe{\trg{\lambda}}$ and $\trg{\lambda} = \trg{(\ifl{0})^\sigma}$, it follows that $\sigma = \safeta$ and $\trg{\lambda} = \trg{(\ifl{0})^\safeta}$.
From this and $\trg{0 :\sigma} \approx \trg{v':\sigma'}$, it follows that $v' = 0$.
Hence, $\trg{B' \triangleright e\bigredt 0 : \safeta'}$ holds.
Therefore, we can apply the E-\TR-if-true rule to derive $\trg{\Omega_0' \xtot{\lambda'} \Omega_1'}$ where $\trg{\Omega_1} \isdef \trg{C, H', \OB{B'}\cdot B' \triangleright s}$ and $\trg{\lambda'}= \trg{(\ifl{0})^\safeta}$.
Hence,  $\trg{\Omega_1} \approx \trg{\Omega_1'}$ follows from $\trg{\Omega_0} \approx \trg{\Omega_0'}$ and $\trg{\lambda} = \trg{\lambda'}$ follows from $\trg{\lambda} = \trg{(\ifl{0})^\safeta}$ and $\trg{\lambda'}= \trg{(\ifl{0})^\safeta}$.

\item[E-\TR-if-false:]
The proof of this case is similar to that of the E-\TR-if-true case.

\item[E-\TR-letin:]
Then, $\trg{\Omega_0} \isdef \trg{C, H, \OB{B}\cdot B \triangleright \letin{x}{e}{s}}$, $\trg{\lambda} = \trg{\epsilon}$, $\trg{B \triangleright e\bigredt v : \sigma}$, and $\trg{\Omega_1} \isdef \trg{C, H, \OB{B}\cdot B\cup x\mapsto v : \sigma \triangleright s}$.
From $\trg{\Omega_0} \approx \trg{\Omega_0'}$, we have that  $\trg{\Omega_0'} \isdef \trg{C, H', \OB{B'}\cdot B' \triangleright \letin{x}{e}{s}}$ where $\trg{H} \approx \trg{H'}$, $\trg{\OB{B}} \approx \trg{\OB{B'}}$, and $\trg{B} \approx \trg{B'}$.
From $\trg{B} \approx \trg{B'}$, $\trg{B \triangleright e\bigredt v : \sigma}$, and \cref{lemma:low-equiv-bindings-low-equiv-results}, we have $\trg{B' \triangleright e\bigredt v' : \sigma'}$ and $\trg{v:\sigma} \approx \trg{v':\sigma'}$.
Therefore, we can apply the rule E-\TR-letin to derive $\trg{\Omega_0' \xtot{\lambda'} \Omega_1'}$ where $\trg{\lambda'} = \trg{\epsilon}$ and $\trg{\Omega_1} \isdef \trg{C, H', \OB{B'}\cdot B'\cup x\mapsto v' : \sigma' \triangleright s}$.
Hence, $\trg{\Omega_1} \approx \trg{\Omega_1'}$ follows from $\trg{\Omega_0} \approx \trg{\Omega_0'}$ and $\trg{v:\sigma} \approx \trg{v':\sigma'}$ whereas $\trg{\lambda} = \trg{\lambda'}$ follows from $\trg{\lambda} = \trg{\epsilon}$ and $\trg{\lambda'}= \trg{\epsilon}$.

\item[E-\TR-write:]
Then, $\trg{\Omega_0} \isdef \trg{C, H_0, \OB{B}\cdot B \triangleright \asgn{e_1}{e_2}}$, $\trg{\lambda} = \trg{\wrl{\abs{n}\mapsto v}^{\sigma_1\sqcup\sigma_2}}$, $\trg{B \triangleright e_1\bigredt n : \sigma_1}$, $\trg{B \triangleright e_2\bigredt v : \sigma_2}$, $\trg{H_0}=\trg{H_2; \abs{n}\mapsto v_0 : \sigma_0 ; H_3}$, $\trg{H_1}=\trg{H_2; \abs{n}\mapsto v : \safeta ; H_3}$, and $\trg{\Omega_1} \isdef \trg{C, H_1, \OB{B}\cdot B \triangleright \skipt}$.

From $\trg{\Omega_0} \approx \trg{\Omega_0'}$, $\trg{\Omega_0'} \isdef \trg{C, H_0', \OB{B}\cdot B \triangleright \asgn{e_1}{e_2}}$ where $\trg{H_0} \approx \trg{H_0'}$, $\trg{\OB{B}} \approx \trg{\OB{B'}}$, and $\trg{B} \approx \trg{B'}$. 
From $\trg{H_0} \approx \trg{H_0'}$ and $\trg{H_0}=\trg{H_2; \abs{n}\mapsto v_0 : \sigma_0 ; H_3}$, it follows that there are $\trg{H_2'}, \trg{H_3'}$ such that $\trg{H_0'}=\trg{H_2'; \abs{n}\mapsto v_0' : \sigma_0 ; H_3'}$.
From $\trg{B} \approx \trg{B'}$, $\trg{B \triangleright e_1\bigredt n : \sigma_1}$, $\trg{B \triangleright e_2\bigredt v : \sigma_2}$, and \cref{lemma:low-equiv-bindings-low-equiv-results}, we get  $\trg{B' \triangleright e_1\bigredt n' : \sigma_1'}$ and $\trg{B' \triangleright e_2\bigredt v' : \sigma_2'}$ such that $\trg{n:\sigma_1} \approx \trg{n':\sigma_1'}$ and $\trg{v:\sigma_2} \approx \trg{v':\sigma_2'}$.
From $\safe{\trg{\lambda}}$, we get $\trg{\sigma_1\sqcup\sigma_2} = \trg{\safeta}$ and therefore $\trg{\sigma_1} = \trg{\safeta}$ and $\trg{\sigma_2} = \trg{\safeta}$.
From  $\trg{\sigma_1} = \trg{\safeta}$, $\trg{\sigma_2} = \trg{\safeta}$, $\trg{n:\sigma_1} \approx \trg{n':\sigma_1'}$, and $\trg{v:\sigma_2} \approx \trg{v':\sigma_2'}$, we therefore get that $\trg{n:\sigma_1} = \trg{n':\sigma_1'}$, and $\trg{v:\sigma_2} = \trg{v':\sigma_2'}$.
Therefore, we can apply the rule E-\TR-write to derive  $\trg{\Omega_0' \xtot{\lambda'} \Omega_1'}$ where $\trg{\lambda'} = \trg{\wrl{\abs{n'} \mapsto v'}^{\sigma_1'\sqcup\sigma_2'}}$, $H_1' = \trg{H_2'; \abs{n}\mapsto v' : \safeta ; H_3'}$, $\trg{\Omega_1} \isdef \trg{C, H_1', \OB{B'}\cdot B' \triangleright \skipt}$.

From $\trg{n:\sigma_1} = \trg{n':\sigma_1'}$,  $\trg{v:\sigma_2} = \trg{v':\sigma_2'}$, $\trg{\lambda'} = \trg{\wrl{\abs{n'}\mapsto v'}^{\sigma_1'\sqcup\sigma_2'}}$, and $\trg{\lambda} = \trg{\wrl{\abs{n} \mapsto v}^{\sigma_1\sqcup\sigma_2}}$, we immediately get that $\trg{\lambda} = \trg{\lambda'}$.

Finally,  $\trg{\Omega_1} \approx \trg{\Omega_1'}$  follows from $\trg{\Omega_0} \approx \trg{\Omega_0'}$, $\trg{\Omega_1} \isdef \trg{C, H_1', \OB{B'}\cdot B' \triangleright \skipt}$, $\trg{\Omega_1'} \isdef \trg{C, H_1', \OB{B'}\cdot B' \triangleright \skipt}$, $\trg{H_1}=\trg{H_2; \abs{n}\mapsto v : \safeta ; H_3}$,  $\trg{H_1'} = \trg{H_2';}$ $\trg{\abs{n'}\mapsto v' : \safeta ;}$ $\trg{H_3'}$, $\trg{n:\sigma_1} = \trg{n':\sigma_1'}$,  $\trg{v:\sigma_2} = \trg{v':\sigma_2'}$,  $\trg{\OB{B}} \approx \trg{\OB{B'}}$, and $\trg{B} \approx \trg{B'}$.

\item[E-\TR-read:]
Then, $\trg{\Omega_0} \isdef \trg{C, H, \OB{B}\cdot B \triangleright \letread{x}{e}{s}}$, $\trg{B \triangleright e\bigredt n : \sigma_1}$, $\trg{H}=\trg{H_1; \abs{n}\mapsto v : \sigma_0 ; H_2}$,  $\trg{\lambda} = \trg{\rdl{\abs{n}}^{\sigma_1}}$, and $\trg{\Omega_1} \isdef \trg{C, H, \OB{B}\cdot B\cup x\mapsto v : \sigma_0 \triangleright s}$.
From $\trg{B} \approx \trg{B'}$, $\trg{B \triangleright e\bigredt n : \sigma_1}$, and \cref{lemma:low-equiv-bindings-low-equiv-results}, we get $\trg{B' \triangleright e\bigredt n' : \sigma_1'}$ and $\trg{n : \sigma_1} \approx \trg{n' : \sigma_1'}$.
From $\safe{\trg{\lambda}}$, we have that $\trg{\sigma_1} = \trg{\safeta}$.
From $\trg{\sigma_1} = \trg{\safeta}$ and $\trg{n : \sigma_1} \approx \trg{n' : \sigma_1'}$, we get $\trg{n : \sigma_1} = \trg{n' : \sigma_1'}$.

From $\trg{\Omega_0} \approx \trg{\Omega_0'}$ and $\trg{n : \sigma_1} = \trg{n' : \sigma_1'}$, we get that $\trg{\Omega_0'} \isdef \trg{C, H', \OB{B'}\cdot B' \triangleright \letread{x}{e}{s}}$ where  $\trg{H'}=\trg{H_1'; \abs{n}\mapsto v' : \sigma_0' ; H_2'}$, $\trg{H} \approx \trg{H'}$, $\trg{\OB{B}} \approx \trg{\OB{B'}}$,  $\trg{B} \approx \trg{B'}$, and $\trg{v : \sigma_0} \approx \trg{v' : \sigma_0'}$.

Therefore, we can apply the rule E-\TR-read to derive  $\trg{\Omega_0' \xtot{\lambda'} \Omega_1'}$ where $\trg{\lambda'} = \trg{\rdl{\abs{n}}^{\sigma_1'}}$, and  $\trg{\Omega_1} \isdef \trg{C, H', \OB{B'}\cdot B' \cup x \mapsto v' : \taintt_0'\triangleright \skipt}$.
Hence, $\trg{\Omega_1} \approx \trg{\Omega_1'}$  and $\trg{\lambda} = \trg{\lambda'}$ follow from $\trg{H} \approx \trg{H'}$, $\trg{\OB{B}} \approx \trg{\OB{B'}}$,  $\trg{B} \approx \trg{B'}$, $\trg{v : \sigma_0} \approx \trg{v' : \sigma_0'}$, and $\trg{n : \sigma_1} = \trg{n' : \sigma_1'}$.

\item[E-\TR-write-prv:]
Then, we have that $\trg{\Omega_0} \isdef \trg{C, H, \OB{B}\cdot B \triangleright \asgnp{e}{e'}}$, $\trg{B \triangleright e\bigredt n : \sigma_0}$, $\trg{B \triangleright e'\bigredt v : \sigma_1}$, $\trg{\lambda} = \trg{\wrl{-\abs{n}}^{\sigma_0}}$, $\trg{H}=\trg{H_2; -\abs{n}\mapsto v_0 : \sigma_2 ; H_3}$, $\trg{H_1}=\trg{H_2; -\abs{n}\mapsto v : \sigma_1 ; H_3}$, and $\trg{\Omega_1} \isdef \trg{ {C, H_1, \OB{B}\cdot B \triangleright \skipt} }$.

From $\trg{\Omega_0} \approx \trg{\Omega_0'}$, we get that  $\trg{\Omega_0'} \isdef \trg{C, H', \OB{B'}\cdot B' \triangleright \asgnp{e}{e'}}$, $\trg{H'}=\trg{H_2'; -\abs{n}\mapsto v_0' : \sigma_2' ; H_3'}$, $\trg{H} \approx \trg{H'}$, $\trg{\OB{B}} \approx \trg{\OB{B'}}$, and $\trg{B} \approx \trg{B'}$.

From $\safe{\trg{\lambda}}$, we get that $\trg{\sigma_0} = \trg{\safeta}$.
From $\trg{\Omega_0} \approx \trg{\Omega_0'}$, $\trg{B \triangleright e\bigredt n : \sigma_0}$, and \cref{lemma:low-equiv-bindings-low-equiv-results}, we get $\trg{B \triangleright e\bigredt n' : \sigma_0'}$ and $\trg{n : \sigma_0} \approx \trg{n' : \sigma_0'}$.
From this and $\trg{\sigma_0} = \trg{\safeta}$, we get $\trg{n : \sigma_0} = \trg{n' : \sigma_0'}$.

From  $\trg{\Omega_0} \approx \trg{\Omega_0'}$, $\trg{B \triangleright e'\bigredt v : \sigma_1}$, and \cref{lemma:low-equiv-bindings-low-equiv-results}, we get $\trg{B \triangleright e'\bigredt v' : \sigma_1'}$ and $\trg{v : \sigma_1} \approx \trg{v' : \sigma_1'}$.

Hence, we can apply the rule E-\TR-write-prv to derive  $\trg{\Omega_0' \xtot{\lambda'} \Omega_1'}$ where $\trg{\lambda} = \trg{\wrl{-\abs{n'}}^{\sigma_0'}}$ and  $\trg{\Omega_1} \isdef \trg{C, H', \OB{B'}\cdot B' \triangleright \skipt}$ where $\trg{H_1'}=\trg{H_2'; -\abs{n'}\mapsto v' : \sigma_1' ; H_3}$.
Therefore, $\trg{\lambda} = \trg{\lambda'}$ follows from $\trg{n : \sigma_0} = \trg{n' : \sigma_0'}$ and $\trg{\Omega_1} \approx \trg{\Omega_1'}$ follows from $\trg{\Omega_0} \approx \trg{\Omega_0'}$, $\trg{n : \sigma_0} = \trg{n' : \sigma_0'}$, and $\trg{v : \sigma_1} \approx \trg{v' : \sigma_1'}$.

\item[E-\TR-read-prv:]
Then, $\trg{\Omega_0} \isdef \trg{C, H, \OB{B}\cdot B \triangleright \letreadp{x}{e}{s}}$, $\trg{B \triangleright e\bigredt n : \sigma_0}$, $\trg{H}=\trg{H_1; -\abs{n}\mapsto v : \sigma_1 ; H_2}$, 	$\trg{\lambda} = \trg{\rdl{-\abs{n}}^{\sigma_0}}$, and $\trg{\Omega_1} \isdef \trg{C, H, \OB{B}\cdot B\cup x\mapsto v : \unta \triangleright s}$.
	
From $\trg{\Omega_0} \approx \trg{\Omega_0'}$, we get that $\trg{\Omega_0'} \isdef \trg{C, H', \OB{B'}\cdot B' \triangleright \letreadp{x}{e}{s}}$ where $\trg{H} \approx \trg{H'}$, $\trg{\OB{B}} \approx \trg{\OB{B'}}$, and $\trg{B} \approx \trg{B'}$.
From $\trg{H} \approx \trg{H'}$ and $\trg{H}=\trg{H_1; -\abs{n}\mapsto v : \sigma_1 ; H_2}$, we get that $\trg{H}=\trg{H_1'; -\abs{n}\mapsto v' : \sigma_1' ; H_2'}$ and $\trg{v:\sigma_1} \approx \trg{v':\sigma'}$.
Moreover, from $\trg{B \triangleright e\bigredt n : \sigma_0}$, $\trg{B} \approx \trg{B'}$, and \cref{lemma:low-equiv-bindings-low-equiv-results}, we get that $\trg{B' \triangleright e\bigredt n' : \sigma_0'}$ and $\trg{n : \sigma_0} \approx \trg{n':\sigma_0'}$.

From $\safe{\trg{\lambda}}$, we get that $\trg{\sigma_0} = \trg{\safeta}$.
From this and $\trg{n : \sigma_0} \approx \trg{n':\sigma_0'}$, we get $\trg{n : \sigma_0} = \trg{n':\sigma_0'}$.

We can apply the rule E-\TR-read-prv to derive  $\trg{\Omega_0' \xtot{\lambda'} \Omega_1'}$ where $\trg{\lambda} = \trg{\rdl{-\abs{n'}}^{\sigma_0'}}$,  and  $\trg{\Omega_1'} \isdef \trg{C, H', \OB{B'}\cdot B' \cup x \mapsto v':\unta \triangleright \skipt}$.
Therefore, $\trg{\lambda} = \trg{\lambda'}$ follows from $\trg{n : \sigma_0} = \trg{n':\sigma_0'}$, whereas $\trg{\Omega_1} \approx \trg{\Omega_1'}$ follows from $\trg{\Omega_0} \approx \trg{\Omega_0'}$ and the fact that the values assigned to $\trg{x}$ has tag $\trg{\unta}$. 

\item[E-\TR-call-internal:]
Then, $\trg{\Omega_0} \isdef \trg{C, H, \OB{B}\cdot B \triangleright \proc{{\call{f}~e}}{\OB{f'}}}$, $\trg{\lambda} = \trg{\epsilon}$, $		\trg{C}.\mtt{intfs}\vdash\trg{f,f'}:\trg{internal}$, $\trg{\OB{f'}} = \trg{\OB{f''};f'}$, $\trg{f(x)\mapsto s;\ret}\in\trg{C}.\mtt{funs}$,  $\trg{B \triangleright e\bigredt v : \sigma}$, and $\trg{\Omega_1} \isdef \trg{C, H, \OB{B}\cdot B\cdot x\mapsto v : \sigma \triangleright \proc{{s;\ret}}{\OB{f'};f}}$.
From $\trg{\Omega_0} \approx \trg{\Omega_0'}$, we get that $\trg{\Omega_0'} \isdef \trg{C, H', \OB{B'}\cdot B' \triangleright \proc{{\call{f}~e}}{\OB{f'}}}$ where $\trg{H} \approx \trg{H'}$, $\trg{\OB{B}} \approx \trg{\OB{B'}}$, and $\trg{B} \approx \trg{B'}$.
From  $\trg{B} \approx \trg{B'}$, $\trg{B \triangleright e\bigredt v : \sigma}$, and \cref{lemma:low-equiv-bindings-low-equiv-results}, we get that $\trg{B \triangleright e\bigredt v' : \sigma'}$ and $\trg{v:\sigma} \approx \trg{v:\sigma'}$.
Therefore, we can apply the rule E-\TR-call-internal to derive  $\trg{\Omega_0' \xtot{\lambda'} \Omega_1'}$ where $\trg{\lambda'} = \trg{\epsilon}$,  and  $\trg{\Omega_1'} \isdef \trg{C, H', \OB{B'}\cdot B'\cdot x\mapsto v' : \sigma' \triangleright \proc{{s;\ret}}{\OB{f'};f}}$.
Hence, $\trg{\lambda} = \trg{\lambda'}$ (which follows from $\trg{\lambda} = \trg{\epsilon}$ and $\trg{\lambda'} = \trg{\epsilon}$), and $\trg{\Omega_1} \approx \trg{\Omega_1'}$ (which follows from $\trg{\Omega_0} \approx \trg{\Omega_0'}$ and $\trg{v:\sigma} \approx \trg{v:\sigma'}$).

\item[E-\TR-callback:]

Then, $\trg{\Omega_0} \isdef \trg{C, H, \OB{B}\cdot B \triangleright \proc{{\call{f}~e}}{\OB{f'}}}$, $\trg{\lambda} = \trg{\cbh{f}{v}{H}^\sigma}$, 
$\trg{\OB{f'}} = \trg{\OB{f''};f'}$, $\trg{f(x)\mapsto s;\ret}\in\trg{C}.\mtt{funs}$, ${\trg{C}}.\mtt{intfs}\vdash\trg{f',f}:\trg{out}$, $\trg{B \triangleright e \bigredt v:\sigma}$, and $\trg{\Omega_1} = \trg{C, H, \OB{B}\cdot B\cdot x\mapsto v:\sigma \triangleright \proc{{s;\ret}}{\OB{f'};f}}$.
From $\trg{\Omega_0} \approx \trg{\Omega_0'}$, we get that $\trg{\Omega_0'} \isdef \trg{C, H', \OB{B'}\cdot B' \triangleright \proc{{\call{f}~e}}{\OB{f'}}}$  where $\trg{H} \approx \trg{H'}$, $\trg{\OB{B}} \approx \trg{\OB{B'}}$, and $\trg{B} \approx \trg{B'}$.
From  $\trg{B} \approx \trg{B'}$, $\trg{B \triangleright e\bigredt v : \sigma}$, and \cref{lemma:low-equiv-bindings-low-equiv-results}, we get that $\trg{B \triangleright e\bigredt v' : \sigma'}$ and $\trg{v:\sigma} \approx \trg{v:\sigma'}$.
Therefore, we can apply the rule E-\TR-callback to derive  $\trg{\Omega_0' \xtot{\lambda'} \Omega_1'}$ where $\trg{\lambda'} = \trg{\cbh{f}{v'}{H}^{\sigma'}}$ and   $\trg{\Omega_1'} \isdef \trg{C, H', \OB{B'}\cdot B'\cdot x\mapsto v':\sigma' \triangleright \proc{{s;\ret}}{\OB{f'};f}}$.
From $\safe{\trg{\lambda}}$, we get that $\trg{\sigma} = \trg{\safeta}$.
From this and $\trg{v:\sigma} \approx \trg{v:\sigma'}$, we get that $\trg{v:\sigma} = \trg{v:\sigma'}$.
Hence, $\trg{\lambda} = \trg{\lambda'}$ (which follows from $\trg{\lambda} = \trg{\cbh{f}{v}{H}^\sigma}$, $\trg{\lambda} = \trg{\cbh{f}{v'}{H}^{\sigma'}}$, and $\trg{v:\sigma} = \trg{v:\sigma'}$), and $\trg{\Omega_1} \approx \trg{\Omega_1'}$ (which follows from $\trg{\Omega_0} \approx \trg{\Omega_0'}$ and $\trg{v:\sigma} \approx \trg{v:\sigma'}$).

\item[E-\TR-call:]
The proof of this case is similar to that of the case E-\TR-callback.

\item[E-\TR-ret-internal:]
Then, $\trg{\Omega_0} \isdef \trg{C, H, \OB{B}\cdot B \triangleright \proc{{\ret}}{\OB{f'};f}}$, $\trg{\OB{f'}} = \trg{\OB{f''};f'}$, $\trg{{C}}.\mtt{intfs}\vdash\trg{f,f'}:\trg{internal}$, $\trg{\lambda} = \trg{\epsilon}$, and $\trg{\Omega_1} = \trg{C, H, \OB{B} \triangleright \proc{\skipt}{\OB{f'}}}$.
From $\trg{\Omega_0} \approx \trg{\Omega_0'}$, we get that $\trg{\Omega_0'} \isdef  \trg{C, H', \OB{B'}\cdot B' \triangleright \proc{{\ret}}{\OB{f'};f}}$  where $\trg{H} \approx \trg{H'}$, $\trg{\OB{B}} \approx \trg{\OB{B'}}$, and $\trg{B} \approx \trg{B'}$.
Therefore, we can apply the rule E-\TR-ret-internal to derive  $\trg{\Omega_0' \xtot{\lambda'} \Omega_1'}$ where $\trg{\lambda'} = \trg{\epsilon}$ and   $\trg{\Omega_1'} \isdef \trg{C, H', \OB{B'} \triangleright \proc{\skipt}{\OB{f'}}}$.
Hence, $\trg{\lambda} = \trg{\lambda'}$ (which follows from  $\trg{\lambda} = \trg{\epsilon}$ and  $\trg{\lambda'} = \trg{\epsilon}$) and $\trg{\Omega_1} \approx \trg{\Omega_1'}$ (which follows from $\trg{H} \approx \trg{H'}$ and $\trg{\OB{B}} \approx \trg{\OB{B'}}$).

\item[E-\TR-retback:]
Then, $\trg{\Omega_0} \isdef \trg{C, H, \OB{B}\cdot B \triangleright \proc{{\ret}}{\OB{f'};f}}$, $\trg{\OB{f'}} = \trg{\OB{f''};f'}$, $\trg{{C}}.\mtt{intfs}\vdash\trg{f,f'}:\trg{internal}$, $\trg{\lambda} = \trg{\rbh{}{H}^\safeta}$, and $\trg{\Omega_1} = \trg{C, H, \OB{B} \triangleright \proc{\skipt}{\OB{f'}}}$.
From $\trg{\Omega_0} \approx \trg{\Omega_0'}$, we get that $\trg{\Omega_0'} \isdef  \trg{C, H', \OB{B'}\cdot B' \triangleright \proc{{\ret}}{\OB{f'};f}}$  where $\trg{H} \approx \trg{H'}$, $\trg{\OB{B}} \approx \trg{\OB{B'}}$, and $\trg{B} \approx \trg{B'}$.
Therefore, we can apply the rule E-\TR-ret-retback to derive  $\trg{\Omega_0' \xtot{\lambda'} \Omega_1'}$ where $\trg{\lambda'} = \trg{\rbh{}{H}^\safeta}$ and   $\trg{\Omega_1'} \isdef \trg{C, H', \OB{B'} \triangleright \proc{\skipt}{\OB{f'}}}$.
Hence, $\trg{\lambda} = \trg{\lambda'}$ (which follows from  $\trg{\lambda} = \trg{\rbh{}{H}^\safeta}$ and  $\trg{\lambda'} = \trg{\rbh{}{H}^\safeta}$) and $\trg{\Omega_1} \approx \trg{\Omega_1'}$ (which follows from $\trg{H} \approx \trg{H'}$ and $\trg{\OB{B}} \approx \trg{\OB{B'}}$).

\item[E-\TR-return:]
The proof of this case is similar to that of E-\TR-retback.

\item[E-\TR-lfence:]
Then, $\trg{\Omega_0} \isdef \trg{C, H, \OB{B} \triangleright \lfence}$, $\trg{\lambda} = \trg{\epsilon}$, and $\trg{\Omega_1} \isdef \trg{C, H, \OB{B} \triangleright \skipt}$.
From $\trg{\Omega_0} \approx \trg{\Omega_0'}$, we get that  $\trg{\Omega_0'} \isdef \trg{C, H', \OB{B}' \triangleright \lfence}$ where $\trg{H'} \approx \trg{H}$ and $\trg{\OB{B}} \approx \trg{\OB{B}'}$.
Hence, we can apply the E-\TR-lfence rule to derive $\trg{\Omega_0' \xtot{\lambda'} \Omega_1'}$ where $\trg{\lambda'} =\trg{\epsilon}$ and $\trg{\Omega_1'} \isdef \trg{C, H', \OB{B}' \triangleright \skipt}$.
Hence, $\trg{\Omega_1} \approx \trg{\Omega_1'}$ and $\trg{\lambda} = \trg{\lambda'}$ hold.
\item[E-\TR-cmove-true:]
Then, $\trg{\Omega_0} \isdef \trg{C, H, \OB{B}\cdot B \triangleright \cmove{x}{e_0}{e_1}{s}}$, $\trg{x}\in\dom{\trg{B}}$, $\trg{\lambda} = \trg{\epsilon}$,  $\trg{\Omega_1} \isdef \trg{C, H, \OB{B}\cdot B\cup x\mapsto v_0 : \taintt \triangleright s}$, $\trg{B \triangleright e_0\bigred v_0 : \taintt_0}$, $\trg{B \triangleright e_1\bigred 0 : \taintt_1}$, and $\trg{\taintt} = \trg{\taintt_0} \sqcup \trg{\taintt_1}$.
From $\trg{\Omega_0} \approx \trg{\Omega_0'}$, we have that $\trg{\Omega_0'} \isdef \trg{C, H', \OB{B'}\cdot B' \triangleright \cmove{x}{e_0}{e_1}{s}}$ where $\trg{H \approx H'}$, $\trg{\OB{B}\cdot B} \approx \trg{\OB{B}'\cdot B'}$, and $\trg{B} \approx \trg{B'}$.
From $\trg{B} \approx \trg{B'}$, $\trg{B \triangleright e_0\bigred v_0 : \taintt_0}$, $\trg{B \triangleright e_1\bigred 0 : \taintt_1}$, and \cref{lemma:low-equiv-bindings-low-equiv-results}, we get that $\trg{B' \triangleright e_0\bigred v_0' : \taintt_0'}$, $\trg{B' \triangleright e_1\bigred v_1' : \taintt_1'}$ where $\trg{v_0:\taintt_0} \approx \trg{v_0' : \taintt_0'}$ and $\trg{0:\taintt_1} \approx \trg{v_1' : \taintt_1'}$.
There are two cases:
\begin{description}
\item[$\trg{\taintt_1} = \trg{\safeta}$:]
Then, $\trg{\taintt_1'} = \trg{\safeta}$ and $\trg{v_1'} = \trg{0}$.
Hence, we can apply the E-\TR-cmove-true rule to derive $\trg{\Omega_0' \xtot{\lambda'} \Omega_1'}$ where $\trg{\Omega_1} \isdef \trg{C, H', \OB{B'}\cdot B'\cup x\mapsto v_0' : \taintt' \triangleright s}$, and $\trg{\taintt'} = \trg{\taintt_0'} \sqcup \trg{\taintt_1'}$.
Observe that $\trg{\Omega_1} \approx \trg{\Omega_1'}$ immediately follows from $\trg{\Omega_0} \approx \trg{\Omega_0'}$ and $\trg{v_0:\taintt_0} \approx \trg{v_0' : \taintt_0'}$.
Hence, $\trg{\Omega_1} \approx \trg{\Omega_1'}$ and $\trg{\lambda} = \trg{\lambda'}$ hold.

\item[$\trg{\taintt_1} = \trg{\unta}$:]
Then, $\trg{\taintt_1'} = \trg{\unta}$ holds as well.
There are two cases:
\begin{description}
\item[$\trg{v_1'} = \trg{0}$:]
Then, we can apply the E-\TR-cmove-true rule to derive $\trg{\Omega_0' \xtot{\lambda'} \Omega_1'}$ where $\trg{\Omega_1} \isdef \trg{C, H', \OB{B'}\cdot B'\cup x\mapsto v_0' : \taintt' \triangleright s}$, and $\trg{\taintt'} = \trg{\taintt_0'} \sqcup \trg{\taintt_1'}$.
Observe that $\trg{\Omega_1} \approx \trg{\Omega_1'}$ immediately follows from $\trg{\Omega_0} \approx \trg{\Omega_0'}$ and $\trg{\taintt} = \trg{\taintt'} = \trg{\unta}$.
Hence, $\trg{\Omega_1} \approx \trg{\Omega_1'}$ and $\trg{\lambda} = \trg{\lambda'}$ hold.

\item[$\trg{v_1'} > \trg{0}$:]
Then, we can apply the E-\TR-cmove-false rule to derive $\trg{\Omega_0' \xtot{\lambda'} \Omega_1'}$ where $\trg{\Omega_1} \isdef \trg{C, H', \OB{B'}\cdot B'\cup x\mapsto v' : \taintt' \triangleright s}$, $\trg{B(x)} = \trg{v' : \taintt_0'}$, and  $\trg{\taintt'} = \trg{\taintt_0 \sqcup \taintt_1'}$.
Observe that $\trg{\Omega_1} \approx \trg{\Omega_1'}$ immediately follows from $\trg{\Omega_0} \approx \trg{\Omega_0'}$ and $\trg{\taintt} = \trg{\taintt'} = \trg{\unta}$.
Hence, $\trg{\Omega_1} \approx \trg{\Omega_1'}$ and $\trg{\lambda} = \trg{\lambda'}$ hold.
\end{description}
\end{description}

\item[E-\TR-cmove-false:]
The proof of this case is similar to that of the E-\TR-cmove-true case.
\end{description}

\item[Induction step:]
Then, $\trg{\Omega_0 \xtot{\lambda} \Omega_1}$ using the E-\TR-step rule.
Then, $\trg{\Omega_0} \isdef \trg{C, H, \OB{B} \triangleright s;s''}$, $\trg{C,}$ $\trg{H,}$ $\trg{\OB{B} \triangleright s} \xtot{\trg{\lambda}} \trg{C_1,}$ $\trg{H_1, \OB{B_1} \triangleright s_1}$, and $\trg{\Omega_1} = \trg{C, H_1, \OB{B_1} \triangleright s_1;s''}$.
From $\trg{\Omega_0} \approx \trg{\Omega_0'}$, we get that $\trg{\Omega_0} \isdef \trg{C, H', \OB{B'} \triangleright s;s''}$ where $\trg{H} \approx \trg{H'}$ and $\trg{\OB{B}} \approx \trg{\OB{B'}}$.
From the induction hypothesis, $\safe{\trg{\lambda}}$, and $\trg{\Omega_0} \approx \trg{\Omega_0'}$, we get that $\trg{C, H', \OB{B'} \triangleright s} \xtot{\trg{\lambda'}} \trg{C, H_1', \OB{B_1'} \triangleright s_1}$ such that $\trg{H_1} \approx \trg{H_1'}$, $\trg{\OB{B_1}} \approx \trg{\OB{B_1'}}$, and $\trg{\lambda} = \trg{\lambda'}$.
Therefore, we can apply the E-\SR-step rule to derive $\trg{\Omega_0' \xtot{\lambda'} \Omega_1'}$ where $\trg{\Omega_1} \isdef \trg{C, H_1', \OB{B_1'} \triangleright s_1;s''}$.
Observe that $\trg{\Omega_1} \approx \trg{\Omega_1'}$ and $\trg{\lambda} = \trg{\lambda'}$ hold.

\end{description}
This concludes the proof of our lemma. 
\end{proof}

\begin{lemma}[Steps of $\xtot{}$ without observations preserve low-equivalence]\label{lemma:no-obs-preserve-low-equivalence-for-non-spec-semantics}
\begin{align*}
	\forall \trg{\Omega_0}, \trg{\Omega_0'}, \trg{\lambda}, \trg{\Omega_1}.
	& \text{ if } 
		\trg{\Omega_0 \xtot{\epsilon} \Omega_1}, \trg{\Omega_0} \approx \trg{\Omega_0'} \\
	& \text{ then } \exists \trg{\Omega_1}.\  \trg{\Omega_0' \xtot{\epsilon'} \Omega_1'}, \trg{\Omega_1} \approx \trg{\Omega_1'}
\end{align*}
\end{lemma}

\begin{proof}
Special case of \cref{lemma:safe-obs-preserve-low-equivalence-for-non-spec-semantics} since $\safe{\trg{\epsilon}}$ is always satisfied.	
\end{proof}

\begin{lemma}[Steps of $\xtot{}$ that agree on observation and low-equivalence cannot disagree on label]\label{lemma:same-obs-cannot-disagree-on-labels-for-low-equivalence-confs-non-spec-semantics}
\begin{align*}
	\forall \trg{\Omega_0}, \trg{\Omega_0'}, \trg{\lambda}, \trg{\sigma}, \trg{\sigma'}, \trg{\Omega_1}, \trg{\Omega_1'}.
	& \text{ if } 
		\trg{\Omega_0 \xtot{\lambda^{\sigma}} \Omega_1},  \trg{\Omega_0' \xtot{\lambda^{\sigma'}} \Omega_1'}, \trg{\Omega_0} \approx \trg{\Omega_0'}\\
	& \text{ then } \trg{\sigma} = \trg{\sigma'}
\end{align*}
\end{lemma}

\begin{proof}
By structural induction on the rules defining $\xtot{}$.
It follows from (1) there are no two rules producing the same observation, (2) labels are always derived by computation over bindings and heaps which are low-equivalent from $\trg{\Omega_0} \approx \trg{\Omega_0'}$, and (3) computation over low-equivalent bindings produce the same labels (\cref{lemma:low-equiv-bindings-low-equiv-results1}).
\end{proof}

\subsubsection{Bindings}

\begin{lemma}[Low-equivalent bindings produce low-equivalent results - 1]
\label{lemma:low-equiv-bindings-low-equiv-results1}
\begin{align*}
	&\forall \trg{B}, \trg{B'}, \trg{e}, \trg{v}, \trg{v'}, \trg{\taintt}, \trg{\taintt'}. \\
	& \qquad \text{ if } \trg{B} \approx \trg{B'}, \trg{B \triangleright e\bigred v : \taintt}, \trg{B' \triangleright e\bigred v' : \taintt'}\\
	& \qquad \text{ then }  \trg{v:\taintt} \approx \trg{v':\taintt'} 
\end{align*}	
\end{lemma}

\begin{proof}
Let $\trg{B}$, $\trg{B'}$, $\trg{e}$, $\trg{v}$, $\trg{v'}$, $\trg{\taintt}$, and $\trg{\taintt'}$ be such that $\trg{B} \approx \trg{B'}$, $\trg{B \triangleright e\bigred v : \taintt}$, $\trg{B' \triangleright e\bigred v' : \taintt'}$.
From $\trg{B} \approx \trg{B'}$, $\trg{B \triangleright e\bigred v : \taintt}$, and \cref{lemma:low-equiv-bindings-low-equiv-results}, there are $\trg{v''}$, $\trg{\taintt''}$ such that $\trg{B' \triangleright e\bigred v'' : \taintt''}$ and $\trg{v:\taintt} \approx \trg{v'':\taintt''}$.
Since  $\trg{\bigred}$ is deterministic, we immediately get that $\trg{v'':\taintt''} = \trg{v':\taintt'}$.
Hence, $\trg{v:\taintt} \approx \trg{v':\taintt'}$.
\end{proof}

\begin{lemma}[Low-equivalent bindings produce low-equivalent results - 2]
\label{lemma:low-equiv-bindings-low-equiv-results}
\begin{align*}
	&\forall \trg{B}, \trg{B'}, \trg{e}, \trg{v}, \trg{\taintt}. \\
	& \qquad \text{ if } \trg{B} \approx \trg{B'}, \trg{B \triangleright e\bigred v : \taintt}, \\
	& \qquad \text{ then } \exists \trg{v'}, \trg{\sigma'}.\ \trg{B' \triangleright e\bigred v' : \taintt'}, \trg{v:\taintt} \approx \trg{v':\taintt'} 
\end{align*}	
\end{lemma}

\begin{proof}
Let $\trg{B}, \trg{B'}, \trg{e}, \trg{v},  \trg{\taintt}$ be such that $\trg{B} \approx \trg{B'}$ and $\trg{B \triangleright e\bigred v : \taintt}$.
We show that the lemma holds by structural induction on the rule used to derive  $\trg{B \triangleright e\bigred v : \taintt}$:
\begin{description}
	\item[Base case:]
	There are two cases based on the rule used to derive $\trg{B \triangleright e\bigred v : \taintt}$:
	\begin{description}
		\item[E-\TR-val:]
		Then, $\trg{e} = \trg{v}$ and $\trg{\taintt} = \trg{\safeta}$.
		Hence, we can derive $\trg{B' \triangleright e\bigred v' : \taintt'}$ using the  E-\TR-val rule, by picking $\trg{v'} = \trg{v}$ and  $\trg{\taintt'} = \trg{\safeta'}$.
		Therefore, $\trg{v:\taintt} \approx \trg{v':\taintt'}$.

		\item[E-\TR-var:]
		Then, $\trg{e} = \trg{x}$ and $\trg{B}(\trg{x}) = \trg{v:\taintt}$.
		From $\trg{B}(\trg{x}) = \trg{v:\taintt}$, we get that $\vdash \trg{B(x)} : \mathit{def}$.
		From $\trg{B} \approx \trg{B'}$, $\vdash \trg{B(x)} : \mathit{def}$, and $\trg{B}(\trg{x}) = \trg{v:\taintt}$, we get $\trg{B'}(\trg{x}) = \trg{v':\taintt'}$ and $\trg{v:\taintt} \approx \trg{v':\taintt'}$.
		Hence, we can derive  $\trg{B' \triangleright e\bigred v' : \taintt'}$ using the  E-\TR-var rule and $\trg{v:\taintt} \approx \trg{v':\taintt'}$.
	\end{description}
	
	\item[Induction step:]
	There are two cases based on the rule used to derive $\trg{B \triangleright e\bigred v : \taintt}$:
	\begin{description}
		\item[E-\TR-op:]
		Then, $\trg{e} = \trg{e_1 \op e_2}$, $\trg{B \triangleright e_1\bigredt n_1 : \sigma_1}$, $\trg{B \triangleright e_2\bigredt n_2 : \sigma_2}$, $\trg{v} = [\trg{n_1\op n_2}]$, and $\trg{\sigma} = \trg{\sigma_1}\sqcup\trg{\sigma_2}$.
		From  $\trg{B} \approx \trg{B'}$, $\trg{B \triangleright e_1\bigredt n_1 : \sigma_1}$, $\trg{B \triangleright e_2\bigredt n_2 : \sigma_2}$, and the induction hypothesis, we get that there are $\trg{n_1'}$, $\trg{n_2}'$, $\trg{\sigma_1}'$, $\trg{\sigma_2'}$ such that $\trg{B' \triangleright e_1\bigredt n_1' : \sigma_1'}$, $\trg{B' \triangleright e_2\bigredt n_2' : \sigma_2'}$, $\trg{n_1' : \sigma_1'} \approx \trg{n_1 : \sigma_1}$ and $\trg{n_2 : \sigma_2} \approx \trg{n_2' : \sigma_2'}$.
		Hence, we can apply the E-\TR-rule to derive $\trg{B' \triangleright e\bigred v' : \taintt'}$, where $\trg{v'} = [\trg{n_1'\op n_2'}]$, and $\trg{\sigma'} = \trg{\sigma_1'}\sqcup\trg{\sigma_2'}$.
		There are two cases:
		\begin{description}
			\item[$\trg{\sigma} = \trg{\safeta}$:]
			Then, $\trg{\sigma_1} = \trg{\safeta}$ and $\trg{\sigma_2} = \trg{\safeta}$.
			From this, $\trg{n_1' : \sigma_1'} \approx \trg{n_1 : \sigma_1}$, and $\trg{n_2 : \sigma_2} \approx \trg{n_2' : \sigma_2'}$, we get $\trg{v_1} = \trg{v_1'}$, $\trg{v_2} = \trg{v_2'}$, $\trg{\sigma_1'} = \trg{\safeta}$, and $\trg{\sigma_2'} = \trg{\safeta}$.
			Hence, $[\trg{n_1\op n_2}] = [\trg{n_1'\op n_2'}]$ and $\trg{\sigma'} = \trg{\safeta}$.
			Thus, $\trg{v:\sigma} = \trg{v' : \sigma'}$ and thus $\trg{v:\taintt} \approx \trg{v':\taintt'}$.
			
			\item[$\trg{\sigma} = \trg{\unta}$:]
			Then, $\trg{\sigma_1} = \trg{\unta} \vee \trg{\sigma_2} = \trg{\unta}$. 
			From this, $\trg{n_1' : \sigma_1'} \approx \trg{n_1 : \sigma_1}$, and $\trg{n_2 : \sigma_2} \approx \trg{n_2' : \sigma_2'}$, we also get that $\trg{\sigma_1'} = \trg{\unta} \vee \trg{\sigma_2'} = \trg{\unta}$. 
			Hence, $\trg{\sigma'} = \trg{\unta}$ as well.
			Therefore, $\trg{v:\taintt} \approx \trg{v':\taintt'}$ holds.

		\end{description}

		\item[E-\TR-comparison:]
		The proof of this case is similar to the case for the E-\TR-op rule.
	\end{description}
\end{description}
This completes the proof.
\end{proof}

\subsubsection{Non-speculative projection $\nspecProject{\cdot}$}

\begin{lemma}[Properties of non-speculative projection $\nspecProject{\cdot}$]\label{lemma:non-speculative-projection}
	\begin{align*}
		\forall \trg{P}, \trg{\Sigma}, \trg{\Sigma'}, \trg{\lambda}, \trg{\lambda'}, \trg{n}.\
			& \text{if } \trg{({0},\SInit{P}) \Xtot{\lambda} ({n},\Sigma)}, \trg{\Sigma \xltot{\lambda'} \Sigma'},
				\safe{\trg{\Sigma'}}\\
			& \text{then } \trg{\nspecProject{(\lambda \cdot \lambda')}} = \trg{\nspecProject{\lambda} \cdot \lambda'} \\
			& \text{if } \trg{({0},\SInit{P}) \Xtot{\lambda} ({n},\Sigma)}, \trg{\Sigma \xltot{\lambda'} \Sigma'},
				\unsafe{\trg{\Sigma'}}\\
			& \text{then } \trg{\nspecProject{(\lambda \cdot \lambda')}} = \trg{\nspecProject{\lambda} }
	\end{align*}
\end{lemma}

\begin{proof}
The lemma follows by inspection of the semantics and of the non-speculative projection.	
\end{proof}

\subsubsection{Properties of components}

\begin{lemma}[Only safe values in components]
\label{lemma:trans-semantics-in-components-only-safe-values}
\begin{align*}
\forall \trg{P}, \trg{\Sigma}, \trg{\lambda}, \trg{n}.
	& \text{if } \trg{(0,\SInit{P}) \Xtot{\lambda} (n,\Sigma)}, \trg{\Sigma} \isdef {\trg{(w, \OB{(\Omega,m,\taintt)}\cdot (\OB{F},\OB{I},
	H, \OB{B} \cdot B\triangleright \proc{s}{\OB{f}\cdot f} , n, \taintt')) }}, \trg{f} \in \trg{\OB{I}},\\
	& \text{then } \forall \trg{x}, \trg{v}, \trg{\sigma}.\ \trg{B}(\trg{x})  = \trg{v: \sigma} \rightarrow \trg{\sigma} = \trg{\safeta} 
\end{align*}
\end{lemma}

\begin{proof}
By induction on $\trg{n}$ combined with (1)  contexts can only write and read information 	from the public heap which is always labelled $\trg{\safeta}$, and (2) E-\TR-call and E-\TR-callback rules tag the variable $\trg{x}$ with $\trg{\safeta}$.
\end{proof}

\subsubsection{Weak Variants}

\begin{theorem}[\ssdef{} implies \snidef (weak)]\label{thm:ss-impl-sni-weak}
\begin{align*}
	\forall \trg{P} \in \weak{\TR}.
	\text{ if } \vdash \trg{P} : \ss(\weak{\TR}),  
	\text{ then } \vdash \trg{P} : \sni(\weak{\TR})
\end{align*}
\end{theorem}
\begin{proof}
The proof of this result is similar to the one of \Thmref{thm:ss-impl-sni}.
The key differences is that all data loaded non-speculatively is tagged as $\trg{\safeta}$ (rather than $\trg{\unta}$).
Therefore, one has to show that executing a non-speculative memory load preserves indistinguishability or results in different non-speculative observations.
This immediately follows from the fact that actions generated by non-speculative memory loads disclose both the memory address and the loaded value (so either we have different non-speculative observations or we load the same value from memory and therefore we preserve indistinguishability). 
\end{proof}

\subsection{RSS implies RSNI}

\begin{corollary}[\rssdef{} implies \rsnidef{}]\label{corollary:rss-impl-rsni}
\begin{align*}
	\forall \trg{P} \in \TR.
	\text{ if } \vdash \trg{P} : \rss(\TR)
	\text{ then } \vdash \trg{P} :\rsni(\TR)
\end{align*}
\end{corollary}

\begin{proof}
Let $\trg{P}$ be an arbitrary program in $\TR$ such that $\vdash \trg{P} :\rss$ holds.
Let $\trg{A}$ be an arbitrary context. 
Then, there are two cases:
\begin{enumerate}
\item [$ \vdash \ctxt{} : \com{atk}$:]
From $\vdash \trg{P}:\rss$, it follows that $\vdash\trg{\ctxt{}\hole{P}}:\ss(\TR)$.
By applying \cref{thm:ss-impl-sni}, we have $\vdash\trg{\ctxt{}\hole{P}}:\sni(\TR)$.
Hence, $\vdash \ctxt{} : \com{atk} \Rightarrow \vdash\trg{\ctxt{}\hole{P}}:\sni(\TR)$ holds.
\item [$ \not \vdash \ctxt{} : \com{atk}$:]
Then, $\vdash \ctxt{} : \com{atk} \Rightarrow \vdash\trg{\ctxt{}\hole{P}}:\sni(\TR)$ trivially holds.
\end{enumerate}
Since $\vdash \ctxt{} : \com{atk} \Rightarrow \vdash\trg{\ctxt{}\hole{P}}:\sni(\TR)$ holds for an arbitrary context $\ctxt{}$ and program $\trg{P}$, we have that $\forall \ctxt{}. \vdash \ctxt{} : \com{atk} \Rightarrow \vdash\trg{\ctxt{}\hole{P}}:\sni(\TR)$ holds as well.
Hence, $\vdash \trg{P} :\rsni(\TR)$ holds.
This completes the proof of our corollary.
\end{proof}

\begin{corollary}[\rssdef{} implies \rsnidef{}]\label{corollary:rss-impl-rsni-weak}
\begin{align*}
	\forall \trg{P} \in \weak{\TR}.
	\text{ if } \vdash \trg{P} : \rss(\weak{\TR})
	\text{ then } \vdash \trg{P} :\rsni(\weak{\TR})
\end{align*}
\end{corollary}

\begin{proof}
The proof is similar to the one of \cref{corollary:rss-impl-rsni} except that we use \cref{thm:ss-impl-sni-weak} instead of \cref{thm:ss-impl-sni}.
\end{proof}

\subsection{SNI does not imply SS}

\begin{theorem}[\snidef{} not imply \ss]\label{thm:sni-not-impl-ss}
\begin{align*}
	\exists \trg{P} \in \TR. \vdash \trg{P} :\sni(\TR) \wedge \not\vdash \trg{P}:\ss(\TR) \\
	\exists \trg{P} \in \weak{\TR}. \vdash \trg{P} :\sni(\weak{\TR}) \wedge \not\vdash \trg{P}:\ss(\weak{\TR})
\end{align*}
\end{theorem}

\begin{proof}
Consider the following program $\trg{P}$:

	\trg{
					\begin{aligned}[c]
						&
						\ifztet{
							\trg{(y < size)}
						}{
							\\
							&\ \
							\begin{aligned}
								&
								\trg{\letreadpt{x_a}{0+y}{}}
								\\
								&
								\trg{\letreadpt{x_b}{4+x_a}{}}
								\\
								&
								\trg{\letint{temp}{x_b}{\skipt}}
							\end{aligned}
							\\
							&
						}{
							\\
							&\ \
							\begin{aligned}
								&
								\trg{\letreadpt{x_a}{0+y}{}}
								\\
								&
								\trg{\letreadpt{x_b}{4+x_a}{}}
								\\
								&
								\trg{\letint{temp}{x_b}{\skipt}}
							\end{aligned}
							\\
							&						}		
					\end{aligned}
				}

The above program clearly satisfies speculative non-interference for, e.g., 	$\trg{w} = \trg{10}$ since the program leaks the same information under the speculative and non-speculative semantics.
However, the program violates speculative safety (see the step-by-step example in  \cref{sec:example}).	
Additionally, observe that the program satisfy $\sni(\weak{\TR})$ but still violates $\ss(\weak{\TR})$. 	
\end{proof}

\subsection{RSNI does not imply RSS}

\begin{corollary}[\rsnidef{} not imply \rss]\label{corollary:rsni-not-impl-rss}
\begin{align*}
	\exists \trg{P} \in \TR. \vdash \trg{P} :\rsni(\TR) \wedge \not\vdash \trg{P}:\rss(\TR) \\
		\exists \trg{P} \in \weak{\TR}. \vdash \trg{P} :\rsni(\weak{\TR}) \wedge \not\vdash \trg{P}:\rss(\weak{\TR})
\end{align*}
\end{corollary}

\begin{proof}
The counter-example provided in \cref{thm:sni-not-impl-ss} can be directly ported to the robust setting (by moving the code in the component and having a context simply calling the component).	
\end{proof}

\newpage
\section{Compiler Criteria and their Implications}\label{sec:compcrit}
\subsection{Strong Criteria for Secure Compilers}
\begin{definition}[Robust Speculative Safety-Preserving Compiler (\rdsspdef)]\label{def:rdssp}
	\begin{align*}
		\vdash \comp{\cdot} : \strong{\rdssp} \isdef&\ 
		\forall\src{P}\ldotp 
		\text{ if } \vdash\src{P} : \rss(\SR)
		\text{ then } \vdash\comp{\src{P}} : \rss(\TR)
	\end{align*}
\end{definition}
This gets expanded to:
\begin{align*}
	\forall\src{P}\ldotp 
	\text{ if }&\ \forall \ctxs{}\ldotp \forall \tras{^\sigma}\in\behavs{\ctxs{}\hole{P}}\ldotp \forall \acas{^\sigma}\in\tras{^\sigma}\ldotp \src{\sigma}\equiv\src{\safeta}
	\\
	\text{ then }&\ \forall \ctxt{}\ldotp \forall \trat{^{\taintt}}\in\behavt{\ctxt{}\hole{\comp{\src{P}}}}\ldotp \forall \acat{^{\taintt}}\in\trat{^{\taintt}}\ldotp \taintt\equiv\trg{\safeta}
\end{align*}

We say that a source and a target trace are related ($\rels$) if the latter contains the source trace plus interleavings of only safe actions.
The trace relation relies on a relation on actions which in turn relies on a relation on values and heaps.

The last two are compiler-dependent, so they are presented later, for lfence in \Cref{sec:sim-code-fence} and for slh in \Cref{sec:rels-slh}.

\mytoprule{\text{Trace relation} \reldef }
\begin{center}
	\typerule{Trace-Relation}{
	}{
		\srce\rels\trge
	}{tr-rel-empty}
	\typerule{Trace-Relation-Same-Act}{
		\src{\tras{^\sigma} } \rels \trg{\trat{^{\taintt}} }	
		&
		\src{\alpha^\sigma} \arel \trg{ \acat{^{\taintt}} }
	}{
		\src{\tras{^\sigma} \cdot \alpha^\sigma} \rels \trg{\trat{^{\taintt}} \cdot \acat{^{\taintt}} }
	}{tr-rel-same}
	\typerule{Trace-Relation-Same-Heap}{
		\src{\tras{^\sigma} } \rels \trg{\trat{^{\taintt}} }	
		&
		\src{\delta^\sigma} \arel \trg{ \trgb{\delta}^{\taintt} }
	}{
		\src{\tras{^\sigma} \cdot \delta^\sigma} \rels \trg{\trat{^{\taintt}} \cdot \trgb{\delta}^{\taintt} }
	}{tr-rel-same-h}
	\typerule{Trace-Relation-Safe-Act}{
		\src{\tras{^\sigma} } \rels \trg{\trat{^{\taintt}} }
		&
		\src{\epsilon} \arel \acat{^{\safeta}}
	}{
		\src{\tras{^\sigma}} \rels \trg{\trat{^{\taintt}} \cdot \acat{^{\safeta}} }
	}{tr-rel-safe-a}
	\typerule{Trace-Relation-Safe-Heap}{
		\src{\tras{^\sigma} } \rels \trg{\trat{^{\taintt}} }	
		&
		\src{\epsilon} \arel \trg{\trgb{\delta}^{\safeta}}
	}{
		\src{\tras{^\sigma} } \rels \trg{\trat{^{\taintt}} \cdot \trgb{\delta}^{\safeta} }
	}{tr-rel-safe-h}
	\typerule{Trace-Relation-Rollback}{
		\src{\tras{^\sigma} } \rels \trg{\trat{^{\taintt}} }	
		&
		\src{\epsilon}\arel\trgb{\rollbl}^{\taintt}
	}{
		\src{\tras{^\sigma} } \rels \trg{\trat{^{\taintt}} \cdot \trgb{\rollbl}^{\taintt} }
	}{tr-rel-rollb}

\mytoprule{\text{Action relation} \areldef }

	\typerule{Action Relation - call}{
		\src{f}\equiv\trg{f}
		&
		\src{v}\vrel\trg{v}
		&
		\src{\sigma}\equiv\taintt
	}{
		\src{\clh{f}{v}{H}^\sigma} \arel \trg{\clh{f}{v}{H}^{\taintt}}
	}{ac-rel-cl}
	\typerule{Action Relation - return}{
	}{
		\src{\rth{}{H}^\safeta} \arel \trg{\rth{}{H}^\safeta}
	}{ac-rel-rt}
	\typerule{Action Relation - callback}{
		\src{f}\equiv\trg{f}
		&
		\src{v}\vrel\trg{v}
		&
		\src{\sigma}\equiv\taintt
	}{
		\src{\cbh{f}{v}{H}^\sigma} \arel \trg{\cbh{f}{v}{H}^{\taintt}}
	}{ac-rel-cb}
	\typerule{Action Relation - returnback}{
	}{
		\src{\rbh{}{H}^\safeta} \arel \trg{\rbh{}{H}^\safeta}
	}{ac-rel-rb}
	\typerule{Action Relation - read}{
		\src{n} \vrel \trg{n}
		&
		\src{\sigma}\equiv\taintt
	}{
		\src{\rdl{n}^\sigma} \arel \trg{\rdl{n}^{\taintt}}
	}{ac-rel-rd}
	\typerule{Action Relation - write}{
		\src{n} \vrel \trg{n}
		&
		\src{\sigma}\equiv\taintt
	}{
		\src{\wrl{n}^\sigma} \arel \trg{\wrl{n}^{\taintt}}
	}{ac-rel-wr}
	\typerule{Action Relation - write 2}{
		\src{n} \vrel \trg{n}
		&
		\src{v} \vrel \trg{v}
		&
		\src{\sigma}\equiv\taintt
	}{
		\src{\wrl{n\mapsto v}^\sigma} \arel \trg{\wrl{n\mapsto v}^{\taintt}}
	}{ac-rel-wr2}
	\typerule{Action Relation - if}{
		\src{n} \vrel \trg{n}
		&
		\src{\sigma}\equiv\taintt
	}{
		\src{\ifl{n}^\sigma} \arel \trg{\ifl{n}^{\taintt}}
	}{ac-rel-if}
	\typerule{Action Relation - epsi alpha}{
		\taintt\equiv \trg{\safeta}
	}{
		\src{\epsilon} \arel \acat{^{\taintt}}
	}{ac-rel-ep-al}
	\typerule{Action Relation - epsi heap}{
		\taintt\equiv \trg{\safeta}
	}{
		\src{\epsilon} \arel \trgb{\delta}^{\taintt}
	}{ac-rel-ep-hp}
	\typerule{Action Relation - rlb}{
		\taintt\equiv \trg{\safeta}
	}{
		\src{\epsilon}\arel\trgb{\rollbl}^{\taintt}
	}{ac-rel-rlb}
\end{center}
\botrule

\begin{definition}[Robust Speculative Safety Compilation (\rdssdef)]\label{def:rdss}
	\begin{align*}
		\vdash \comp{\cdot} : \strong{\rdss} \isdef&\
			\forall\src{P}, \ctxt{}, \trat{^{\taintt}},
			\exists \ctxs{}, \tras{^\sigma} \ldotp 
			\\
			\text{ if }&\
			\trg{\ctxt{}\hole{\comp{\src{P}}}} \semt \trat{^{\taintt}}
			\text{ then }
			\src{\ctxs{}\hole{P}} \sems \tras{^\sigma}
			\text{ and }
			\tras{^\sigma}\rels\trat{^{\taintt}}
	\end{align*}
\end{definition}

\begin{theorem}[\rdss implies \rdssp]\label{thm:rdss-impl-rdsp}(\showproof{rdss-impl-rdsp})
With the relation \relref:
\begin{align*}
	\forall \comp{\cdot}
	\text{ if } \vdash\comp{\cdot} : \strong{\rdss}
	\text{ then } \vdash\comp{\cdot} : \strong{\rdssp}
\end{align*}
\end{theorem}

\begin{theorem}[\rdssp implies \rdss]\label{thm:rdssp-impl-rdss}(\showproof{rdssp-impl-rdss})
With the relation \relref:
\begin{align*}
	\forall \comp{\cdot}
	\text{ if } \vdash\comp{\cdot} :\strong{\rdssp}
	\text{ then } \vdash\comp{\cdot} : \strong{\rdss}
\end{align*}
\end{theorem}

\begin{theorem}[\rdss and \rdssp are equivalent]\label{thm:rdss-eq-rdsp}(\showproof{rdss-eq-rdsp})
With the relation \relref:
\begin{align*}
	\forall \comp{\cdot}
	\vdash\comp{\cdot} : \strong{\rdss}
	\iff
	\vdash\comp{\cdot} : \strong{\rdssp}
\end{align*}
\end{theorem}

\subsection{Strong Criteria for Insecure Compilers}
\begin{definition}[\snipdef]\label{def:rsnip}
	\begin{align*}
		\vdash \comp{\cdot} : \rsnip \isdef&\ 
		\forall\src{P}\ldotp 
		\text{ if } \src{\vdash\src{P} : \rsni(\SR)}
		\text{ then } \trg{\vdasht\comp{\src{P}} : \rsni(\TR)}
	\end{align*}
\end{definition}
\begin{corollary}[\com{\nvdash \comp{\cdot}:\rsnip}]\label{def:not-rsnip}
	\begin{align*}
		\nvdash \comp{\cdot} : \rsnip \isdef&\ 
		\exists \src{P}\ldotp 
		\text{ } \src{\vdash\src{P} : \rsni(\SR)}
		\text{ and } \trg{\trgb{\nvdash}\comp{\src{P}} : \rsni(\TR)}
	\end{align*}
\end{corollary}

\subsection{Weak Criteria for Secure Compilers}

\begin{definition}[Robust Speculative Safety-Preserving Compiler (\weak{\rdssp})]\label{def:rdssp-weak}
	\begin{align*}
		\vdash \comp{\cdot} : \weak{\rdssp} \isdef&\ 
		\forall\src{P}\ldotp 
		\text{ if } \vdash\src{P} : \rss(\weak{\SR})
		\text{ then } \vdash\comp{\src{P}} : \rss(\weak{\TR})
	\end{align*}
\end{definition}

\begin{definition}[Robust Speculative Safety Compilation (\weak{\rdss})]\label{def:rdss-weak}
	\begin{align*}
		\vdash \comp{\cdot} : \weak{\rdss} \isdef&\
			\forall\src{P}, \ctxt{}, \trat{^{\taintt}},
			\exists \ctxs{}, \tras{^\sigma} \ldotp 
			\\
			\text{ if }&\
			\trg{\ctxt{}\hole{\comp{\src{P}}}} \semt \trat{^{\taintt}}
			\text{ then }
			\src{\ctxs{}\hole{P}} \sems \tras{^\sigma}
			\text{ and }
			\tras{^\sigma}\rels\trat{^{\taintt}}
	\end{align*}
\end{definition}

\begin{theorem}[\weak{\rdss} implies \weak{\rdssp}]\label{thm:rdss-impl-rdsp-weak}(\showproof{rdss-impl-rdsp-weak})
With the relation \relref:
\begin{align*}
	\forall \comp{\cdot}
	\text{ if } \vdash\comp{\cdot} : \weak{\rdss}
	\text{ then } \vdash\comp{\cdot} : \weak{\rdssp}
\end{align*}
\end{theorem}

\begin{theorem}[\weak{\rdssp} implies \weak{\rdss}]\label{thm:rdssp-impl-rdss-weak}(\showproof{rdssp-impl-rdss-weak})
With the relation \relref:
\begin{align*}
	\forall \comp{\cdot}
	\text{ if } \vdash\comp{\cdot} :\weak{\rdssp}
	\text{ then } \vdash\comp{\cdot} : \weak{\rdss}
\end{align*}
\end{theorem}

\begin{theorem}[\weak{\rdss} and \weak{\rdssp} are equivalent]\label{thm:rdss-eq-rdsp-weak}(\showproof{rdss-eq-rdsp-weak})
With the relation \relref:
\begin{align*}
	\forall \comp{\cdot}
	\vdash\comp{\cdot} : \weak{\rdss}
	\iff
	\vdash\comp{\cdot} : \weak{\rdssp}
\end{align*}
\end{theorem}

\subsection{Weak Criteria for Insecure Compilers}
\begin{definition}[\weak{\snipdef}]\label{def:rsnip-weak}
	\begin{align*}
		\vdash \comp{\cdot} : \weak{\rsnip} \isdef&\ 
		\forall\src{P}\ldotp 
		\text{ if } \src{\vdash\src{P} : \rsni(\weak{\SR})}
		\text{ then } \trg{\vdasht\comp{\src{P}} : \rsni(\weak{\TR})}
	\end{align*}
\end{definition}
\begin{corollary}[\com{\nvdash \comp{\cdot}:\weak{\rsnip}}]\label{def:not-rsnip-weak}
	\begin{align*}
		\nvdash \comp{\cdot} : \weak{\rsnip} \isdef&\ 
		\exists \src{P}\ldotp 
		\text{ } \src{\vdash\src{P} : \rsni(\weak{\SR})}
		\text{ and } \trg{\trgb{\nvdash}\comp{\src{P}} : \rsni(\weak{\TR})}
	\end{align*}
\end{corollary} 
\newpage
\section{Compiler Insecurity Results}\label{sec:compinsec}

\subsection{Unsafe SLH }\label{sec:form-slh-not-sec-general}
In the following, assume that the compilation of \lstinline{A} (i.e., \trg{n_a} or \trg{n_a-1}) contains a value \trg{v_a} and its low-equivalent counterpart contains \trg{v_a'}.
As before, assume \trg{size} is \trg{4} and \trg{y} is \trg{8}.
We indicate the two traces for the two low-equivalent states as \trg{t} and \trg{t'} respectively and highlight in \hl{yellow} where they differ.
Note that these traces contain more heap actions, specifically those required to read and write the predicate bit when it is stored on the heap (location \trg{-1}).

\begin{theorem}[This SLH compiler is not \strong{\rsnip}]\label{thm:slh-comp-not-strong-rsnip}
	\begin{align*}
		\nvdash \compslh{\cdot} : \strong{\rsnip}
	\end{align*}
\end{theorem}
\begin{proof}

The attacker is the same:
\begin{align*}
	\ctxt{^{\!8}} \isdef&\ \trg{main(x)\mapsto \call{get}~8; \ret}
\end{align*}

Below are the two different target traces for the code of \Cref{sec:code2}.

\vspace{-1em}
{\small
\begin{align*}
	\trg{t'} =&\ 
		\trg{ 
			\clh{get}{8}{}^\safeta \cdot 
			\rdl{1}^\safeta \cdot
			\rdl{-(n_a-1+8+1)}^\safeta \cdot
		}
		\\&\ \ 
		\trg{
			\rdl{-1}^\safeta \cdot
			\ifl{1}^\safeta \cdot
			\rdl{-1}^\safeta \cdot
			\wrl{-1}^\safeta \cdot
		}
		\\&\ \ 
		\trg{
			\hl{\trg{\rdl{n_b + v_a}^\unta}} \cdot
			\rollbl^\safeta \cdot
			\rdl{-1}^\safeta \cdot
			\wrl{-1}^\safeta \cdot
			\rth{}{}^\safeta
		}
	\\
	\trg{t'} =&\  
		\trg{ 
			\clh{get}{8}{}^\safeta \cdot 
			\rdl{1}^\safeta \cdot
			\rdl{-(n_a-1+8+1)}^\safeta \cdot
		}
		\\&\ \ 
		\trg{
			\rdl{-1}^\safeta \cdot
			\ifl{1}^\safeta \cdot
			\rdl{-1}^\safeta \cdot
			\wrl{-1}^\safeta \cdot
		}
		\\&\ \ 
		\trg{
			\hl{\trg{\rdl{n_b + v_a'}^\unta}} \cdot
			\rollbl^\safeta \cdot
			\rdl{-1}^\safeta \cdot
			\wrl{-1}^\safeta \cdot
			\rth{}{}^\safeta
		}
	\\
	\nspecProject{\trg{t}} =&\ \nspecProject{\trg{t'}} =
		\trg{ 
			\clh{get}{8}{}^\safeta \cdot 
			\rdl{1}^\safeta \cdot
			\rdl{-(n_a-1+8+1)}^\safeta \cdot
		}
		\\&\ \ 
		\trg{
			\ifl{1}^\safeta \cdot
			\rdl{-1}^\safeta \cdot
			\wrl{-1}^\safeta \cdot
			\rth{}{}^\safeta
		}
\end{align*}
}
	
\end{proof}

\subsection{Unsafe Inter-procedural SLH}
This SLH compiler does not pass the pr state across procedures and stores it in a local variable.
\begin{align*}
	\compslht{ {f(x)\mapsto s;\ret} } &= 
		\trg{f(x)\mapsto  
			\begin{aligned}[t]
				&
				\letint{\trg{x_{pr}}}{\falset}{\compslht{s}};
				\trg{\ret}
			\end{aligned}
		}
	\\
	\compslht{ \ifzte{e}{s}{s'} } &= \trg{ 
		\begin{aligned}[t]
			&
			\letint{\trg{x_g}}{\compslh{e}}{
			\\&\
				\ifztet{\trg{x_g}
					}{ 
					\letint{\trg{x_{pr}}}{\trg{x_{pr} \vee \neg x_g}}{\compslht{s}}
				\\
				&\
				}{
					\letint{\trg{x_{pr}}}{\trg{x_{pr} \vee x_g}}{\compslht{s'}}
				}
			}
		\end{aligned}
	} \\
	\compslht{\letreadp{x}{e}{s}} &= 
		\trg{
			\letreadpt{\trg{x}}{\trg{e}}{ 
			\cmovet{\trg{x}}{\trg{0}}{\trg{x_{pr}}}{\compslht{s}}
			}			
		} 
\end{align*}

In order to prove \rdss for this compiler, we need a strong relation between states that instead of asserting that \trg{H(-1)} keeps a bool of the speculation, each state has the first binding for a variable which captures speculation.

When proving \Thmref{thm:spec-most-omega}, in the case of a call from a context to compiled component, we are not able to instate this invariant.
So, there, we need to add \trg{\lfence}, so that we stop speculation altogether when jumping into compiled code.
This is noted in \showproof{spec-most-omega}.

Crucially, this compiler is not \rsnip.
\begin{theorem}[\compslht{\cdot} is not \rsnip]\label{thm:compslht-not-rsni}
	\begin{align*}
		\nvdash \compslht{\cdot} : \rsnip
	\end{align*}
\end{theorem}
\begin{proof}
\begin{align*}
	&
	\src{get(y)\mapsto}
	\letreads{
		\src{size}
	}{
		\src{1}
	}{
		\letreadps{
			\src{x}
		}{
			\src{n_a+y}
		}{
			\ifztes{
				\src{y<size}
			\\
			&\quad
			}{
				\src{\call{get2}~ x}
			\quad
			\quad
			}{
				\skips
			}
		}
	}
	\\
	&
	\src{get2(x)\mapsto}
	\
	\src{
		\letreads{
			{temp}
		}{
			{n_B + x} 
		}{
			\skips
		}
	}
\\
	&
	\trg{get(y) \mapsto}
	\\
	&\enskip
		\letint{
			\trg{\predState}
		}{
			\trg{1}
		}{
			\letreadt{
				\trg{size}
			}{
				\trg{1}
			}{
				\letreadpt{
					\trg{x}
				}{
					\trg{n_a + y}
				}{
				\\
				&\enskip
					\letint{
						\trg{x_g}
					}{
						\trg{y < size}
					}{
						\ifztet{
							\trg{x_g}
						\\
						&\quad
						}{
							\letint{
								\trg{\predState}
							}{
								\trg{\predState \vee \neg x_g }
							}{
								\trg{\call{get2}~x}
							}
						\\
						&\quad
						}{
							\letint{
								\trg{\predState}
							}{
								\trg{\predState \vee x_g }
							}{
								\skipt
							}
						}
					}
				}
			}
		}
	\\
	&
	\trg{get2(x) \mapsto}
	\trg{
		\letint{
			\trg{\predState}
		}{
			\trg{1}
		}{
			\letreadt{
				\trg{temp}
			}{
				\trg{n+b + x}
			}{
				\skipt
			}
		}
	}
\end{align*}

\begin{align*}
	\trg{t'} =&\ 
		\trg{ 
			\clh{get}{8}{}^\safeta \cdot 
			\rdl{1}^\safeta \cdot
			\rdl{-(n_a+42)}^\safeta \cdot
		}
		\\&\ \ 
		\trg{
			\ifl{1}^\safeta \cdot
			\hl{\trg{\rdl{n_b + v_a}^\unta}} \cdot
			\rollbl^\safeta \cdot
			\rth{}{}^\safeta
		}
	\\
	\trg{t'} =&\  
		\trg{ 
			\clh{get}{8}{}^\safeta \cdot 
			\rdl{1}^\safeta \cdot
			\rdl{-(n_a+42)}^\safeta \cdot
		}
		\\&\ \ 
		\trg{
			\ifl{1}^\safeta \cdot
			\hl{\trg{\rdl{n_b + v_a'}^\unta}} \cdot
			\rollbl^\safeta \cdot
			\rth{}{}^\safeta
		}
	\\
	\nspecProject{\trg{t}} =&\ \nspecProject{\trg{t'}} =
		\trg{ 
			\clh{get}{8}{}^\safeta \cdot 
			\rdl{1}^\safeta \cdot
			\rdl{-(n_a+42)}^\safeta \cdot
		}
		\\&\ \ 
		\trg{
			\ifl{1}^\safeta \cdot
			\rth{}{}^\safeta
		}
\end{align*}	
\end{proof}

\subsection{MSVC Details}\label{sec:msvc-details}
\lstset{language=Asm}
Unlike ICC, MSVC tries to reduce the number of \lstinline{lfence}s by selectively determining which branches to patch. 
\Cref{example:msvc:secure} below illustrates how MSVC works on the standard Spectre v1 snippet while \Cref{example:mvcc:insecure} (and as pointed out in~\cite{kocher2018examples,spectector}) illustrates that MSVC sometimes omits necessary \lstinline{lfence}s.

\begin{example}[MSVC in action]
\label{example:msvc:secure}
\Cref{lis:spectre-1-listing-compiled-msr} presents the (simplified) assembly produced by MSVC on the code of \Cref{lis:spectrev1}.
\begin{lstlisting}[label=lis:spectre-1-listing-compiled-msr,caption={\Cref{lis:spectrev1} compiled with MSVC with \texttt{/Qspectre} flag enabled.}] 
	mov     rax, size // load size
	cmp     rcx, rax // compare y (in rcx) and size 
	jae     END // jump if y is out-of-bound
	lfence // halt speculative execution
	movzx   eax, A[rcx] // load A[y]
	movzx   eax, B[rax] // load B[A[y]]
	mov     temp, al  // assignment to temp
END:
	ret     0
\end{lstlisting} 
In this case, MSVC correctly inserts the \lstinline{lfence} (line 4) just after the branch instruction that checks whether \lstinline{x} (stored in register \lstinline{rcx}) is in-bound.
The \lstinline{lfence} stops the mis-speculated execution of the two memory accesses and it effectively prevents speculative leaks.
\end{example}

\begin{example}[MSVC is not \rsnip]\label{example:mvcc:insecure}
\lstset{language=Java}
Consider now the code in \Cref{lis:spectre-10-listing}, which is adapted from from~\cite[Example~10]{kocher2018examples}.
In contrast to \Cref{lis:spectrev1}, this example speculatively leaks whether \lstinline{A[y]} is \lstinline{0} through the branch statement in line 3.

When compiling this code with the \lstinline{lfence}-countermeasure enabled, MSVC produces the snippet shown in  \Cref{lis:spectre-10-listing-compiled}.
\lstset{language=Asm}
\begin{lstlisting}[label=lis:spectre-10-listing-compiled,caption={\Cref{lis:spectre-10-listing} compiled with MVCC with \texttt{/Qspectre} flag enabled.}] 
		mov     rax, size // load size
		cmp     rcx, rax // compare y (in rcx) and size
		jae     END // jump if out-of-bound
		cmp     [A+rcx], 0 // compare A[y] and 0
		jne     END // jump if A[y] is not 0
		movzx   eax, B // load B[0]
		mov     temp, al // assignment to temp            
	END:
		ret     0
\end{lstlisting} 
In this case, the compiler does not insert an \lstinline{lfence} after the first branch instruction on line 3.
Therefore, the compiled program still contains a speculative leak.

In our framework, the source program from \Cref{lis:spectre-10-listing} trivially satisfies \rsni, because the source language \SR{} does not allow speculative execution.
Its compilation in \Cref{lis:spectre-10-listing-compiled}, however, violates \rsni.
To show this, consider two low-equivalent initial states $\trg{\Omega^0}$, $\trg{\Omega^1}$ where \trg{y} is out-of-bound, \lstinline{A[y]} is \trg{0} in $\trg{\Omega^0}$ and \trg{1} in $\trg{\Omega^1}$, and the value of \trg{y} is \trg{42} in both.
The corresponding traces are: 

\vspace{-1em}
{\small
\begin{align*}
	\trg{t_{\Omega^0}} =&\ 
		\trg{ \clh{get}{42}{}^\safeta \cdot \ifl{0}^\safeta \cdot \rdl{n_A+42}{}^\safeta \cdot  }
		\\
		& \quad 
		\trg{ \ifl{0}^\unta \cdot  \rollbl^\safeta \cdot \rdl{n_B+0}^\safeta  \cdot \rollbl^\safeta }
	\\
	\trg{t_{\Omega^1}} =&\ 
		\trg{ \clh{get}{42}{}^\safeta \cdot \ifl{0}^\safeta \cdot \rdl{n_A+42}{}^\safeta \cdot } 
		\\
		& \quad 
		\trg{  \ifl{1}^\unta \cdot \rdl{n_B+0}^\safeta   \cdot \rollbl^\safeta \cdot \rollbl^\safeta }
\end{align*}
}
These two traces have the same non-speculative projection 

\noindent$\trg{ \clh{get}{42}{}^\safeta \cdot \ifl{0}^\safeta}$ but they differ in the observation associated with the branch instruction from line 5 
(which is $\trg{\ifl{0}^\unta}$ in \trg{t_{\Omega^0}} and $\trg{\ifl{1}^\unta}$ in \trg{t_{\Omega^1}}).
Therefore, they are a counterexample to \rsni.
As a result, MVCC violates \rsnip{} since it does not preserve \rsni{}.
\end{example} %

\subsection{SLH Details}\label{sec:slh-details}
\begin{example}[SLH in action]\label{example:clang:slh} 
Consider again the Spectre v1 snippet from \Cref{lis:spectrev1}.
Clang with SLH enabled compiles the program into the (simplified) assembly in \Cref{lis:spectre-1-listing-compiled-clang}.
\lstset{language=Asm}
\begin{lstlisting}[label=lis:spectre-1-listing-compiled-clang,caption={Compiled version of  \Cref{lis:spectrev1} produced by Clang with \texttt{-x86-speculative-load-hardening} flag enabled.}] 
	mov     rax, rsp // load predicate bit from stack pointer
	sar     rax, 63 // initialize mask (0xF...F if left-most bit of rax is 1)
	mov     edx, size // load size
	cmp     rdx, rdi // compare size and y
	jbe     ELSE // jump if out-of-bound
THEN:
	cmovbe  rax, rcx // set mask to -1 if out-of-bound
	movzx    ecx, [A + rdi] // load A[y]
	or      rcx, rax // mask A[y]
	mov     cl, [B + rcx] // load B[mask(A[y])]
	or      cl, al // mask B[mask(A[y])]
	mov     temp, cl // assignment to temp
	jmp     END
ELSE:
	cmova   rax, -1 // set mask to -1 if in bound
END:
	shl     rax, 47
	or      rsp, rax // store predicate bit on stack pointer
	ret
\end{lstlisting} 
\lstset{language=Java}
The masking introduced by SLH is sufficient to avoid speculative leaks.
Indeed, if the processor speculates over the branch instruction in line 5 and speculatively executes the first memory access on line 7, the loaded value is masked immediately afterwards (line 8) and it is set to \trg{0xF..F}.
Thus, the second memory access (line 9) will not depend on sensitive information; thereby preventing the leak.
\end{example}
\begin{example}[SLH is not \rsnip]\label{example:clang:slh:insecure}
Consider the variant of Spectre v1 illustrated in \Cref{lis:spectre-variant-listing}.
The main difference with the standard Spectre v1 example (\Cref{lis:spectrev1}) is that the first memory access is performed non-speculatively (line 2).
Its value, however, is still leaked through the speculatively-executed memory access in line 4.
Clang with SLH compiles this code into the snippet of \Cref{lis:spectre-variant-listing-compiled-clang}.
\lstset{language=Asm}
\begin{lstlisting}[label=lis:spectre-variant-listing-compiled-clang,caption={Compiled version of  \Cref{lis:spectre-variant-listing} produced by Clang with \texttt{-x86-speculative-load-hardening} flag enabled.}] 
	mov     rax, rsp 	// load predicate bit from stack pointer
	sar     rax, 63 // initialize mask (0xF...F if left-most bit of rax is 1)
	movzx   edx, [A + rdi] 	// load A[y]
	or      edx, eax 	// mask A[y]
	mov     x, edx 	// assignment to x
	mov     esi, size 	// load size
	cmp     rsi, rdi 	// compare size and y
	jbe     ELSE 	// jump if out-of-bound
THEN:
	cmovbe  rax, -1 	// set mask to -1 if out-of-bound
	mov     cl, [B + rdx] 	// load B[x]
	or      cl, al 	// mask B[x]
	mov     temp, cl 	// assignment to temp
	jmp     END
ELSE:
	cmova   rax, -1 	// set mask to -1 if in-bound
END:
	shl     rax, 47 
	or      rsp, rax 	// store predicate bit on stack pointer
	ret
\end{lstlisting} 
In the compiled code, the value of \lstinline{A[y]} is hardened using the mask retrieved from the stack pointer (line 4).
As a result, if the \lstinline{get} function is invoked non-speculatively, then the mask is set to \trg{0x0..0} and the value of \lstinline{A[y]} is not protected.
Therefore, speculatively executing the load in line 11 may still leak the value of \lstinline{A[y]} speculatively, which will be different in traces generated from different, low-equivalent states.
\end{example}

\begin{example}[Non-interprocedural SLH is not \rsnip]\label{ex:slh-nonint-insec}
	\lstset{language=Java}
	The program of \Cref{lis:spectre-variant-listing-proc} splits the memory accesses of \lstinline{A} and \lstinline{B} of the classical snippet across functions \lstinline{get} and \lstinline{get_2}.
	\begin{lstlisting}[mathescape,label=lis:spectre-variant-listing-proc,caption={Inter-procedural variant of the Spectre v1 snippet~\cite{crossproc}.}] 
	void get (int y) 
		x = A[y] ; 
		if (y < size) then  get_2 (x);

	void get_2 (int x)  temp = B[x];
	\end{lstlisting}
	\lstset{language=Asm}
	Intuitively, once compiled, \lstinline{get} starts the speculative execution (line 3), then the compiled code corresponding to \lstinline{get_2} is executed speculatively.
	However, the predicate bit of \lstinline{get_2} is set to $\trg{0}$ upon calling the function and therefore the memory access corresponding to \lstinline{B[x]} is not masked and it leaks the value of \lstinline{x} (which is equivalent to \lstinline{A[y]}).
\end{example}

\newpage
\section{The lfence Compiler \complfence{\cdot}}\label{sec:comp-lfence} 

The lfence compiler (as implemented in Intel ICC).

The main feature is that the `then' and `else' branches of the conditionals start with an \trg{\lfence}, so no speculation is possible in the branches.
We do not add a speculation barrier at function boundaries for the same reason why we do not let the context speculate (see \Cref{src:trg-sem-com}).
Since the context speculates and since the only source of speculation is branching, we do not need to add \trg{\lfence} at function boundaries.
We would need to do so were we to model speculation on return addresses too.

\begin{align*}
	\complfence{ H ; \OB{F} ; \OB{I}} &= \trg{ \complfence{H} ; \complfence{\OB{F}} ; \complfence{\OB{I}}}
	\\\\
	\complfence{ \srce} &= \trge
	\\
	\complfence{ \OB{I}\cdot f } &= \trg{ \complfence{ \OB{I} }\cdot f }
	\\\\
	\complfence{ H ; -n\mapsto v : \unta} &= \trg{\complfence{H} ; -\complfence{n} \mapsto \complfence{v} : \unta}
	\\\\
	\complfence{ {f(x)\mapsto s;\ret} } &= \trg{f(x)\mapsto \complfence{s};\ret}
	\\
	\complfence{ {s;s'} } &= \complfence{s}\trg{;} \complfence{s'}
	\\
	\complfence{ \skips } &= \trg{\skipt}
	\\
	\complfence{ \letin{x}{e}{s} } &= \trg{ \letin{x}{\complfence{e}}{\complfence{s}}}
	\\
	\complfence{ \ifzte{e}{s}{s'} } &= \trg{ \ifzte{\complfence{e}}{ \{\lfence; \complfence{s} \}}{\{\lfence; \complfence{s'}\}}}
	\\
	\complfence{ \call{f}~e } &= \trg{ \call{f}~\complfence{e} }
	\\
	\complfence{ \asgn{e}{e'} } &= \trg{ \asgn{\complfence{e}}{\complfence{e'}} }
	\\
	\complfence{ \letread{x}{e}{s} } &= \trg{ \letread{x}{\complfence{e}}{ \complfence{s} } }
	\\
	\complfence{ \asgnp{e}{e'} } &= \trg{ \asgnp{\complfence{e}}{\complfence{e'}} }
	\\
	\complfence{ \letreadp{x}{e}{s} } &= \trg{ \letreadp{x}{\complfence{e}}{ \complfence{s} } }
	\\
	\\
	\complfence{ n } &= \trg{n}
	\\
	\complfence{ e \op e' } &= \trg{ \complfence{ e } \trgb{\op} \complfence{ e' }}
	\\
	\complfence{ e \bop e' } &= \trg{ \complfence{ e } \trgb{\bop} \complfence{ e' }}
\end{align*}

\begin{theorem}[The lfence compiler is \strong{\rdss}]\label{thm:lfence-comp-rdss}(\showproof{lfence-comp-rdss})
	\begin{align*}
		\vdash \complfence{\cdot} : \strong{\rdss}
	\end{align*}
\end{theorem}

\begin{theorem}[All lfence-compiled programs are \rss(\SR)]\label{thm:all-lfence-comp-are-rdss}(\showproof{19})
	\begin{align*}
		\forall\src{P}\ldotp \vdash\complfence{P} : \rss(\SR)
	\end{align*}
\end{theorem}

\subsection{Backtranslation}\label{sec:bt-fence}
We need a backtranslation for the proof.
In this case, given that the languages are so close, we build both a context-based backtranslation (\Cref{sec:ctx-bt-fence}) and a trace-based backtranslation (analogous to the one of the SLH compiler).

\subsubsection{Context-based Backtranslation}\label{sec:ctx-bt-fence}
\begin{align*}
	\backtrfencec{ \trg{H;\OB{{F}}} } =&\ \src{\backtrfencec{ \trg{H} };\backtrfencec{ \trg{\OB{F}} } }
	\\\\
	\backtrfencec{ \trge } =&\ \srce
	\\
	\backtrfencec{ \trg{H; n \mapsto v : \taintt} } =&\ \src{\backtrfencec{ \trg{H} } ; \backtrfencec{\trg{n}}\mapsto\backtrfencec{\trg{v}}:\backtrfencec{\taintt}}
	\\\\
	\backtrfencec{ \trg{f(x) \mapsto s;\ret}} =&\ \src{f(x)\mapsto \backtrfencec{\trg{s}};\ret}
	\\\\
	\backtrfencec{\taintt} =&\ \src{\sigma}
	\\\\
	\backtrfencec{ \trg{n} } =&\ \src{n}
	\\
	\backtrfencec{ \trg{e\op e'} } =&\  \src{\backtrfencec{ \trg{e} }\op\backtrfencec{ \trg{e'} }}
	\\
	\backtrfencec{ \trg{e\bop e'} } =&\ \src{\backtrfencec{ \trg{e} }\bop\backtrfencec{ \trg{e'} }}
	\\\\
	\backtrfencec{ \trg{\skipt} } =&\ \src{\skips}
	\\
	\backtrfencec{ \trg{s;s'} } =&\ \src{\backtrfencec{ \trg{s} };\backtrfencec{ \trg{s'} }}
	\\
	\backtrfencec{ \trg{\letin{x}{e}{s}} } =&\ \src{\letin{x}{\backtrfencec{ \trg{e} }}{\backtrfencec{ \trg{s} }}}
	\\
	\backtrfencec{ \trg{\ifzte{e}{s}{s'}} } =&\ \src{\ifzte{\backtrfencec{ \trg{e} }}{\backtrfencec{ \trg{s} }}{\backtrfencec{ \trg{s'} }}}
	\\
	\backtrfencec{ \trg{\call{f}~ e} } =&\ \src{\call{f}~\backtrfencec{ \trg{e} }}
	\\
	\backtrfencec{ \trg{\asgn{e}{e'}} } =&\ \src{\asgn{\backtrfencec{ \trg{e} }}{\backtrfencec{ \trg{e'} }}}
	\\
	\backtrfencec{ \trg{\letread{x}{e}{s}} } =&\ \src{\letread{x}{\backtrfencec{ \trg{e} }}{\backtrfencec{ \trg{s} }}}
	\\
	\backtrfencec{ \trg{\lfence} } =&\ \src{\skips}
	\\
	\backtrfencec{ \trg{\cmove{x}{e}{e'}{s}} } =&\ \src{ \ifzte{\backtrfencec{ \trg{e'} }}{\letin{x}{\backtrfencec{\trg{e}}}{\skips}}{\skips};\backtrfencec{ \trg{s} } }
\end{align*}
Note that we define the backtranslation of heaps because attackers define them.
We do not define compilation of heaps because components do not define them, though adding them would be simple.

We can use this backtranslation to prove \rdss.

\subsubsection{Properties of the Context-based Backtranslation}\label{sec:ctx-bt-fence-props}
We want the backtranslation to be correct, so given a compiled program and a context, the backtranslation generates a source context that with the program generates a trace that is related to the target one.

\begin{theorem}[Correctness of the Backtranslation for lfence]\label{thm:corr-bt-lfence}(\showproof{corr-bt-lfence})
	\begin{align*}
		\text{ if } 
			&\
			\trg{\ctxt{}\hole{\complfence{P}} \sem \trat{^{\taintt}}}
		\\
		\text{ then }
			&\
			\src{\backtrfencec{\ctxt{}}\hole{P} \sem \tras{^{\sigma}}}
		\\
		\text{ and }
			&\
			\tras{^{\sigma}} \rels \trat{^{\taintt}}
	\end{align*}
\end{theorem}

\begin{theorem}[Generalised Backward Simulation for lfence]\label{thm:bwd-sim-lfence}(\showproof{bwd-sim-lfence})
	\begin{align*}
		\text{ if }
			&\
			\text{ if } \src{f}\in\src{\OB{f''}} \text{ then } \trg{s} = \complfence{s} \text{ else } \src{s} = \backtrfencec{\trg{s}}
			\\
		\text{ and}
			&\
			\text{ if } \src{f'}\in\src{\OB{f''}} \text{ then } \trg{s'} = \complfence{s'} \text{ else } \src{s'} = \backtrfencec{\trg{s'}}
		\\
		\text{ and }
			&\
			\trg{\Sigma}=\trg{w( C, H, \OB{B} \triangleright \proc{s;s''}{\OB{f}\cdot f}, \bot,\safeta)}
		\\
		\text{ and }
			&\
			\trg{\Sigma'}=\trg{ w(C, H', \OB{B'} \triangleright \proc{s';s''}{\OB{f'}\cdot f'},\bot,\safeta)}
		\\
		\text{ and }
			&\
			\trg{(n,\Sigma) \Xtot{\trat{^{\taintt}}} (n',\Sigma')}
		\\
		\text{ and }
			&\
			\src{\Omega} \srel_{\src{\OB{f''}}} \trg{\Sigma}
		\\
		\text{ then }
			&\
			\src{\Omega}=\src{C, H, \OB{B} \triangleright \proc{s;s''}{\OB{f}\cdot f} \Xtos{\tras{^\sigma}} C, H', \OB{B'} \triangleright \proc{s';s''}{\OB{f'}\cdot f'}}=\src{\Omega'}
		\\
		\text{ and }
			&\
			\src{\tras{^\sigma}} \tracerel \trg{\trat{^{\taintt}}}
		\\
		\text{ and }
			&\
			\src{\Omega'} \srelref_{\src{\OB{f''}}}\, \trg{\Sigma'}
	\end{align*}
\end{theorem}

\subsubsection{Simulation and Relation for Compiled Code}\label{sec:sim-code-fence}
We need the usual backward simulation result which we derive from forward simulation plus determinism of the semantics.

Values are only nats, so values are related if they are the same nat.
Heaps are related if they map related nats (the same address) to related values.

\mytoprule{\text{Heap relation} \hreldef \text{Value relation} \vreldef }
\begin{center}
	\typerule{Heap - base }{}{
		\srce \hrel \trge
	}{hrel-b}
	\typerule{Heap - ind }{
		\src{H} \hrel \trg{H}
		\\
		\src{z} \vrel \trg{z}
		&
		\src{v} \vrel \trg{v}
		&
		\src{\sigma} \equiv \taintt
	}{
		\src{H ; z\mapsto v:\sigma} \hrel \trg{H; z\mapsto v:\taintt}
	}{hrel-i}
	\typerule{Heap - start }{
		\src{H} \hrel \trg{H}
		&
		\src{H'} \hrel \trg{H'}
		\\
		\src{v} \vrel \trg{v}
		&
		\src{\sigma} \equiv \taintt
	}{
		\src{H ; 0\mapsto v:\sigma ; H'} \hrel \trg{H;0\mapsto v:\taintt ; H'}
	}{hrel-start}
	\typerule{Value - num }{
		\src{z} \equiv \trg{z} &
		\com{z}\in\mb{Z}
	}{
		\src{z} \vrel \trg{z}
	}{vr}
\end{center}
\botrule

Bindings are related if they map same-named variables to related values.
Components are related according to a list of function names \src{\OB{f}} which identify compiled code.
All functions in the list have a compiled counterpart in the target component while all functions not in the list have a backtranslated counterpart in the source component.
States are related if the target is not speculating and the $\Omega$ sub-component of the target state is related to the source state. 
This relation is the key one for this compiler, the SLH compiler will need a different one.	
 
\mytoprule{\text{Binding relation} \breldef \text{ Component relation} \creldef  \text{ State relation} \sreldef}
\begin{center}
	\typerule{ Binding - base }{}{
		\srce \brel \trge
	}{brel-b}
	\typerule{ Binding - ind }{
		\src{B} \brel \trg{B}
		&
		\src{v} \vrel \trg{v}	
		&
		\src{\sigma} \equiv \taintt
	}{
		\src{B \cdot x\mapsto v:\sigma} \brel \trg{B\cdot x\mapsto v:\taintt}
	}{brel-i}
	\typerule{ Bindings }{
		\src{\OB{B}} \brel \trg{\OB{B}}
		&
		\src{B} \brel \trg{B}
	}{
		\src{\OB{B}\cdot B} \brel \trg{\OB{B}\cdot B}
	}{bsrel}
	\typerule{ Components }{
		\forall \src{f}\in\src{\OB{f}}.
		\text{ if } \src{f(x)\mapsto s} \in \src{\OB{F}}
		\text{ then } \trg{f(x)\mapsto \complfence{s}} \in \trg{\OB{F}}
		\\
		\forall \trg{f(x)\mapsto s} \in \trg{\OB{F}}
		\text{ if } \trg{f}\notin\src{\OB{f}}
		\text{ then } \src{f(x)\mapsto \backtrfencec{\trg{s}}} \in \src{\OB{F}}
		\\
		\src{\OB{I}}\equiv\trg{\OB{I}}
	}{
		\src{\OB{F};\OB{I}} \crel_{\src{\OB{f}}} \trg{\OB{F};\OB{I}} 
	}{comps}
	\typerule{ States }{
		\src{\OB{B}} \brel \trg{\OB{B}}
		&
		\src{H} \hrel \trg{H}
		&
		\src{\OB{f}} \equiv \trg{\OB{f}}
		&
		\src{C} \crel_{\src{\OB{f''}}} \trg{C}
	}{
		\src{C,H,\OB{B}\triangleright \proc{s}{\OB{f}} } \srel_{\src{\OB{f''}}} \trg{w, (C,H,\OB{B}\triangleright \proc{s}{\OB{f}}, \bot, \safeta) }
	}{stats}
\end{center}
\botrule
The state relation is very powerful because the target cannot be in a speculation state.
We can break this invariant temporarily (at the beginning of a compiled if) but we need to reinstate it (by executing the lfence).

Starting with a related stack frame, if a source expression with a substitution star-reduces to a value, then the compiled expression with the compiled substitution star-reduces to the compiled value.
\begin{lemma}[Forward Simulation for Expressions in lfence]\label{thm:fwd-sim-exp-lfence}(\showproof{fwd-sim-exp-lfence})
	\begin{align*}
		\text{ if }
			&\ 
			\src{B\triangleright e \bigreds v : \sigma}
		\\
		\text{ and }
			&\
			\src{B}\brelref\trg{B}
		\\
		\text{ and }
			&\
			\src{\sigma}\equiv\taintt
		\\
		\text{ then }
			&\
			\trg{B \triangleright \complfence{e}\bigredt\complfence{v} : \taintt}
	\end{align*}
\end{lemma}
Starting with related components, heaps and stack frames, if a source statement takes a step emitting a label, then the compiled statement can take several steps and emit a trace that is related to the label.
Effectively, the target also only emits a single label, but we need to account for multiple steps (for the compilation of the if).
We keep track of arbitrary source and target continuations \src{s''} and \trg{s''} that are not touched by the reductions in order to use this result in a general setting.
\begin{lemma}[Forward Simulation for Compiled Statements in lfence]\label{thm:fwd-sim-stm-lfence}(\showproof{fwd-sim-stm-lfence})
	\begin{align*}
		\text{ if }
			&\
			\src{\Omega}=\src{C, H, \OB{B} \triangleright \proc{s;s''}{\OB{f}} \xtos{\lambda^\sigma} C, H', \OB{B'} \triangleright \proc{s';s''}{\OB{f'}}}=\src{\Omega'}
		\\
		\text{ and }
			&\
			\src{\Omega} \srelref_{\src{\OB{f''}}}\, \trg{\Sigma}
		\\
		\text{ and }
			&\
			\trg{\Sigma}=\trg{w, (C, H, \OB{B} \triangleright \proc{\complfence{s};s''}{\OB{f}}, \bot, \safeta) }
		\\
		\text{ and }
			&\
			\trg{\Sigma'}=\trg{ w (C, H', \OB{B'} \triangleright \proc{\complfence{s'};s''}{\OB{f'}}, \bot, \safeta)}
		\\
		\text{ then }
			&\
			\trg{(n,\Sigma) \Xtot{\trat{^{\taintt}}} (n',\Sigma')}
		\\
		\text{ and }
			&\
			\src{\lambda^\sigma} \tracerel \trg{\trat{^{\taintt}}} \qquad \text{( using the trace relation!)}
		\\
		\text{ and }
			&\
			\src{\Omega'} \srelref_{\src{\OB{f''}}}\, \trg{\Sigma'}
	\end{align*}
\end{lemma}

In the general theorem we need backward simulation, which we derive from forward simulation in a standard way.
Note that by ending up in a state with a compiled statement, we rule out cross-boundary calls and returns.
These cases pop up in the proofs using this theorem and we deal with them there. 
\begin{theorem}[Backward Simulation for Compiled Steps in lfence]\label{thm:bwd-sim-comp-steps-lfence}(\showproof{bwd-sim-comp-steps-lfence})
	\begin{align*}
		\text{ if }
			&\
			\trg{\Sigma}=\trg{w, (C, H, \OB{B} \triangleright \proc{\complfence{s};s''}{\OB{f}}, \bot, \safeta) }
		\\
		\text{ and }
			&\
			\trg{\Sigma'}=\trg{ w (C, H', \OB{B'} \triangleright \proc{\complfence{s'};s''}{\OB{f'}}, \bot, \safeta)}
		\\
		\text{ and }
			&\
			\trg{(n,\Sigma) \Xtot{\trat{^{\taintt}}} (n',\Sigma')}
		\\
		\text{ and }
			&\
			\src{\Omega} \srelref_{\src{\OB{f''}}}\, \trg{\Sigma}
		\\
		\text{ then }
			&\
			\src{\Omega}=\src{C, H, \OB{B} \triangleright \proc{s;s''}{\OB{f}} \xtos{\lambda^\sigma} C, H', \OB{B'} \triangleright \proc{s';s''}{\OB{f'}}}=\src{\Omega'}
		\\
		\text{ and }
			&\
			\src{\lambda^\sigma} \tracerel \trg{\trat{^{\taintt}}} \qquad \text{( using the trace relation!)}
		\\
		\text{ and }
			&\
			\src{\Omega'} \srelref_{\src{\OB{f''}}}\, \trg{\Sigma'}
	\end{align*}
\end{theorem}

\subsubsection{Simulation and Relation for Backtranslated Code}\label{sec:bwd-sim-lfence}
We need backward simulation for backtranslated code.
Since we are doing a context-based backtranslation, and since backtranslated values are related to source values using $\vrelref$ as for compiled values, we do not need extra relations.

As before, we need two lemmas on the backward simulation for backtranslated expressions and on the backward simulation for statements.
\begin{lemma}[Backward Simulation for Backtranslated Expressions in lfence]\label{thm:back-sim-bte-lfence}(\showproof{back-sim-bte-lfence})
	\begin{align*}
		\text{ if }
			&\ 
			\trg{B \triangleright {e}\bigredt{v} : \taintt}
		\\
		\text{ and }
			&\
			\src{B}\brelref\trg{B}
		\\
		\text{ and }
			&\
			\src{\sigma}\equiv\taintt
		\\
		\text{ then }
			&\
			\src{B\triangleright \backtrfencec{\trg{e}} \bigreds \backtrfencec{\trg{v}} : \sigma}
	\end{align*}
\end{lemma}

\begin{lemma}[Backward Simulation for Backtranslated Statements]\label{thm:back-sim-bts-lfence}(\showproof{back-sim-bts-lfence})
	\begin{align*}
		\text{ if }
			&\
			\trg{\Sigma}=\trg{w(C, H, \OB{B} \triangleright \proc{{s};s''}{\OB{f}},\bot,\safeta) }
		\\
		\text{ and }
			&\
			\trg{\Sigma'}=\trg{ w(C, H', \OB{B'} \triangleright \proc{{s'};s''}{\OB{f'}},\bot\safeta)}
		\\
		\text{ and }
			&\
			\trg{\Sigma \xltot{{\trgb{\lambda}^{\taintt}}} \Sigma'}
		\\
		\text{ and }
			&\
			\src{\Omega} \srel_{\src{\OB{f''}}} \trg{\Sigma}
		\\
		\text{ then }
			&\
			\src{\Omega}=\src{C, H, \OB{B} \triangleright \proc{\backtrfencec{\trg{s}};s''}{\OB{f}} \xtos{\lambda^\sigma} C, H', \OB{B'} \triangleright \proc{\backtrfencec{\trg{s'}};s''}{\OB{f'}}}=\src{\Omega'}
		\\
		\text{ and }
			&\
			\src{\lambda^\sigma} \arelref \trg{{\trgb{\lambda}^{\taintt}}} \qquad \text{( using the action relation!)}
		\\
		\text{ and }
			&\
			\src{\Omega'} \srel_{\src{\OB{f''}}} \trg{\Sigma'}
	\end{align*}
\end{lemma}
Note that by ending up in a state with a backtranslated statement, we rule out cross-boundary calls and returns, they are dealt with in \Cref{thm:bwd-sim-lfence}.

The final result we will need for the general theorem is that initial states made of a compiled program and a backtranslated context are related.
This relies on heap and value cross-language relation holding for compiled and backtranslated heaps and values.

\begin{lemma}[Initial States are Related]\label{thm:ini-state-rel}(\showproof{ini-state-rel})
	\begin{align*}
		&
		\forall \src{P}, \forall \src{\OB{f}}=\dom{\src{P}.\src{F}}, \forall \ctxt{}
		\\
		&
		\SInits{\backtrfencec{\ctxt{}}\hole{P}} \srelref_{\src{\OB{f}}}\, \SInitt{\ctxt{}\complfence{\src{P}}}
	\end{align*}
\end{lemma}

\begin{lemma}[A Value is Related to its Compilation for lfence]\label{thm:val-rel-comp-lfence}(\showproof{val-rel-comp-lfence})
	\begin{align*}
		{\src{v}}\vrel\complfence{v}
	\end{align*}
\end{lemma}

\begin{lemma}[A Heap is Related to its Compilation for lfence]\label{thm:heap-rel-comp-lfence}(\showproof{heap-rel-comp-lfence})
	\begin{align*}
		{\src{H}}\hrel\complfence{H}
	\end{align*}
\end{lemma}

\begin{lemma}[A Value is Related to its Backtranslation]\label{thm:val-rel-bt-lfence}(\showproof{val-rel-bt-lfence})
	\begin{align*}
		\backtrfencec{\trg{v}}\vrel\trg{v}
	\end{align*}
\end{lemma}

\begin{lemma}[A Taint is Related to its Backtranslation]\label{thm:taint-rel-bt-lfence}(\showproof{taint-rel-bt-lfence})
	\begin{align*}
		\backtrfencec{\taintt}\equiv\taintt
	\end{align*}
\end{lemma}

\begin{lemma}[A Heap is Related to its Backtranslation]\label{thm:heap-rel-bt-lfence}(\showproof{heap-rel-bt-lfence})
	\begin{align*}
		\backtrfencec{\trg{H}}\hrel\trg{H}
	\end{align*}
\end{lemma}

\subsection{A Stronger Property for \complfence{\cdot}}\label{sec:stronger-prop-lfence}
\begin{theorem}[The lfence compiler is \rdss with more leaks]\label{thm:lfence-comp-stc}(\showproof{lfence-comp-stc})
	\begin{align*}
		\vdash \complfence{\cdot} : \rdss 
	\end{align*}
\end{theorem}

Let's add a reduction that generates an \trg{\eta} action:
\begin{center}
	\typerule{E-\TR-speculate-eta}{
		\trg{\Omega \xtot{\epsilon} \Omega'}	
		&
		\trg{\Omega}\equiv\trg{C, H, \OB{B} \triangleright s;s'}
		&
		\trg{s}\not\equiv\trg{\lfence}
	}{
		\trg{w, \OB{(\Omega,\trgb{\omega},\taintt)}\cdot(\Omega,n+1,\taintt) \xltot{\trgb{\eta}^{\taintt}} w, \OB{(\Omega,\trgb{\omega},\taintt)}\cdot(\Omega',n,\taintt)}
	}{et-eta}
\end{center}
Any silent step that is not lfence now generates an \trgb{\eta} action whose tag is the same as the pc's.
Thus, eta actions are safe when not speculating and they are unsafe when speculating.

\newpage
\section{The Current SLH Compiler \compslh{\cdot}}\label{sec:comp-slh-weak}
\begin{align*}
	\compslh{ \call{f}~e } &= \trg{ \call{f}~\compslh{e} }
	\\
	\compslh{ \asgn{e}{e'} } &= \trg{ \asgn{\compslh{e}}{\compslh{e'}} }
	\\
	\compslh{ \letread{x}{e}{s} } &= \trg{ \letread{x}{\compslh{e}}{ \compslh{s} } }
	\\
	\compslh{ \asgnp{e}{e'} } &= \trg{ \asgnpt{\compslh{e}+1}{\compslh{e'}} }
	\\
	\compslh{ \ifzte{e}{s}{s'} } &= \trg{ 
		\begin{aligned}[t]
			&
			\letint{\trg{x_g}}{\compslh{e}}{
			\\
			&\
				\ifztet{\trg{x_g}
					}{ 
					\letreadpt{\trg{x}}{\trg{-1}}{\asgnpt{\trg{-1}}{\trg{x \vee \neg x_g}}};
					\compslh{s}
				\\
				&\ \
				}{
					\letreadpt{\trg{x}}{\trg{-1}}{\asgnpt{\trg{-1}}{\trg{x \vee x_g}}};
					\compslh{s'}
				}
			}
		\end{aligned}
	} 
	\\
	\compslh{\letreadp{x}{e}{s}} &= 
		\trg{
			\letreadpt{\trg{x}}{\compslh{e}+1}{ 
				\letreadpt{\trg{\predState}}{\trg{-1}}{
				\cmovet{\trg{x}}{\trg{0}}{\trg{\predState}}{\compslh{s}}
				}
			}			
		}
\end{align*}
Omitted cases are as in \Cref{sec:comp-slh}.

\begin{theorem}[This SLH compiler is \weak{\rdss}]\label{thm:slh-comp-rdss-weak}(\showproof{slh-comp-rdss-weak})
	\begin{align*}
		\vdash \compslh{\cdot} : \weak{\rdss}
	\end{align*}
\end{theorem}

\newpage
\section{The Strong SLH Compiler \compsslh{\cdot}}\label{sec:comp-slh}

\begin{align*}
	\compsslh{ H ; \OB{F} ; \OB{I}} &= \trg{  \compsslh{H}\cup(-1\mapsto\falset : \safeta) ; \compsslh{\OB{F}} ; \compsslh{\OB{I}}}
	\\
	\\
	\compsslh{ \srce} &= \trge
	\\
	\compsslh{ \OB{I}\cdot f } &= \trg{ \compsslh{ \OB{I} }\cdot f }
	\\\\
	\compsslh{H , -n\mapsto v:\unta} &= \trg{\compsslh{H}, -\compsslh{n}-1\mapsto\compsslh{v}:\unta}
	\\\\
	\compsslh{ {f(x)\mapsto s;\ret} } &= 
		\trg{f(x)\mapsto  
				\compsslh{s};
				\trg{\ret}
		}
	\\
	\\
	\compsslh{ n } &= \trg{n}
	\\
	\compsslh{ e \op e' } &= \trg{ \compsslh{ e } \trgb{\op} \compsslh{ e' }}
	\\
	\compsslh{ e \bop e' } &= \trg{ \compsslh{ e } \trgb{\bop} \compsslh{ e' }}
	\\
	\\
	\compsslh{ {s;s'} } &= \compsslh{s}\trg{;} \compsslh{s'}
	\\
	\compsslh{ \skips } &= \trg{\skipt}
	\\
	\compsslh{ \letin{x}{e}{s} } &= \trg{ \letin{x}{\compsslh{e}}{\compsslh{s}}}
	\\
	\compsslh{ \call{f}~e } &= 
		\trg{
			\begin{aligned}[t]
				&
				\letint{\trg{x_f}}{\compsslh{e}}{
				\\&\
					\letreadpt{\trg{\predState}}{\trg{-1}}{
					\\&\ \ 
						\cmovet{\trg{x_f}}{\trg{0}}{\trg{\predState}}{\trg{\call{f}~x_f}}
					}
				}
			\end{aligned}
		}
	\\
	\compsslh{ \asgn{e}{e'} } &= 
		\trg{
			\begin{aligned}[t]
				&
				\letint{\trg{x_f}}{\compsslh{e}}{
					\\
					&\
					\letint{\trg{x_f'}}{\compsslh{e'}}{
						\\
						&\ \ 
						\letreadpt{\trg{\predState}}{-1}{
						\\
						&\  \ \
							\cmovet{\trg{x_f}}{\trg{0}}{\trg{\predState}}{ 
							\\
							&\ \ \ \
								\cmovet{\trg{x_f'}}{\trg{0}}{\trg{\predState}}{ 
								\\
								&\ \ \ \ \ 
									\trg{\asgn{x_f}{x_f'}}
								}
							}
						}
					}
				}
			\end{aligned}
		}
	\\
	\compsslh{ \letread{x}{e}{s} } &= 
		\trg{ 
			\begin{aligned}[t]
				&
				\letreadt{\trg{x}}{\compsslh{e}}{ 
				\\&\
					\letreadpt{\trg{\predState}}{\trg{-1}}{
					\\&\ \ 
						\cmovet{\trg{x}}{\trg{0}}{\trg{\predState}}{ 
						\\
						&\ \ \ 
							\compsslh{s} 
						}
					}
				} 
			\end{aligned}
		}
	\\
	\compsslh{ \asgnp{e}{e'} } &= 
		\trg{
			\begin{aligned}[t]
				&
				\letint{\trg{x_f}}{\compsslh{e}\trg{+1}}{
					\\
					&\
					\letint{\trg{x_f'}}{\compsslh{e'}}{
						\\
						&\ \ 
						\letreadpt{\trg{\predState}}{\trg{-1}}{
						\\
						&\ \ \
							\cmovet{\trg{x_f}}{\trg{0}}{\trg{\predState}}{ 
							\\
							&\ \ \ \
								\cmovet{\trg{x_f'}}{\trg{0}}{\trg{\predState}}{ 
								\\
								&\ \ \ \ \ 
									\trg{\asgnp{x_f}{x_f'}}
								}
							}
						}
					}
				}
			\end{aligned}
		}
	\\
	\compsslh{ \ifzte{e}{s}{s'} } &= \trg{ 
		\begin{aligned}[t]
			&
			\letint{\trg{x_g}}{\compsslh{e}}{
			\\
			&\
				\letreadpt{\trg{\predState}}{\trg{-1}}{
					\\
					&\ \
					\cmovet{\trg{x_g}}{\trg{0}}{\trg{\predState}}{
					\\
					&\ \ \ 
						\ifztet{\trg{x_g}
							\\
							&\ \ \ \ 
							}{ 
							\letreadpt{\trg{x}}{\trg{-1}}{\asgnpt{\trg{-1}}{\trg{x \vee \neg x_g}}};
							\compsslh{s}
						\\
						&\ \ \ \ 
						}{
							\letreadpt{\trg{x}}{\trg{-1}}{\asgnpt{\trg{-1}}{\trg{x \vee x_g}}};
							\compsslh{s'}
						}
					}
				}
			}
		\end{aligned}
	} 
	\\
	\compsslh{\letreadp{x}{e}{s}} &= 
		\trg{
			\begin{aligned}[t]
				&
				\letint{\trg{x}}{\compsslh{e}\trg{+1}}{
					\\
					&\ 
					\letreadpt{\trg{\predState}}{\trg{-1}}{
					\\
					&\ \ 
						\cmovet{\trg{x}}{\trg{0}}{\trg{\predState}}{
						\\
						&\ \ \ 
							\letreadpt{\trg{x}}{\trg{x}}{\compsslh{s}}
						}
					}	
				}	
			\end{aligned}
		} 
\end{align*}

We compile all private reads and write to operate on an address that is greater than the expected one by 1.
This ensures that location \trg{-1} is untouched by the compiled code, so it can be used to keep information, namely the \trg{pr} state.

\begin{theorem}[Our SLH compiler is \strong{\rdss}]\label{thm:slh-comp-rdss}(\showproof{slh-comp-rdss})
	\begin{align*}
		\vdash \compsslh{\cdot} : \strong{\rdss}
	\end{align*}
\end{theorem}

\subsection{Inter-procedural SLH}
This SLH compiler does not pass the pr state across procedures and it needs the lfence to still be \strong{\rdss}.
\begin{align*}
	\compslht{ {f(x)\mapsto s;\ret} } &= 
		\trg{f(x)\mapsto  
			\begin{aligned}[t]
				&
				\trg{\lfence;}
				\\
				&
				\letint{\trg{x_{pr}}}{\falset}{\compslht{s}};
				\trg{\ret}
			\end{aligned}
		}
	\\
	\compslht{ \ifzte{e}{s}{s'} } &= \trg{ 
		\begin{aligned}[t]
			&
			\letint{\trg{x_g}}{\compsslh{e}}{
			\\&\
				\ifztet{\trg{x_g}
					}{ 
					\letint{\trg{x_{pr}}}{\trg{x_{pr} \vee \neg x_g}}{\compslht{s}}
				\\
				&\
				}{
					\letint{\trg{x_{pr}}}{\trg{x_{pr} \vee x_g}}{\compslht{s'}}
				}
			}
		\end{aligned}
	} \\
	\compslht{\letreadp{x}{e}{s}} &= 
		\trg{
			\letreadpt{\trg{x}}{\trg{e}}{ 
			\cmovet{\trg{x}}{\trg{0}}{\trg{x_{pr}}}{\compslht{s}}
			}			
		} 
\end{align*}

In order to prove \rdss for this compiler, we need a strong relation between states that instead of asserting that \trg{H(-1)} keeps a bool of the speculation, each state has the first binding for a variable which captures speculation.

When proving \Thmref{thm:spec-most-omega}, in the case of a call from a context to compiled component, we need to add \trg{\lfence}, so that we stop speculation altogether to instate this invariant.
This is noted in \showproof{spec-most-omega}.

\begin{theorem}[Our inter-procedural SLH compiler is \strong{\rdss}]\label{thm:slh-comp-rdss-proc}(\showproof{slh-comp-rdss-proc})
	\begin{align*}
		\vdash \compslht{\cdot} : \strong{\rdss}
	\end{align*}
\end{theorem}

\subsection{Backtranslation}\label{sec:bt-slh}
We can use the same context-based backtranslation of \Cref{sec:ctx-bt-fence}, what changes are the cross-language relations.

\subsubsection{Relations for the SLH Compiler}\label{sec:rels-slh}
The state relation extends the previous one because now we can relate a source state to a target state that is speculating.
Given a target state \trg{\Sigma} that is a stack of operational states \trg{\OB{\Omega}}, we require that the first element of the stack is in the strong state relation ($\srel$) with a source state \src{\Omega}.
All other states need to be related at a weaker state relation ($\ssrel$), which only enforces that all bindings map variables to \trg{\safeta} values.
That is, a target an a source list of bindings are related if all that the target binds is \trg{\safeta}, though the target can bind possibly more variables and the target stack of bindings can be 
While speculating, target state do not mimick source reductions anymore but we need to enforce this taint on variables in order to ensure that any target action will be \trg{\safeta}.

Since our target language uses possibly more variables than the source (e.g., \trg{\predState}), we need to sometimes forget them.
So the binding relation is parametrised by a stack of possibly growing list of name variables \trg{\OB{D}} that tell us when to not care for a binding, i.e., when we find a variable in \trg{B} that is in the current \trg{D}.
We adapt the state relation to forward that (later).

\mytoprule{\text{Binding relation} \breldef }
\begin{center}
	\typerule{ Binding - base }{}{
		\srce \brel^{\trge} \trge
	}{sbrel-b}
	\typerule{ Binding - ind }{
		\src{B} \brel^\trg{D} \trg{B}
		&
		\src{v} \vrel \trg{v}	
		&
		\src{\sigma} \equiv \taintt
	}{
		\src{B \cdot x\mapsto v:\sigma} \brel^\trg{D} \trg{B\cdot x\mapsto v:\taintt}
	}{sbrel-i}
	\typerule{ Binding - ind }{
		\src{B} \brel^\trg{D} \trg{B}
		&
		\trg{x}\in\trg{D}
	}{
		\src{B} \brel^\trg{D} \trg{B\cdot x\mapsto v:\taintt}
	}{sbrel-f}
	\typerule{ Bindings }{
		\src{\OB{B}} \brel^{\trg{\OB{D}}} \trg{\OB{B}}
		&
		\src{B} \brel^{\trg{D}} \trg{B}
	}{
		\src{\OB{B}\cdot B} \brel^{\trg{\OB{D}\cdot D}} \trg{\OB{B}\cdot B}
	}{sbsrel}
\end{center}
\botrule

We need to change the heap relation to account for the fact that private heaps are related when the addresses are -1 in the target.
So we define the relation $\shrel$ for private heaps (whose domain is negative integers) to account for this.

\mytoprule{\text{Heap relation} \shreldef}
\begin{center}
	\typerule{Heap - base }{}{
		\srce \shrel \trge
	}{shrel-b}
	\typerule{Heap - negative }{
		\src{H} \shrel \trg{H}
		\\
		\src{z-1} \vrel \trg{z}
		&
		\src{v} \vrel \trg{v}
		&
		\src{\sigma} \equiv \taintt
	}{
		\src{H ; z\mapsto v:\sigma} \shrel \trg{H; z\mapsto v:\taintt}
	}{shrel-i}
	\typerule{Heap - start }{
		\src{H} \shrel \trg{H}
		&
		\src{H'} \hrel \trg{H'}
		&
		\src{v} \vrel \trg{v}
		&
		\src{\sigma} \equiv \taintt
	}{
		\src{H ; 0\mapsto v:\sigma ; H'} \hrel \trg{H;-1 \mapsto \falset : \safeta;0\mapsto v:\taintt ; H'}
	}{shrel-start}
\end{center}
\botrule

\mytoprule{\text{Binding relation} \sbreldef  \text{ State relation} \ssreldef}
\begin{center}
	\typerule{States relation}{
		\src{\Omega}\srel_{\src{\OB{f}}}^\trg{\OB{D}}\trg{(w,(\Omega,\bot,\safeta))}
		&
		\forall \trg{\Omega_s}\in\trg{\OB{\Omega}},~
		\src{\Omega}\ssrel_{\src{\OB{f}}}\trg{\Omega_s}
	}{
		\src{\Omega}\ssrel_{\src{\OB{f}}}^\trg{\OB{D}}\trg{(w,(\Omega,\bot,\safeta)\cdot\OB{(\Omega,W,\unta)})}
	}{}
	\typerule{ Base States }{
		\src{\OB{B}} \brel^\trg{\OB{D}} \trg{\OB{B}}
		&
		\src{H} \hrel \trg{H}
		&
		\src{\OB{f}} \equiv \trg{\OB{f}}
		&
		\src{C} \crel_{\src{\OB{f''}}} \trg{C}
	}{
		\src{C,H,\OB{B}\triangleright \proc{s}{\OB{f}} } \srel_{\src{\OB{f''}}}^\trg{\OB{D}} \trg{w, (C,H,\OB{B}\triangleright \proc{s}{\OB{f}}, \bot, \safeta) }
	}{sstats}
	\typerule{Single state relation - ctx}{
		(\text{ if }\trg{f}\notin\src{\OB{f'}} \text{ then } \vdash \trg{{B}} : \sbrel)
		&
		\src{C}\crel_{\src{\OB{f'}}}\trg{C}
		&
		\vdash\trg{H}:\shrel
	}{
		\src{C, H, \OB{B}\triangleright \proc{s}{\OB{f}} } \ssrel_{\src{\OB{f'}}} \trg{C, H, \OB{B}\cdot B\triangleright \proc{s}{\OB{f}\cdot f} } 
	}{}

	\typerule{Heap relation same}{
		\forall \trg{n\mapsto v : \taintt} \in \trg{H} 
		\text{ if } \trg{n}\geq\trg{0} \text{ then } \taintt=\trg{\safeta}
		&
		\trg{H(-1)}=\trg{\truet : \safeta}
	}{
		\vdash \trg{H} : \shrel
	}{}
	\typerule{Target Bindings ok base}{
	}{
		\vdash \trge : \sbrel
	}{}
	\typerule{Target Bindings ok sing}{
		\vdash \trg{{B}} : \sbrel	
	}{
		\vdash \trg{ B \cdot x \mapsto v : \safeta} : \sbrel
	}{}

\end{center}
\botrule
When speculating, any heap is related: the target may write extra things speculatively.

\subsubsection{Relation for Inter-procedural SLH}\label{sec:rel-interproc-slh}

\mytoprule{\text{Binding relation} \csbreldef  \text{ State relation} \cssreldef}
\begin{center}
	\typerule{States relation}{
		\src{\Omega}\srel_{\src{\OB{f}}}^\trg{\OB{D}}\trg{(w,(\Omega,\bot,\safeta))}
		&
		\forall \trg{\Omega_s}\in\trg{\OB{\Omega}},~
		\src{\Omega}\cssrel_{\src{\OB{f}}}\trg{\Omega_s}
	}{
		\src{\Omega}\cssrel_{\src{\OB{f}}}^\trg{\OB{D}}\trg{(w,(\Omega,\bot,\safeta)\cdot\OB{(\Omega,W,\unta)})}
	}{}
	\typerule{Single state relation - ctx}{
		(\text{ if }\trg{f}\notin\src{\OB{f'}} \text{ then } \vdash \trg{{B}} : \csbrel)
		&
		\src{C}\crel_{\src{\OB{f'}}}\trg{C}
		&
		\vdash\trg{H}:\cshrel
	}{
		\src{C, H, \OB{B}\triangleright \proc{s}{\OB{f}} } \cssrel_{\src{\OB{f'}}} \trg{C, H, \OB{B}\cdot B\triangleright \proc{s}{\OB{f}\cdot f} } 
	}{aasd}
	\typerule{Heap relation same}{
		\forall \trg{n\mapsto v : \taintt} \in \trg{H} 
		\text{ if } \trg{n}\geq\trg{0} \text{ then } \taintt=\trg{\safeta}
	}{
		\vdash \trg{H} : \cshrel
	}{}
	\typerule{Target Bindings ok base}{
	}{
		\vdash \trg{\predState\mapsto 0 : \safeta} : \csbrel
	}{}
	\typerule{Target Bindings ok sing}{
		\vdash \trg{{B}} : \csbrel	
	}{
		\vdash \trg{ B \cdot x \mapsto v : \safeta} : \csbrel
	}{}
\end{center}
\botrule

Speculating states are related $\cssrel$ when, in case the executing function is not attacker (\Cref{tr:aasd}), the heap is whatever but the bindings contain the predicate bit set to true.

\begin{lemma}[Initial States are Related for SLH]\label{thm:ini-state-rel-slh}(\showproof{ini-state-rel-slh})
	\begin{align*}
		&
		\forall \src{P}, \forall \src{\OB{f}}=\dom{\src{P}.\src{F}}, \forall \ctxt{}
		\\
		&
		\SInits{\backtrfencec{\ctxt{}}\hole{P}} \ssrelref_{\src{\OB{f}}}\, \SInitt{\ctxt{}\compsslh{\src{P}}}
	\end{align*}
\end{lemma}

\begin{lemma}[A Value is Related to its Compilation for SLH]\label{thm:val-rel-comp-slh}(\showproof{val-rel-comp-slh})
	\begin{align*}
		{\src{v}}\vrel\compsslh{v}
	\end{align*}
\end{lemma}

\begin{lemma}[A Heap is Related to its Compilation for SLH]\label{thm:heap-rel-comp-slh}(\showproof{heap-rel-comp-slh})
	\begin{align*}
		{\src{H}}\hrel\compsslh{H}
	\end{align*}
\end{lemma}

\begin{lemma}[Forward Simulation for Expressions in SLH]\label{thm:fwd-sim-exp-slh}(\showproof{fwd-sim-exp-slh})
	\begin{align*}
		\text{ if }
			&\ 
			\src{B\triangleright e \bigreds v : \sigma}
		\\
		\text{ and }
			&\
			\src{B}\brelref\trg{B}
		\\
		\text{ and }
			&\
			\src{\sigma}\equiv\taintt
		\\
		\text{ then }
			&\
			\trg{B \triangleright \compsslh{e}\bigredt\compsslh{v} : \taintt}
	\end{align*}
\end{lemma}

\begin{lemma}[Forward Simulation for Compiled Statements in SLH]\label{thm:fwd-sim-stm-slh}(\showproof{fwd-sim-stm-slh})
	\begin{align*}
		\text{ if }
			&\
			\src{\Omega}=\src{C, H, \OB{B} \triangleright \proc{s;s''}{\OB{f}} \xtos{\lambda^\sigma} C, H', \OB{B'} \triangleright \proc{s';s''}{\OB{f'}}}=\src{\Omega'}
		\\
		\text{ and }
			&\
			\src{\Omega} \ssrelref_{\src{\OB{f''}}}^{\trg{\OB{D}}} \trg{\Sigma}
		\\
		\text{ and }
			&\
			\trg{\Sigma}=\trg{w(C, H, \OB{B} \triangleright \proc{\compsslh{s};s''}{\OB{f}},\bot,\safeta) }
		\\
		\text{ and }
			&\
			\trg{\Sigma'} = \trg{w(C, H', \OB{B'} \triangleright \proc{\compsslh{s'};s''}{\OB{f'}},\bot,\safeta)}
		\\
		\text{ then }
			&\
			\trg{(n,\Sigma) \Xtot{\trat{^{\taintt}}} (n',\Sigma')}
		\\
		\text{ and }
			&\
			\src{\lambda^\sigma} \tracerel \trg{\trat{^{\taintt}}} \qquad \text{( using the trace relation!)}
		\\
		\text{ and }
			&\
			\exists {\trg{\OB{D'}}}\supseteq{\trg{\OB{D}}}.~
			\src{\Omega'} \ssrelref_{\src{\OB{f''}}}^{\trg{\OB{D'}}} \trg{\Sigma'}
	\end{align*}
\end{lemma}

\begin{theorem}[Backward Simulation for Compiled Statements in SLH]\label{thm:back-sim-stm-slh}(\showproof{back-sim-stm-slh})
	\begin{align*}
		\text{ if }
			&\
			\trg{(n,\Sigma) \Xtot{\trat{^{\taintt}}} (n',\Sigma')}
		\\
		\text{ and }
			&\
			\src{\Omega} \ssrelref_{\src{\OB{f''}}}^{\trg{\OB{D}}} \trg{\Sigma}
		\\
		\text{ and }
			&\
			\trg{\Sigma}=\trg{w(C, H, \OB{B} \triangleright \proc{\compsslh{s};s''}{\OB{f}},\bot,\safeta) }
		\\
		\text{ and }
			&\
			\trg{\Sigma'} = \trg{w(C, H', \OB{B'} \triangleright \proc{\compsslh{s'};s''}{\OB{f'}},\bot,\safeta)}
		\\
		\text{ then }
			&\
			\src{\Omega}=\src{C, H, \OB{B} \triangleright \proc{s;s''}{\OB{f}} \xtos{\lambda^\sigma} C, H', \OB{B'} \triangleright \proc{s';s''}{\OB{f'}}}=\src{\Omega'}
		\\
		\text{ and }
			&\
			\src{\lambda^\sigma} \tracerel \trg{\trat{^{\taintt}}} \qquad \text{( using the trace relation!)}
		\\
		\text{ and }
			&\
			\exists {\trg{\OB{D'}}}\supseteq{\trg{\OB{D}}}.~
			\src{\Omega'} \ssrelref_{\src{\OB{f''}}}^{\trg{\OB{D'}}} \trg{\Sigma'}
	\end{align*}
\end{theorem}

\begin{lemma}[Expression Reductions with Safe Bindings are Safe]\label{thm:exp-red-safe}(\showproof{exp-red-safe})
	\begin{align*}
		\text{ if }
			&\
			\vdash \trg{{B}} : \sbrel	
		\\
		\text{ then }
			&\
			\trg{B \triangleright e \bigredt v : \safeta}
	\end{align*}
\end{lemma}

\begin{lemma}[Any Speculation from Related States is Safe]\label{thm:spec-rel-satte-safe}(\showproof{spec-rel-satte-safe})
	\begin{align*}
		\text{ if }
			&\
			\src{\Omega} \ssrelref_{\src{\OB{f_c}}}^{\trg{\OB{D}}} \trg{\Sigma}
		\\
		\text{ and }
			&\
			\trg{\Sigma}=\trg{w (C, H_b, \OB{B_b} \triangleright \proc{s_b}{\OB{f_b}},\bot,\safeta) \cdot \OB{(C, H', \OB{B'} \triangleright \proc{s'}{\OB{f'}},\omega',\unta)} \cdot (C, H, \OB{B} \triangleright \proc{s}{\OB{f}\cdot f},\omega,\unta)}
		\\
		\text{ and }
			&\
			\trg{\Sigma'}=\trg{w (C, H_b, \OB{B_b} \triangleright \proc{s_b}{\OB{f_b}},\bot,\safeta)}
		\\
		\text{ and }
			&\
			\trg{n'}\leq \trg{n+w}
		\\
		\text{ and }
			&\
			\text{ let } \trg{\OB{\omega'}} =  \trg{\omega_1', \cdots, \omega_k'}, \trg{\omega_1'+ \cdots+ \omega_k' + \omega} \leq \trg{w}
		\\
		\text{ then }
			&\
			\trg{(n,\Sigma) \Xtot{\trat{^{\taintt}}} (n',\Sigma')}
		\\
		\text{ and }
			&\
			\srce \tracerel \trg{\trat{^{\taintt}}}
		\\
		\text{ and }
			&\
			\src{\Omega} \ssrelref_{\src{\OB{f_c}}}^{\trg{\OB{D}}} \trg{\Sigma'}
	\end{align*}
\end{lemma}

\begin{lemma}[Speculation Lasts at Most Omega]\label{thm:spec-most-omega}(\showproof{spec-most-omega})
	\begin{align*}
		\text{ if }
			&\
			\src{\Omega} \ssrelref_{\src{\OB{f_c}}}^{\trg{\OB{D}}} \trg{\Sigma}
		\\
		\text{ and }
			&\
			\trg{\Sigma}=\trg{w (C, H_b, \OB{B_b} \triangleright \proc{s_b}{\OB{f_b}},\bot,\safeta) \cdot \OB{(C, H', \OB{B'} \triangleright \proc{s'}{\OB{f'}},\omega',\unta)} \cdot (C, H, \OB{B} \triangleright \proc{s}{\OB{f}\cdot f},\omega,\unta)}
		\\
		\text{ and }
			&\
			\trg{\Sigma'}=\trg{w (C, H_b, \OB{B_b} \triangleright \proc{s_b}{\OB{f_b}},\bot,\safeta) \cdot \OB{(C, H', \OB{B'} \triangleright \proc{s'}{\OB{f'}},\omega',\unta)} \cdot (C, H'', \OB{B''} \triangleright \proc{s''}{\OB{f''}\cdot f''},0,\unta)}
		\\
		\text{ and }
			&\
			\trg{n'}\leq \trg{n+\omega}
		\\
		\text{ then }
			&\
			\trg{(n,\Sigma) \Xtot{\trat{^{\taintt}}} (n',\Sigma')}
		\\
		\text{ and }
			&\
			\srce \tracerel \trg{\trat{^{\taintt}}}
		\\
		\text{ and }
			&\
			\src{\Omega} \ssrelref_{\src{\OB{f_c}}}^{\trg{\OB{D}}} \trg{\Sigma'}
	\end{align*}
\end{lemma}

\begin{lemma}[Context Speculation Lasts at Most Omega]\label{thm:ctx-spec-most-omega}(\showproof{ctx-spec-most-omega})
	\begin{align*}
		\text{ if }
			&\
			\src{\Omega} \ssrelref_{\src{\OB{f_c}}}^{\trg{\OB{D}}} \trg{\Sigma}
		\\
		\text{ and }
			&\
			\trg{\Sigma}=\trg{w (C, H_b, \OB{B_b} \triangleright \proc{s_b}{\OB{f_b}},\bot,\safeta) \cdot \OB{(C, H', \OB{B'} \triangleright \proc{s'}{\OB{f'}},\omega',\unta)} \cdot (C, H, \OB{B} \triangleright \proc{s}{\OB{f}\cdot f},\omega,\unta)}
		\\
		\text{ and }
			&\
			\trg{\Sigma'}=\trg{w (C, H_b, \OB{B_b} \triangleright \proc{s_b}{\OB{f_b}},\bot,\safeta) \cdot \OB{(C, H', \OB{B'} \triangleright \proc{s'}{\OB{f'}},\omega',\unta)} \cdot (C, H'', \OB{B''} \triangleright \proc{s''}{\OB{f''}\cdot f''},0,\unta)}
		\\
		\text{ and }
			&\
			\trg{n'}\leq \trg{n+\omega}
		\\
		\text{ and }
			&\
			\trg{f},\trg{f''}\notin\src{\OB{f_c}}
		\\
		\text{ then }
			&\
			\trg{(n,\Sigma) \Xtot{\trat{^{\taintt}}} (n',\Sigma')}
		\\
		\text{ and }
			&\
			\srce \tracerel \trg{\trat{^{\taintt}}}
		\\
		\text{ and }
			&\
			\src{\Omega} \ssrelref_{\src{\OB{f_c}}}^{\trg{\OB{D}}} \trg{\Sigma'}
	\end{align*}
\end{lemma}

\begin{lemma}[Single Context Speculation is Safe]\label{thm:ctx-spec-sing-safe}(\showproof{ctx-spec-sing-safe})
	\begin{align*}
		\text{ if }
			&\
			\src{\Omega} \ssrelref_{\src{\OB{f_c}}}^{\trg{\OB{D}}} \trg{\Sigma}
		\\
		\text{ and }
			&\
			\trg{\Sigma}=\trg{w (C, H_b, \OB{B_b} \triangleright \proc{s_b}{\OB{f_b}},\bot,\safeta) \cdot \OB{(C, H', \OB{B'} \triangleright \proc{s'}{\OB{f'}},\omega',\unta)} \cdot (C, H, \OB{B} \triangleright \proc{s}{\OB{f}\cdot f},\omega,\unta)}
		\\
		\text{ and }
			&\
			\trg{\Sigma'}=\trg{w (C, H_b, \OB{B_b} \triangleright \proc{s_b}{\OB{f_b}},\bot,\safeta) \cdot \OB{(C, H', \OB{B'} \triangleright \proc{s'}{\OB{f'}},\omega',\unta)} \cdot (C, H'', \OB{B''} \triangleright \proc{s''}{\OB{f''}\cdot f''},\omega-1,\unta)}
		\\
		\text{ and }
			&\
			\trg{f},\trg{f''}\notin\src{\OB{f_c}}
		\\
		\text{ then }
			&\
			\trg{(\Sigma) \xltot{\acat{^{\taintt}}} (\Sigma')}
		\\
		\text{ and }
			&\
			\srce \tracerel \trg{\acat{^{\taintt}}}
		\\
		\text{ and }
			&\
			\src{\Omega} \ssrelref_{\src{\OB{f_c}}}^{\trg{\OB{D}}} \trg{\Sigma'}
	\end{align*}
\end{lemma}

\begin{lemma}[Compiled Speculation Lasts at Most Omega]\label{thm:comp-spec-most-omega}(\showproof{comp-spec-most-omega})
	\begin{align*}
		\text{ if }
			&\
			\src{\Omega} \ssrelref_{\src{\OB{f_c}}}^{\trg{\OB{D}}} \trg{\Sigma}
		\\
		\text{ and }
			&\
			\trg{\Sigma}=\trg{w (C, H_b, \OB{B_b} \triangleright \proc{s_b}{\OB{f_b}},\bot,\safeta) \cdot \OB{(C, H', \OB{B'} \triangleright \proc{s'}{\OB{f'}},\omega',\unta)} \cdot (C, H, \OB{B} \triangleright \proc{\compsslh{s}}{\OB{f}\cdot f},\omega,\unta)}
		\\
		\text{ and }
			&\
			\trg{\Sigma'}=\trg{w (C, H_b, \OB{B_b} \triangleright \proc{s_b}{\OB{f_b}},\bot,\safeta) \cdot \OB{(C, H', \OB{B'} \triangleright \proc{s'}{\OB{f'}},\omega',\unta)} \cdot (C, H'', \OB{B''} \triangleright \proc{\compsslh{s''}}{\OB{f''}\cdot f''},0,\unta)}
		\\
		\text{ and }
			&\
			\trg{n'}\leq \trg{n+\omega}
		\\
		\text{ and }
			&\
			\trg{f},\trg{f''}\in\src{\OB{f_c}}
		\\
		\text{ then }
			&\
			\trg{(n,\Sigma) \Xtot{\trat{^{\taintt}}} (n',\Sigma')}
		\\
		\text{ and }
			&\
			\srce \tracerel \trg{\trat{^{\taintt}}}
		\\
		\text{ and }
			&\
			\src{\Omega} \ssrelref_{\src{\OB{f_c}}}^{\trg{\OB{D}}} \trg{\Sigma'}
	\end{align*}
\end{lemma}

\begin{lemma}[Compiled Speculation is Safe]\label{thm:comp-spec-safe}(\showproof{comp-spec-safe})
	\begin{align*}
		\text{ if }
			&\
			\src{\Omega} \ssrelref_{\src{\OB{f_c}}}^{\trg{\OB{D}}} \trg{\Sigma}
		\\
		\text{ and }
			&\
			\trg{\Sigma}=\trg{w (C, H_b, \OB{B_b} \triangleright \proc{s_b}{\OB{f_b}},\bot,\safeta) \cdot \OB{(C, H', \OB{B'} \triangleright \proc{s'}{\OB{f'}},\omega',\unta)} \cdot (C, H, \OB{B} \triangleright \proc{\compsslh{s}}{\OB{f}\cdot f},\omega,\unta)}
		\\
		\text{ and }
			&\
			\trg{\Sigma'}=\trg{w (C, H_b, \OB{B_b} \triangleright \proc{s_b}{\OB{f_b}},\bot,\safeta) \cdot \OB{(C, H', \OB{B'} \triangleright \proc{s'}{\OB{f'}},\omega',\unta)} \cdot (C, H'', \OB{B''} \triangleright \proc{\compsslh{s''}}{\OB{f''}\cdot f''},\omega'',\unta)}
		\\
		\text{ and }
			&\
			\trg{f},\trg{f''}\in\src{\OB{f_c}}
		\\
		\text{ then }
			&\
			\trg{(n,\Sigma) \Xtot{\trat{^{\taintt}}} (n',\Sigma')}
		\\
		\text{ and }
			&\
			\srce \tracerel \trg{\trat{^{\taintt}}}
		\\
		\text{ and either }
			&\
			\src{\Omega} \ssrelref_{\src{\OB{f_c}}}^{\trg{\OB{D}}} \trg{\Sigma'}
		\\
		\text{ or }
			&\
			\trg{\omega''} = \trg{0}
	\end{align*}
\end{lemma}

\begin{theorem}[Correctness of the Backtranslation for SLH]\label{thm:corr-bt-slh}(\showproof{corr-bt-slh})
	\begin{align*}
		\text{ if } 
			&\
			\trg{\ctxt{}\hole{\compsslh{P}} \sem \trat{^{\taintt}}}
		\\
		\text{ then }
			&\
			\src{\backtrfencec{\ctxt{}}\hole{P} \sem \tras{^{\sigma}}}
		\\
		\text{ and }
			&\
			\tras{^{\sigma}} \rels \trat{^{\taintt}}
	\end{align*}
\end{theorem}

\begin{theorem}[Generalised Backward Simulation for SLH]\label{thm:bwd-sim-slh}(\showproof{bwd-sim-slh})
	\begin{align*}
		\text{ if }
			&\
			\text{ if } \src{f}\in\src{\OB{f''}} \text{ then } \trg{s} = \compsslh{s} \text{ else } \src{s} = \backtrfencec{\trg{s}}
			\\
		\text{ and}
			&\
			\text{ if } \src{f'}\in\src{\OB{f''}} \text{ then } \trg{s'} = \compsslh{s'} \text{ else } \src{s'} = \backtrfencec{\trg{s'}}
		\\
		\text{ and }
			&\
			\trg{\Sigma}=\trg{w( C, H, \OB{B} \triangleright \proc{s;s''}{\OB{f}\cdot f}, \bot,\safeta)}
		\\
		\text{ and }
			&\
			\trg{\Sigma'}=\trg{ w(C, H', \OB{B'} \triangleright \proc{s';s''}{\OB{f'}\cdot f'},\bot,\safeta)}
		\\
		\text{ and }
			&\
			\trg{(n,\Sigma) \Xtot{\trat{^{\taintt}}} (n',\Sigma')}
		\\
		\text{ and }
			&\
			\src{\Omega} \srel_{\src{\OB{f''}}}^{\trg{\OB{D}}} \trg{\Sigma}
		\\
		\text{ then }
			&\
			\src{\Omega}=\src{C, H, \OB{B} \triangleright \proc{s;s''}{\OB{f}\cdot f} \Xtos{\tras{^\sigma}} C, H', \OB{B'} \triangleright \proc{s';s''}{\OB{f'}\cdot f'}}=\src{\Omega'}
		\\
		\text{ and }
			&\
			\src{\tras{^\sigma}} \tracerel \trg{\trat{^{\taintt}}}
		\\
		\text{ and }
			&\
			\exists {\trg{\OB{D'}}}\supseteq{\trg{\OB{D}}}.~
			\src{\Omega'} \srelref_{\src{\OB{f''}}}^{\trg{\OB{D'}}} \trg{\Sigma'}
	\end{align*}
\end{theorem}

\begin{lemma}[Compiled Speculation is Safe Inter-procedurally]\label{thm:comp-spec-safe-proc}(\showproof{comp-spec-safe-proc})
	\begin{align*}
		\text{ if }
			&\
			\src{\Omega} \cssrelref_{\src{\OB{f_c}}}^{\trg{\OB{D}}} \trg{\Sigma}
		\\
		\text{ and }
			&\
			\trg{\Sigma}=\trg{w (C, H_b, \OB{B_b} \triangleright \proc{s_b}{\OB{f_b}},\bot,\safeta) \cdot \OB{(C, H', \OB{B'} \triangleright \proc{s'}{\OB{f'}},\omega',\unta)} \cdot (C, H, \OB{B} \triangleright \proc{\compslht{s}}{\OB{f}\cdot f},\omega,\unta)}
		\\
		\text{ and }
			&\
			\trg{\Sigma'}=\trg{w (C, H_b, \OB{B_b} \triangleright \proc{s_b}{\OB{f_b}},\bot,\safeta) \cdot \OB{(C, H', \OB{B'} \triangleright \proc{s'}{\OB{f'}},\omega',\unta)} \cdot (C, H'', \OB{B''} \triangleright \proc{\compslht{s''}}{\OB{f''}\cdot f''},\omega'',\unta)}
		\\
		\text{ and }
			&\
			\trg{f},\trg{f''}\in\src{\OB{f_c}}
		\\
		\text{ then }
			&\
			\trg{(n,\Sigma) \Xtot{\trat{^{\taintt}}} (n',\Sigma')}
		\\
		\text{ and }
			&\
			\srce \tracerel \trg{\trat{^{\taintt}}}
		\\
		\text{ and either }
			&\
			\src{\Omega} \cssrelref_{\src{\OB{f_c}}}^{\trg{\OB{D}}} \trg{\Sigma'}
		\\
		\text{ or }
			&\
			\trg{\omega''} = \trg{0}
	\end{align*}
\end{lemma}

\newpage
\section{A Complete Insight on the Proofs}\label{sec:full-insight}

This section describes how to prove that compiler countermeasures are secure (i.e, \rdss), starting with \complfence{\cdot}.

To prove that a compiler is \rdss we need to backtranslate target attackers \ctxt{} into source ones \ctxs{}.
Our setup has very similar source and target languages to enable this kind of backtranslation; the alternative would have been to rely on the target trace in order to build \ctxs{}.
In this case, the backtranslation function (\backtr{\cdot}) homomorphically translates target heaps, functions, statements etc$\ldotp$ into source ones (see \Cref{sec:bt-fence}).
Since our compilers are devised for essentially the same languages (or at least, languages with the same notions of attacker), we can define a single backtranslation to use for the security of all compilers we define.

To prove the compiler is \rssc we need to show that given a trace produced by the execution of attacker and compiled code, the backtranslated attacker and the source code produce a related trace (according to the trace relation of \Cref{sec:compcrit}).
This can be broken down in a sequence of canonical steps:
\begin{itemize}
 	\item we first set up a cross-language relation between source and target states;
 	\item we prove that initial states are related;
 	\item we prove that reductions preserve this relation and generate related traces.

 		We reason about reductions depending on whether the pc that is triggering the reduction is in attacker or in program code.
 \end{itemize}
The state relation we use is strong: a source state is related to a target state if the latter is a singleton stack and all the sub-part of the state are identical.
That is, the target state must not be speculating (starting speculation adds a state to the stack, which is not a singleton in that case), the heaps bind the same locations to the same values, the bindings bind the same variables to the same values.

\myfig{
	\centering
	\tikzpic{
		\node[draw=\stlccol,rounded corners,font=\footnotesize](sc1){};
		\node[draw=\stlccol,rounded corners,font=\footnotesize, right = .5 of sc1](sc2){};
		\node[draw=\stlccol,rounded corners,font=\footnotesize, right = .7 of sc2](sc3){};
		\node[draw=\stlccol,rounded corners,font=\footnotesize, right = .5 of sc3](sc4-if){\src{ifz}};

		\node[draw=\ulccol,rounded corners,font=\footnotesize, below = .8 of sc1](tc1){};
		\node[draw=\ulccol,rounded corners,font=\footnotesize,] at(sc2|-tc1) (tc2){};
		\node[draw=\ulccol,rounded corners,font=\footnotesize,] at(sc3|-tc1) (tc3){};
		\node[draw=\ulccol,rounded corners,font=\footnotesize,] at(sc4-if|-tc1) (tc4-if){\complfence{\src{ifz}}};

		\node[draw=\ulccol,dashed,rounded corners,font=\footnotesize, below= .5 of tc4-if](ts1){\trg{\lfence}};
		\node[draw=\ulccol,dashed,rounded corners,font=\footnotesize, right = .5 of ts1](ts2){\trg{w}=\trg{0}};		

		\node[draw=\ulccol,rounded corners,font=\footnotesize,] at(ts2|-tc1) (tc5){};
		\node[draw=\ulccol,rounded corners,font=\footnotesize, right = .5 of tc5](tc6){};
		\node[draw=\ulccol,rounded corners,font=\footnotesize, right = .7 of tc6](tc7){};
		\node[draw=\ulccol,rounded corners,font=\footnotesize, right = .5 of tc7](tc8){};

		\node[draw=\stlccol,rounded corners,font=\footnotesize,] at(tc5|-sc1) (sc5){};
		\node[draw=\stlccol,rounded corners,font=\footnotesize,] at(tc6|-sc1) (sc6){};
		\node[draw=\stlccol,rounded corners,font=\footnotesize,] at(tc7|-sc1) (sc7){};
		\node[draw=\stlccol,rounded corners,font=\footnotesize,] at(tc8|-sc1) (sc8){};

		\draw[dashed, draw = \stlccol] (sc1) to (sc2);
		\draw[dashed, draw = \stlccol, bend left] (sc2) to node[above,font=\footnotesize](as1){\src{\acas{?^\sigma}}} (sc3);
		\draw[dashed, draw = \stlccol] (sc3) to (sc4-if);
		\draw[->, draw = \stlccol] (sc4-if) to (sc5);
		\draw[dashed, draw = \stlccol] (sc5) to (sc6);
		\draw[dashed, draw = \stlccol, bend left] (sc6) to node[above,font=\footnotesize](as2){\src{\acas{!^\sigma}}} (sc7);
		\draw[dashed, draw = \stlccol] (sc7) to (sc8);
		\draw[draw = \stlccol, - ] ([yshift=.3em]sc3.north) to ([yshift=.6em]sc3.north) to node[above,font=\footnotesize](sd1){\src{\OB{\delta^\sigma}}} ([yshift=.6em]sc6.north) to ([yshift=.3em]sc6.north);

		\draw[dashed, draw = \ulccol] (tc1) to (tc2);
		\draw[dashed, draw = \ulccol, bend left] (tc2) to node[above,font=\footnotesize](at1){\trg{\acat{?^{\taintt}}}} (tc3);
		\draw[dashed, draw = \ulccol] (tc3) to (tc4-if);
		\draw[dashed, draw = \ulccol] (tc5) to (tc6);
		\draw[dashed, draw = \ulccol, bend left] (tc6) to node[above,font=\footnotesize](at2){\src{\acat{!^{\taintt}}}} (tc7);		
		\draw[dashed, draw = \ulccol] (tc7) to (tc8);
		\draw[draw = \ulccol, - ] ([yshift=.3em]tc3.north) to ([yshift=.6em]tc3.north) to node[above,font=\footnotesize](td1){\trg{\OB{\trgb{\delta}^{\taintt}}}} ([yshift=.6em]tc6.north) to ([yshift=.3em]tc6.north);

		\draw[draw = \ulccol, decoration={zigzag, segment length=4, amplitude=.9, post=lineto, post length=2pt},decorate,->] (tc4-if) to (ts1);
		\draw[draw = \ulccol, decoration={zigzag, segment length=4, amplitude=.9, post=lineto, post length=2pt},decorate,->] (ts1) to (ts2);
		\draw[draw = \ulccol, decoration={zigzag, segment length=4, amplitude=.9, post=lineto, post length=2pt},decorate,->] (ts2) to  node[right,font=\footnotesize](ar){\trg{\rollbl}} (tc5);

		\draw[rounded corners, dotted, fill=yellow, opacity = .2 ] (as1.north) -| ([xshift=.2em]sc3.east) |- ([yshift=-.2em]tc2.south) -| ([xshift=-.2em]sc1.west) |- (as1.north);
		\draw[rounded corners, dotted, fill=yellow, opacity = .2 ] ([yshift=.2em]sc7.north) -| ([xshift=.2em]sc8.east) |- ([yshift=-.2em]tc8.south) -| ([xshift=-.2em]sc7.west) |- ([yshift=.2em]sc7.north);

		\draw[rounded corners, dotted, fill=green, opacity = .1 ] (sd1.north) -| ([xshift=.2em]sc7.east) |- ([yshift=-.2em]ts1.south) -| ([xshift=-.2em]sc3.west) |- (sd1.north);

		\draw[rounded corners, dotted, fill=black, opacity = .1 ] (tc4-if.south) -- (tc4-if.south east) -- (tc5.south) -| ([xshift=.2em]ar.east) |- ([yshift=-.2em]ts1.south) -| ([xshift=-.2em]ts1.north west) |- (tc4-if.south);		

		\draw[-] ([yshift=3em]sc3.west) -- ([yshift=-4em]tc3.west);
		\draw[-] ([yshift=3em]sc7.east) -- ([yshift=-4em]tc7.east);

		\node[font = \footnotesize,above = .15 of sd1.center](y1){\src{P} / \complfence{P} executes};
		\node[font = \footnotesize, left = of y1, align = center](yl){\backtr{\ctxt{}} / \ctxt{} \\ executes};
		\node[font = \footnotesize,right = 1.5 of y1, align = center](yr){\backtr{\ctxt{}} / \ctxt{} \\ executes};

		\draw[-,thin, dotted] (sc1) to (tc1);
		\draw[-,thin, dotted] (sc2) to (tc2);
		\draw[-,thin, dotted] (sc3) to (tc3);
		\draw[-,thin, dotted] (sc4-if) to (tc4-if);
		\draw[-,thin, dotted] (sc5) to (tc5);
		\draw[-,thin, dotted] (sc6) to (tc6);
		\draw[-,thin, dotted] (sc7) to (tc7);
		\draw[-,thin, dotted] (sc8) to (tc8);
		\draw[-,thin, dotted] (as1) to (at1);
		\draw[-,thin, dotted] (as2) to (at2);
		\draw[-,thin, dotted] (sd1) to (td1);
	}
}{proof-lfence}{A diagram depicting the proof that \complfence{\cdot} is \rssc.}
We depict our proof approach for \complfence{\cdot} in \Cref{fig:proof-lfence}.
The top half of the picture represents the source reductions: source states (circles) perform multiple steps (dashed lines) and reduce producing traces (annotations on reductions).
We highlight the single reduction caused by the execution of an \src{ifz} statement since that has relevance in the target language.

The bottom half of the picture represents the target states and their reductions, here we have an additional kind of program states: dashed ones.
Intuitively, these states are not related to any source state, while other states are.

This is symbolised by black dotted connections between source and target states.
The same kind of connection between source and target actions indicates that these actions are related.

In our setup, execution either happens with the pc in attacker code or in component code.
We now describe how to reason about these reductions.

To reason about attacker code, we use a lock-step simulation: we show that starting from related states, if \ctxt{} does a step, then its backtranslation \backtr{\ctxt{}} does the same step and ends up in related states.
This is what happens in the yellow areas in the picture.

To reason about component code, we adapt a reasoning commonly used in compiler correctness results~\cite{Leroy09b,barthe-ct2}.
That is, if \src{s} steps and emits a trace, then \complfence{s} does one or more steps and emits a trace such that
\begin{itemize}
	\item the ending states are related;
	\item the emitted source and target traces are related;
\end{itemize}
This is what happens in the green area, recall that related traces are connected by black-dotted lines.

The only place where this proof is not straightforward is the case for compilation of \src{ifz} i.e., the statement that triggers speculation in \TR.
This is what happens in the grey area.
When observing target-level executions for \complfence{ifz}, we see that the cross-language state relation is temporarily broken.
After the \complfence{ifz} is executed, speculation starts, so the stack of target states is not a singleton and therefore the cross-language state relation cannot hold.
However, if we unfold the reductions, we see that compiled code immediately triggers an \trg{\lfence}, which rolls the speculation back (the speculation window \trg{w} is \trg{0}) reinstating the cross-language state relation.
Thus, for the case of \src{ifz} to go through, we see that the target effectively does more steps than the source (it starts speculation, it executes the \trg{\lfence} and then it rolls speculation back) before ending up in a state related to the source one.

\myfig{
	\centering
	\tikzpic{
		\node[draw=\stlccol,rounded corners,font=\footnotesize](sc1){};
		\node[draw=\stlccol,rounded corners,font=\footnotesize, right = .5 of sc1](sc2){};
		\node[draw=\stlccol,rounded corners,font=\footnotesize, right = .7 of sc2](sc3){};
		\node[draw=\stlccol,rounded corners,font=\footnotesize, right = .5 of sc3](sc4-if){\src{ifz}};

		\node[draw=\ulccol,rounded corners,font=\footnotesize, below = .8 of sc1](tc1){};
		\node[draw=\ulccol,rounded corners,font=\footnotesize,] at(sc2|-tc1) (tc2){};
		\node[draw=\ulccol,rounded corners,font=\footnotesize,] at(sc3|-tc1) (tc3){};
		\node[draw=\ulccol,rounded corners,font=\footnotesize,] at(sc4-if|-tc1) (tc4-if){\compslh{\src{ifz}}};

		\node[draw=\ulccol,dashed,rounded corners,font=\footnotesize, below= .5 of tc4-if](ts1){};
		\node[draw=\ulccol,dashed,rounded corners,font=\footnotesize, right = .5 of ts1](ts3){};
		\node[draw=\ulccol,dashed,rounded corners,font=\footnotesize, right = .5 of ts3](ts4){};
		\node[draw=\ulccol,dashed,rounded corners,font=\footnotesize, right = .5 of ts4](ts2){\trg{w}=\trg{0}};		
		\node[font=\footnotesize, left = .5 of ts1](ph){};

		\node[draw=\ulccol,rounded corners,font=\footnotesize,] at(ts2|-tc1) (tc5){};
		\node[draw=\ulccol,rounded corners,font=\footnotesize, right = .5 of tc5](tc6){};
		\node[draw=\ulccol,rounded corners,font=\footnotesize, right = .7 of tc6](tc7){};
		\node[draw=\ulccol,rounded corners,font=\footnotesize, right = .5 of tc7](tc8){};

		\node[draw=\stlccol,rounded corners,font=\footnotesize,] at(tc5|-sc1) (sc5){};
		\node[draw=\stlccol,rounded corners,font=\footnotesize,] at(tc6|-sc1) (sc6){};
		\node[draw=\stlccol,rounded corners,font=\footnotesize,] at(tc7|-sc1) (sc7){};
		\node[draw=\stlccol,rounded corners,font=\footnotesize,] at(tc8|-sc1) (sc8){};
		\node[font=\footnotesize, ] at (ph -| sc6)(ph2){};

		\draw[dashed, draw = \stlccol] (sc1) to (sc2);
		\draw[dashed, draw = \stlccol, bend left] (sc2) to node[above,font=\footnotesize](as1){\src{\acas{?^\sigma}}} (sc3);
		\draw[dashed, draw = \stlccol] (sc3) to (sc4-if);
		\draw[->, draw = \stlccol] (sc4-if) to node[](sd0){} (sc5);
		\draw[dashed, draw = \stlccol] (sc5) to (sc6);
		\draw[dashed, draw = \stlccol, bend left] (sc6) to node[above,font=\footnotesize](as2){\src{\acas{!^\sigma}}} (sc7);
		\draw[dashed, draw = \stlccol] (sc7) to (sc8);
		\draw[draw = \stlccol, - ] ([yshift=.6em]sc3.center) to ([yshift=.9em]sc3.center) to node[above,font=\footnotesize](sd1){\src{\OB{\delta_1^\sigma}}} ([yshift=.9em]sc4-if.center) to ([yshift=.6em]sc4-if.center);
		\draw[draw = \stlccol, - ] ([yshift=.3em]sc5.north) to ([yshift=.6em]sc5.north) to node[above,font=\footnotesize](sd2){\src{\OB{\delta_2^\sigma}}} ([yshift=.6em]sc6.north) to ([yshift=.3em]sc6.north);

		\draw[dashed, draw = \ulccol] (tc1) to (tc2);
		\draw[dashed, draw = \ulccol, bend left] (tc2) to node[above,font=\footnotesize](at1){\trg{\acat{?^{\taintt}}}} (tc3);
		\draw[dashed, draw = \ulccol] (tc3) to (tc4-if);
		\draw[dashed, draw = \ulccol] (tc5) to (tc6);
		\draw[dashed, draw = \ulccol, bend left] (tc6) to node[above,font=\footnotesize](at2){\src{\acat{!^{\taintt}}}} (tc7);		
		\draw[dashed, draw = \ulccol] (tc7) to (tc8);
		\draw[draw = \ulccol, - ] ([yshift=.6em]tc3.center) to ([yshift=.9em]tc3.center) to node[above,font=\footnotesize](td1){\trg{\OB{\trgb{\delta}_1^{\taintt}}}} ([yshift=.9em]tc4-if.center) to ([yshift=.6em]tc4-if.center);
		\draw[draw = \ulccol, - ] ([yshift=.3em]tc5.north) to ([yshift=.6em]tc5.north) to node[above,font=\footnotesize](td2){\trg{\OB{\trgb{\delta}_2^{\taintt}}}} ([yshift=.6em]tc6.north) to ([yshift=.3em]tc6.north);

		\draw[draw = \ulccol, - ] ([yshift=-.6em]ts1.center) to ([yshift=-.9em]ts1.center) to node[below,font=\footnotesize](td0){\trat{^{\taintt}}} ([yshift=-.9em]ts2.center) to ([yshift=-.6em]ts2.center);

		\draw[draw = \ulccol, decoration={zigzag, segment length=4, amplitude=.9, post=lineto, post length=2pt},decorate,->] (tc4-if) to (ts1);
		\draw[draw = \ulccol, decoration={zigzag, segment length=4, amplitude=.9, post=lineto, post length=2pt},decorate,->] (ts1) to (ts3);
		\draw[draw = \ulccol, loosely dotted, thick ] (ts3) to (ts4);
		\draw[draw = \ulccol, decoration={zigzag, segment length=4, amplitude=.9, post=lineto, post length=2pt},decorate,->] (ts4) to (ts2);
		\draw[draw = \ulccol, decoration={zigzag, segment length=4, amplitude=.9, post=lineto, post length=2pt},decorate,->] (ts2) to  node[pos=.4,left,font=\footnotesize](ar){\trg{\rollbl}} (tc5);

		\draw[rounded corners, dotted, fill=yellow, opacity = .2 ] (as1.north) -| ([xshift=.2em]sc3.east) |- ([yshift=-.2em]tc2.south) -| ([xshift=-.2em]sc1.west) |- (as1.north);
		\draw[rounded corners, dotted, fill=yellow, opacity = .2 ] ([yshift=.2em]sc7.north) -| ([xshift=.2em]sc8.east) |- ([yshift=-.2em]tc8.south) -| ([xshift=-.2em]sc7.west) |- ([yshift=.2em]sc7.north);

		\draw[rounded corners, dotted, fill=green, opacity = .1 ] (sd1.north) -| ([xshift=.2em]sc7.east) |- ([yshift=-.2em]tc4-if.south) -| ([xshift=-.2em]sc3.west) |- (sd1.north);

		\draw[rounded corners, dotted, fill=black, opacity = .1 ] (tc4-if.south) -| ([xshift=.2em]ts2.east) |- (td0.south) -| ([xshift=-.2em]ts1.north west) |- (tc4-if.south);		

		\draw[-] ([yshift=3em]sc3.west) -- ([yshift=-1em]tc3.west);
		\draw[-] ([yshift=3em]sc7.east) -- ([yshift=-1em]tc7.east);
		\draw[-] ([yshift=-.9em,xshift=-4.5em]tc3.west) -- ([yshift=-.9em,xshift=2em]tc7.east);

		\node[font = \footnotesize,above = .5 of sd0.north](y1){\src{P} / \compslh{P} executes};
		\node[font = \footnotesize, left = 1 of y1, align = center](yl){\backtr{\ctxt{}} / \ctxt{} \\ executes};
		\node[font = \footnotesize,right = 2 of y1, align = center](yr){\backtr{\ctxt{}} / \ctxt{} \\ executes};

		\node[font = \footnotesize, align = center] at (as1 |- ts2) (yx){either \ctxt{} \\ or \compslh{P} \\ executes};

		\draw[-,thin, dotted] (sc1) to (tc1);
		\draw[-,thin, dotted] (sc2) to (tc2);
		\draw[-,thin, dotted] (sc3) to (tc3);
		\draw[-,thin, dotted] (sc4-if) to (tc4-if);
		\draw[-,thin, dotted] (sc5) to (tc5);
		\draw[-,thin, dotted] (sc6) to (tc6);
		\draw[-,thin, dotted] (sc7) to (tc7);
		\draw[-,thin, dotted] (sc8) to (tc8);
		\draw[-,thin, dotted] (as1) to (at1);
		\draw[-,thin, dotted] (as2) to (at2);
		\draw[-,thin, dotted] (sd1) to (td1);
		\draw[-,thin, dotted] (sd2) to (td2);
	}
	\vspace{-.5em}
}{proof-slh2}{A diagram depicting the proof that \compslh{\cdot}, \compsslh{\cdot} and, \compslht{\cdot} are \protect\rssc.}
This is the part where the proofs of the SLH-related countermeasures gets more complicated (\Cref{fig:proof-slh2}), though the general structure remains unchanged.
In compiled code speculation is not rolled back immediately after the \trg{then} or \trg{else} branch start executing.
Instead, execution can continue for \trgb{\omega} steps, spanning both attacker and compiled code and generating a trace \trat{^{\taintt}}.
Our proof here relies on an auxiliary lemma stating the following:
\begin{itemize}
	\item in the target, speculation lasts at most \trgb{\omega} steps and then it will be rolled back;
	\item after the rollback, the strong state relation we need is reinstated;
	\item during this speculation any trace produced in the target is related to the empty source trace.

	This is needed because such a relation is only possible when all actions in the target trace \trat{^{\taintt}} are tainted \trg{\safeta}: i.e., they do not leak.

\end{itemize}
Ensuring that all target actions are \trg{S} is achieved through declaring a property on target (speculating) states and prove that any speculating transition \emph{preserves} that property.
Specifically, the property is that the bindings always contain \trg{\safeta} values.
From this property we can easily see that any generated action is \trg{\safeta}.
To prove that this property holds right after speculation, we need 
\begin{itemize}
 	\item \trg{\predState} correctly captures whether speculation is ongoing or not;
 	\item the mask used by the compiler taints the variable it is applied to as \trg{\safeta}.
 \end{itemize} 
As already shown, both conditions hold for \compslh{\cdot}, \compsslh{\cdot} and, \compslht{\cdot}, so we can conclude that they are \rssc.

\subsection{Failing \rssc Proofs }\label{sec:proofexists}
When a countermeasure is not \rssc we can use the insights of its failed proof to understand whether it is also not \rsnip.
In fact, while MSVC was already known to be insecure, this was not true for SLH.
When we modelled vanilla SLH and started proving \rssc, the proof broke in the ``gray area''.
While this does not directly mean that SLH is insecure, the way the proof broke provided insights on the insecurity of SLH.
Concretely, we were not able to show that the property on speculating target states holds when speculating reductions are done and this led to \Cref{ex:slh-nonint-insec,example:clang:slh:insecure}. %
We believe the insights of this proof technique can guide proofs of (in)security of other countermeasures too.
 
\newpage
\section{Proofs for Countermeasures and Criteria}

\begin{proof}[Proof of \Thmref{thm:rdss-impl-rdsp}]\proofref{}{rdss-impl-rdsp}\hfill
	
	We have (HPSSC)
	\begin{align*}
		\vdash \comp{\cdot} : \rdss \isdef&\
			\forall\src{P}, \ctxt{}, \trat{^{\taintt}},
			\exists \ctxs{}, \tras{^\sigma} \ldotp 
			\\
			\text{ if }&\
			\trg{\ctxt{}\hole{\comp{\src{P}}}} \semt \trat{^{\taintt}}
			\text{ then }
			\src{\ctxs{}\hole{P}} \sems \tras{^\sigma}
			\text{ and }
			\tras{^\sigma}\rels\trat{^{\taintt}}
	\end{align*}
	We need to prove:
	\begin{align*}
		\forall\src{P}\ldotp 
		\text{ if }&\ \forall \ctxs{}\ldotp \forall \tras{^\sigma}\in\behavs{\ctxs{}\hole{P}}\ldotp \forall \acas{^\sigma}\in\tras{^\sigma}\ldotp \src{\sigma}\equiv\src{\safeta}
		\\
		\text{ then }&\ \forall \ctxt{}\ldotp \forall \trat{^{\taintt}}\in\behavt{\ctxt{}\hole{\comp{\src{P}}}}\ldotp \forall \acat{^{\taintt}}\in\trat{^{\taintt}}\ldotp \taintt\equiv\trg{\safeta}
	\end{align*}

	We proceed by contradiction:
	\begin{align*}
		&\ \forall \ctxs{}\ldotp \forall \tras{^\sigma}\in\behavs{\ctxs{}\hole{P}}\ldotp \forall \acas{^\sigma}\in\tras{^\sigma}\ldotp \src{\sigma}\equiv\src{\safeta} \text{ (HPS) }
		\\
		\text{ and }&\ \exists \ctxt{}\ldotp \exists \trat{^{\taintt}}\in\behavt{\ctxt{}\hole{\comp{\src{P}}}}\ldotp \exists \acat{^{\taintt}}\in\trat{^{\taintt}} \text{ (HPT)}
		\\
		\text{ and }&\ \taintt\equiv\trg{\unta} \text{ HPU }
	\end{align*}	

	We instantiate HPSSC with HPS and HPT so we get that $\tras{^\sigma}\relref\trat{^{\taintt}}$.

	By definition of $\relref$ we have two main cases:
	\begin{itemize}
		\item source and target actions are the same \Cref{tr:tr-rel-same,tr:tr-rel-same-h}.

			By \Cref{tr:ac-rel-cl,tr:ac-rel-cb,tr:ac-rel-rt,tr:ac-rel-rb,tr:ac-rel-rd,tr:ac-rel-wr} and by \Cref{thm:all-s-rdss}, we conclude that all target taint is \trg{\safeta}, which contradicts (HPU);
		
		\item there is no source action and a single target one \Cref{tr:tr-rel-safe-h,tr:tr-rel-safe-a,tr:tr-rel-rollb} 

			By \Cref{tr:ac-rel-ep-al,tr:ac-rel-ep-hp,tr:ac-rel-rlb}, the target taint is \trg{\safeta}, which contradicts (HPU)
	\end{itemize}

	Having found a contradiction in all cases, this theorem holds.
\end{proof}

\BREAK

\begin{proof}[Proof of \Thmref{thm:rdssp-impl-rdss}]\proofref{}{rdssp-impl-rdss}\hfill
	
	We have (RSSP)
	\begin{align*}
		\forall\src{P}\ldotp 
		\text{ if }&\ \forall \ctxs{}\ldotp \forall \tras{^\sigma}\in\behavs{\ctxs{}\hole{P}}\ldotp \forall \acas{^\sigma}\in\tras{^\sigma}\ldotp \src{\sigma}\equiv\src{\safeta}
		\\
		\text{ then }&\ \forall \ctxt{}\ldotp \forall \trat{^{\taintt}}\in\behavt{\ctxt{}\hole{\comp{\src{P}}}}\ldotp \forall \acat{^{\taintt}}\in\trat{^{\taintt}}\ldotp \taintt\equiv\trg{\safeta}
	\end{align*}
	We need to prove:
	\begin{align*}
		&\
			\forall\src{P}, \ctxt{}, \trat{^{\taintt}},
			\exists \ctxs{}, \tras{^\sigma} \ldotp 
			\\
			\text{ if }&\
			\trg{\ctxt{}\hole{\comp{\src{P}}}} \semt \trat{^{\taintt}}
			\text{ then }
			\src{\ctxs{}\hole{P}} \sems \tras{^\sigma}
			\text{ and }
			\tras{^\sigma}\rels\trat{^{\taintt}}
	\end{align*}

	We proceed by contradiction, assuming that the target reduction is not related to the source one:
	\begin{align*}
			&\
			\exists \ctxt{}, \trat{^{\taintt}},
			\forall \ctxs{}, \tras{^\sigma} \ldotp 
			\\
			&\
			\trg{\ctxt{}\hole{\comp{\src{P}}}} \semt \trat{^{\taintt}}
			\text{ and }
			\src{\ctxs{}\hole{P}} \sems \tras{^\sigma} \text{ (HPSR) }
			\text{ and }
			\tras{^\sigma}\nrels\trat{^{\taintt}} \text{ (HPA)}
	\end{align*}
	We analyse HPA.

	By $\relref$ we determine when the two traces are not related.

	By the universal quantification over \ctxs{} rules out all cases when there are trivial mismatches: a source call/ret and a target ret/call, calls to two different functions, a source write/read and a target read/write.

	So we are left with these cases:
	\begin{itemize}
		\item a target call matched by a source call with unrelated argument %

			This contradicts \Cref{tr:ac-rel-cl,tr:ac-rel-cb}.			
		\item a target call matched by a source call with related arguments but with different taint.

			In this case, since all source taints are \src{\safeta}, we conclude that we have a target action that is tagged as \trg{\unta}.
		\item a target return matched by a source return with unrelated heaps

			This contradicts \Cref{tr:ac-rel-rt,tr:ac-rel-rb}.
		\item a target read/write matched by no action in the source
			
			This contradicts \Cref{tr:ac-rel-ep-hp,tr:ac-rel-ep-al}.
		\item a target read/write matched by a source read/write to a different location.

			This contradicts \Cref{tr:ac-rel-rd,tr:ac-rel-wr}.
		\item a target read/write matched by a source read/write to the same location but with different taint.

			In this case, since all source taints are \src{\safeta}, we conclude that we have a target action that is tagged as \trg{\unta}.
	\end{itemize}
	We can therefore conclude that $\exists \ctxt{}\ldotp \exists \trat{^{\taintt}}\in\behavt{\ctxt{}\hole{\comp{\src{P}}}}\ldotp \exists \acat{^{\taintt}}\in\trat{^{\taintt}}\ldotp \taintt\equiv\trg{\unta}$ (HPU).

	We instantiate RSSP with HPSR and conclude 
	\begin{align*}
		&\ \forall \ctxt{}\ldotp \forall \trat{^{\taintt}}\in\behavt{\ctxt{}\hole{\comp{\src{P}}}}\ldotp 
		\\
		&\ \forall \acat{^{\taintt}}\in\trat{^{\taintt}}\ldotp \taintt\equiv\trg{\safeta} \text{ (HPS)}
	\end{align*}
	We obtain the contradiction between HPU and HPS.
\end{proof}

\BREAK

\begin{proof}[Proof of \Thmref{thm:rdss-eq-rdsp}]\proofref{}{rdss-eq-rdsp}\hfill
	
	By \Thmref{thm:rdss-impl-rdsp} and \Thmref{thm:rdssp-impl-rdss}.
\end{proof}

\BREAK

\begin{proof}[Proof of \Thmref{thm:rdss-impl-rdsp-weak}]\proofref{}{rdss-impl-rdsp-weak}\hfill
	
	Trivial adaptation of \Thmref{thm:rdss-impl-rdsp}.
\end{proof}

\BREAK

\begin{proof}[Proof of \Thmref{thm:rdssp-impl-rdss-weak}]\proofref{}{rdssp-impl-rdss-weak}\hfill
	
	Trivial adaptation of \Thmref{thm:rdssp-impl-rdss}.
\end{proof}

\BREAK

\begin{proof}[Proof of \Thmref{thm:rdss-eq-rdsp-weak}]\proofref{}{rdss-eq-rdsp-weak}\hfill
	
	By \Thmref{thm:rdss-impl-rdsp-weak} and \Thmref{thm:rdssp-impl-rdss-weak}.
\end{proof}

\BREAK

\begin{proof}[Proof of \Thmref{thm:lfence-comp-rdss}]\proofref{}{lfence-comp-rdss}\hfill

	Instantiate \ctxs{} with \backtrfencec{\ctxt{}}.

	This holds by \Thmref{thm:corr-bt-lfence}.
\end{proof}

\BREAK

\begin{proof}[Proof of \Thmref{thm:all-lfence-comp-are-rdss}]\proofref{}{19}\hfill

	By \Cref{thm:all-s-rdss} we have HPS: $\forall\src{P}\ldotp \vdash\src{P}:\rss$.

	By \Cref{thm:lfence-comp-rdss} we have HPC: $\vdash\complfence{\cdot}:\rdss$.

	By \Cref{thm:rdss-impl-rdsp} with HPC we have HPP: $\vdash\complfence{\cdot}:\rdssp$.

	By \Thmref{def:rdssp} of HPP with HPS we conclude that $\forall\src{P}\ldotp \vdash\complfence{P} : \rss$.
\end{proof}

\BREAK

\begin{proof}[Proof of \Thmref{thm:corr-bt-lfence}]\proofref{}{corr-bt-lfence}\hfill

	This holds by \Thmref{thm:ini-state-rel} and by \Thmref{thm:bwd-sim-lfence}.
\end{proof}

\BREAK

\begin{proof}[Proof of \Thmref{thm:bwd-sim-lfence}]\proofref{}{bwd-sim-lfence}\hfill
	
	We proceed by induction on the reduction \Xtot{\trat{^{\taintt}}}
	\begin{description}
		\item[Base]

		no reductions, this case is trivial

		\item[Inductive]

		we have $n$ target steps to states $\trg{\Sigma_i} = \trg{C, H_i, \OB{B_i} \triangleright \proc{s_i;s_i''}{\OB{f_i}\cdot f_i}}$ and $\src{\Omega_i}=\src{C, H_i, \OB{B_i} \triangleright \proc{s_i;s_i''}{\OB{f_i}\cdot f_i}}$

		where $\src{\Omega_i}\srel\trg{\Sigma_i}$ and the traces produced so far are related via $\rels$.

		We proceed by case analysis on \src{f_i}
		\begin{description}
			\item[in \src{\OB{f''}} (in the compiled component)] 

			By case analysis on \src{f'}
			\begin{description}
				\item[in \src{\OB{f''}} (in the compiled component)] 

				This holds by \Thmref{thm:bwd-sim-comp-steps-lfence};

				\item[not in \src{\OB{f''}} (in the context)] 

				This happens by the execution of two statements:
				\begin{description}
					\item[call]

					This is a call from a compiled function to a context function.

					In this case we have that \trg{\Sigma_i} = \trg{C, H_i, \OB{B_i}\cdot B_i \triangleright \proc{\complfence{\call{f'}\ e};s_i'' }{\OB{f_i}\cdot f_i}}

					so by \Cref{tr:et-sp-act} with \Cref{tr:eus-call}(in the target ofc) we know 

					\trg{\Sigma_i \xtot{ (\clh{f'}{\complfence{v}}{H_i})^\safeta } \Sigma'} = \trg{w, (C, H', \OB{B'}\cdot x\mapsto \complfence{v} \triangleright \proc{ s_f ;s_i'' }{\OB{f'}\cdot f'}, \bot, \safeta)}

					and we know \trg{B_i\triangleright \complfence{e} \bigredt \complfence{v}} (HPE)

					where $\trg{f(x)\mapsto s_f}\in\trg{C}$ and \trg{\OB{B'}} = \trg{\OB{B_i}\cdot B_i} and \trg{\OB{f'}} = \trg{\OB{f_i}\cdot f_i} and \trg{H'}=\trg{H_i}

					and the taint is \trg{\safeta} by definition of $\sqcap$ since by $\srel$ the pc taint is \trg{\safeta}.

					By \Thmref{thm:back-sim-bte-lfence}, HPE yields  \src{B_i\triangleright {e} \bigreds {v}} (HPSE)

					We have that \src{\Omega_i} = \src{C, H_i, \OB{B_i}\cdot B_i \triangleright \proc{\call{f}\ e ; s_i''}{\OB{f_i}\cdot f_i}}

					by definition of \backtrfencec{\cdot} we know that $\src{f(x)\mapsto \backtrfencec{\trg{s_f}}}\in\src{C}$

					we take \src{\OB{B'}} = \src{\OB{B_i}\cdot B_i} so that by hypothesis we have that $\src{\OB{B'}}\brel\trg{\OB{B'}}$ (HPB)
 
					we take \src{\OB{f'}} = \src{\OB{f_i}\cdot f_i} so that by hypothesis we have that $\src{\OB{f'}}\equiv\trg{\OB{f'}}$ (HPF)
 
					we take \src{H'} = \src{H_i} so that by hypothesis we have that $\src{H'}\hrel\trg{H'}$ (HPH)

					By \Cref{tr:eus-call}(in the source) with HPSE and the hypotheses above, we have that 

					\src{\Omega_i \xtos{ (\clh{f'}{v}{H_i})^\safeta } \Omega'} = \src{C,H',\OB{B'}\cdot x\mapsto v \triangleright \proc{\backtrfencec{\trg{s_f}};s_i''}{\OB{f'}\cdot f'}}

					We need to prove that:
					\begin{itemize}
						\item $\src{\Omega'}\srelref\trg{\Sigma'}$, which by \Cref{tr:stats} means proving that:
						\begin{itemize}
							\item $\src{\OB{B'} \cdot x\mapsto v}\brel\trg{\OB{B'}\cdot x \mapsto \complfence{v}}$, which holds by (HPB) and by \Cref{tr:brel-i} with \Cref{tr:vr} with \Thmref{thm:val-rel-comp-lfence};
							\item $\src{\OB{f'}}\equiv\trg{\OB{f'}}$, which holds by (HPF);
							\item $\src{C}\crel_{\src{\OB{f''}}}\trg{C}$, which holds by $\srel$ of the initial states since components do not change;
							\item $\src{H'}\hrel\trg{H'}$, which holds by (HPH)
						\end{itemize}
						\item $\src{(\clh{f'}{v}{H_i})^\safeta}\arelref\trg{(\clh{f'}{\complfence{v}}{H_i})^\safeta}$, which by \Cref{tr:ac-rel-cl} means proving that:
						\begin{itemize}
							\item $\src{f'}\equiv\trg{f'}$, which holds;
							\item $\src{v}\vrel\complfence{v}$, which holds by \Cref{thm:val-rel-comp-lfence};
							\item $\src{\safeta}\equiv\trg{\safeta}$
						\end{itemize}
					\end{itemize}
					
					so this case holds.

					\item[ret]  

					This is a return from a compiled function to a context function.

					This is the dual of the case below for return from context to code.
				\end{description}
			\end{description}

			\item[not in \src{\OB{f''}} (in the context)] 

			By case analysis on \src{f'}
			\begin{description}
				\item[in \src{\OB{f''}} (in the compiled component)] 

				This happens by the execution of two statements:
				\begin{description}
					\item[call]

					This is a call from a context function to a compiled function.

					This is the dual of the case for call above.

					\item[ret]  

					This is a return from a context function to a compiled function.
					
					In this case we have that \trg{\Sigma_i} = \trg{C, H_i, \OB{B_i}\cdot B_i \triangleright \proc{\ret;\complfence{s_f};s_i'' }{\OB{f_i}\cdot f_i}}

					so by \Cref{tr:et-sp-act} with \Cref{tr:eus-ret}(in the target ofc) we know 

					\trg{\Sigma_i \xtot{ (\rth{}{H_i})^\safeta } \Sigma'} = \trg{w, (C, H_i, \OB{B_i} \triangleright \proc{\complfence{s_f};s_i'' }{\OB{f_i}}, \bot, \safeta)}

					We have that \src{\Omega_i} = \src{C, H_i, \OB{B_i}\cdot B_i \triangleright \proc{\ret ; s_f ; s_i''}{\OB{f_i}\cdot f_i}}

					By \Cref{tr:eus-call}(in the source), we have that 

					\src{\Omega_i \xtos{ (\rth{}{H_i})^\safeta } \Omega'} = \src{C,H_i,\OB{B_i} \triangleright \proc{s_f;s_i''}{\OB{f_i}}}

					We need to prove that:
					\begin{itemize}
						\item $\src{\Omega'}\srelref\trg{\Sigma'}$, which by \Cref{tr:stats} means proving that:
						\begin{itemize}
							\item $\src{\OB{B_i}}\brel\trg{\OB{B_i}}$, which holds by hypothesis;
							\item $\src{\OB{f_i}}\equiv\trg{\OB{f_i}}$, which holds by hypothesis;
							\item $\src{C}\crel_{\src{\OB{f''}}}\trg{C}$, which holds by $\srel$ of the initial states since components do not change;
							\item $\src{H_i}\hrel\trg{H_i}$, which holds by hypothesis
						\end{itemize}
						\item $\src{(\rth{}{H_i})^\safeta}\arelref\trg{(\rth{}{H_i})^\safeta}$, which by \Cref{tr:ac-rel-rt} is trivially true.
					\end{itemize}
					
					so this case holds.
				\end{description}

				\item[not in \src{\OB{f''}} (in the context)] 

				This holds by \Thmref{thm:back-sim-bts-lfence};
			\end{description}
		\end{description}
	\end{description}
\end{proof}

\BREAK

\begin{proof}[Proof of \Thmref{thm:bwd-sim-comp-steps-lfence}]\proofref{}{bwd-sim-comp-steps-lfence}\hfill
	
	By contradiction, assume the source ends up in a different state $\src{\Omega''}$ such that $\src{\Omega''}\neq\src{\Omega'}$.

	By \Thmref{thm:fwd-sim-stm-lfence} we have that the target also ends up in state \trg{\Sigma''} such that $\src{\Omega''}\srel_{\src{\OB{f''}}}\trg{\Sigma''}$ and such that $\trg{\Sigma''}\neq\trg{\Sigma'}$ (HPC).

	So we have that \trg{\Sigma} reaches both \trg{\Sigma'} and \trg{\Sigma''}.

	Since the semantics is deterministic, it must be that $\trg{\Sigma''}=\trg{\Sigma'}$ (HPE).
	
	We now have a contradiction between HPC and HPE.
\end{proof}

\BREAK

\begin{proof}[Proof of \Thmref{thm:fwd-sim-exp-lfence}]\proofref{}{fwd-sim-exp-lfence}\hfill
	
	Trivial induction on \src{e}.
\end{proof}

\BREAK

\begin{proof}[Proof of \Thmref{thm:fwd-sim-stm-lfence}]\proofref{}{fwd-sim-stm-lfence}\hfill
	
	The proof proceeds by structural induction on \src{s}.
	\begin{description}
		\item[Base]
		\begin{description}
			\item[skip] trivial;
			\item[call]
			Two cases arise:
			\begin{enumerate}
				\item \src{f} is defined by the component.

				This holds by \Thmref{thm:fwd-sim-exp-lfence} and by relatedness of heaps;
				\item \src{f} is defined by the context.

				This cannot arise as in the target we don't step to a compiled statement \complfence{s'}.
			\end{enumerate}
			\item[return]
			Two cases arise:
			\begin{enumerate}
				\item \src{f} is defined by the component.

				This holds by relatedness of heaps;
				\item \src{f} is defined by the context.

				This cannot arise as in the target we don't step to a compiled statement \complfence{s'}.
			\end{enumerate}
			\item[write] by \Thmref{thm:fwd-sim-exp-lfence} and \Cref{tr:ac-rel-wr2};
			\item[private write] by \Thmref{thm:fwd-sim-exp-lfence} and \Cref{tr:ac-rel-wr}.
		\end{description}
		\item[Inductive] 
		\begin{description}
			\item[sequencing] by IH;
			\item[letin] by IH and \Thmref{thm:fwd-sim-exp-lfence};
			\item[if zero] by IH and \Thmref{thm:fwd-sim-exp-lfence}.

				By definition of \complfence{\cdot}, this is the only case where we need to account for multiple steps in the target, since there is an \trg{\lfence}.

				By \Cref{tr:eut-lf}, a rollback is triggered via \Cref{tr:et-sp-rb}.

				Relatedness of states is therefore ensured, while relatedness of traces is ensured by: \Cref{tr:tr-rel-rollb} and \Cref{tr:ac-rel-rlb}.
			\item[let read] by IH and \Thmref{thm:fwd-sim-exp-lfence} and \Cref{tr:ac-rel-rd};
			\item[let private read] by IH and \Thmref{thm:fwd-sim-exp-lfence} and \Cref{tr:ac-rel-rd};
		 	\item[conditional letin] by IH and \Thmref{thm:fwd-sim-exp-lfence}.
		 \end{description} 
	\end{description}
\end{proof}

\BREAK

\begin{proof}[Proof of \Thmref{thm:back-sim-bte-lfence}]\proofref{}{back-sim-bte-lfence}\hfill

	The proof proceeds by structural induction on \trg{e}.
	\begin{description}
		\item[Base]
		\begin{description}
			\item[number] trivial;
			\item[variable] this follows from the relatedness of stack frames;
		\end{description}
		\item[Inductive] 
		\begin{description}
			\item[ops] by IH;
			\item[bops] by IH.
		\end{description}
	\end{description}
\end{proof}

\BREAK

\begin{proof}[Proof of \Thmref{thm:back-sim-bts-lfence}]\proofref{}{back-sim-bts-lfence}\hfill

	The proof proceeds by structural induction on \trg{s}.
	\begin{description}
		\item[Base]
		\begin{description}
			\item[skip] trivial;
			\item[call]
			Two cases arise:
			\begin{enumerate}
				\item \trg{f} is defined by the component.

				This holds by \Thmref{thm:back-sim-bte-lfence} and relatedness of heaps.
				\item \trg{f} is defined by the context.

				This cannot arise as in the source we don't step to a backtranslated statement \backtrfencec{\trg{s'}}.
			\end{enumerate}
			\item[return]
			Two cases arise:
			\begin{enumerate}
				\item \trg{f} is defined by the component.

				This holds by relatedness of heaps.
				\item \trg{f} is defined by the context.

				This cannot arise as in the source we don't step to a backtranslated statement \backtrfencec{\trg{s'}}.
			\end{enumerate}
			\item[write] by \Thmref{thm:back-sim-bte-lfence}.

				One complexity is showing that the action taint \taintt=\trg{\safeta}, but this follows from $\srel$ which tells that the pc taint is \trg{\safeta}.
			\item[private write] this cannot arise by \Cref{def:atk};
			\item[lfence] trivial.
		\end{description}
		\item[Inductive] 
		\begin{description}
			\item[sequencing] by IH;
			\item[letin] by IH and \Thmref{thm:back-sim-bte-lfence};
			\item[if zero] by IH and \Thmref{thm:back-sim-bte-lfence}.

				In this case, \Cref{tr:et-sp-if} is not applicable, so we cannot speculate.
			\item[let read] by IH and \Thmref{thm:back-sim-bte-lfence}.

				One complexity is showing that the action taint \taintt=\trg{\safeta}, but this follows from $\srel$ which tells that the pc taint is \trg{\safeta}.
			\item[let private read] this cannot arise by \Cref{def:atk};
		 	\item[conditional letin] by IH and \Thmref{thm:back-sim-bte-lfence}.
		 \end{description} 
	\end{description}
\end{proof}

\BREAK

\begin{proof}[Proof of \Thmref{thm:ini-state-rel}]\proofref{}{ini-state-rel}\hfill
	
	By definition of \srelref, we need to prove that:
	\begin{itemize}
		\item bindings are related via $\brelref$ (by \Cref{tr:bsrel}, then \Cref{tr:brel-i} where $\src{0}\vrel\trg{0}$ by \Cref{tr:vr});
		\item heaps are related via $\hrelref$: 
			\begin{itemize}
				\item for the program heap, \Cref{tr:} tells us that its domain is negative numbers, and the negative heap relatedness holds by \Thmref{thm:heap-rel-comp-lfence};
				\item for the context heap, \Cref{def:atk} tells us that its domain is natural numbers, so this holds by \Cref{thm:heap-rel-bt-lfence};
			\end{itemize}
		\item components are related via $\crelref$ by simple inspection of \complfence{\cdot} and \backtrfencec{\cdot};
		\item the target taint is \trg{\safeta}: this holds by \Cref{tr:ini-ut};
		\item the target window is \trg{\bot}: this holds by \Cref{tr:ini-ut};
	\end{itemize}
\end{proof}

\BREAK

\begin{proof}[Proof of \Thmref{thm:val-rel-comp-lfence}]\proofref{}{val-rel-comp-lfence}\hfill
	
	Trivial analyisis of the compiler.
\end{proof}

\BREAK

\begin{proof}[Proof of \Thmref{thm:heap-rel-comp-lfence}]\proofref{}{heap-rel-comp-lfence}\hfill
	
	Trivial analyisis of the compiler.
\end{proof}

\BREAK

\begin{proof}[Proof of \Thmref{thm:val-rel-bt-lfence}]\proofref{}{val-rel-bt-lfence}\hfill
	
	Trivial analyisis of the backtranslation.
\end{proof}

\BREAK

\begin{proof}[Proof of \Thmref{thm:taint-rel-bt-lfence}]\proofref{}{taint-rel-bt-lfence}\hfill
	
	Trivial analyisis of the backtranslation.
\end{proof}

\BREAK

\begin{proof}[Proof of \Thmref{thm:heap-rel-bt-lfence}]\proofref{done}{heap-rel-bt-lfence}\hfill 
	
	Trivial analyisis of the backtranslation with \Cref{thm:val-rel-bt-lfence,thm:taint-rel-bt-lfence}.
\end{proof}

\BREAK

\begin{proof}[Proof of \Thmref{thm:lfence-comp-stc}]\proofref{}{lfence-comp-stc}\hfill
	Instantiate \ctxs{} with \backtrfencec{\ctxt{}}.

	This holds by an adaptation of \Thmref{thm:corr-bt-lfence} to the extra reduction, which in turn holds by 
		\Thmref{thm:ini-state-rel} and by an adaptation of \Thmref{thm:bwd-sim-lfence} to the extra reduction,
			which in turn holds by \Thmref{thm:back-sim-bte-lfence} and by an adaptation of \Thmref{thm:bwd-sim-comp-steps-lfence} to the extra reduction,
				where the additional reduction must be considered.

				Since that reduction is not triggered, these adaptations trivially hold.
\end{proof}

\BREAK

\begin{proof}[Proof of \Thmref{thm:ini-state-rel-slh}]\proofref{}{ini-state-rel-slh}\hfill
	
	Analogous to the proof of \Thmref{thm:ini-state-rel} but with \Thmref{thm:val-rel-comp-slh} and \Thmref{thm:heap-rel-comp-slh}.
\end{proof}

\BREAK

\begin{proof}[Proof of \Thmref{thm:val-rel-comp-slh}]\proofref{}{val-rel-comp-slh}\hfill

	Trivial analysis of \compsslh{\cdot}.
\end{proof}

\BREAK

\begin{proof}[Proof of \Thmref{thm:heap-rel-comp-slh}]\proofref{}{heap-rel-comp-slh}\hfill
	
	Trivial analysis of \compsslh{\cdot} given that \trg{-1} is allocated by the compiler and that all addresses are shifted by 1, so they account for \Cref{tr:shrel-i}.
\end{proof}

\BREAK

\begin{proof}[Proof of \Thmref{thm:fwd-sim-exp-slh}]\proofref{}{fwd-sim-exp-slh}\hfill
	Trivial induction on \src{e}.
\end{proof}

\BREAK

\begin{proof}[Proof of \Thmref{thm:fwd-sim-stm-slh}]\proofref{}{fwd-sim-stm-slh}\hfill

	The proof proceeds by induction on \src{s}.
	\begin{description}
		\item[Base]	
			\begin{description}
				\item[skip] Trivial.
				\item[call f] 
					We have two cases:
					\begin{itemize}
						\item \src{f} is component-defined.

						This follows from \Thmref{thm:fwd-sim-exp-slh}.

						\item \src{f} is context-defined.

						This is a contradiction because the ending statement is not a compiled one.
					\end{itemize}

				\item[assign]

				This follows from \Thmref{thm:fwd-sim-exp-slh}.

				\item[private assign]    

				This follows from \Thmref{thm:fwd-sim-exp-slh}.

				\item[return]

				Trivial. 
			\end{description}
		\item[Inductive]  
			\begin{description}
				\item[sequencing] 

				By IH.

				\item[let-in] 

				By IH and \Thmref{thm:fwd-sim-exp-slh}.

				\item[if then else]

				In this case we have this source reduction, wlog assume HE \src{B \triangleright e \bigreds \trues : \sigma }
				\begin{align*}
					\src{C, H, \OB{B}\cdot B \triangleright \ifzte{e}{s}{s'};s'' \xtos{(\ifl{0})^\safeta} C, H, \OB{B}\cdot B \triangleright s;s''}
				\end{align*}

				By \Thmref{thm:fwd-sim-exp-slh} with HE we get HET : \trg{B \triangleright \compsslh{e} \bigreds \compsslh{\trues} : \taintt}.

				In the target we have these reductions (by HET):
				\begin{align*}
					&
					\trg{\Sigma} =
					\\
					&
					\trg{ w (C, H, \OB{B}\cdot B, \bot, \safeta) \triangleright \compsslh{\ifzte{e}{s}{s'}};s''}
					\\
					\equiv
					&
					\trg{ w (C, H, \OB{B}\cdot B, \bot, \safeta) \triangleright
						\begin{aligned}[t]
							&
							\letint{\trg{x_g}}{\compsslh{e}}{
							\\
							&\
								\letreadpt{\trg{\predState}}{-1}{}
								\\&\
								\cmovet{\trg{x_g}}{\trg{0}}{\trg{\predState}}{
								\\&\ 
									\ifztet{\trg{x_g}
										}{ 
										\letreadpt{\trg{x}}{\trg{-1}}{\asgnpt{\trg{-1}}{\trg{x \vee \neg x_g}}};
										\compsslh{s}
									\\
									&\ \
									}{
										\letreadpt{\trg{x}}{\trg{-1}}{\asgnpt{\trg{-1}}{\trg{x \vee x_g}}};
										\compsslh{s'}
									}
								}
							}
						\end{aligned};s''
					}
					\\
					&
					\text{ let } \trg{B'} = \trg{B}\cdot \trg{x_g\mapsto \truet:\taintt}
					\\
					\xltot{}
					&
					\trg{ w (C, H, \OB{B}\cdot B', \bot, \safeta) \triangleright
						\begin{aligned}[t]
							&
							\letreadpt{\trg{\predState}}{-1}{}
							\\&\
							\cmovet{\trg{x_g}}{\trg{0}}{\trg{\predState}}{
								\\
								&\
								\ifztet{\trg{x_g}
									}{ 
									\letreadpt{\trg{x}}{\trg{-1}}{\asgnpt{\trg{-1}}{\trg{x \vee \neg x_g}}};
									\compsslh{s}
								\\
								&\ \
								}{
									\letreadpt{\trg{x}}{\trg{-1}}{\asgnpt{\trg{-1}}{\trg{x \vee x_g}}};
									\compsslh{s'}
								}
							}
						\end{aligned};s''
					}
					\\
					&
					\text{ since } \trg{H(-1)\mapsto \falset : \safeta}
					\\
					&
					\text{ let } \trg{B''} = \trg{B'}\cup\trg{\predState\mapsto \falset:\safeta}
					\\
					\xltot{(\rdl{-1})^\safeta}
					&
					\trg{ w (C, H, \OB{B}\cdot B'', \bot, \safeta) \triangleright
						\begin{aligned}[t]
							&
							\cmovet{\trg{x_g}}{\trg{0}}{\trg{\predState}}{
								\\
								&\
								\ifztet{\trg{x_g}
									}{ 
									\letreadpt{\trg{x}}{\trg{-1}}{\asgnpt{\trg{-1}}{\trg{x \vee \neg x_g}}};
									\compsslh{s}
								\\
								&\ \
								}{
									\letreadpt{\trg{x}}{\trg{-1}}{\asgnpt{\trg{-1}}{\trg{x \vee x_g}}};
									\compsslh{s'}
								}
							}
						\end{aligned};s''
					}
					\\
					&
					\text{ since } \trg{\predState\mapsto \falset:\safeta}
					\\
					\xltot{}
					&
					\trg{ w (C, H, \OB{B}\cdot B'', \bot, \safeta) \triangleright
						\begin{aligned}[t]
							&
							\ifztet{\trg{x_g}
								}{ 
								\letreadpt{\trg{x}}{\trg{-1}}{\asgnpt{\trg{-1}}{\trg{x \vee \neg x_g}}};
								\compsslh{s}
							\\
							&\ \
							}{
								\letreadpt{\trg{x}}{\trg{-1}}{\asgnpt{\trg{-1}}{\trg{x \vee x_g}}};
								\compsslh{s'}
							}
						\end{aligned};s''
					}
					\\
					\\
					&
					\text{ By \Cref{tr:et-sp-if} }
					\\
					\xltot{ (\ifl{0})^\safeta }
					&
					\trg{ w 
						\begin{aligned}[t]
							& \trg{
								(C, H, \OB{B}\cdot B'', \bot, \safeta \triangleright \letreadpt{\trg{x}}{\trg{-1}}{\asgnpt{\trg{-1}}{\trg{x \vee \neg x_g}}}; \compsslh{s} ; s'') 
								}
							\\
							\cdot
							&
							\trg{
								(C, H, \OB{B}\cdot B'', w, \unta \triangleright 
									\begin{aligned}[t]
										&
										\letreadpt{\trg{x}}{\trg{-1}}{\asgnpt{\trg{-1}}{\trg{x \vee x_g}}};
										\compsslh{s'}
									\end{aligned};s''
								)
							}
						\end{aligned}
					}
					\\
					&
					\text{ since \trg{H(-1)\mapsto \falset : \safeta}}
					\\
					&
					\text{ let } \trg{B'''} = \trg{B''}\cdot \trg{x\mapsto \falset:\safeta}
					\\
					\xltot{ (\rdl{-1})^\safeta }
					&
					\trg{ w 
						\begin{aligned}[t]
							& \trg{
								(C, H, \OB{B}\cdot B'', \bot, \safeta \triangleright \letreadpt{\trg{x}}{\trg{-1}}{\asgnpt{\trg{-1}}{\trg{x \vee \neg x_g}}}; \compsslh{s} ; s'') 
								}
							\\
							\cdot
							&
							\trg{
								(C, H, \OB{B}\cdot B''', w, \unta \triangleright 
									\begin{aligned}[t]
										&
										\asgnpt{\trg{-1}}{\trg{x \vee x_g}};
										\compsslh{s'}
									\end{aligned};s''
								)
							}
						\end{aligned}
					}
					\\
					&
					\text{ since \trg{B''' \triangleright x\vee x_g \bigredt \truet : \safeta} let } \trg{H'} = \trg{H}\cup\trg{-1\mapsto\truet:\safeta}
					\\
					\xltot{ (\wrl{-1})^\safeta }
					&
					\trg{ w 
						\begin{aligned}[t]
							& \trg{
								(C, H, \OB{B}\cdot B'', \bot, \safeta \triangleright \letreadpt{\trg{x}}{\trg{-1}}{\asgnpt{\trg{-1}}{\trg{x \vee \neg x_g}}}; \compsslh{s} ; s'') 
								}
							\\
							\cdot
							&
							\trg{
								(C, H', \OB{B}\cdot B''', w, \unta \triangleright 
									\begin{aligned}[t]
										&
										\compsslh{s'}
									\end{aligned};s''
								)
							}
						\end{aligned}
					}
				\end{align*}
				Call this last state \trg{\Sigma_i}.

				Let \trg{\OB{D}} = \trg{\OB{D}\cdot D}, consider \trg{\OB{D'}} = \trg{\OB{D}\cdot D,x_g}.

				We can easily prove that $\src{\Omega}\ssrel_{\src{\OB{f_c}}}^{\trg{\OB{D'}}}\trg{\Sigma_i}$ since by HP the first states are related at $\srel_{\src{\OB{f_c}}}^{\trg{\OB{D'}}}$ and the speculating states in \trg{\Sigma_i} only have safe bindings.

				\begin{center}
				\fbox{
					\parbox{.75\textwidth}{
					Note: this is where the proof for \compslhd{\cdot} breaks, because there we need a stronger invariant, namely that all bindings are safe when speculating, and we can't prove that here because we inherit \trg{B} with all it can have, including \trg{\unta} bindings.
					If we do not require the bindings to all be safe when speculating, the speculation lemmas break (\Thmref{thm:comp-spec-safe}): without this we cannot show that actions are safe.
					}
				}
				\end{center}

				By \Thmref{thm:spec-rel-satte-safe} we know that: 
				\begin{align*}
					&
					\trg{(n'',\Sigma_i) \Xtot{\trat{^{\taintt}}} (n',\Sigma'')} 
				\end{align*}
				Where 
				\begin{align*}
					\trg{\Sigma''} =&\ 
						\trg{ w 
						\begin{aligned}[t]
							& \trg{
								(C, H, \OB{B}\cdot B'', \bot, \safeta \triangleright \letreadpt{\trg{x}}{\trg{-1}}{\asgnpt{\trg{-1}}{\trg{x \vee \neg \truet}}}; \compsslh{s} ; s'') 
								}
						\end{aligned}
					}
				\end{align*}
				and (HL): $\srce\tracerel\trat{^{\taintt}}$

				and (HS): $\src{\Omega}\ssrel_{\src{\OB{f_c}}}^{\trg{\OB{D'}}}\trg{\Sigma''}$

				and that the first element in the stack of states in \trg{\Sigma''} is the same as the first element in the stack of states of \trg{\Sigma'}, which is still $\srel$ with \src{\Omega}.
				
				The reductions proceed as follows:
				\begin{align*}
					&
					\trg{\Sigma''}
					\\
					\xltot{ (\rdl{-1})^\safeta }
					& 
					\trg{
						(C, H, \OB{B}\cdot B'', \bot, \safeta \triangleright {\asgnpt{\trg{-1}}{\trg{\falset \vee \neg \truet}}}; \compsslh{s} ; s'') 
					}
					\\
					&
					\text{ since \trg{\_ \triangleright \falset\vee\neg\truet \bigredt \falset : \safeta}}
					\\
					\xltot{ (\wrl{-1})^\safeta }
					& 
					\trg{
						(C, H, \OB{B}\cdot B'', \bot, \safeta \triangleright \compsslh{s} ; s'') 
					}
				\end{align*}

				At this point, we have this target trace:
				\begin{align*}
					&
					\trg{(n,\Sigma) \Xtot{(\rdl{-1})^\safeta \cdot (\ifl{0})^\safeta \cdot (\rdl{-1})^\safeta \cdot (\wrl{-1})^\safeta } (n'',\Sigma_i) \Xtot{\trat{^{\taintt}}} (n'',\Sigma'') }
					\\
					&\
					\trg{ \Xtot{ %
						(\rdl{-1})^\safeta \cdot (\wrl{-1})^\safeta} (n',\Sigma')}
					\\
					&
					\text{ i.e.,}
					\\
					&
					\trg{(n,\Sigma) \Xtot{ (\rdl{-1})^\safeta \cdot (\ifl{0})^\safeta \cdot (\rdl{-1})^\safeta \cdot (\wrl{-1})^\safeta \cdot \trat{^{\taintt}} \cdot 
						(\rdl{-1})^\safeta \cdot (\wrl{-1})^\safeta} (n',\Sigma')}
				\end{align*}
				We need to show that this trace is $\tracerel$ to the source trace \src{(\ifl{0})^\safeta}.

				This holds because 
				\begin{itemize}
					\item the first read action can be dropped by \Cref{tr:tr-rel-safe-h} and \Cref{tr:ac-rel-ep-hp};
					\item the if actions are related by \Cref{tr:tr-rel-same-h} and \Cref{tr:ac-rel-if};
					\item the second read action can be dropped by \Cref{tr:tr-rel-safe-h} and \Cref{tr:ac-rel-ep-hp};
					\item the first write action can be dropped by \Cref{tr:tr-rel-safe-h} and \Cref{tr:ac-rel-ep-hp};
					\item the trace \trat{^{\taintt}} can be dropped by HL;
					\item the third read action can be dropped by \Cref{tr:tr-rel-safe-h} and \Cref{tr:ac-rel-ep-hp};
					\item the second write action can be dropped by \Cref{tr:tr-rel-safe-h} and \Cref{tr:ac-rel-ep-hp}.
				\end{itemize}
				So this case holds.

				\item[read]

				By IH and \Thmref{thm:fwd-sim-exp-slh}.

				\item[private read]   

				By IH and \Thmref{thm:fwd-sim-exp-slh}.
			\end{description}
	\end{description}
\end{proof}

\BREAK

\begin{proof}[Proof of \Thmref{thm:back-sim-stm-slh}]\proofref{}{back-sim-stm-slh}\hfill

	Analogous to \Thmref{thm:bwd-sim-comp-steps-lfence} with \Thmref{thm:fwd-sim-stm-slh}.
\end{proof}

\BREAK

\begin{proof}[Proof of \Thmref{thm:exp-red-safe}]\proofref{}{exp-red-safe}\hfill
	
	Trivial induction on \trg{e}, the only nontrivial case is when $\trg{e}=\trg{x}$ but this follows from the safety of bindings.
\end{proof}

\BREAK

\begin{proof}[Proof of \Thmref{thm:spec-most-omega}]\proofref{}{spec-most-omega}\hfill

	This proceeds by cases on \trg{f} and \trg{f''}:
	\begin{description}
		\item[Both in \src{\OB{f_c}}] 

			These are compiled reductions, this holds by \Thmref{thm:comp-spec-most-omega};

		\item[Both not in \src{\OB{f_c}}] 

			These are context reductions, this holds by \Thmref{thm:ctx-spec-most-omega};

		\item[$\trg{f}\in \src{\OB{f_c}}$ and $\trg{f''}\notin \src{\OB{f_c}}$] 

			This is a reduction going from compiled code to context.

			We proceed by induction on \trgb{\omega}.

			The base case is trivial by \Cref{tr:eut-tr-init}, the inductive case has two cases:
			\begin{description}
				\item[call]  

					By analising the case of \compsslh{\cdot} for call, we have:

					\trg{
						\letin{x_f}{\compsslh{e}}{
							\cmovet{\trg{x_f}}{\trg{0}}{\trg{\predState}}{\call{f}~x_f}
						}
					}

					We have two cases: \trg{B\triangleright \compsslh{e}\bigredt v : \taintt} or \trg{B\triangleright \compsslh{e}\bigredt e' : \taintt}, i.e., the execution of \compsslh{e} gets stuck.

					The latter case is trivially true by \Cref{tr:et-sp-rb-s}: when speculation gets stuck it gets rolled back.

					The former case proceeds as follows.

					We have two cases, either there speculation window is long enough ($\trg{\omega}>3$) or not.

					In the latter case, some of the reductions below happen and then a \Cref{tr:et-sp-rb} is triggered, so this holds.

					Otherwise, if the window is long enough, we have the following.

					By HP we have that \trg{H(-1)\mapsto \truet:\safeta} so the code above will step as follows:

					(for simplicity we only keep track of the top of the stack of execution states)
					\begin{align*}
						&
						\trg{
							(C,H,\OB{B}\cdot B \triangleright
														\letin{x_f}{\compsslh{e}}{
															\cmovet{\trg{x_f}}{\trg{0}}{\trg{\predState}}{\call{f}~x_f}
														}), \omega ,\unta
						}						
						\\
						&\text{ assuming } \trg{B\triangleright \compsslh{e}\bigredt v : \taintt}
						\\
						\xtot{}
						&
						\trg{
							(C,H,\OB{B}\cdot B\cdot x_f\mapsto v:\taintt \triangleright
															\cmovet{\trg{x_f}}{\trg{0}}{\trg{\predState}}{\call{f}~x_f}), \omega-1 ,\unta
						}						
						\\
						&\text{ since } \trg{H(-1)\mapsto \truet:\safeta}
						\\
						\xtot{}
						&
						\trg{
							(C,H,\OB{B}\cdot B\cdot x_f\mapsto 0:\safeta \triangleright
															{\call{f}~x_f}), \omega-2 ,\unta
						}						
						\\
						&\text{ where }\trg{\OB{B'}}\text{ is the current stack and the body of \trg{f} is \trg{s}}
						\\
						\xtot{\cbh{f}{0}{H}^\safeta}
						&
						\trg{
							(C,H,\OB{B'}\cdot x\mapsto 0:\safeta \triangleright
														s), \omega-3 ,\unta
						}
					\end{align*}

					We need to prove that the new binding is safe (which is true) and that the action is droppable (which holds by \Cref{tr:ac-rel-ep-al}), so this case holds.

					\begin{center}
					\fbox{
						\parbox{.75\textwidth}{
						Note: this is where the proof for \compslht{\cdot} without \trg{\lfence} breaks, 
						because there we need a stronger invariant, namely that all bindings start with a variable capturing speculation, and here we cannot set that up correctly.
						
						\compslht{\cdot} with \trg{\lfence} goes through because a rollback is triggered, and the stack of bindings goes back to what was related.
						}
					}
					\end{center}

				\item[ret] 

					This is analogous to the point above.

					Additionally, we need to prove that the bindings we go back to are all safe.

					This trivially holds because a context cannot create unsafe bindings (\Thmref{thm:ctx-spec-sing-safe})
					 and the binding created in a call is safe (\Cref{tr:eus-call}).
			\end{description}

		\item[$\trg{f''}\in \src{\OB{f_c}}$ and $\trg{f}\notin \src{\OB{f_c}}$] 

			This is a reduction going from context to compiled code.

			We proceed by induction on \trgb{\omega}.

			The base case is trivial by \Cref{tr:eut-tr-init}, the inductive case has two cases:
			\begin{description}
				\item[call] 
				\item[ret] 
			\end{description}
			Both are trivially true since nothing extra needs to be enforced.
	\end{description}
\end{proof}

\BREAK

\begin{proof}[Proof of \Thmref{thm:ctx-spec-most-omega}]\proofref{}{ctx-spec-most-omega}\hfill

	By induction on the reduction and with \Thmref{thm:ctx-spec-sing-safe}.
\end{proof}

\BREAK

\begin{proof}[Proof of \Thmref{thm:ctx-spec-sing-safe}]\proofref{}{ctx-spec-sing-safe}\hfill

	By induction on \trg{s}:
	\begin{description}
		\item[Base]
		\begin{description}
			\item[skip]
			\item[assignment] 
			\item[lfence] 
		\end{description}
		\item[Inductive]  
		\begin{description}
			\item[call]  
			\item[sequence]
			\item[letin]
			\item[if]
			\item[read]
			\item[cmove]
		\end{description}
	\end{description}
	All cases are trivial, all expressions evaluate to \trg{\safeta} due to \Thmref{thm:exp-red-safe} and no reduction can load \trg{\unta} values, so the conditions are met.

	All actions are tagged \trg{\safeta} so they can be related to \srce.
\end{proof}

\BREAK

\begin{proof}[Proof of \Thmref{thm:comp-spec-most-omega}]\proofref{}{comp-spec-most-omega}\hfill

	By induction on the reduction with \Thmref{thm:comp-spec-safe}
\end{proof}

\BREAK

\begin{proof}[Proof of \Thmref{thm:comp-spec-safe}]\proofref{}{comp-spec-safe}\hfill

	By induction on \src{s}:
	\begin{description}
		\item[Base]
			\begin{description}
				\item[skip] Trivial.

				\item[assign]

				This is analogous to the call case, save that there are two case analyses for both expressions.

				\item[private assign] 

				This is analogous to the call case, save that there are two case analyses for both expressions.

			\end{description}
		\item[Inductive] 
			\begin{description}
				\item[call] 

				If \trg{\omega} = \trg{0} then this trivially holds by \Cref{tr:eut-tr-init}.

				If \trg{\omega} = \trg{n+1} then we have two cases:
				\begin{itemize}
					\item \src{f} is compiled code.

					We have two cases:
					\begin{enumerate}
						\item \trg{B\triangleright \compsslh{e}\bigredt v:\taintt}

						We have two cases here:

						\begin{enumerate}
							\item \trg{\omega>3}

							By HP we have that \trg{H(-1)\mapsto \truet:\safeta} so we have:

							(for simplicity we only keep track of the top of the stack of execution states)
							\begin{align*}
								&
								\trg{
									(C,H,\OB{B}\cdot B \triangleright
																\letin{x_f}{\compsslh{e}}{
																	\cmovet{\trg{x_f}}{\trg{0}}{\trg{\predState}}{\call{f}~x_f}
																}), \omega ,\unta
								}						
								\\
								&\text{ assuming } \trg{B\triangleright \compsslh{e}\bigredt v : \taintt}
								\\
								\xtot{}
								&
								\trg{
									(C,H,\OB{B}\cdot B\cdot x_f\mapsto v:\taintt \triangleright
																	\cmovet{\trg{x_f}}{\trg{0}}{\trg{\predState}}{\call{f}~x_f}), \omega-1 ,\unta
								}						
								\\
								&\text{ since } \trg{H(-1)\mapsto \truet:\safeta}
								\\
								\xtot{}
								&
								\trg{
									(C,H,\OB{B}\cdot B\cdot x_f\mapsto 0:\safeta \triangleright
																	{\call{f}~x_f}), \omega-2 ,\unta
								}						
								\\
								&\text{ where }\trg{\OB{B'}}\text{ is the current stack and the body of \trg{f} is \compsslh{s}}
								\\
								\xtot{\cbh{f}{0}{H}^\safeta}
								&
								\trg{
									(C,H,\OB{B'}\cdot x\mapsto 0:\safeta \triangleright
																\compsslh{s}), \omega-3 ,\unta
								}
							\end{align*}

							So this case holds by IH and by \Cref{tr:ac-rel-ep-al} since the action is safe.

							\item \trg{\omega\leq 3}

							In this case the execution will run out of steps and the case holds by \Cref{tr:et-sp-rb}.
						\end{enumerate}

						\item \trg{B\triangleright \compsslh{e}\bigredt e:\taintt}

						If \compsslh{e} gets stuck, this holds by \Cref{tr:et-sp-rb-s}.
					\end{enumerate}

					\item \src{f} is context.

					This is a contradiction. 
				\end{itemize}

				\item[seq]

				This is analogous to the if case save for the considerations on the expressions.

				\item[letin]

				This is analogous to the if case.

				\item[if]

				We have the same cases as in the call case, so we take a look at the most interesting one, namely when all reductions go through:

				We have these reductions:

				\begin{align*}
					&
					\trg{ 
						\begin{aligned}[t]
							&
							\trg{C, H, \OB{B}\cdot B \triangleright}
							\\
							&
							\letint{\trg{x_g}}{\compsslh{e}}{
							\\
							&\
								\letreadpt{\trg{\predState}}{-1}{}
								\\
								&\
								\cmovet{\trg{x_g}}{\trg{0}}{\trg{\predState}}{
								\\
								&\
									\ifztet{\trg{x_g}
										}{ 
										\letreadpt{\trg{x}}{\trg{-1}}{\asgnpt{\trg{-1}}{\trg{x \vee \neg x_g}}};
										\compsslh{s}
									\\
									&\ \
									}{
										\letreadpt{\trg{x}}{\trg{-1}}{\asgnpt{\trg{-1}}{\trg{x \vee x_g}}};
										\compsslh{s'}
									}
								}
							}
						\end{aligned}
					} 
					\\
					\xtot{}
					&
					\trg{ 
						\begin{aligned}[t]
							&
							\trg{C, H, \OB{B}\cdot B	\cdot x_g \mapsto v : \taintt \triangleright}
							\\
							&
								\letreadpt{\trg{\predState}}{-1}{}
								\\
								&\
								\cmovet{\trg{x_g}}{\trg{0}}{\trg{\predState}}{
								\\
								&\
									\ifztet{\trg{x_g}
										}{ 
										\letreadpt{\trg{x}}{\trg{-1}}{\asgnpt{\trg{-1}}{\trg{x \vee \neg x_g}}};
										\compsslh{s}
									\\
									&\ \
									}{
										\letreadpt{\trg{x}}{\trg{-1}}{\asgnpt{\trg{-1}}{\trg{x \vee x_g}}};
										\compsslh{s'}
									}
								}
						\end{aligned}
					} 
					\\
					\xtot{}
					&
					\trg{ 
						\begin{aligned}[t]
							&
							\trg{C, H, \OB{B}\cdot B	\cdot x_g \mapsto v :\taintt \cdot \predState \mapsto \truet:\safeta \triangleright}
							\\
								&\
								\cmovet{\trg{x_g}}{\trg{0}}{\trg{\predState}}{
								\\
								&\
									\ifztet{\trg{x_g}
										}{ 
										\letreadpt{\trg{x}}{\trg{-1}}{\asgnpt{\trg{-1}}{\trg{x \vee \neg x_g}}};
										\compsslh{s}
									\\
									&\ \
									}{
										\letreadpt{\trg{x}}{\trg{-1}}{\asgnpt{\trg{-1}}{\trg{x \vee x_g}}};
										\compsslh{s'}
									}
								}
						\end{aligned}
					} 
					\\
					&
					\text{ this step is key: note that \trg{x_g} becomes \trg{\safeta}}
					\\
					\xtot{\ifl{0}^\safeta}
					&
					\trg{ 
						\begin{aligned}[t]
							&
							\trg{C, H, \OB{B}\cdot B	\cdot x_g \mapsto 0 : \safeta \cdot \predState \mapsto \truet:\safeta \triangleright}
							\\
								&\
									\ifztet{\trg{x_g}
										}{ 
										\letreadpt{\trg{x}}{\trg{-1}}{\asgnpt{\trg{-1}}{\trg{x \vee \neg x_g}}};
										\compsslh{s}
									\\
									&\ \
									}{
										\letreadpt{\trg{x}}{\trg{-1}}{\asgnpt{\trg{-1}}{\trg{x \vee x_g}}};
										\compsslh{s'}
									}
						\end{aligned}
					} 
					\\
					\xtot{\rdl{-1}^\safeta}
					&
					\trg{ 
						\begin{aligned}[t]
							&
							\trg{C, H, \OB{B}\cdot B	\cdot x_g \mapsto 0 : \safeta \cdot \predState \mapsto \truet:\safeta \triangleright}
							\\
							&
								\letreadpt{\trg{x}}{\trg{-1}}{\asgnpt{\trg{-1}}{\trg{x \vee \neg x_g}}};
										\compsslh{s}
						\end{aligned}
					} 
					\\
					\xtot{\rdl{-1}^\safeta}
					&
					\trg{ 
						\begin{aligned}[t]
							&
							\trg{C, H, \OB{B}\cdot B	\cdot x_g \mapsto 0 : \safeta \cdot \predState \mapsto \truet:\safeta \cdot x \mapsto \truet :\safeta \triangleright}
							\\
							&
								\asgnpt{\trg{-1}}{\trg{x \vee \neg x_g}};
										\compsslh{s}
						\end{aligned}
					} 
					\\
					\xtot{\wrl{-1}^\safeta}
					&
					\trg{ 
						\begin{aligned}[t]
							&
							\trg{C, H\cup -1\mapsto\truet:\safeta, \OB{B}\cdot B	\cdot x_g \mapsto 0 : \safeta \cdot \predState \mapsto \truet:\safeta \cdot x \mapsto \truet :\safeta \triangleright}
							\\
							&
										\compsslh{s}
						\end{aligned}
					} 
				\end{align*}
				The rest holds by IH so long as the states are related by $\ssrel$, which is trivially true and if the trace is related to \srce by $\tracerel$.
 
				We have this trace
				\begin{align*}
					\Xtot{\rdl{-1}^\safeta \cdot \ifl{0}^\safeta \cdot \rdl{-1}^\safeta \cdot \wrl{-1}^\safeta}
				\end{align*}
				and each action is related to \srce by \Cref{tr:ac-rel-ep-hp,tr:ac-rel-ep-al}, so the whole trace is related to \srce.

				Thus this case holds.

				\item[read]

				This is analogous to the if case.

				\item[private read]     

				This is analogous to the if case.
			\end{description}
	\end{description}
\end{proof}

\BREAK

\begin{proof}[Proof of \Thmref{thm:spec-rel-satte-safe}]\proofref{}{spec-rel-satte-safe}\hfill

	This proof proceeds by induction on the stack of configurations.
	\begin{description}
		\item[Base] 
			Empty stack:

			By \Thmref{thm:spec-most-omega} we have the first reductions and $\srce \tracerel \trg{\trat{^{\taintt}}}$:
			\begin{align*}
				&\
				\trg{w (C, H_b, \OB{B_b} \triangleright \proc{s_b}{\OB{f_b}},\bot,\safeta) \cdot (C, H, \OB{B} \triangleright \proc{s}{\OB{f}},\omega,\unta)}
				\\
				\Xtot{\trat{^{\taintt}}}
				&\
				\trg{w (C, H_b, \OB{B_b} \triangleright \proc{s_b}{\OB{f_b}},\bot,\safeta) \cdot (C, H', \OB{B'} \triangleright \proc{s'}{\OB{f'}},0,\unta)}
				\\
				\Xtot{\rollbl}
				&\
				\trg{w (C, H_b, \OB{B_b} \triangleright \proc{s_b}{\OB{f_b}},\bot,\safeta)}
			\end{align*}
			Given the rollback reduction and \Cref{tr:ac-rel-rlb} this case holds.

		\item[Inductive]  

			This holds by IH plus the same reasoning as in the base case.
	\end{description}
\end{proof}

\BREAK

\begin{proof}[Proof of \Thmref{thm:bwd-sim-slh}]\proofref{}{bwd-sim-slh}\hfill

	We have these kinds of reductions:
	\begin{itemize}
		\item compiled-to-compiled code: this holds by \Thmref{thm:back-sim-stm-slh};
		\item backtranslated-to-backtranslated code: this holds by \Thmref{thm:back-sim-bts-lfence};
		\item compiled-to-backtranslated or backtranslated-to-compiled code: this is analogous to the cases discussed in \Thmref{sec:bwd-sim-lfence} since these do not trigger any speculation.
	\end{itemize}
\end{proof}

\BREAK

\begin{proof}[Proof of \Thmref{thm:corr-bt-slh}]\proofref{}{corr-bt-slh}\hfill

	By \Thmref{thm:bwd-sim-slh} with \Thmref{thm:ini-state-rel-slh}.
\end{proof}

\BREAK

\begin{proof}[Proof of \Thmref{thm:slh-comp-rdss}]\proofref{}{slh-comp-rdss}\hfill

	Instantiate \ctxs{} with \backtrfencec{\ctxt{}}.

	This holds by \Thmref{thm:corr-bt-slh}.
\end{proof}

\BREAK

\begin{proof}[Proof of \Thmref{thm:slh-comp-rdss-proc}]\proofref{}{slh-comp-rdss-proc}\hfill

	Instantiate \ctxs{} with \backtrfencec{\ctxt{}}.

	This holds by a variation of \Thmref{thm:corr-bt-slh} to account for the different state relation ($\cssrel$) which in turns holds by a variation of both \Thmref{thm:bwd-sim-slh} and \Thmref{thm:ini-state-rel-slh}.

	The latter is a trivial variation of the same theorem to account for the different trace relation.

	The former holds by a variation of three results : \Thmref{thm:back-sim-stm-slh} and \Thmref{thm:back-sim-bts-lfence} and \Thmref{sec:bwd-sim-lfence} with a variation to account for the different state relation.

	Of these three, only the first is effectively affected by the change of state relation, so the former holds by a variation of \Thmref{thm:fwd-sim-stm-slh}, which relies on \Thmref{thm:spec-rel-satte-safe} and then on \Thmref{thm:spec-most-omega} and then on an adaptation of \Thmref{thm:comp-spec-safe}, where the new trace relation plays a role.

	We provide only the proof of the last theorem (in \Thmref{thm:comp-spec-safe-proc}) since it is the only one with any change.
\end{proof}

\BREAK
\begin{proof}[Proof of \Thmref{thm:comp-spec-safe-proc}]\proofref{}{comp-spec-safe-proc}\hfill

	By induction on \src{s}:
	\begin{description}
		\item[Base]
			\begin{description}
				\item[skip] Trivial.

				\item[assign]

				This is analogous to the call case, save that there are two case analyses for both expressions.

				\item[private assign] 

				This is analogous to the call case, save that there are two case analyses for both expressions.

			\end{description}
		\item[Inductive] 
			\begin{description}
				\item[call] 

				If \trg{\omega} = \trg{0} then this trivially holds by \Cref{tr:eut-tr-init}.

				If \trg{\omega} = \trg{n+1} then we have two cases:
				\begin{itemize}
					\item \src{f} is compiled code.

					We have two cases:
					\begin{enumerate}
						\item \trg{B\triangleright \compsslh{e}\bigredt v:\taintt}

						We have two cases here:

						\begin{enumerate}
							\item \trg{\omega>1}

							Compiled code starts with an \trg{\lfence} so the execution is immediately rolled back and this case holds by \Cref{tr:et-sp-rb}.

							\item \trg{\omega\leq 1}

							In this case the execution will run out of steps and the case holds by \Cref{tr:et-sp-rb}.
						\end{enumerate}

						\item \trg{B\triangleright \compsslh{e}\bigredt e:\taintt}

						If \compsslh{e} gets stuck, this holds by \Cref{tr:et-sp-rb-s}.
					\end{enumerate}

					\item \src{f} is context.

					This is a contradiction. 
				\end{itemize}

				\item[seq]

				This is analogous to the if case save for the considerations on the expressions.

				\item[letin]

				This is analogous to the if case.

				\item[if]

				This is analogous to \showproof{comp-spec-safe} of \Thmref{thm:comp-spec-safe}.

				\item[read]

				This is analogous to the if case.

				\item[private read]     

				This is analogous to the if case.
			\end{description}
	\end{description}
\end{proof}

\BREAK

\begin{proof}[Proof of \Thmref{thm:slh-comp-rdss-weak}]\proofref{}{slh-comp-rdss-weak}\hfill

	This is analogous to the proof of \Thmref{thm:slh-comp-rdss}.

	The only changes are in the compilation of calls, public reads, private reads, public writes and private writes in two theorems:

	The former is the adaptation of \Thmref{thm:fwd-sim-stm-slh} from \compsslh{\cdot} to \compslh{\cdot}.

	The second is the adaptation of \Thmref{thm:comp-spec-most-omega} from \compsslh{\cdot} to \compslh{\cdot}.

\end{proof}

\BREAK

\end{document}